\documentclass[10pt]{article}
\usepackage[margin=.8in,left=.8in]{geometry}

\usepackage{amsthm}
\usepackage{amssymb}
\usepackage{amsmath}
\usepackage{mathtools}
\usepackage{adjustbox}

\usepackage[table]{xcolor}

\usepackage{stmaryrd}    
\usepackage{xfrac}

\usepackage{tensor}

\usepackage{enumitem}\setlist{nosep}

\usepackage{slashed}

\usepackage[safe]{tipa} 

\usepackage{hhline}

\usepackage{tikz}
\usetikzlibrary{
  backgrounds, 
  decorations.pathmorphing, 
  decorations.markings,
  decorations.text
}
\usetikzlibrary{cd}

\usepackage{mathrsfs} 
\DeclareMathAlphabet{\mathpzc}{OT1}{pzc}{m}{it} 

\usepackage{hyperref}
\hypersetup{
  colorlinks = true,
  linkcolor= darkblue, citecolor=darkblue            
}

\hypersetup{
        colorlinks=true,
        linkcolor=darkblue,
        filecolor=darkblue,      
        urlcolor=black,
        citecolor=darkblue,
        linktoc=page
    }

\usepackage{cleveref}

\definecolor{darkblue}{rgb}{0.05,0.25,0.65}
\definecolor{darkgreen}{RGB}{20,140,10}
\definecolor{lightgray}{rgb}{0.9,0.9,0.9}
\definecolor{darkorange}{RGB}{200,100,5}
\definecolor{darkyellow}{rgb}{.91,.91,0}
\definecolor{orangeii}{RGB}{200,100,5}
\definecolor{lightblue}{RGB}{243, 250, 255}

\newtheorem{theorem}{Theorem}[section]

\newtheorem{lemma}[theorem]{Lemma}

\newtheorem{proposition}[theorem]{Proposition}

\theoremstyle{definition}
\newtheorem{definition}[theorem]{Definition}
\newtheorem{notation}[theorem]{Notation}
\newtheorem{example}[theorem]{Example}

\newtheorem{remark}[theorem]{Remark}

\usepackage{accents}
\newlength{\dhatheight}
\newcommand{\doublehat}[1]{%
    \hspace{1.3pt}
    \settoheight{\dhatheight}{\ensuremath{\widehat{#1}}}%
    \addtolength{\dhatheight}{-0.23ex}%
    \widehat{\vphantom{\rule{1pt}{\dhatheight}}%
    \hspace{-1.3pt}
    \smash{\widehat{#1}}}}

\newcommand{\SecondFundamentalForm}{\mbox{\rm I\hspace{-1.25pt}I}}

\newcommand{\GaugePotential}[1]{\vec{#1}}

\newcommand{\CoverOf}[1]{\widetilde{#1}}

\newcommand{\boldpi}{\pi\hspace{-5.7pt}\pi}

\newcommand{\HilbertSpace}[1]{\mathcal{#1}}

\let\PLAINthebibliography\thebibliography
\renewcommand\thebibliography[1]{
  \PLAINthebibliography{#1}
  \setlength{\parskip}{0.5pt}
  \setlength{\itemsep}{0.5pt plus .3ex}
}

\newcommand{\proofstep}[1]{\scalebox{.85}{#1}}

\newcommand{\K}{\tilde H_3^2}

\newcommand{\shape}{
  \raisebox{1pt}{\rm\normalfont\textesh}
}

\newcommand{\ZTwo}{\mathbb{Z}_2}

\newcommand{\defneq}{\equiv}

\newcommand\bos[1]{\mathstrut\mkern2.5mu#1\mkern-14mu\raise1.7ex%
  \hbox{$\scriptstyle\rightsquigarrow$}}

\newcommand\bosonic[1]{\mathstrut\mkern2.5mu#1\mkern-14mu\raise1.7ex%
  \hbox{$\scriptstyle\rightsquigarrow$}}

\usepackage[cmtip,all]{xy}
\newcommand{\longsquiggly}{\xymatrix{{}\ar@{~>}[r]&{}}}

\usepackage{bbm} 

\usepackage[new]{old-arrows}   

\begin{document}

\setlength{\abovedisplayskip}{3pt}
\setlength{\belowdisplayskip}{3pt}
\setlength{\abovedisplayshortskip}{-3pt}
\setlength{\belowdisplayshortskip}{3pt}

\title{Flux Quantization on M5-branes}

\author{
  Grigorios Giotopoulos${}^{\ast}$,
  \;\;
  Hisham Sati${}^{\ast \dagger}$,
  \;\;
  Urs Schreiber${}^{\ast}$
}

\maketitle

\begin{abstract}
  We highlight the need for global completion of the field content in the M5-brane sigma-model analogous to Dirac's charge/flux quantization, and we point out that the superspace Bianchi identities on the worldvolume and on its ambient supergravity background constrain the M5's flux-quantization law to be in a non-abelian cohomology theory rationally equivalent to a twisted form of co-Homotopy. In order to clearly bring out this subtle point we give a streamlined re-derivation of the worldvolume 3-flux via M5 ``super-embeddings''. Finally, assuming the flux-quantization law to actually be in co-Homotopy (``Hypothesis H'') we show how this implies Skyrmion-like solitons on general M5-worldvolumes and (abelian) anyonic solitons on the boundaries of ``open M5-branes'' in heterotic M-theory.    
\end{abstract}

\vspace{1cm}

\begin{center}
\begin{minipage}{10cm}
  \tableofcontents
\end{minipage}
\end{center}

\medskip

\vfill

\hrule
\vspace{5pt}

{
\footnotesize
\noindent
\def\arraystretch{1}
\tabcolsep=0pt
\begin{tabular}{ll}
${}^*$\,
&
Mathematics, Division of Science; and
\\
&
Center for Quantum and Topological Systems,
\\
&
NYUAD Research Institute,
\\
&
New York University Abu Dhabi, UAE.  
\end{tabular}
\hfill
\adjustbox{raise=-15pt}{
\href{https://ncatlab.org/nlab/show/Center+for+Quantum+and+Topological+Systems}{
\includegraphics[width=3cm]{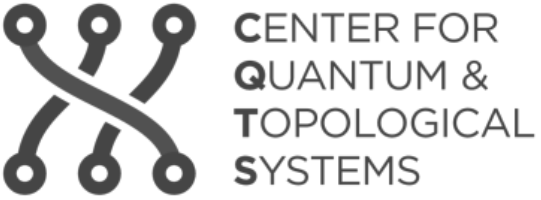}}
}

\vspace{1mm} 
\noindent ${}^\dagger$The Courant Institute for Mathematical Sciences, NYU, NY.

\vspace{.2cm}

\noindent
The authors acknowledge the support by {\it Tamkeen} under the 
{\it NYU Abu Dhabi Research Institute grant} {\tt CG008}.
}

\newpage

\section{Introduction and Overview}
\label{IntroductionAndOverview}

\noindent
{\bf The quest for non-perturbative strongly-coupled quantum theory.}
The key contemporary open question in the foundations of theoretical physics --- which traditionally relied on perturbation theory and mean field theory for weakly-coupled systems --- remains (cf. \cite{BakulevShirkov10}\cite[\S 4.1]{HollandsWald15}) the general analytic understanding of strongly-{\it coupled} quantum systems (such as confined chromodynamics, c.f. \cite{BrambillaEtAl14}). With it (cf. \cite{AFFK15}\cite{FGSS20}) comes the closely related issue of strongly-{\it interacting} and -{\it correlated} (long-range entangled) quantum systems (such as topologically ordered quantum materials \cite[\S 6.3]{ZCZW19}\cite{SS23-ToplOrder}, envisioned to provide future hardware for robust topological quantum computers, e.g. \cite{FKLW03}\cite{MySS24}).

\medskip

\noindent
{\bf The expected M5-brane model for strongly-coupled quantum systems.}
Meanwhile, the refinement of quantum field theory by string theory (e.g. \cite{BLT13}) --- for whatever else its motivations and aims have been at any time --- has had the remarkable effect of leading to a glimpse of a general non-perturbative formulation of fundamental quantum physics, famously going by the working title ``M-theory'' (cf. \cite{Duff96}\cite{Duff99}). Notably for holographic arguments about strongly-coupled quantum systems 
\cite{ZLSS15}\cite[\S IV]{RhoZahed16}\cite{HartnollLucasSachdev18} 
to become realistic by {\it not} passing to their large-$N$ limit requires including M-theoretic corrections in the dual gravitational theory \cite[p. 60]{AGMOY00}\cite[Fig. 4]{SS23-ToplOrder}.

\smallskip

Within M-theory, the object most directly expressing (strongly coupled) quantum field theory in realistic numbers of spacetime dimensions is the ``Fivebrane'' or ``M5-brane'' (cf. \cite[\S 3]{Duff99} and below \eqref{TwistedBianchiIdentityInIntroduction}). Here, apart from the widely appreciated implications on understanding weak/strong coupling S-dualities \cite{Witten04}\cite[\S 3-4]{Witten09} (see also 
\cite{MathaiSati}), it is worth highlighting that:

\vspace{1mm} 
\begin{itemize}[
  leftmargin=.8cm,
  topsep=1pt,
  itemsep=2pt
]
\item[\bf (i)]
the quantum fields on suitable M5-brane configurations engineer at least an approximation to confined quantum chromodynamics (the Witten-Sakai-Sugimoto model \cite[\S 4]{Witten98} for ``holographic QCD'' \cite{Rebhan15}, including the realization of hadron bound states \cite{Sugimoto16}\cite{ILP18}) via skyrmions -- cf. \S\ref{ResultingSolitons} below;
\item[\bf (ii)]
 co-dimension=2 defects inside M5-branes have been argued \cite{CGK20}\cite{SS23-DefectBranes}\cite{SS24-AbAnyons} to carry anyonic quantum observables (cf. \eqref{GelfandRaikovTheorem} below) 
 as expected in topologically ordered quantum materials \cite{SS23-ToplOrder}.
\end{itemize}

\smallskip

This means that a deep understanding of M5-branes may plausibly go a long way towards resolving outstanding problems in contemporary strongly-coupled/correlated quantum systems.

\smallskip

At the same time, along with the ambient M-theory also a complete theory of the M-branes has been missing (after all, the term ``M-theory'' is the ``non-committal'' abbreviation \cite[p. 2]{HoravaWitten96} of ``Membrane Theory'' due to doubts about the nature of the M2-brane theory):
The above results are all deduced just from expected subsectors of a would-be M5-brane theory. This lack of a complete ``M5-brane model'' is commonplace for the quantum theory of coincident M5-worldvolumes in its decoupling limit expected to be given by a largely elusive $D=6$ $\mathcal{N}=(2,0)$ superconformal field theory (e.g.  \cite{HeckmanRudelius18}\cite[\S 3]{Lambert19}\cite{SaemannSchmidt18}). However, the problem reaches deeper:

\medskip

\noindent
{\bf The open problem of flux quantization on the M5-brane.}
Our starting point here is to highlight that an open question had already remained both at a more fundamental level, namely
concerning the full definition of the classical (or pre-quantum\footnote{Flux quantization applies to the classical fields of a higher gauge theory, but it thereby determines the ``pre-quantum line bundle'' for sigma-model branes which are charged under (i.e., sense the Lorentz force exerted by) these gauge fields.}) on-shell field content on the single M5-brane worldvolume, as well as on a more universal level, namely concerning the classification of ``fractional'' (i.e., torsion-charged, see \S\ref{ResultingSolitons}) solitons on the M5-worldvolume. This is the open problem of {\it flux quantization} (pointers and survey in \cite{SS24-Flux}) on the M5-brane worldvolume.

\smallskip

The root cause is that alongside the ambient 11d supergravity field (for pointers and review see \cite{MiemiecSchnakenburg06}\cite{GSS24-SuGra}) with its 4-form and dual 7-form flux densities $G_4$, $G_7$ on 11-dimensional target spacetime $X$, also the field theory on the worldvolume $\Sigma$ of an M5-brane
$$
  \begin{tikzcd}[
    column sep=76pt
  ]
    \scalebox{.7}{
      \color{orangeii}
      \bf
      \def\arraystretch{.9}
      \begin{tabular}{c}
        M5-brane
        \\
        worldvolume
      \end{tabular}
    }
    \Sigma
    \ar[
      r, 
      "{ \phi }",
      "{
        \scalebox{.6}{
          \color{darkgreen}
          \bf
          ``embedding'' field
        }
      }"{swap}
    ]
    &
    X
    \scalebox{.7}{
      \color{darkblue}
      \bf
      \def\arraystretch{.9}
      \begin{tabular}{c}
        11-dimensional
        \\
        spacetime
      \end{tabular}
    }
  \end{tikzcd}
$$
famously carries a 3-form flux density $H_3$, subject to the following {\it Bianchi identity} (\cite[(36)]{HoweSezgin97b}\cite[(5.75)]{Sorokin00}, we give a streamlined account below in \S\ref{SuperEmbeddingConstruction})
\begin{equation}
  \label{TwistedBianchiIdentityInIntroduction}
  \mathrm{d}
  \,
  H_3
  \;=\;
  \phi^\ast G_4
  \,.
\end{equation}
So far this is very well known. However, most if not all authors now consider the corresponding higher gauge field configuration to be exhibited by a {\it globally defined} gauge potential 2-form $B_2$ on $\Sigma$ (e.g. \cite[(27)]{APPS97}\cite[(37)]{HoweSezgin97b}\cite[(4)]{BLNPST97}\cite[(5.3)]{CKvP98}\cite[(27)]{SezginSundell98}\cite[(3.21)]{CKKTvP98}\cite[(5.74)]{Sorokin00}\cite[(4.44)]{Bandos11}\cite[(5.78)]{BandosSorokin23}). 
Furthermore,  together with gauge potential 3- and 6-forms $C_3$, $C_6$ on $X$  this is taken to define the 3-form flux density $H_3$ by the following formulas, which we (re-)derive these below in \S\ref{GaugePotentialsOnM5Branes} from systematic flux quantization, cf. Rem. \ref{ReproducingTraditionalLocalGaugePotentials}:
\begin{equation}
  \label{TraditionalFormulaForBFieldGaugePotential}
  H_3
  \;=\;
  \mathrm{d}\, B_2
  \;+\;
  \phi^\ast 
  C_3
  \,,
  \hspace{.7cm}
  G_4
  \,=\, 
  \mathrm{d}\, C_3 
  \,,
  \hspace{.7cm}
  G_7
  \,=\, 
  \mathrm{d}\, C_6 
  +
  \tfrac{1}{2}
  C_3 \, G_4
  \,.
\end{equation}

\smallskip 
\noindent The problem here is that these equations \eqref{TraditionalFormulaForBFieldGaugePotential} in general only make sense locally, namely on (open covers here to be incarnated as) surjective  submersions $\widehat X$ 
\eqref{AnOpenCover}
onto $X$ with compatible surjective submersion $\CoverOf{\Sigma}$ \eqref{MapLiftedToOpenCovers} onto $\Sigma$:
$$
  \begin{tikzcd}[row sep=small, 
    column sep=86pt
  ]
  \mathllap{
    \scalebox{.7}{
      \color{orangeii}
      \bf
      \def\arraystretch{.9}
      \begin{tabular}{c}
        submersion
        \\
        onto...
      \end{tabular}
    }
  }
    \CoverOf{\Sigma}
    \ar[
      r,
      "{
        \CoverOf{\phi}
      }"
    ]
    \ar[
      d,
      ->>,
      shorten=-3pt,
      "{\;
        p_{{}_\Sigma}
      }"
    ]
    &
    \CoverOf{X}
    \ar[
      d,
      ->>,
      shorten=-3pt,
      "{\;
        p_{{}_X}
      }"
    ]
  \mathrlap{
    \scalebox{.7}{
      \color{darkblue}
      \bf
      \def\arraystretch{.9}
      \begin{tabular}{c}
        submersion
        \\
        onto...
      \end{tabular}
    }
  }
    \\
    \mathllap{
    \scalebox{.7}{
      \color{orangeii}
      \bf
      \def\arraystretch{.9}
      \begin{tabular}{c}
        ...M5-brane
        \\
        worldvolume
      \end{tabular}
    }
    }
    \Sigma
    \ar[
      r, 
      "{ \phi }",
      "{
        \scalebox{.6}{
          \color{darkgreen}
          \bf
          ``embedding'' field
        }
      }"{swap}
    ]
    &
    X
    \mathrlap{
    \scalebox{.7}{
      \color{darkblue}
      \bf
      \def\arraystretch{.9}
      \begin{tabular}{c}
        ...11-dimensional
        \\
        spacetime
      \end{tabular}
    }
    }
  \end{tikzcd}
$$
But this means that the local gauge-potential form fields
$$
  \left(
  \def\arraycolsep{2pt}
  \def\arraystretch{1.4}
  \begin{array}{lcl}
    C_3 
    &\in&
    \Omega^3_{\mathrm{dR}}
    \big(
      \CoverOf{X}
    \big)
    \\
    C_6 
    &\in&
    \Omega^6_{\mathrm{dR}}
    \big(
      \CoverOf{X}
    \big)
    \\
    B_2 
    &\in&
    \Omega^2_{\mathrm{dR}}
    \big(
      \CoverOf{\Sigma}
    \big)
  \end{array}
  \middle\vert
  \def\arraystretch{1.4}
  \begin{array}{l}
    \mathrm{d}\, C_3 \,=\, G_4
    \\
    \mathrm{d}\, C_6 \,=\, G_7 - \tfrac{1}{2} C_3\, G_4
    \\
    \mathrm{d}\, B_2 \,=\, 
    H_3 - \phi^\ast C_3
  \end{array}
 \!\! \right)
$$
by themselves do {\it not} represent a {\it complete} higher gauge field configuration: Further field content
is needed to {\it glue} ($C_3$ and) $B_2$ on all higher intersections \eqref{FiberProductOfCoveringSubmersion}:
\vspace{-2mm} 
$$
  \begin{tikzcd}[
    column sep=40pt
  ]
  \CoverOf{\Sigma} \times_{{}_X} \cdots \times_{{}_X} \CoverOf{\Sigma}
  \ar[
    rr,
    "{
      \CoverOf{\phi}
      \,
      \times_{{}_X}
      \cdots
      \times_{{}_X}
      \CoverOf{\phi}
    }"
  ]
  &&
  \CoverOf{X} \times_{{}_X} \cdots \times_{{}_X} \CoverOf{X}
  \end{tikzcd}
$$
in a coherent manner that we recall in a moment. 

\smallskip 
The admissible choices of this remaining global field content are the {\it flux quantization laws} \cite{SS24-Flux} --- these determine which non-rational 
torsion brane charges (fractional branes, \S\ref{ResultingSolitons}) may appear as sources of the field fluxes; and it is only with this extra choice that
even the classical on-shell field content of the higher gauge theory is actually complete (cf. \cite{SS24-Phase}).

\medskip

\noindent
{\bf The analog problem of Dirac charge quantization.} To appreciate the issue, we quickly compare it to the analogous traditional situation 
in vacuum electromagnetism (cf. \cite[\S 2.1]{SS24-Flux}\cite[\S 3.1]{SS24-Phase}): Here the flux density is a 2-form $F_2$ on spacetime $X$ 
(the Faraday tensor), and the local gauge potential is a 1-form $A_1$ on a submersion $\CoverOf{X}$ onto spacetime. There are in fact infinitely 
many admissible flux quantization laws even in this basic situation (cf. \cite[\S 2]{SS23-Obs}), but the standard one going back to \cite{Dirac31}
says (as maybe first made explicit by \cite{Alvarez85a}\cite{Alvarez85b}) that alongside the local gauge potential
\vspace{-1mm} 
$$
  \big(
  A_1 \,\in\, \Omega^1_{\mathrm{dR}}\big(\CoverOf{X}\big)
  \;\big\vert\;
  \mathrm{d}
  \,A_1
  \,=\,
  F_2
  \big)
$$

\vspace{1mm} 
\noindent an electromagnetic field configuration consists, in addition, of a {\it transition function}
\begin{equation}
  \label{AbelianTransitionFunction}
  \left(
    A_0
    \;\in\;
  \Omega^0_{\mathrm{dR}}\big(
    \CoverOf{X}
    \!\times_X\!
    \CoverOf{X}
  \big)
  \;\middle\vert\;
  \def\arraystretch{1.4}
  \begin{array}{l}
  \mathrm{d} 
  \,
  A_0
  \;=\;
  \mathrm{pr}_2^\ast
  A_1
  -
  \mathrm{pr}_1^\ast A_1
  \,,
  \\
  \mathrm{pr}_{12}^\ast A_0
  +
  \mathrm{pr}_{23}^\ast A_0
  \;=\;
  \mathrm{pr}_{13}^\ast A_0
  \;\;
  \mathrm{mod}
  \;\;
  C^0\big(
    \CoverOf{X}
    \!\times_{{}_X}\!
    \CoverOf{X}
    \!\times_{{}_X}\!
    \CoverOf{X}
    ;\,
    \mathbb{Z}
  \big)
  \end{array}
 \!\!\! \right).
\end{equation}
This data is of course (the {\v C}ech-Deligne cocycle presentation, for exposition see \cite[\S 2]{FSS13CupCS}\cite{FSS15-Stacky}) equivalent to a connection on a $\mathrm{U}(1)$-principal bundle over $X$.
The point here is that this extra field content \eqref{AbelianTransitionFunction} has experimentally measurable consequences, as it implies that
\begin{itemize}[
  leftmargin=.8cm,
  topsep=1pt,
  itemsep=2pt
]
  \item[\bf (i)]
    there is a smallest unit of electric charge\footnote{
      Electric charge quantization is traditionally said to be conditioned on the assumption that magnetic monopoles exist (which remains hypothetical) -- but we highlight that the actual logic is slightly different: Magnetic monopoles are (just as the experimentally observed Abrikosov vortices) compatible with and hence indicative of the assumption of EM-flux quantization in ordinary differential ({\v C}ech-Deligne) cohomology, and {\it this assumption} generally implies that the exponentiated Lorentz-(inter)action is the holonomy $\mbox{``}\exp\big( \mathrm{e} \int_{S^1} A
      \big)\mbox{''} \,=\,\mathrm{hol} : \mathrm{Map}(S^1,X) \xrightarrow{\;} \mathrm{U}(1)$
      (for electrically unit charged particles)
      of the flux-quantized background EM-field $\GaugePotential{A}$ regarded as a  $\mathrm{U}(1)$-principal connection, or at most an integer power 
      $\mbox{``}\exp\big( n\mathrm{e} \int_{S^1} \widehat A\big)\mbox{''} \,=\,\mathrm{hol}^n$
      (for particles carrying $n$ units of fundamental electric charge).
    } 
    (as observed)
    and 
  \item[\bf (ii)] Dirac-monopoles (hypothetical) but also Abrikosov-vortices (experimentally observed) carry integer multiples of a unit magnetic charge
  \item[\bf (iii)] should the cosmos have the non-trivial topology such as that of a lens space with its torsion cohomology group in degree 2
  (for which, incidentally, there is mild observational support \cite{AL12}), then it may carry torsion magnetic charge 
  (cf. \cite[p. 28]{FMS07})
  of ``fractional monopoles'' (cf. \S\ref{ResultingSolitons}).
\end{itemize}

 \smallskip 
These simple examples already show that (determination and understanding of) flux quantization is (or should be) a crucial step in understanding the full phenomenological scope of a higher gauge theory, in particular with respect to non-perturbative phenomena. 

\newpage 

\smallskip 
Clearly then, it is important to understand flux quantization also for the M5-brane worldvolume theory. We may approach this problem as follows.

\smallskip

\noindent
{\bf The general rule for flux quantization on phase space.} \label{GeneralRuleOfFluxQuantization} There is a 
beautiful general definition and rule for admissible flux quantization laws (\cite{FSS23Char}\cite{SS24-Phase}, review in \cite{SS24-Flux}) which applies naturally to all higher gauge theories of higher Maxwell-type and recovers existing proposals for flux quantization, showing how to systematically extend these to previously neglected theories:

\smallskip

If spacetime $X^{d+1}$ is globally hyperbolic with spatial Cauchy surface $X^d$
\begin{equation}
  \label{GloballyHyperbolic Spacetimes}
  X^{1+d}
  \;\simeq\;
  \mathbb{R}^1
    \times 
  X^d 
  \,,
\end{equation}
so that the flux densities restricted to $X^d$ constitute Cauchy data for the solution space of on-shell flux densities on $X^{d+1}$,
then the simple idea is that flux quantization means to 
\begin{itemize}[
  leftmargin=.8cm,
  topsep=1pt,
  itemsep=2pt
]
\item[\bf (i)]
accompany 
\begin{equation}
  \label{FluxDensitiesSatisfyingBianchiIdentities}
  \scalebox{1}{
    \bf
    flux densities 
  }
  \vec F
  \;\defneq\;
  \big(
    F^{(i)}
    \,\in\,
    \Omega^{\mathrm{deg}_i}_{\mathrm{dR}}
    (X^d)
  \big)_{i \in I}
  \;\;
  \scalebox{1}{
      satisfying 
      Bianchi ids.
  }
  \begin{array}{l}
    \mathrm{d}
    \,F^{(i)}
    \;=\;
    P^{(i)}\big(
      \vec F
    \big)
  \end{array}
\end{equation}
\item[\bf (ii)]
with
\begin{equation}
  \label{ChargesInCohomology}
  \scalebox{1}{
    {\bf
    charges
    }
    $\vec c\;\;$
    defining classes
    $
    \;\;
    \big[
      \vec c
      \,
    \big]
    \;\in\;
    H^1(X^d;\Omega \mathcal{A})
    \;\;
    $
   in a generalized cohomology theory $\mathcal{A}$
  }
\end{equation}
\item[\bf (iii)]
subject to an 
$$
  \scalebox{1}{
    {\bf
    identification 
    }
    $
    \;
    \GaugePotential{A} 
    :
    \vec{F} 
      \Rightarrow 
    \vec{c}
    \;
    $
    of 
    the flux densities
    with the charges 
    in de Rham cohomology,
  }
$$
\end{itemize}
where the key technical issue is to understand what it means for both the flux densities and the charges to be regarded as cocycles in a suitable form of de Rham cohomology where they can be compared.
This is accomplished by observing that (cf. \cite[\S 3]{SS24-Flux}\cite[\S 2.1]{GSS24-SuGra}):

\begin{itemize}[
  topsep=1pt,
  itemsep=2pt,
]
  \item[\bf (i)]
    any system of Bianchi identities \eqref{FluxDensitiesSatisfyingBianchiIdentities} is equivalently the closure condition on differential forms with coefficients in a {\it characteristic $L_\infty$-algebra} $\mathfrak{a}$
    \begin{equation}
      \label{LInfinityValuedFluxForms}
      \vec F
      \;\in\;
      \Omega^1_{\mathrm{dR}}
      (
        X^d
        ;\,
        \mathfrak{a}
      )_{\mathrm{clsd}}
    \end{equation}
    whose concordance classes
    constitute the corresponding $\mathfrak{a}$-valued de Rham cohomology
    \begin{equation}
      \label{NonabelianDeRhamCohomology}
      \big[\vec F\big]
      \;\in\;
      H^1_{\mathrm{dR}}(
        X^d
        ;\,
        \mathfrak{a}
      )
      \;\;
      =
      \;\;
      \Omega^1_{\mathrm{dR}}
      (
        X^d
        ;\,
        \mathfrak{a}
      )_{\mathrm{clsd}}
      \big/\mathrm{cncrd}
      \,,
    \end{equation}
  \item[\bf (ii)]
  any generalized cohomology theory \eqref{ChargesInCohomology}
  is characterized by its (pointed) classifying space $\mathcal{A}$ as
  $$
    H^1(
      X^d;
      \Omega\mathcal{A}
    )
    \;:=\;
    \pi_0
    \,
    \mathrm{Maps}\big(
      X^d
      ;\,
      \mathcal{A}
    \big)
  $$
  whose {\it character map} landing in de Rham cohomology with coefficients \eqref{NonabelianDeRhamCohomology}  in the {\it Whitehead $L_\infty$-algebra} $\mathfrak{l}\mathcal{A}$ is classified by the $\mathbb{R}$-rationalization map on $\mathcal{A}$:
  $$
    \begin{tikzcd}[
      row sep=3pt, column sep=large
    ]
      \mathcal{A}
      \ar[
        rr,
        "{
          \eta^{\scalebox{.8}{$\mathbb{R}$}}
        }"
      ]
      &&
      L^{\mathbb{R}}
      \mathcal{A}
      \\
      H^1\big(
        X^d
        ;\,
        \mathcal{A}
      \big)
      \ar[
        rr,
        "{
          H^1(
            X^d
            ;\,\eta^{\mathbb{R}}
          )
        }"
      ]
      &&
      H^1\big(
        X^d
        ;\,
        L^{\mathbb{R}}
        \mathcal{A}
      \big)      
      \ar[
        r,
        phantom,
        "{ \simeq }"
      ]
      &
      H^1_{\mathrm{dR}}\big(
        X^d
        ;\,
        \mathfrak{l}
        \mathcal{A}
      \big)\,,            
    \end{tikzcd}
  $$
  \item[\bf (iii)]
  the required identification is a homotopy in the deformation $\infty$-stack of closed $\mathfrak{a} \simeq \mathfrak{l}\mathcal{A}$-valued differential forms

  \vspace{-.6cm}
  \begin{equation}
    \label{GaugePotentialAsHomotopy}
    \begin{tikzcd}[row sep=25pt, column sep=huge]
      &[+20pt]
      &[-15pt]
      \mathcal{A}
      \ar[
        d,
        "{
          \eta^{\mathbb{R}}
        }"
      ]
      \ar[
        dd,
        rounded corners,
         to path={
             ([xshift=+00pt]\tikztostart.east)  
          -- ([xshift=+6pt]\tikztostart.east)  
          -- node{
            \scalebox{.7}{
              \colorbox{white}{
                $\mathbf{ch}$
              }
            }
          } ([xshift=+14pt]\tikztotarget.north) 
        }
      ]
      \\
      &&
      L^{\mathbb{R}} \mathcal{A}
      \ar[
        d,
        shorten >=-2pt,
        "{
          \sim
        }"{sloped}
      ]
      \\[-10pt]
      X^d
      \ar[
        r,
        dashed,
        "{ 
          \vec F
        }"{name=t, pos=.7},
        "{
          \scalebox{.7}{
            \color{darkgreen}
            \bf
            flux densities
          }
        }"{swap}
      ]
      \ar[
        uurr,
        dashed,
        bend left=20,
        "{
          \scalebox{.7}{
            \color{darkgreen}
            \bf
            charges
          }
        }"{sloped, pos=.35},
        "{
          \vec c
        }"{swap, name=s, pos=.3}
      ]
      &
      \Omega^1_{\mathrm{dR}}(-; \mathfrak{a})_{\mathrm{clsd}}
      \ar[
        r,
        shorten <=-3pt,
        "{
          \eta^{\scalebox{.6}{$\shape$}}
        }"
      ]
      &
      \shape
      \,
      \Omega^1_{\mathrm{dR}}(-; \mathfrak{a})_{\mathrm{clsd}}
      \ar[
        from=s,
        to=t,
        Rightarrow,
        dashed,
        "{
          \GaugePotential{A}
          \mathrlap{
            \scalebox{.7}{
              \color{darkorange}
              \bf
              gauge potentials
            }
          }
        }"{pos=-.1}
      ]
    \end{tikzcd}
  \end{equation}
\end{itemize}

\vspace{2pt}

\noindent
This general prescription notably subsumes (as reviewed in \cite{SS24-Flux} with extensive pointers to the literature):
\begin{itemize}[
  topsep=1pt,
  itemsep=2pt,
  leftmargin=.6cm
]
  \item traditional Dirac charge quantization of electromagnetism in (differential) ordinary integral 2-cohomology \cite[\S 3.1]{SS24-Phase}\cite[\S 2]{SS23-Obs},
  
  \item B-field flux quantization in (differential) ordinary integral 3-cohomology 
  \cite[Ex. 3.10]{SS24-Flux},
  
  \item 
    flux quantization of self-dual higher gauge fields in (differential) ordinary cohomology \cite[\S 3.2]{SS24-Phase},
    
  \item quantization of RR-fluxes in topological K-theory \cite[\S 3.3]{SS24-Phase}\cite[\S 4.1]{SS24-Flux} (cf. \cite{GS22}), 
  
  \item 
  a couple of proposed flux quantization laws for the C-field in 11d supergravity
  \cite[\S 3.4]{SS24-Phase}\cite[\S 4.2]{SS24-Flux},

  including the C-field model of \cite{DFM07} following \cite{Witten97a}\cite{Witten97b},

\item 
flux quantization for additional Green-Schwarz terms with (abelian) 2-flux densities \cite{Sati19}\cite{FSS22-Twistorial}\cite{SS20EquChar}.
\end{itemize}

\newpage 

\noindent
{\bf The issue of (self-)duality in flux quantization.}
What makes the above examples work is -- besides the assumption of globally hyperbolic spacetime \eqref{GloballyHyperbolic Spacetimes} --  that in all these cases the full equations of motion of the higher gauge theories are
\begin{itemize}
\item[\bf (i)]
the Bianchi identities \eqref{FluxDensitiesSatisfyingBianchiIdentities} \eqref{LInfinityValuedFluxForms}
$$
  \mathrm{d}
  \,
  F^{(i)}
  \;=\;
  P^{(i)}
  \big(
    \vec F
  \big)
$$
\item[\bf (ii)] and one more {\it linear} system of (self-)duality constraints
\begin{equation}
  \label{LinearSelfDualityConstraint}
  \star \, F^{(i)}
  \;=\;
  \mu^{(i)}\big(
    \vec F
  \big)
  \,,
\end{equation}
\end{itemize}
because it turns out \cite[Thm. 2.2]{SS24-Phase} that for flux densities on a Cauchy surface the linear duality constraint \eqref{LinearSelfDualityConstraint} is entirely absorbed into the isomorphism between the space of such Cauchy data and the solution space of flux densities over all of spacetime.

\smallskip 
However, each of the assumptions, that of globally hyperbolic spacetimes  \eqref{GloballyHyperbolic Spacetimes} and that of the linear self-duality constraint \eqref{LinearSelfDualityConstraint} is somewhat restrictive; in particular, the latter is actually violated for the 3-flux density on M5-branes (\cite{HoweSezginWest97}, cf. Rem. \ref{NonSelfDualityOf3FluxDensity} below).

\medskip

\noindent
{\bf Resolution by flux-quantization on super-spaces.}
Our key move now for solving the problem of flux quantization also on M5-branes is the observation that both of the above problems go away when considering the M5-brane as  immersed {\it in super-spacetime} (following \cite{HoweSezgin97b}\cite[\S 5.2]{Sorokin00}). The reason is that here the self-duality condition on the flux forms turns out to be all absorbed into the Bianchi identities on their super-field versions (re-derived as Prop. \ref{BianchiIdentityOnM5BraneInComponents})!
Since this result --- which in light of the problem of flux quantization is now revealed to be quite profound ---  has perhaps remained under-appreciated outside the original specialist literature, our main contribution in \S\ref{TheM5EquationsOfMotion} below is to give a streamlined and rigorous re-derivation, based on some more mathematically informed commentary on the proper definition of the underlying concept of ``super-embeddings'' \cite{HSW98}\cite{HoweRaetzelSezgin98} \cite{Sorokin00}\cite{BandosSorokin23} in \S\ref{SuperEmbeddings}.
 
\smallskip

This resolution of the problem of flux quantization on the M5 builds on  and extends the super-flux quantization of the background 11d supergravity fields which we established in \cite{GSS24-SuGra}:

\medskip

\noindent
{\bf Super-flux of 11d Supergravity.} Namely, the analogous miracle of on-shell 11d supergravity (going back to \cite{CF80}\cite{BrinkHowe80}\cite[\S III.8.5]{CDF91}), is that on super-spacetimes $X$ (``curved superspace'', namely super-manifolds of super-dimension $(1,10) \,\vert\, \mathbf{32}$ equipped with super-coframe fields $(E,\Psi)$ and super-torsion-free spin connection $\Omega$) the Bianchi identities 
\begin{equation}
  \label{SuperCFieldBianchiIdentityInIntroduction}
  \mathrm{d} \, G_4^s \;=\; 0
  \,,
  \hspace{.4cm}
  \mathrm{d}
  \,
  G_7^s \;=\;
  \tfrac{1}{2}G_4^s \, G_4^s
\end{equation}
on the duality-symmetric {\it super-}flux densities 
\begin{equation}
  \label{SuperFluxDensitiesOf11dSugra}
  \def\arraystretch{1.1}
  \begin{array}{l}
  G_4^s
  \;\defneq\;
  \underbrace{
  (G_4)_{a_1 \cdots a_4}
  E^{a_1} \cdots E^{a_4}
  }_{G_4}
  \,+\,
  \underbrace{
  \tfrac{1}{2}
  \big(
    \overline{\Psi}
    \,
    \Gamma_{a_1 a_2}
    \,
    \Psi
  \big)
  E^{a_1}
  \,
  E^{a_2}
  }_{ G_4^0 }
  \\
  G_7^s
  \;\defneq\;
  \underbrace{
  (G_7)_{a_1 \cdots a_7}
  E^{a_1} \cdots E^{a_7}
  }_{G_7}
  \,+\,
  \underbrace{
  \tfrac{1}{5!}
  \big(
    \overline{\Psi}
    \,
    \Gamma_{a_1 \cdots a_5}
    \,
    \Psi
  \big)
  E^{a_1}
  \cdots
  E^{a_5}
  }_{ G_7^0 }
  \end{array}
\end{equation}
are already equivalent (as brought out in this form in \cite[Thm. 3.1]{GSS24-SuGra}) to the equations of motion of 11d SuGra; in particular they {\it imply} the Hodge duality relation
\begin{equation}
 \label{CFieldSelfDualityInIntroduction}
  G_7 \;=\; \star \, G_4
  \;\;\;
  \in
  \;
  \Omega^7_{\mathrm{dR}}
  \big(
    \bosonic{X}
  \big)
\end{equation}
on the underlying bosonic spacetime manifold $\bosonic{X} \xhookrightarrow{\;} X$.

\smallskip 
This is remarkable, because it means (\cite[Claim 1.1]{GSS24-SuGra}) that the flux quantization \eqref{GaugePotentialAsHomotopy} applied to the super-flux densities \eqref{SuperFluxDensitiesOf11dSugra} on $(1,10\vert \mathbf{32})$-dimensional super-spacetime may be regarded the full globally completed field content of on-shell 11d supergravity, reflecting also all torsion/fractional charges of M-branes (according to that choice of flux quantization).

\smallskip

Here we are concerned with further refining this result to include the M5-brane worldvolume theory in such a background.

\medskip

\noindent
{\bf Super-flux on the M5-brane.}
Namely, a similar miracle occurs (this goes back to \cite{HoweSezgin97b}\cite[\S 5.2]{Sorokin00}, re-derived in \S\ref{TheM5EquationsOfMotion} below) for M5-brane worldvolumes $\Sigma \xhookrightarrow{\phi} X$ understood as ``super-embeddings'' (\cite{HSW98}\cite{HoweRaetzelSezgin98} \footnote{
  Early discussion of the idea of brane sigma-models where both the worldvolume as well as the target spacetime are treated as supermanifolds is due to  \cite{GatesHishino86}\cite{BMG86}, under the name ``supersymmetry squared''.
}, or more precisely {\it $\sfrac{1}{2}$BPS immersions}, cf. Defs. \ref{BPSImmersion}, \ref{M5SuperImmersion} below) of  $(1,5 \,\vert\, 2\cdot \mathbf{8})$-dimensional super-worldvolumes (with induced super-coframe fields $(e,\psi)$ and spin-connection $\omega$) into $(1,10\vert \mathbf{32})$-dimensional super-spacetimes as above.
Here the worldvolume Bianchi identity \eqref{TwistedBianchiIdentityInIntroduction}
\vspace{1mm} 
\begin{equation}
  \label{SuperBFieldBianchiIdentityInIntroduction}
  \mathrm{d}
  \,
  H_3^s
  \;\;
  =
  \;\;
  \phi^\ast
  G^s_4
\end{equation}
is now imposed on the super-flux densities \eqref{SuperFluxDensitiesOf11dSugra}. Furthermore, the expression 
\begin{equation}
  \label{WorldvolumeSuperFluxDensity}
  H^s_3 
  \;\;
  \defneq
  \;\;
  \underbrace{
    (H_3)_{a_1 a_2 a_3}
    \,
    e^{a_1}
    e^{a_2}
    e^{a_3}
  }_{ H_3 }
  \;+\;
  \underbrace{
    0
  }_{
    H_3^0
  }
\end{equation}
already implies (re-derived as Prop. \ref{BianchiIdentityOnM5BraneInComponents} below)
the subtle non-linear Hodge self-duality property of $H_3$ (cf. Rem. \ref{NonSelfDualityOf3FluxDensity}), namely that it is expressed as a rational function of a super-3-form 
which is actually self-dual:
\begin{equation}
  \label{The3FluxSolutionInIntroduction}
  \Rightarrow
  \;\;
  (H_3)_{a b c}
  \;=\;
  \frac
    {
      -4
    }
    {
    \mathclap{\phantom{
      \vert^{\vert}
    }}
    1 
    -
    \sfrac{2}{3}
    \,
    \mathrm{tr}(\K \cdot \K)
  }
  \big(
    \delta
      ^{ a  }
      _{ a' }
    +
    2
    \,
    (\K)
      ^{ a' }
      _{ a  }
  \big)
  (\tilde H_3)_{a' b c}
  \,,
  \;\;\;
  \mbox{for}
 \;  \left\{\!\!\!
  \def\arraystretch{1.4}
  \begin{array}{l}
    (\tilde H_3)
    \;\;
    \defneq
    \;\;
    \tfrac{1}{3!}
    (\tilde H_3)_{a_1 a_2 a_3}
    e^{a_1}
    e^{a_2}
    e^{a_3}
    \\
    (\tilde H_3)_{a_1 a_2 a_3}
    \;=\;
    \tfrac{1}{3!}
    \epsilon_{
      a_1 a_2 a_3
      \,
      b_1 b_2 b_3
    }
    (\tilde H_3)^{
      b_1 b_2 b_3
    }.
  \end{array}
  \right.
\end{equation}

\smallskip

With this result in hand, the admissible flux-quantization laws for completed field content on M5-branes now follows from general considerations:

\medskip

\noindent
{\bf Characteristic flux $L_\infty$-algebra on the M5.}
The immediate consequence is that the on-shell flux densities on the M5 super-worldvolume are entirely characterized by the Bianchi identities \eqref{SuperCFieldBianchiIdentityInIntroduction} on the super-fluxes
and \eqref{SuperBFieldBianchiIdentityInIntroduction}; and the first step towards flux-quantizing them is hence to identify their characteristic $L_\infty$-algebra \eqref{LInfinityValuedFluxForms}.

\vspace{1mm} 
\begin{itemize}
\item For the 11d Sugra C-field Bianchi identity \eqref{SuperCFieldBianchiIdentityInIntroduction}, this is (\cite[\S 4]{Sati10}\cite[(24)]{SS24-Flux}\cite[Ex. 2.29]{GSS24-SuGra}) the ``M-theory gauge algebra'' \cite[(2.6)]{CJLP98}
\begin{equation}
  \label{WhiteheaOf4Sphere}
  \mathfrak{l}S^4
  \;\simeq\;
    \mathbb{R}
    \left\langle
      \def\arraystretch{.9}
      \def\arraycolsep{0pt}
      \begin{array}{c}
        v_3
        \\
        v_6
      \end{array}
    \right\rangle
    \Big/
    \big(
      \def\arraystretch{.9}
      \def\arraycolsep{0pt}
      \begin{array}{c}
        [v_3,v_3] 
        \,=\, 
        v_6
      \end{array}
    \big)
    \hspace{.7cm}
    \Leftrightarrow
    \hspace{.7cm}
    \mathrm{CE}\big(
      \mathfrak{l}S^4
    \big)
    \;\;
    \simeq
    \;\;
    \mathbb{R}
    \left[
      \def\arraystretch{.9}
      \def\arraycolsep{0pt}
      \begin{array}{c}
        g_4
        \\
        g_7
      \end{array}
    \right]
    \Big/
    \left(
      \def\arraystretch{.9}
      \def\arraycolsep{0pt}
      \begin{array}{l}
        \mathrm{d}\, g_4 \;=\; 0
        \\
        \mathrm{d}\, g_7 \;=\; 
        \tfrac{1}{2} g_4 g_4
      \end{array}
    \right)
    ,
\end{equation}
\item while with the M5's B-field Bianchi identity
\eqref{SuperBFieldBianchiIdentityInIntroduction} {\it adjoined}, in addition a non-trivial unary bracket $[-]$ appears:
\begin{equation}
  \label{WhiteheadOfHHopfFib}
  \mathfrak{l}_{_{S^4}}S^7
  \;\simeq\;
    \mathbb{R}
    \left\langle
      \def\arraystretch{.9}
      \def\arraycolsep{0pt}
      \begin{array}{c}
        v_3
        \\
        v_6
        \\
        \color{purple}
        v_2
      \end{array}
    \right\rangle
    \Big/
    \left(
      \def\arraystretch{1.1}
      \def\arraycolsep{1pt}
      \begin{array}{ccc}
        [v_3,v_3] 
        &=&
        v_6
        \\
        \color{purple}
        {[v_3]} 
        &\color{purple}=&
        \color{purple}
        v_2
      \end{array}
    \right)
    \hspace{.7cm}
    \Leftrightarrow
    \hspace{.7cm}
    \mathrm{CE}\big(
      \mathfrak{l}_{{}_{S^4}}S^7
    \big)
    \;\;
    \simeq
    \;\;
    \mathbb{R}
    \left[
      \def\arraystretch{1}
      \def\arraycolsep{1pt}
      \begin{array}{c}
        g_4
        \\
        g_7
        \\
        \color{purple}
        h_3
      \end{array}
    \right]
    \Big/
    \left(
      \def\arraystretch{1}
      \def\arraycolsep{0pt}
      \begin{array}{l}
        \mathrm{d}\, g_4 \;=\; 0
        \\
        \mathrm{d}\, g_7 \;=\; 
        \tfrac{1}{2} g_4 g_4
        \\
        \color{purple}
        \mathrm{d}\, h_3 \;=\; g_4
      \end{array}
    \right)
    .
\end{equation}
\end{itemize}
Here the notation on the left indicates that these $L_\infty$-algebras happen to coincide (by \cite[Prop. 3.20]{FSS20-H}, cf. \cite[(38)]{FSSHopf}) with the (relative) $\mathbb{R}$-Whitehead $L_\infty$-algebras (cf. \cite[Prop. 5.16]{FSS23Char}\cite{SV2}) of (the homotopy type of) the 4-sphere and of the quaternionic Hopf fibration $S^7 \xrightarrow{ h_{\mathbb{H}} } S^4$, respectively, corresponding to the $L_\infty$-fibration
$$
  \begin{tikzcd}
    \mathfrak{l}_{{}_{S^4}}
    S^7
    \ar[
      d,
      "{
        \mathfrak{l}
        h_{\mathbb{H}}
      }"
    ]
    &
    v_3
    \ar[
      d,
      shorten=2pt,
      |->
    ]
    &[-15pt]
    v_6
    \ar[
      d,
      shorten=2pt,
      |->
    ]
    &[-15pt]
    v_2
    \ar[
      d,
      shorten=2pt,
      |->
    ]
    \\
    \mathfrak{l}S^4
    &
    v_3
    &
    v_6
    &
    0
    \,.
  \end{tikzcd}
$$

This means that the on-shell flux density content
on the M5-brane is concisely embodied by commutative diagrams of smooth super-spaces (as explained in  \cite[\S 2.1]{GSS24-SuGra}, for background see \cite{GS23}\cite{GSS24-SuperSet}) of the following form:
\begin{equation}
  \label{DiagrammaticEOMs}
  \begin{tikzcd}[
    column sep=50pt,
    row sep=35pt
  ]
    \mathllap{
      \scalebox{.7}{
        \color{darkblue}
        \bf
        Super-worldvolume
      }
      \;\;
    }
    \Sigma^{1,5\vert 2 \cdot \mathbf{8}}
    \ar[
      d,
      "{
        \phi
      }",
      "{
        \scalebox{.7}{
          \color{darkgreen}
          \bf
          \def\arraystretch{.9}
          \begin{tabular}{c}
            M5 super-immersion
            \\
            (Def. \ref{M5SuperImmersion})
          \end{tabular}
        }
      }"{swap, xshift=-4pt}
    ]
    \ar[
      rr,
      dashed,
      "{
        \big(
          \phi^\ast
          G_4^s
          ,\;
          \phi^\ast
          G_7^s
          ,\;
          {
            \color{purple}
            H_3^s
          }
        \big)
      }",
      "{
        \scalebox{.7}{
          \color{purple}
          \bf
          \def\arraystretch{.9}
          \begin{tabular}{c}
            worldvolume
            \\
            B-field flux
          \end{tabular}
        }
      }"{swap}
    ]
    &&
    \Omega^1_{\mathrm{dR}}
    \big(
      -
      ;\,
      \mathfrak{l}_{S^4}
      S^7
    \big)
    _{\mathrm{clsd}}
    \ar[
      d,
      ->>,
      "{
        (
        \mathfrak{l}
        h_{\mathbb{H}}
        )_\ast
      }"
    ]
    \\
    \mathllap{
      \scalebox{.7}{
        \color{darkblue}
        \bf
        11d super-spacetime
      }
      \;\;
    }
    X^{1,10\vert \mathbf{32}}
    \ar[
      rr,
      "{
        (
          G_4^s
          ,\;
          G_7^s
        )
      }",
      "{
        \color{darkgreen}
        \bf
        \scalebox{.7}{
          \color{darkgreen}
          \bf
          \def\arraystretch{.9}
          \begin{tabular}{c}
            background 
            \\
            C-field flux
          \end{tabular}
        }
      }"{swap}
    ]
    &&
    \Omega^1_{\mathrm{dR}}
    \big(
      -
      ;\,
      \mathfrak{l}S^4
    \big)_{\mathrm{clsd}}
  \end{tikzcd}
  \hspace{.6cm}
  \Leftrightarrow
  \hspace{.6cm}
  \left\{\!\!\!
  \def\arraystretch{1.3}
  \begin{array}{l}
  \mathrm{d}\, H_3^s
  \;=\;
  \phi^\ast G_4^s
  \\
  \\
  \mathrm{d}\, G_4^s \,=\, 0
  \\
  \mathrm{d}\, G_7^s \,=\, 
  \tfrac{1}{2} G_4^s \, G_4^s.
  \end{array}
  \right.
\end{equation}
Here the bottom map reflects a solution to 11d supergravity (by \cite[Ex. 2.30, Thm. 3.1]{GSS24-SuGra}, following \cite{CF80}\cite{BrinkHowe80}), the left map reflects a solution to the M5-brane's equations of motion (discussed in \S\ref{SuperEmbeddingConstruction}, following \cite{HoweSezgin97b}\cite[\S 5.2]{Sorokin00}) in this background, and the unique top map making the diagram commute extracts the corresponding worldvolume flux density (re-derived in Prop. \ref{BianchiIdentityOnM5BraneInComponents} below).

\medskip

\noindent
{\bf Super-flux quantization on M5-branes.}
The upshot of casting the equations of motion
of the M5-brane in the diagrammatic form \eqref{DiagrammaticEOMs} is that it reveals the admissible flux quantization laws 
(cf. \cite[\S 3.2]{SS24-Flux})
to be classified by fibrations 
$p : \mathcal{A} \xrightarrow{\;} \mathcal{B}$ (of connected nilpotent spaces of finite rational homotopy type $\mathfrak{l}p$, cf. \cite[Def. 5.1]{FSS23Char}) whose rational homotopy type
(the target of the twisted character map, cf. \cite[\S V]{FSS23Char})
is that of the quaternionic Hopf fibration $h_{\mathbb{H}}$, 
\begin{equation}
  \label{WorldvolumeCharges}
  \begin{tikzcd}[
    column sep=45pt,
    row sep=10pt
  ]
    \Sigma^{1,5\vert 2 \cdot \mathbf{8}}
    \ar[
      dd,
      "{
        \phi
      }"
    ]
    \ar[
      rr,
      dashed,
      "{
        (
          \phi^\ast
          c_3
          ,\,
          \phi^\ast c_6
          ,\,
          {
            \color{purple}
            b_2
          } 
        )
      }",
      "{
        \scalebox{.7}{
          \color{purple}
          \bf 
          \def\arraystretch{.9}
          \begin{tabular}{c}
            worldvolume
            \\
            charges
          \end{tabular}
        }
      }"{swap}
    ]
    &&
    \mathcal{A}
    \ar[
      dd,
      "{
        p
      }"
    ]
    &
    \mathfrak{l}_{{}_\mathcal{B}}
    \mathcal{A}
    \ar[
      r,
      "{
        \overset{!}{\sim}
      }"
    ]
    \ar[
      dd,
      "{
        \mathfrak{l}p
      }"
    ]
    &[-10pt]
    \mathfrak{l}_{{}_{S^4}}
    S^7
    \ar[
      dd,
      "{
        \mathfrak{l}
        h_{\mathbb{H}}
      }"
    ]
    \\
    &&
    \ar[
      r,
      phantom,
      shift right=3pt,
      "{
        \longmapsto
      }"
    ]
    &
    {}
    \\
    X^{1,10\vert \mathbf{32}}
    \ar[
      rr,
      "{
        (c_3
        ,\, 
        c_6)
      }",
      "{
        \scalebox{.7}{
          \color{darkgreen}
          \bf
          \def\arraystretch{.9}
          \begin{tabular}{c}
            background
            \\
            M-brane charges
          \end{tabular}
        }
      }"{swap}
    ]
    &&
    \mathcal{B}
    &
    \mathfrak{l}
    \mathcal{B}
    \ar[
      r,
      "{
        \overset{!}{\sim}
      }"
    ]
    &
    \mathfrak{l}S^4
  \end{tikzcd}
\end{equation}
On the bottom of the diagram this is the situation discussed in \cite{GSS24-SuGra}:
Given such a choice $\mathcal{B}$ for the base classifying space of the cohomology theory in which  M-brane charge takes values, the flux-quantized gauge fields of the 11d supergravity background according to \eqref{GaugePotentialAsHomotopy} are given by dashed maps as on the bottom of the following diagram of supergeometric $\infty$-groupoids, constituting a cocycle in differential $\mathcal{B}$-cohomology \cite[Def. 9.3]{FSS23Char}:
\begin{equation}
  \label{M5BraneFluxQuantizationInIntroduction}
  \begin{tikzcd}[
    row sep=5pt, column sep=huge
  ]
    &[+10pt]
      \mathcal{A}
      \ar[
        dddd
      ]
    \ar[
      drr,
      "{
        \mathbf{ch}_{
          \mathcal{A}
        }^{\mathcal{B}}
      }"{sloped}
    ]
    &[-20pt]
    &[-10pt]
    \\
    \Sigma^{1,5\vert 2\cdot \mathbf{8}}
    \ar[
      dddd,
      "{
        \phi
      }"
    ]
    \ar[
      ur,
      dashed,
      "{
       (
         \phi^\ast
         c_3
         ,\,
         \phi^\ast
         c_6
         ,\,
         {
         \color{purple}
         b_2
         }
        )
      }"{
        sloped,
        pos=.45
      },
      "{\ }"{pos=.4, swap, name=s1}      
    ]
    \ar[
      drr,
      dashed,
      crossing over,
      "{
        \scalebox{.95}{
        \colorbox{white}{$
        (
          \phi^\ast
          G_4^s
          ,\,
          \phi^\ast
          G_7^s
          ,\,
          {
          \color{purple}
          H_3^s
          }
        )
        $}
        }
      }"{sloped, swap, pos=.4},
      "{\ }"{pos=.3, name=t1}
    ]
    \ar[
      from=s1,
      to=t1,
      Rightarrow,
      dashed, color=purple, 
      "{
        \widehat{B}_2 
      }"{pos=.8}
    ]
    &&&
    \shape
    \big(
    \Omega^1_{\mathrm{dR}}(
      -
      ;\,
      \mathfrak{l}_{
        \mathfrak{l}S^4
      }
      S^7
    )_{\mathrm{clsd}}
    \big)
    \ar[dddd]
    \\
    &&
    \Omega^1_{\mathrm{dR}}\big(
      -;
      \,
      \mathfrak{l}_{S^4}
      S^7
    \big)_{\mathrlap{\mathrm{clsd}}}
    \ar[
      ur,
      shorten <=-5pt,
      shorten >=-5pt,
      "{
        \eta^{\,\scalebox{.6}{$\shape$}}
      }"{description}
    ]
    \\
    \\
    &
    \mathcal{B}
    \ar[
      drr,
      "{
        \mathbf{ch}_{\mathcal{B}}
      }"{sloped, pos=.3}
    ]
    \\
    X^{1,10\vert\mathbf{32}}
    \ar[
      ur,
      "{
        (c_3, c_6)
      }"{sloped},
      "{\ }"{swap, name=s2, pos=.4}
    ]
    \ar[
      drr,
      "{
        (G_4^s, G_7^s)
      }"{sloped, swap, pos=.4},
      "{\ }"{name=t2, pos=.3}
    ]
    \ar[
      from=s2,
      to=t2,
      Rightarrow,
      color=orangeii, 
      "{
        (
          \widehat{C}_3
          ,\
          \widehat{C}_6
        )
      }"{pos=.8}
    ]
    &&&
    \shape
    \big(
    \Omega^1_{\mathrm{dR}}(
      -;
      \,
      \mathfrak{l}S^4
    )_{\mathrm{clsd}}
    \big)
    \\
    &&
    \Omega^1_{\mathrm{dR}}\big(
      -;
      \,
      \mathfrak{l}S^4
    \big)_{\mathrlap{\mathrm{clsd}}}
    \ar[
      ur,
      shorten <=-5pt,
      shorten >=-5pt,
      "{
        \eta^{\,\scalebox{.6}{$\shape$}}
      }"{description}
    ]
    \ar[
      from=uuuu,
      crossing over
    ]
  \end{tikzcd}
\end{equation}
Now, the completion of the bottom part to the full diagram via the top dashed maps constitutes the construction of complete flux-quantized on-shell fields on the M5-brane worldvolume, being cocycles in the {\it twisted} differential $\mathcal{A}$-cohomology \cite[Def. 11.2]{FSS23Char}, with the twist here being the pullback of the bulk differential cocycle to the worldvolume.

\smallskip

It is instructive to re-write this a little (as in \cite[(11.4)]{FSS23Char}):
Noticing that in \eqref{M5BraneFluxQuantizationInIntroduction} we have homotopy ``cones'' over a ``cospan'' (of supergeometric $\infty$-groupoids), the diagram factors equivalently through the homotopy fiber products (of the character maps $\mathbf{ch}$ with the flux images $\eta^{\scalebox{.7}{$\shape$}}$):
$$
  \begin{tikzcd}[
    column sep=15pt
  ]
  \Sigma^{1,5\vert 2\cdot \mathbf{8}}
  \ar[
    dd,
    "{
      \phi
    }"
  ]
  \ar[
    rrrr,
    dashed,
    "{
      (
        \phi^\ast c_3
        ,\, 
        \phi^\ast c_6
        ,\, 
      {\color{purple} b_2})
      \,
      \tensor
        [_{\mathbf{ch}}]
        {\times}
        {_{\eta^{\scalebox{.6}{$\shape$}}}}
       \,
       (
         \phi^\ast
         G_4^s
         ,\,
         \phi^\ast
         G_7^s
         ,\,
         {
           \color{purple}
           H_3^s
         }
       )
    }",
    "{
      \scalebox{.7}{
        \color{purple}
        \bf
        \def\arraystretch{.9}
        \begin{tabular}{c}
          complete flux-quantized
          \\
          B-field on M5-brane
        \end{tabular}
      }
    }"{swap}
  ]
  &&
  &[+60pt]
  &[+60pt]
  \mathcal{A}_{\mathrm{diff}/\mathcal{B}}
  \ar[
    dd,
    "{
      p_{\mathrm{diff}}
    }"
  ]
  &[-5pt]
  :=
  &[-5pt]
  \mathcal{A}
  \;
  \tensor
    [_{\mathbf{ch}}]
    {\times}
    {_{\eta^{\scalebox{.6}{$\shape$}}}}
    \;
      \Omega^1_{\mathrm{dR}}
      (
        -
        ;\,
        \mathfrak{l}_{S^4} S^7
      )_{\mathrm{clsd}}    
  \ar[
    dd,
    "{
      p
      \,
    \tensor
      [_{\mathbf{ch}}]
      {\times}
      {_{\eta^{\scalebox{.6}{$\shape$}}}}
      \;
      \mathfrak{l}p_\ast
    }"
  ]
    \\
    \\
  X^{1,10\vert \mathbf{32}}
  \ar[
    rrrr,
    "{
      (c_3, c_6)
      \,
      \tensor
        [_{\mathbf{ch}\!}]
        {\times}
        {_{\!\eta^{\scalebox{.6}{$\shape$}}}}
      \,
      (G_4^s, G_7^s)
    }",
    "{
      \scalebox{.7}{
        \color{darkgreen}
        \bf
        \def\arraystretch{.9}
        \begin{tabular}{c}
          complete flux-quantized
          \\
          C-field of 11 SuGra background
        \end{tabular}
      }
    }"{swap}
  ]
  &&&&
  \mathcal{B}_{\mathrm{diff}}
  &:=&
  \mathcal{B}
  \;
  \tensor
    [_{\mathbf{ch}}]
    {\times}
    {_{\eta^{\scalebox{.6}{$\shape$}}}}
    \;
      \Omega^1_{\mathrm{dR}}
      (
        -
        ;\,
        \mathfrak{l}S^4
      )_{\mathrm{clsd}}    
  \end{tikzcd}
$$
in refinement of the front square diagram in \eqref{M5BraneFluxQuantizationInIntroduction}, which expresses in diagrammatic form
\eqref{DiagrammaticEOMs}
the super-flux Bianchi identity \eqref{TwistedBianchiIdentityInIntroduction}
that we started with. 

\smallskip

This solves the problem of flux-quantization on the M5-brane in generality. In order to say more,  we next turn attention to a particular choice of admissible flux-quantization law.

\medskip

\noindent
{\bf Hypothesis ``H'' about Flux quantization on the M5-brane.}
The admissibility condition \eqref{WorldvolumeCharges} on flux-quantization laws on the M5-brane is a strong constraint, but it still leaves infinitely many in-equivalent choices. For instance, with every admissible flux quantization law $p : \mathcal{A} \xrightarrow{\;} \mathcal{B}$ also its Cartesian product with the classifying space $B K$ of any {\it finite} group $K$ is again admissible (because the homotopy groups of such $B K$ are purely torsion, so that the corresponding charges have vanishing reflection in the flux densities).
At the same time, the form of \eqref{WorldvolumeCharges} suggests an evident choice among this infinitude of choices: The condition that $p : \mathcal{A} \xrightarrow{\;} \mathcal{B}$ be of the same rational homotopy type as the quaternionic Hopf fibration is of course solved by taking $p := h_{\mathbb{H}}$ to be the quaternionic Hopf fibration itself!

\smallskip 
Now the (twisted) generalized cohomology theory classified by the (quaternionic Hopf fibration between) spheres is called (twisted, unstable) co-{\it Homotopy} theory (review and pointers in \cite{FSS23Char}, more on this in \cref{ResultingSolitons} below), denoted as follows:
\begin{equation}
  \label{TwistedCohomotopyInIntroduction}
  \hspace{-.7cm}
  \def\arraystretch{1}
  \begin{array}{crclcl}
  {
    \scalebox{.7}{
      \color{darkblue}
      \bf
      \def\arraystretch{.9}
      \begin{tabular}{c}
        M-brane charge in
        \\
        4-Cohomotopy
      \end{tabular}
    }
  }
  &
  [c_3,c_6]
  \in
  \pi^4(X)
  &:=&
  \pi_0
  \, 
  \mathrm{Maps}
  \big(
    X
    ,\,
    S^4
  \big)
  &=&
  \Big\{
  \!
  \begin{tikzcd}[
    column sep=25pt
  ]
    X
    \ar[
      rr,
      dashed,
      "{
        (c_3,c_6)
      }"
    ]
    &&
    S^4
  \end{tikzcd}
  \! \Big\}_{\!\big/ \mathrm{hmtp.}}
  \\[15pt]
  {
    \scalebox{.7}{
      \color{darkblue}
      \bf
      \def\arraystretch{.9}
      \begin{tabular}{c}
        {\color{purple}
        Charges on M5-branes} in
        \\
        twisted 3-Cohomotopy
      \end{tabular}
    }
  }
  &
  \pi^{3+\phi^\ast(c_3,c_6)}(\Sigma)
  &:=&
  \pi_0
  \mathrm{Maps}\big(
    \Sigma
    ,\,
    S^7
  \big)_{/S^4}
  &=&
  \left\{
  \adjustbox{raise=3pt}{
  \begin{tikzcd}
    \Sigma
    \ar[
      d,
      "{ \phi }"{swap},
      "{\  }"{name=t, pos=.8}
    ]
    \ar[
      rr,
      dashed,
      "{
        {
          \color{purple}
          b_2
        }
      }",
      "{\ }"{swap,name=s}
    ]
    &&
    S^7
    \ar[
      d,
      "{ h_{\mathbb{H}} }"
    ]
    \\
    X
    \ar[
      rr,
      "{
        (
        c_3
        ,\,
        c_6
        )
      }"
    ]
    &&
    S^4
    \ar[
      from=s,
      to=t,
      Rightarrow,
      dashed
    ]
  \end{tikzcd}
  }
  \right\}_{
    \!\!\!\big/\mathrm{rel.hmtp.}}
  \end{array}
\end{equation}
Therefore, the hypothesis that this ``evident'' choice of flux-quantization is the ``correct'' one for completing the theory of the M5-brane in M-theory has been called ``Hypothesis H'' in \cite{FSS20-H}\cite{FSSHopf}\cite{FSS20TwistedString}\cite{GS21}\cite{SS23-Mf}, following \cite[\S 2.5]{Sati13} (the corresponding differential co-Homotopy for the M5-brane fields was first considered in \cite{FSS15-M5WZW}), which may be thought of as an M-theoretic version of the traditional ``Hypothesis K'' that D-brane charges are in (twisted) K-theory (cf. \cite[Rem. 4.1]{SS23-DefectBranes}). 

\smallskip

Notice that maps to the 4-sphere 
also appear on $N$ M5-worldvolumes $\Sigma$ in the guise of the scalar fields $\Phi$ taking values in the surrounding $\mathbb{R}^5 \setminus \{0\} \,\simeq\, \mathbb{R}_{> 0} \times S^4$, from which the $G_4$-flux density $\simeq N \mathrm{dvol}_{S^4}$ on spacetime thus gets pulled back to the worldvolume as $\Phi^\ast G_4 \,=\, N \cdot \Phi^\ast \mathrm{dvol}_{S^4}$
(first highlighted in \cite[\S 3]{GanorMotl98}\cite{Intriligator00}, as such reviewed in \cite[p. 20]{Berman08} and further developed in \cite{FSSHopf}). 
Now by {\it Hypothesis H} already $G_4$ itself is the ``Cohomotopical character'' (\cite[\S 12]{FSS23Char}) on spacetime of --- hence (by \cite[(47)]{FSS20-H}) the pullback of the unit volume form along  --
a map $c_3$ from the ambient spacetime to the homotopy type of $S^4$ (playing the role of the Cohomotopy classifying space), which near the horizon of a black M5-brane is (cf. \cite[Rem. 1]{SS21M5Anomaly}\cite{GSS24-AdS7})
just a multiple $N \in \mathbb{Z}$ of the projection map $p_{S^4}$ identifying the geometric $S^4$ around the M5-brane with the ``classifying'' $S^4$ for Cohomotopy theory:
\begin{equation}
  \label{CFieldFluxPulledBackToNM5Breanes}
  \hspace{-2cm} 
  \begin{tikzcd}[row sep=25pt,
    column sep=75pt
  ]
    \Sigma
    \ar[
      dd,
      "{
        \mathrm{id}_\Sigma
        \times
        \Phi
      }"{description},
      "{
        \scalebox{.7}{
          \color{darkgreen}
          \bf
          \begin{tabular}{c}
            M5-embedding map 
            \\
            for scalar fields \scalebox{1.2}{$\Phi$}
          \end{tabular}
        }
      }"{
        swap,
        xshift=-6pt
      }
    ]
    \ar[
      ddrr,
      end anchor={[yshift=3pt]},
      "{
        N \cdot \Phi^\ast \mathrm{dvol}_{S^4}
        \;\;\;\mapsfrom\;\;\;
        \mathrm{dvol}_{S^4}
      }"{sloped}
    ]
    \\
    \\
    \mathllap{
      \scalebox{.7}{
        \color{darkblue}
        \bf
        \def\arraystretch{.9}
        \begin{tabular}{c}
          Spacetime near
          \\
          black M5 horizon 
        \end{tabular}
      }
    }
    \underbrace{
    \Sigma 
      \times 
    \mathbb{R}_{> 0}
     \times
    S^4
    }_{X}
    \ar[
      rr,
      "{ 
        c_3 
        \,=\,
        N
        \cdot
        p_{S^4}
      }"{pos=.4},
      "{
        \scalebox{.7}{
          \color{darkgreen}
          \bf
          \def\arraystretch{.9}
          \begin{tabular}{c}
            classifying map
            of C-field charge
            \\
            by Hypothesis H
          \end{tabular}
        }
      }"{swap, pos=.4}
    ]
    &&
    S^4
    \mathrlap{
        \scalebox{.7}{
          \color{darkblue}
          \bf
          \def\arraystretch{.9}
          \begin{tabular}{c}
            Classifying space
            \\
            for 4-Cohomotopy
          \end{tabular}
        }    
    }
  \end{tikzcd}
\end{equation}
Here Hypothesis H entails that the global C-field charge on $\Sigma$ is not just the pulled back flux shown on the diagonal above, but the homotopy class of that diagonal map itself (of which the pullback flux is just the ``character'' in de Rham cohomology). E.g.: if $\Sigma \simeq \mathbb{R} \times S^5$ then there is a torsion charge $[S^5 \to S^4] \in \pi^4(\Sigma) \,\simeq\, \ZTwo$ (a homotopy group of spheres, cf. \cite{SS23-Mf}) which is non-trivial in Cohomotopy but not detected by the de Rham class of $\Phi^\ast G_4$.

\medskip 
Hypothesis H is supported by the fact that it {\it implies} a series of subtle topological (anomaly cancellation-) conditions (survey in \cite[Tbl 1]{FSS20-H}\cite[\S 12]{FSS23Char}) that are expected to hold in M-theory (including M5-brane anomaly cancellation \cite{SS21M5Anomaly}), notably the following two
\footnote{
\label{OnTangentialTwists}
The gray terms in \eqref{ShiftedG4FluxQuantization} and \eqref{G7FluxQuantization} do arise on curved spacetimes by flux-quantization in {\it tangentially twisted} co-Homotopy \cite{FSS20-H}.

Here we disregard discussion of the tangential twist for the sake of brevity. Under this simplifying assumption the gray term in \eqref{ShiftedG4FluxQuantization} is not implied by the cohomotopical quantization condition,
and thereby the implied integrality condition on $G_4$ is the expected one (only) on spacetimes for which $\tfrac{1}{4}p_1$ is integral. This however is the case for the main examples of interest here, such as the Freund-Rubin spacetime $\mathrm{AdS}_{7} \times S^{4}$ (cf. Ex. \ref{HolographicM5BraneSuperImmersions}) where the Pontrjagin classes actually vanish \cite[Prop. 22]{SS21M5Anomaly}.



} 
which ensure that the exponentiation of the usual gauge-coupling action functionals for the M-branes 
(\cite[\S 3]{PST97})
are globally well-defined:

\vspace{1mm} 
\begin{itemize}[
  leftmargin=.6cm,
  topsep=1pt,
  itemsep=2pt
]
  \item 
  \cite[Prop. 3.13]{FSS20-H}
  the shifted 
   flux quantization 
   \cite[p 2-3]{Witten97a}
   of the flux sourced by M5-branes:
   \vspace{1mm} 
   \begin{equation}
     \label{ShiftedG4FluxQuantization}
     \big[G_4 {\color{gray}+ \tfrac{1}{4}p_1}\big] 
     \,\in\, 
     H^4(X;\, \mathbb{Z}) \xrightarrow{\quad} H^4_{\mathrm{dR}}(X)
   \,,
   \end{equation}
   which serves as the Wess-Zumono-term on M2-branes,
  \item 
  \cite[Thm. 4.8]{FSSHopf}:
  the integral flux quantization of the
  {\it Page charge} 
  \vspace{1mm} 
  \begin{equation}
    \label{G7FluxQuantization}
    \big[\,
    \widetilde 2 G_7 + H_3 \, \big(G_4 {\color{gray}+ \tfrac{1}{4}p_1}\big)
    \big]
    \;\in\;
    H^7\big(
      \Sigma^7;\,
      \mathbb{Z}
    \big)
    \xrightarrow{\quad}
    H^7_{\mathrm{dR}}(\Sigma^7)
  \end{equation}
  sourced by M2-branes, which serves as the completed {\it Hopf-Wess-Zumino term} on M5-branes.
\end{itemize}

\smallskip

Note here that the homotopy fiber of the quaternionic Hopf fibration is the 3-sphere
$\!\!
  \begin{tikzcd}
    S^3
    \ar[
      r,
      "{
        \mathrm{fib}(h_{\mathbb{H}})
      }"
    ]
    &
    S^7
    \ar[
      r,
      "{
        h_{\mathbb{H}}
      }"
    ]
    &
    S^4
    \,,
  \end{tikzcd}
\!\!$
(which remains the case after $\widehat{\mathrm{Sp}(2)}$-twisting)
and which implies that the charges on the M5 \eqref{TwistedCohomotopyInIntroduction} are {\it locally} cocycles in 3-Cohomotopy, globally twisted by the 4-Cohomotopy of the supergravity background
(\cite[Rem. 3.17]{FSS20-H}\cite[pp. 21]{FSSHopf}\cite[p. 6]{FSS20TwistedString}; a twist of this kind was anticipated early on in \cite{Sati06}\cite{Sati10}). In \cite{FSS20TwistedString} it was explained how, under Hypothesis H with gravitational twists included, as above, this reveals the $H_3$-flux as associated with a kind of ``non-abelian gerbe field'' on the M5-brane not unlike the proposals previously made in \cite[p. 6, 15]{Witten04}\cite[\S 3]{Witten09}\cite{SaemannSchmidt18}\cite{SaemannSchmidt20}.

\medskip

\noindent
{\bf The need for super-flux quantization.}
However, in  previous discussions of Hypothesis H it had been left open which effect, if any, the duality-relations on the flux densities have on the flux quantization process (away from a Cauchy surface, where the issue was solved in \cite{SS24-Phase}). This is the remaining point which we solve here (in tandem with \cite{GSS24-SuGra}), by passing to super-spacetimes and observing that here the duality relations on the bosonic flux densities are absorbed within the Bianchi identities of their super-flux enhancements, so that flux quantization on super-space is revealed to already deal with the exact on-shell field content, not requiring any further constraints.
In order to properly bring out this remarkable but subtle point we proceed as follows:

\smallskip 
\begin{itemize}[
  leftmargin=.6cm,
  topsep=1pt,
  itemsep=2pt
]
\item[in]  \S\ref{SuperEmbeddings} {\bf Revisiting  super-embeddings} we review the idea of ``super-embeddings'' of super $p$-branes with attention to what we feel have remained loose ends: the global nature of co-frame fields, the relation to the classical theory of Darboux co-frame fields, and an observation on how to naturally unify/streamline what actually is a list of traditional ``super-embedding'' conditions.

\item[in] \S\ref{SuperEmbeddingConstruction} {\bf M5-brane super-immersions} we systematically re-derive the equations of motion of the 3-flux on M5-branes using a transparent algebraic description of the all-important reduction and branching of spin representations on the M5-brane.
\end{itemize}

\smallskip

\noindent
{\bf Analyzing quantized flux on M5-branes.}
With these results in hand, we close

\begin{itemize}[
  leftmargin=.6cm,
  topsep=1pt,
  itemsep=2pt
]
\item[in]  
\S\ref{FluxQuantizationOnM5Branes} {\bf Flux quantization on M5 branes} by showing how flux quantization on the M5 reproduces the traditional local formulas for higher gauge potentials while completing these to global fields that may exhibit skyrmionic and anyonic topological properties.
\end{itemize}

\newpage

\section{Revisiting ``super-embeddings''}
\label{SuperEmbeddings}

Here we give a streamlined and rigorous account of the ``super-embedding approach'' to super $p$-brane sigma models (due to \cite{BPSTV95}\cite{HoweSezgin97a}\cite{HoweRaetzelSezgin98}\cite{Sorokin00}, previously reviewed in \cite{Bandos11}\cite{BandosSorokin23}), providing a precise super-geometric formulation (which seems previously to have been missed by super-geometers, cf. \cite[\S13.3]{Rogers07}). Besides setting the scene for the analysis following in \S\ref{SuperEmbeddingConstruction} and clearing up some fine-print not usually considered in the literature, the main observation here is (Prop. \ref{ComponentsOfBPImmersion}) a slick way to unify (Def. \ref{BPSImmersion} of ``BPS super-immersions'') the traditional data of ``super-embeddings'' (Rem. \ref{SuperEmbeddingConditionInTheLiterature}) by generalizing the classical notion of {\it Darboux coframe fields} for Riemannian immersions (discussed in \S\ref{DarbouxCoFrameFields}); see also Rem. \ref{ImmersionsAndEmbeddings} for the distinction between immersions and embeddings.

\medskip

\noindent
{\bf Conventions.} Our conventions are standard, but since the computations in \S\ref{SuperEmbeddingConstruction} crucially depend on the corresponding prefactors, here to briefly make them explicit:
\begin{notation}[\bf Tensor conventions]
$\,$
\begin{itemize}[leftmargin=.4cm]
\item
  The Einstein summation convention applies throughout: Given a product of terms indexed by some $i \in I$, with the index of one factor in superscript and the other in subscript, then a sum over $I$ is implied:
  $
    x_i \, y^i
    :=
    \sum_{i \in I} 
    x_i \, y^i
  $.

\item
Our Minkowski metric is the matrix
\begin{equation}
  \label{MinkowskiMetric}
  \big(\eta_{ab}\big)
    _{a,b = 0}
    ^{ d }
  \;\;
    =
  \;\;
  \big(\eta^{ab}\big)
    _{a,b = 0}
    ^{ d }
  \;\;
    :=
  \;\;
  \Big(
    \mathrm{diag}
      (-1, +1, +1, \cdots, +1)
  \Big)_{a,b = 0}^{d}
\end{equation}
\item
  Shifting position of frame indices always refers to contraction with the  Minkowski metric \eqref{MinkowskiMetric}:
  $$
    V^a 
      \;:=\;
    V_b \, \eta^{a b}
    \,,
    \;\;\;\;
    V_a \;=\; V^b \eta_{a b}
    \,.
  $$
\item Skew-symmetrization of indices is denoted by square brackets ($(-1)^{\vert\sigma\vert}$ is sign of the permutation $\sigma$):
$$
  V_{[a_1 \cdots a_p]}
  \;:=\;
  \tfrac{1}{p!}
  \sum_{
    \sigma \in \mathrm{Sym}(n)
  }
  (-1)^{\vert \sigma \vert}
  V_{ a_{\sigma(1)} \cdots a_{\sigma(p)} }\,.
$$
\item
We normalize the Levi-Civita symbol to \begin{equation}
  \label{transversalizationOfLeviCivitaSymbol}
  \epsilon_{0 1 2 \cdots} 
    \,:=\, 
  +1
  \;\;\;\;\mbox{hence}\;\;\;\;
  \epsilon^{0 1 2 \cdots} 
    \,:=\, 
  -1
  \,.
\end{equation}
\item
We normalize the Kronecker symbol to
$$
  \delta
    ^{a_1 \cdots a_p}
    _{b_1 \cdots b_p}
  \;:=\;
  \delta^{[a_1}_{[b_1}
  \cdots
  \delta^{a_p]}_{b_p]}
  \;=\;
  \delta^{a_1}_{[b_1}
  \cdots
  \delta^{a_p}_{b_p]}
  \;=\;
  \delta^{[a_1}_{b_1}
  \cdots
  \delta^{a_p]}_{b_p}
$$
so that
\begin{equation}
  \label{ContractingKroneckerWithSkewSymmetricTensor}
  V_{
    \color{darkblue}
    a_1 \cdots a_p
  }
  \tensor*
    {\delta}
    {
    ^{ 
       \color{darkblue}
       a_1 \cdots a_p 
    }
    _{b_1 \cdots b_p}
    }
  \;\;
  =
  \;\;
  V_{[b_1 \cdots b_p]}  
  \;\;\;\;
  \mbox{and}
  \;\;\;\;
  \epsilon^{
    {\color{darkblue}  
      c_1 \cdots c_p
    }
    a_1 \cdots a_q
  }
  \,
  \epsilon_{
    {\color{darkblue}
    c_1 \cdots c_p 
    }
    b_1 \cdots b_q
  }
  \;\;
  =
  \;\;
  -
  \,
  p! \cdot q!
  \,
  \delta
    ^{a_1 \cdots a_q}
    _{b_1 \cdots b_q}
  \,.
\end{equation}
\end{itemize}
\end{notation}

\subsection{Darboux co-frame fields}
\label{DarbouxCoFrameFields}

We recall the classical notion of {\it Darboux co-frame fields} for (pseudo-)Riemannian immersions (Def. \ref{DarbouxCoFrame} below, which may not to have found due attention in the super-embedding literature before) and re-cast it into an equivalent form (Prop. \ref{ReformulatingTheDarbouxCondition}) whose super-geometric generalization turns out to be just that of $\sfrac{1}{2}$BPS super-immersions (``super-embeddings'') of super $p$-branes (\S\ref{BPSSuperImmersions} below).

\smallskip

In order to be precise, we begin now by being a little pedantic about some maybe underappreciated global aspects of local coframe fields (which can be and typically are ignored in local analysis but can no longer be ignored for global discussions such as of flux quantization) --- the impatient reader may want to skip ahead to \S\ref{BPSSuperImmersions} and come back here only as need be.

\medskip

\noindent
{\bf Relativistic local co-frame fields.}
In much of the physics literature, coframe fields $E$ on a spacetime $X$ are shown on a single tacitly-assumed chart $U \xhookrightarrow{\;} X$ only, instead of on all of spacetime $X$, leaving their global definition to the imagination of the reader. But since global issues cannot be neglected for our purpose of flux quantization, we introduce a tad of extra notation that allows to elegantly deal with this issue properly.

\medskip

 {\bf Open covers.}
Namely given an open cover of spacetime
$
  \big\{
    U_j
      \xhookrightarrow{\;\iota_i\;}
    X
  \big\}_{j \in J}
$
such that the coframe field $E$ is naively defined on each of the charts $U_j$, then we denote smooth manifold which is the disjoint union of all these charts as follows

\begin{equation}
  \label{AnOpenCover}
  \overset{
   \mathclap{
     \raisebox{3pt}{
     \scalebox{.7}{
       \color{darkblue}
       \bf
       open cover
     }
     }
   }
  }{
  \big\{
    U_j
      \xhookrightarrow{\;\iota_i\;}
    X
  \big\}_{j \in J}
  }
  \hspace{1cm}
  \CoverOf{X}
  \;:=\;
  \overset{
   \mathclap{
     \raisebox{7pt}{
     \scalebox{.7}{
       \color{darkblue}
       \bf
       corresponding open submersion
     }
     }
   }
  }{
  \textstyle{
    \underset{
      j \in J
    }{\coprod}
  }
  \,
  U_j
  }
  \,,
  \hspace{.7cm}
  \mbox{with}
  \hspace{.7cm}
  \begin{tikzcd}[
    sep=0pt
  ]
    \CoverOf{X}
    \ar[
      rr,
      ->>,
      "{ p }"
    ]
    &&
    X
    \\
    \underbrace{
      (x,j)
    }_{
      \in \, U_j
    }
    &\longmapsto&
    x\,.
  \end{tikzcd}
\end{equation}
In terms of this, the co-frame field is a map of the form
$
E
  :
    T\CoverOf{X}
      \longrightarrow 
    \mathbb{R}^{1+d}
$,
namely a $J$-tuple of map
$\big( E_j : T U_j \to \mathbb{R}^{1+d} \big)_{j \in J}$,
satisfy some conditions which we summarize in Def. \ref{CoFrameFields} below.

To that end, notice that
given a pair of open covers $\big\{ U_j \xhookrightarrow{\iota_j} X \big\}_{j \in J}$ and $\big\{ U'_{j'} \xhookrightarrow{\iota'_j} X \big\}_{j' \in J'}$ (which might be the same) of the same space spacetime $X$, the disjoint union of all the intersections $U_j \cap U'_{j'}$ of their charts is their {\it fiber product} with respect to the maps \eqref{AnOpenCover}:
\begin{equation}
  \label{FiberProductOfCoveringSubmersion}
  \hspace{-5mm} 
  \CoverOf{X} 
  \!\times_{\!{}_X}\!
  \CoverOf{X}'
  \;=\;
  \textstyle{
    \underset{
      {j \in J}
      \atop
      {j' \in J'}
    }{\coprod}
  }
  U_j \cap U'_{j'}
  \hspace{.9cm}  
  \begin{tikzcd}[
    row sep=5pt
  ]
    &
    \CoverOf{X} 
    \!\times_{{}_X}\!
    \CoverOf{X}'
    \ar[
      dl,
      ->>,
      "{
        \mathrm{pr}
      }"{swap}
    ]
    \ar[
      dr,
      ->>,
      "{
        \mathrm{pr}'
      }"
    ]
    \\
    \CoverOf{X}
    \ar[
      dr,
      ->>,
      "{ p }"{swap}
    ]
    &&
    \CoverOf{X}'
    \ar[
      dl,
      ->>,
      "{ p' }"
    ]
    \\
    &
    X
  \end{tikzcd}
  \hspace{.7cm}
  \mbox{with}
  \hspace{.7cm}
  \begin{tikzcd}[row sep=-2pt, 
    column sep=0pt
  ]
    \CoverOf{X}
    \!
    \times_{\!{}_X}
    \!
     \CoverOf{X'}
    \ar[
      rr, 
      ->>,
      "{ \mathrm{pr} }"
    ]
    &&
    \CoverOf{X}
    \\
    \underbrace{
      (x,j,j')
    }_{
      \in
      \,
      U_{j}
      \cap 
      U'_{j'}
    }
    &\longmapsto&
    (x,j)
    \mathrlap{\,.}
  \end{tikzcd}
\end{equation}
Therefore the pullback of the co-frame along $\mathrm{pr}$ is the original one restricted from the charts $U_j$ to all
their intersections with the charts $U'_{j'}$:
\vspace{-1mm} 
$$
  \begin{tikzcd}[row sep=-2pt, column sep=small]
    \mathrm{pr}^\ast
    \,
    E
    \ar[r, phantom, "{ : }"]
    &[2pt]
    T
    \CoverOf{X}
    \!
    \times_{\!{}_X}
    \!
    \CoverOf{X}'
    \ar[
      rr,
      ->>,
      "{ \mathrm{pr} }"
    ]
    &&
    T
    \CoverOf{X}
    \ar[
      rr,
      ->>,
      "{ p }"
    ]
    &&
    X
    \\
    &
    \underbrace{
    (v,j,j')
    }_{
      \in
      \,
      T(
        U_j \cap U'_{j'}
      )
    }
    &\longmapsto&
    \underbrace{
      (v,j)
    }_{
      \in \,
      T U_j
    }
    &\longmapsto&
    E_j(v)\,.
  \end{tikzcd}
$$

\begin{definition}[\bf Relativistic local co-frame field]
\label{CoFrameFields}
Given a smooth manifold $X$ of dimension
$1+d$, a {\it Relativistic local co-frame field} $E$ on $X$ is any one of the following six equivalent structures (given in increasing level of concreteness):
\begin{itemize}[
  leftmargin=.8cm,
  topsep=1pt,
  itemsep=2pt
]
\item[\bf (i)] a smooth $\mathrm{O}(1,d)$-structure on $X$;

\item[\bf (ii)] a smooth reduction of the structure group of the frame bundle of $X$ through $\mathrm{O}(1,d) \hookrightarrow \mathrm{GL}(1+d)$;
\item[\bf (iii)] a vector bundle isomorphism from the tangent bundle to a Minkowski-space fiber bundle;
\item[\bf (iv)] a local trivialization of the tangent bundle whose transition functions take values in $\mathrm{O}(1,d) \hookrightarrow \mathrm{GL}(1+d)$;
\item[\bf (v)]
\begin{itemize}[leftmargin=.6cm]
  \item[\bf (a)] an open cover 
  $\begin{tikzcd}
      \CoverOf{X}
    \ar[r, ->>, "{ \mathrm{opn} }"{swap}] 
    & 
    X
    \,,
   \end{tikzcd}$  

   \vspace{-2mm} 
   \item[\bf (b)]
   a fiberwise linear isomorphism
   $
     \begin{tikzcd}[
       column sep=10pt,
       row sep=-1pt
     ]
       T   \CoverOf{X}
       \ar[
         rr, 
         shorten=-1pt,
         "{ \sim }"{swap},
         "{ t }"
       ]
       \ar[
         dr,
         shorten=-2pt
       ]
       &&
       \mathbb{R}^{1,d}
       \!\times\!
         \CoverOf{X}
       \ar[
         dl,
         shorten=-2pt
       ]
       \\
       &
         \CoverOf{X}
     \end{tikzcd}
   $
   \item[\bf (c)]
   whose transition function
   is Lorentz-valued
   \begin{equation}
     \label{TransitionFunction}
     \hspace{-5mm} 
     g 
     \,: 
     \begin{tikzcd}[column sep=32]
       \mathbb{R}^{1,d}
       \!\times\! 
       \CoverOf{X} 
       \!\times_{{}_{X}}\! 
       \CoverOf{X}
       \ar[
         r,
         "{
           \mathrm{pr}_1^\ast 
           \, 
           t^{-1}
         }"
       ]
       &
       T \CoverOf{X} 
       \!\times_{{}_{X}}\! 
       \CoverOf{X}       
       \ar[
         r,
         "{
           \mathrm{pr}_2^\ast 
           \, 
           t          
         }"
       ]
       &
       \mathbb{R}^{1,d}
       \!\times\! 
       \CoverOf{X} 
       \!\times_{X}\! 
       \CoverOf{X}
     \end{tikzcd}
     \;\;\;\mbox{with}\;\;\; 
     g 
     \;\in\;
     C^\infty\big(
       \CoverOf{X} 
         \!\times_{{}_X}\!
       \CoverOf{X}
       ;\,
       \mathrm{O}(1,d)
     \big)\,,
   \end{equation}
  \item[\bf (d)]
where a pair of such local trivializations $(\CoverOf{X}_1, t_1)$ $(\CoverOf{X}_2, t_2)$ is regarded as equivalent if there is a Lorentz-group valued function $k \in C^\infty\big(\CoverOf{X}_1 \!\times_{{}_X}\! \CoverOf{X}_2 ;\, \mathrm{O}(1,d)\big)$ such that
\begin{equation}
  \label{TransformationOfLocaltrivializations}
  \begin{tikzcd}[column sep=huge]
    T \CoverOf{X}_1
    \!\times_{{}_X}\!
    \CoverOf{X}_2
    \ar[
      r,
      "{ \sim }"
    ]
    \ar[
      d,
      "{
        \mathrm{pr}_1^\ast t_1
      }"
    ]
    &
    \CoverOf{X}_1
    \!\times_{{}_X}\!
    T \CoverOf{X}_2
    \ar[
      d,
      "{
        \mathrm{pr}_2^\ast t_2
      }"
    ]
    \\
    \mathbb{R}^{1,d}
    \times
    \CoverOf{X}_1 
      \!\times_{{}_X}\!
    \CoverOf{X}_2
    \ar[
      r,
      "{ k }"
    ]
    &
    \mathbb{R}^{1,d}
    \times
    \CoverOf{X}_1
      \!\times_{{}_X}\!
    \CoverOf{X}_2
  \end{tikzcd}
\end{equation}
\end{itemize}
\item[\bf (vii)]
 \begin{itemize}
 \item[\bf (a)]
   a differential 1-form
   $
     E
     \;\in\;
     \Omega^1_{\mathrm{dR}}
     \big(
         \CoverOf{X}
       ;\,
       \mathbb{R}^{1,d}
     \big)$   
   \vspace{1mm} 
 \item[\bf (b)]
   for which there is $t$ as above with
   \begin{equation}
     \label{CoFrame1Form}
     \begin{tikzcd}
       T\CoverOf{X}
       \ar[
         r,
         "{ t }",
         "{ \sim }"{swap}
       ]
       \ar[
         rr,
         rounded corners,
         to path={
              ([yshift=+00pt]\tikztostart.south)  
           -- ([yshift=-09pt]\tikztostart.south)
           -- node[yshift=5pt]{
               \scalebox{.9}{
                 $E$
               }
           }
              ([yshift=-08pt]\tikztotarget.south)
           -- ([yshift=+00pt]\tikztotarget.south)
         }
       ]
       &
       \mathbb{R}^{1,d}
       \times
        \CoverOf{X}
       \ar[
         r,
         ->>
       ]
       &
       \mathbb{R}^{1,d}
     \end{tikzcd}
   \end{equation}
which means by \eqref{TransitionFunction} that on double overlaps these 1-forms are related by Lorentz transformations given by the transition functions:
\begin{equation}
  \label{TransformationOfCoframe}
  \begin{tikzcd}[column sep=50]
    &
    T \CoverOf{X}
    \times_{{}_X}
      \CoverOf{X}
    \ar[
      d,
      "{
        \mathrm{pr}_1^\ast t
      }"{swap}
    ]
    \ar[
      rr,
      equals
    ]
    &&
    T \CoverOf{X}
    \times_{{}_X}
      \CoverOf{X}
    \ar[
      d,
      "{
        \mathrm{pr}_2^\ast t
      }"
    ]
    \\
    &
    \mathbb{R}^{1,d}
    \!\times\!
    \CoverOf{X}
    \!\times_{{}_X}\!
    \CoverOf{X}
    \ar[
      r,
      "{
        \mathrm{pr}_1^\ast
        t^{-1}
      }"
    ]
    \ar[
      rr,
      rounded corners,
      to path={
           ([xshift=15pt, yshift=-00pt]\tikztostart.south)  
        -- ([xshift=15pt,yshift=-08pt]\tikztostart.south)
        -- node[yshift=5pt]{
              \scalebox{.75}{
                $g$
              }
        }
           ([xshift=-15pt,yshift=-08pt]\tikztotarget.south)
        -- ([xshift=-15pt,yshift=-00pt]\tikztotarget.south)
      }
    ]
    &
    T\CoverOf{X}
    \!\times_{{}_X}\!
      \CoverOf{X}
    \ar[
      r,
      "{
        \mathrm{pr}_2^\ast
        t
      }"
    ]
    &
    \mathbb{R}^{1,d}
    \!\times\!
    \CoverOf{X}
    \!\times_{{}_X}\!
    \CoverOf{X}
  \end{tikzcd}
\end{equation}

\vspace{1mm} 
\item[\bf (c)] again subject to the above notion of equivalence \eqref{TransformationOfLocaltrivializations}.
\end{itemize}
\end{itemize}
\end{definition}
\newpage
\begin{remark}[{\bf {\v C}ech cocycle data from co-frames}]
Beware that in the above Def. \ref{CoFrameFields} there are no further conditions imposed on triple overlaps, because we are  {\it not constructing} a bundle from {\v C}ech cocycle data, but instead are {\it extracting} {\v C}ech data (in order to impose orthogonality conditions on it) from picking local trivializations (the co-frames) of an already given bundle (the tangent bundle).
Indeed, the transition function \eqref{TransitionFunction} induced by a choice of co-frames necessarily satisfies on triple overlaps 
$$
  \begin{tikzcd}[row sep=-2pt, 
    column sep=30pt
  ]
    &
    \CoverOf{X}
    \!\times_X\!
    \CoverOf{X}
    \!\times_X\!
    \CoverOf{X}
    \ar[
      dl,
      "{
        \mathrm{pr}_{12}
      }"{pos=0.7}
    ]
    \ar[
      dd,
      "{
        \mathrm{pr}_{13}
      }"
    ]
    \ar[
      dr,
      "{
        \mathrm{pr}_{23}
      }"
      {swap, pos=.7}
    ]
    \\
    \CoverOf{X}
    \!\times_X\!
    \CoverOf{X}
    &
    &
    \CoverOf{X}
    \!\times_X\!
    \CoverOf{X}
    \\
    &
    \CoverOf{X}
    \!\times_X\!
    \CoverOf{X}
  \end{tikzcd}
$$
its  {\v C}ech cocycle condition 
$$
  \mathrm{pr}_{23}^\ast \, g
  \cdot 
  \mathrm{pr}_{12}^\ast \, g
  \;=\;
  \mathrm{pr}_{13}^\ast \, g
  \hspace{.7cm}
  \mbox{hence}
  \hspace{.5cm}
  \underset{
    {i j k}
    \atop
    x \in U_{i j k}
  }{\forall}
  \;
  g_{j k}(x)
  \cdot
  g_{i j}(x)
  \;=\;
  g_{i k}(x)
$$
because the following diagram commutes by construction:
$$
  \begin{tikzcd}[
    row sep=30pt,
    column sep=35pt
  ]
    & 
    \mathbb{R}^{1+d}
    \!\times\!
    \widehat X
    \ar[
       ddr,
       shift left=2pt,
       bend left=10,
       "{
         \mathrm{pr}_{23}^\ast
         \,
         g
       }"{sloped, description}
     ]
     \\
    &
    p^\ast 
    T X
    \ar[
      u,
      "{
        \mathrm{pr}_{23}^\ast 
        \,
        t
      }"{description}
    ]
    \ar[
      dl,
      "{
        \mathrm{pr}_{2}^\ast 
        \, 
        t
      }"{description}
    ]
    \ar[
      dr,
      "{
        \mathrm{pr}_3^\ast 
        \,
        t
      }"{description}
    ]
    \\[-13pt]
    \mathbb{R}^{1+d}
    \!\times\!
    \widehat X
    \ar[
      uur,
      shift left=2pt,
      bend left=10,
      "{
        \mathrm{pr}_{12}^\ast
        \,
        g
      }"{sloped, description}
    ]
    \ar[
      rr,
      shift right=2pt,
      bend right=10,
      "{
        \mathrm{pr}_{13}^\ast
        \, 
        g
      }"{description}
    ]
    &&
    \mathbb{R}^{1+d}
    \!\times\!
    \widehat X
  \end{tikzcd}
  \hspace{.6cm}
  \mbox{hence}
  \hspace{.6cm}
  \underset{
    {i,j,k}
    \atop
    {x \in U_{i j k}}
  }{\forall}
  \begin{tikzcd}[
    row sep=35pt,
    column sep=45pt
  ]
    & 
    \mathbb{R}^{1+d}
    \ar[
       ddr,
       shift left=2pt,
       bend left=10,
       "{
         g_{j k}(x)
       }"{sloped, description}
     ]
     \\
    &
    T_x X
    \ar[
      u,
      "{
        E_j(x)
      }"{description}
    ]
    \ar[
      dl,
      "{
        E_i(x)
      }"{description}
    ]
    \ar[
      dr,
      "{
        E_k(x)
      }"{description}
    ]
    \\[-13pt]
    \mathbb{R}^{1+d}
    \ar[
      uur,
      shift left=2pt,
      bend left=10,
      "{
        g_{i j}(x)
      }"{sloped, description}
    ]
    \ar[
      rr,
      shift right=2pt,
      bend right=10,
      "{
        g_{i k}(x)
      }"{description}
    ]
    &&
    \mathbb{R}^{1+d}
    \,.
  \end{tikzcd}
$$
This is closely related to the equivalence of co-frame fields to metric tensors, which we come to in Lem. \ref{OrthonormalCoframesAreEquivalent} and Rem. \ref{GroupoidOfCoFrameFieldsEquivalentToSetOfMetrics} below.
\end{remark}
\begin{notation}[\bf Co-frame components]
With the canonical coordinate projection functions denoted
$$
  \begin{tikzcd}
    \mathbb{R}^{1,d} \simeq 
    \mathbb{R}^{1}
    \times
    \mathbb{R}^{1}
    \times
    \cdots
    \times
    \mathbb{R}^{1}
    \ar[
      rr,
      "{ 
        (-)^a 
      }"
    ]
    &&
    \mathbb{R}
    \,,
  \end{tikzcd}
  \;\;\;\;
  a \in \{0,1, \cdots, d\}
  \,,
$$
we have the usual component-expression of a co-frame field (Def. \ref{CoFrameFields}):
$$
  E^a
  \,\in\,
  \Omega^1_{\mathrm{dR}}\big(
    \CoverOf{X}
  \big)
  \;\;\;\;\;
  \mbox{for}
  \;\;\;\;\;
  \begin{tikzcd}
    T \CoverOf{X}
    \ar[
      r,
      "{ E }"
    ]
    &
    \mathbb{R}^{1,10}
    \ar[
      r,
      "{ (-)^a }"
    ]
    &
    \mathbb{R}
    \,.
  \end{tikzcd}
$$
For instance, in these components the transition function \eqref{TransitionFunction} 
and the transformation property \eqref{TransformationOfCoframe} reads:
$$
  \mathrm{pr}^\ast_2 
  (E^a)
  \;=\;
  g^{a}{}_b 
  (\mathrm{pr}^\ast_1 E^b)
  \,.
$$
\end{notation}
\begin{definition}[\bf Orthonormal local co-frames]
\label{OrthonormalCoframes}
A coframe field (Def. \ref{CoFrameFields}) induces a metric tensor $\mathrm{d}s^2$ on $X$, as the tensor square of the associated 1-form $E$ \eqref{CoFrame1Form} by the formula
\begin{equation}
  \label{MetricInducedByCoFrameField}
  \mathrm{d}s^2 
    \;:=\;
  E^a \otimes E_a
  \;:=\;
  \eta_{a b} E^a \otimes E^b 
  \,.
\end{equation}
By \eqref{TransformationOfCoframe} this tensor descends from the cover $\CoverOf{X}$ to $X$, making it a pseudo-Riemannian manifold $\big(X, \mathrm{d}s^2\big)$.

Conversely, given a (pseudo)-Riemannian metric on $X$, a local co-frame field $E$ is called {\it orthonormal} if it represents that metric, in that the above relation holds (cf. e.g. \cite[p. 14]{Lee18}).
\end{definition}

The following Lem. \ref{ExistenceOfOrthonormalCoframes} is a classical fact, but worth recording as preparation for the construction of Darboux co-frames in Lem. \ref{ExistenceOfDarbouxCoframes}, which is needed for ``super-embeddings'' in Def. \ref{BPSImmersion} to be well-defined.

\begin{lemma}[\bf Existence of orthonormal co-frames]
  \label{ExistenceOfOrthonormalCoframes}
  On a (pseudo-)Riemannian manifold $(X, \mathrm{d}s^2)$ there exists a (relativistic) orthonormal co-frame (Def. \ref{OrthonormalCoframes}).
\end{lemma}
\begin{proof}
  Since $X$ is a smooth manifold,
  we may find an open cover $\widehat X$ of $X$ by coordinate charts, with associated coordinate function
  $
      x
        \,:\,
      \CoverOf{X}
      \xrightarrow{\;}
      \mathbb{R}^{1,d}
  $. This locally induces a canonical co-frame given by the tuple of coordinate differentials
  $\big(\mathrm{d}x^0,\, \mathrm{d}x^1, \cdots, \mathrm{d}x^d \big)$ and a 
  {\it frame} given by the tuple of coordinate vector fields $\big(\partial_0, \partial_1, \cdots, \partial_d\big)$. While these will in general not be orthonormal with respect to $\mathrm{d}s^2$, the pseudo-Riemannian version of the Gram-Schmidt algorithm (e.g. \cite[Lem. 2.24]{ONeil83}, here for matrices with coefficients in $C^\infty(\widehat X)$) produces a local frame that is orthonormal.
  \begin{equation}
    \label{AnOrthonormalFrame}
    \big(
      V_a
      \;:=\;
      E_a^\mu \partial_\mu
    \big)_{a = 0}^d
    ,\,
    \;\;\;\;\;
    \mathrm{d}s^2\big(
      V_a
      ,\,
      V_b
    \big)
    \;=\;
    \eta_{a b}
    \,.
  \end{equation}
  The Gram-Schmidt coefficient matrix is invertible, with inverse to be denoted by shifting its indices, as usual:
  \begin{equation}
    \label{InverseFrameCoefficientMatrix}
    \big(
      E^a_\mu
    \big)_{a,\mu= 0}^d
    \;\;
    :=
    \;\;
    \Big(
    \big(
      E_a^\mu
    \big)_{a,\mu= 0}^d
    \Big)^{-1}
    \,.
  \end{equation}
  These inverse coefficients define the desired orthonormal co-frame
  \begin{equation}
    \label{AnOrthonormalCoframe}
    \big(
    E^a \;:=\;
    E^a_\mu \mathrm{d}x^\mu
    \big)_{a=0}^d
    \,,
    \;\;\;
    \eta_{a b} \, E^a \otimes E^b 
    \;=\;
    \mathrm{d}s^2
  \end{equation}
  due to the orthonormality \eqref{AnOrthonormalFrame} and the invertibility \eqref{InverseFrameCoefficientMatrix}:
  \begin{equation}
    \label{OrthonormalityInComponents}
    \eta_{a b}
    \,
    E^a_\mu
    \,
    E^b_\nu
    \;=\;
    \mathrm{d}s^2_{\mu \nu}
    \;\;\;\;\;\;\;\;
    \Leftrightarrow
    \;\;\;\;\;\;\;\;
    \eta_{a b}
    \;=\;
    E_a^\mu
    \,
    E_b^\nu
    \,
    \mathrm{d}s^2_{\mu \nu}
    \,.
  \end{equation}

  It just remains to verify that \eqref{AnOrthonormalCoframe} satisfies the global conditions on a co-frame from Def. \ref{CoFrameFields}, which amounts to checking that on double overlaps $\CoverOf{X} \!\times_{{}_X}\! \CoverOf{X}$ the transition matrix is orthogonal
  \begin{equation}
    \label{TransitionMatrix}
    \Big(
    g^a{}_{a'}
    \;:=\;
    \big(
    \mathrm{pr}_1^\ast E^a_\mu
    \big)
    \big(
      \mathrm{pr}_2^\ast
      E^\mu_{a'}
    \big)
    \Big)_{a,a' = 0}^d
    \;\;\;
    \in
    \;\;
    C^\infty\big(
      \CoverOf{X}
      ;\,
      \mathrm{O}(1,d)
    \big)
    \,,
  \end{equation}
  which is indeed the case:
  \begin{equation}
    \label{CheckingOrthognalityOfFrameTransitions}
    \def\arraystretch{1.6}
    \begin{array}{lll}
    \eta_{a b}
    \,
    g^{a}{}_{a'}
    \,
    g^{b}{}_{b'}
    &
    \;=\;
    \eta_{a b}
    \Big(
    \big(
    \mathrm{pr}_1^\ast E^a_\mu
    \big)
    \big(
      \mathrm{pr}_2^\ast
      E^\mu_{a'}
    \big)
    \Big)
    \Big(
    \big(
    \mathrm{pr}_1^\ast E^b_{\nu}
    \big)
    \big(
      \mathrm{pr}_2^\ast
      E^\nu_{b'}
    \big)
    \Big)
    &
    \proofstep{
      by
      \eqref{TransitionMatrix}
    }
    \\
    &\;=\;
    \underbrace{
    \eta_{a b}
    \big(
    \mathrm{pr}_1^\ast E^a_\mu
    \big)
    \big(
    \mathrm{pr}_1^\ast E^b_{\nu}
    \big)
    }_{
      \mathrm{d}s^2_{\mu \nu}
    }
    \big(
      \mathrm{pr}_2^\ast
      E^\mu_{a'}
    \big)
    \big(
      \mathrm{pr}_2^\ast
      E^\nu_{b'}
    \big)
    &
    \proofstep{
      by
      \eqref{OrthonormalityInComponents}
    }
    \\[-11pt]
    &\;=\;
    \eta_{a' b'}
    &
    \proofstep{
      by
      \eqref{OrthonormalityInComponents}.
    }
    \end{array}
  \end{equation}

\vspace{-.5cm}
\end{proof}
\begin{lemma}[\bf Essential uniqueness of orthonormal co-frame fields]
\label{OrthonormalCoframesAreEquivalent}
  Any pair of orthonormal frames $E$, $\tilde E$ 
  for the same metric $\mathrm{d}s^2$ is equivalent by a unique transformation 
  \eqref{TransformationOfLocaltrivializations}.
\end{lemma}
\begin{proof}
   We may assume without restriction that the two co-frames are given by differential forms $E, \tilde E :  T\CoverOf{X} \xrightarrow{\;} \mathbb{R}^{1+d}$ with respect to the same cover by coordinate charts 
   $x \,:\, \CoverOf{X} \xrightarrow{\;} \mathbb{R}^{1+d}$
   (otherwise pull them back to the common refinement cover $\CoverOf{X}_1 \!\times_{{}_X}\! \CoverOf{X}_2$ and then further to any coordinate atlas $\CoverOf{X}$ of that). Here, both are expanded in components of the corresponding coordinate frame
   $$
     E^a
     \;=\;
     E^a_\mu \, \mathrm{d}x^\mu
     \,,
     \;\;\;\;\;
     \tilde E^a
     \;=\;
     \tilde E^a_\mu \, \mathrm{d}x^\mu
     \,.
   $$
   By their co-frame property, the coefficient matrices are pointwise invertible with inverses to be denoted by shifting the indices, as usual:
   $$
     (E_a^\mu)_{a,\mu=0}^{d}
     \;:=\;
     \big(
       (E^a_\mu)_{a,\mu=0}^{d}
     \big)^{-1}
     \,.
   $$
   Therefore 
   $$
     \big(
       k^a{}_{a'}
       \,:=\,
       \tilde E^a_\mu
       \,
       E^\mu_{a'}
     \big)_{a,a' = 0}^{d}
      \;\;
      \in
      \;
      C^\infty
      \big(
        \CoverOf{X}
        ,\,
        \mathrm{GL}(1+d)
      \big)
   $$
   is the unique 
   transformation
   $$
     k^a{}_{a'} E^{a'}
     \;=\;
     \tilde E^a
     \hspace{1cm}
     \begin{tikzcd}[row sep=small, column sep=large]
       \mathbb{R}^{1+d}
       \ar[
         rr,
         "{ k }"
       ]
       &&
       \mathbb{R}^{1+d}
       \\
       &
       T_x X
       \ar[
         ul,
         "{
           E(x)
         }"
       ]
       \ar[
         ur,
         "{
           \tilde E(x)
         }"
         {swap}
       ]
     \end{tikzcd}
   $$
   and the fact that this is an orthogonal transformation
   $
     \begin{array}{l}
       \eta_{a b}
       \,
       k^a{}_{a'}
       \,
       k^{b}{}_{b'} = \eta_{a' b'} 
     \end{array}
   $, follows as in \eqref{CheckingOrthognalityOfFrameTransitions}. 
\end{proof}
\begin{remark}[\bf Groupoid of co-frame fields equivalent to set of metrics]
\label{GroupoidOfCoFrameFieldsEquivalentToSetOfMetrics}
The pair of Lemmas \ref{ExistenceOfOrthonormalCoframes} and \ref{OrthonormalCoframesAreEquivalent} may jointly be understood as saying that the functor from the groupoid of (relativistic) co-frame fields $E$ to the set of (pseudo-)Riemannian metric tensors $\mathrm{d}s^2$ (on a given smooth manifold $X$) is essentially surjective (in fact surjective) and fully faithful, hence is an equivalence. 
\end{remark}

\begin{remark}[\bf Lifting smooth maps to covers]
\label{LiftingSmoothMapsToCovers}
The notion of equivalence on the data in Def. \ref{CoFrameFields} ensures that the specific choice of open cover $\CoverOf{X}$ is irrelevant --- which justifies our notation suggestive of {\it any one} open cover of $X$. 
But to lift a smooth map
$
  \Sigma
  \xrightarrow{\;\;}
  X
$
to a map between chosen covers
\begin{equation}
  \label{MapLiftedToOpenCovers}
  \begin{tikzcd}[row sep=15pt, column sep=large]
    \CoverOf{\Sigma}
    \ar[
      rr,
      "{ \CoverOf{\phi} }"
    ]
    \ar[
      d,
      ->>
    ]
    &&
    \CoverOf{X}
    \ar[
      d,
      ->>
    ]
    \\
    \Sigma
    \ar[
      rr,
      "{ \phi }"
    ]
    &&
    X
  \end{tikzcd}
\end{equation}
the cover $\CoverOf{\Sigma}$ needs to be fine enough, relative to the given $\CoverOf{X}$. This can always be achieved. Given $\phi$, the canonical choice for the cover of $\Sigma$  is the pullback $\CoverOf{\Sigma} := \Sigma \times_{X} \CoverOf{X}\cong \textstyle{
    \coprod_{j \in J}
  }
  \,
  \phi^{-1}(U_j)$ of the cover on $X$, hence the case where \eqref{MapLiftedToOpenCovers} is a Cartesian square.
\footnote{The general way of dealing with these matters is to work with (model categories of) ``smooth $\infty$-stacks'', where the issue of passing to covers is reflected in the notion of cofibrant resolutions. The interested reader may find these  methods concisely reviewed in \cite[\S 1]{FSS23Char}, but for the present purpose the above considerations are sufficient.}
\end{remark}

\medskip

\noindent
{\bf Immersions.}
The notion of {\it isometric embeddings} of Riemannian manifolds into each other is of course a classical one, with the ``isometric embedding problem'' -- namely the task of finding isometric embeddings of given abstract Riemannian manifolds into large-dimensional but flat Euclidean spaces -- being a seminal problem in the field of Riemannian geometry (cf. \cite{HanLewicka23}).
The key observation of the ``super-embedding approach'' to super $p$-branes is that certain  immersions of supermanifolds are considerably richer than their bosonic counterpart might suggest, in that their super-odd component may encode extra data (Rem. \ref{ExistenceOfDarbouxCoframes}) of {\it differential forms} on the embedded submanifold, playing the role of flux densities of higher gauge fields appearing on these brane worldvolumes.

\begin{remark}[\bf Immersions vs. embeddings]
\label{ImmersionsAndEmbeddings}
In developing this here with mathematical precision, we start by noticing that it is not really embeddings but {\it immersions} 
\begin{equation}
  \label{ImmersionDiagram}
  \begin{tikzcd}[column sep=80]
    T \Sigma
    \ar[
      rr,
      "{
        \mathrm{d}
        \phi
      }",
      "{
        \scalebox{.6}{
          \color{gray}
          fiberwise injection
        }
      }"{swap, pos=.44}
    ]
    \ar[
      dd
    ]
    \ar[
      dr,
      hook
    ]
    &[-70pt]
    &
    T X
    \ar[
      dd
    ]
    \\[-10pt]
    &
    \Sigma
      \!\times_{{}_X}\!
    T X
    \ar[
      ur,
      end anchor={[yshift=-3pt]}
    ]
    \ar[dl]
    \ar[
      dr,
      phantom,
      "{
        \scalebox{.65}{\color{gray}(pb)}
      }"{pos=.2}
    ]
    \\[-3pt]
    \mathllap{
      \scalebox{.7}{
        \color{darkblue}
        \bf
        \def\arraystretch{.9}
        \begin{tabular}{c}
          Immersed manifold
          \\
          (brane worldvolume)
        \end{tabular}
      }
    }
    \Sigma
    \ar[
      rr, 
      "\phi",
      "{
        \scalebox{.7}{
          \color{darkgreen}
          \bf
          \def\arraystretch{.9}
          \begin{tabular}{l}
            immersion (colloq.:
            \\
            ``embedding field'')
          \end{tabular}
        }
      }"{swap}
    ]
    &&
    X
    \mathrlap{
      \scalebox{.7}{
        \color{darkblue}
        \bf
        \def\arraystretch{.9}
        \begin{tabular}{c}
          Ambient manifold
          \\
          (bulk/target space)
        \end{tabular}
      }
    }
  \end{tikzcd}
\end{equation}
that are relevant here --- recalling (e.g. \cite[\S III.4]{Boothby75}) that an embedding of smooth manifolds is 
\begin{itemize}[
  leftmargin=.8cm,
  topsep=1pt,
  itemsep=2pt
]
\item[\bf (i)] an immersion --- namely a smooth map $\phi$ whose differential $\mathrm{d}\phi$ is fiberwise an injection of tangent spaces \eqref{ImmersionDiagram},
\item[\bf (ii)] which in addition is a homeomorphism onto its topological image. 
\end{itemize}
These are non-degeneracy conditions on $\phi$, locally and globally:
The first condition is {\it local} and translates, 
as we will see, to differential equations on $\phi$,
but the second condition is to rule out {\it global} degeneracies of $\phi$, such as points in target space where two distinct points of the embedded manifold touch.

\noindent
Strikingly, it is the differential equations of (i) which, in the super-geometric situation below, translate to the equations of motion of super $p$-branes --- this is the phenomenon of interest here. 
But no global constraints (ii) on $p$-brane dynamics are meant to be imposed. Therefore ``super-embedding approach'' is a little bit of a misnomer -- what the literature is really concerned with is both weaker and stronger than super-embeddings: weaker because only {\it super-immersions} are required, and stronger because these immersions are required to ``preserve half of the local supersymmetry'', hence to be ``$\sfrac{1}{2}$BPS'' (e.g. \cite{DEGKS08}). 

Therefore we speak of ``$\sfrac{1}{2}$BPS super-immersions'' (Def. \ref{BPSImmersion} below, see Rem. \ref{SuperEmbeddingConditionInTheLiterature} for relating back to the ``super-embedding'' terminology). 
\end{remark}

\medskip

\noindent
{\bf Darboux co-frames for immersions.} We observe in \S\ref{SuperEmbeddingConstruction} that what in the super-$p$-brane literature came to be known as the ``super-embedding condition'' is what in classical differential geometry is known as the characterization of {\it Darboux coframes} adapted to immersions \eqref{ImmersionDiagram}. Therefore we here first dwell a little on Darboux coframes over ordinary manifolds.
The following Def. \ref{DarbouxCoFrame} of {\it Darboux co-frames} may be found in \cite[p. 246 (2.11)]{Sternberg64}\cite[(1.13)]{GriffithsHarris79}\cite[p. 426]{Zandi88}\cite[Def. 1.17]{MRS12}\cite[\S 3]{Giron20}; it is the evident higher-dimensional generalization of {\'E}. Cartan's characterization of embedded surfaces by adapted coframes
(cf. \cite[p. 211]{Cartan26}) and the evident dualization of the notion of  {\it Darboux frames} \cite[p. 244 Def. 2.1]{Sternberg64}\cite[(1.12)]{GriffithsHarris79}\cite[p. 818]{BergerBryantGriffiths83}, which in turn are the evident higher-dimensional generalization of the original Darboux frames used in the differential geometry of curves and surfaces embedded into Euclidean 3-space (e.g. \cite[p. 210]{Guggenheimer77}\cite[p. 107]{PetruninBarrera20}).

\begin{notation}[\bf Tangential and transversal components]
\label{TangentialAndTransversal}
Given an immersion $\phi : \Sigma  \xhookrightarrow{\;} X$ of smooth manifolds of dimensions $1+p \leq 1+d \,\in\,\mathbb{N}$, respectively, we write
\begin{equation}
  \label{ProjectorOntoFirstCoordinates}
  \begin{tikzcd}[
    row sep=0pt
  ]
    P
    :
    &[-30pt]
    \mathbb{R}^{1+d}
    \ar[
      r,
      ->>
    ]
    &
    \mathbb{R}^{1+p}
    \ar[
      r,
      hook
    ]
    &
    \mathbb{R}^{1,d}
    \hspace{1cm}
    \scalebox{.7}{
      \color{darkblue}
      \bf
      Tangential projector
    }
  \end{tikzcd}
\end{equation}
for the linear projector onto the first $1+p$ coordinate axes in the corresponding local model space, and 
\begin{equation}
\label{ProjectorOntoLastCoordinates}
  \begin{tikzcd}[
    row sep=0pt
  ]
    \overline{P}
    :
    &[-30pt]
    \mathbb{R}^{1+d}
    \ar[
      r,
      ->>
    ]
    &
    \mathbb{R}^{d-p}
    \ar[
      r,
      hook
    ]
    &
    \mathbb{R}^{1,d}
    \hspace{1cm}
    \scalebox{.7}{
      \color{darkblue}
      \bf
      Transversal projector
    }
  \end{tikzcd}
\end{equation}
for the complentary projector onto the last $d-p$ coordinate axes.

Given moreover a co-frame field   \begin{tikzcd}
    T \CoverOf{X}
    \ar[
      r,
      "{  E }"
    ]
    &
    \mathbb{R}^{1+d}
  \end{tikzcd}
(Def. \ref{CoFrameFields})
and a lift  $\CoverOf{\phi} : \CoverOf{\Sigma} \xrightarrow{\;} \CoverOf{X}$ (Rem. \ref{LiftingSmoothMapsToCovers})
we denote by $e := P \circ E \circ \mathrm{d}\CoverOf{\phi}$ the pullback of $E$ along $\CoverOf{\phi}$ to $\CoverOf{\Sigma}$, post-composed with the projection operator \eqref{ProjectorOntoFirstCoordinates}
  \begin{equation}
    \label{PullbackProjectedFrameField}
    e 
    \,:=\,
    P \circ E \circ \mathrm{d}\CoverOf{\phi}
    \,:\, 
    \begin{tikzcd}
      T \CoverOf{\Sigma}
      \ar[
        rr,
        "{ \mathrm{d}\CoverOf{\phi} }"
      ]
      &&
      T \CoverOf{X}
      \ar[
        rr,
        "{ E }"
      ]
      &&
      \mathbb{R}^{1+d}
      \ar[r,->>]
      \ar[
        rr,
        bend left=15pt,
        "{ P }"
      ]
      &
      \mathbb{R}^{1+p}
      \ar[r, hook]
      &
      \mathbb{R}^{1+d}
      \,.
    \end{tikzcd}
  \end{equation}

  \newpage 
  Notationally, this means that $e$ may still be regarded as carrying an index ranging through both tangential and transverse directions, while it just happens to vanish on all transverse indices:
  $$
    e^a \;=\;
    \left\{\!\!
    \def\arraystretch{1}
    \begin{array}{cl}
      \phi^\ast E^a 
      &
      \mbox{
        for tangential $a$,
      }
      \\
      0 
      &
      \mbox{
        for transversal $a$.
      }
    \end{array}
    \right.
  $$
\end{notation}

This terminology is adapted to the following situation:
\begin{definition}[\bf Darboux co-frame fields]
\label{DarbouxCoFrame}
For $(X,\mathrm{d}s^2)$ a smooth (pseudo-)Riemannian manifold and 
$\phi : \Sigma \xrightarrow{\;} X$ an immersion \eqref{ImmersionDiagram} of a smooth manifold $\Sigma$, then an orthonormal co-frame field $E$ (Def. \ref{OrthonormalCoframes}) is called {\it adapted} or {\it Darboux} for $\phi$ if, in the terminology of Ntn. \ref{TangentialAndTransversal}
\begin{equation}
  \label{PullbackOfTranverseDarbouxFrameComponentsVanishes}
  \phi^\ast E^a \;=\; 0
  \;\;\;
  \mbox{
    for transversal $a$
  }
  \hspace{1cm}
  \Leftrightarrow
  \hspace{1cm}
  \phi^\ast 
  \overline{P} E
  \;=\;
  0
  \,.
\end{equation}
\end{definition}
\begin{remark}[\bf Coframe field implied by Darboux condition]
\label{TheDarbouxCondition}
Due to the split short exact sequence of vector spaces
$$
  \begin{tikzcd}[column sep=large]
    0
    \ar[r]
    &[-10pt]
    \mathbb{R}^{1+p}
    \ar[r, hook]
    &
    \mathbb{R}^{1+d}
    \ar[
      r,
      ->>,
      "{
        \overline{P}
      }"
    ]
    \ar[
      l,
      ->>,
      shift right=2pt,
      bend right=15pt,
      "{
        P
      }"{swap}
    ]
    &
    \mathbb{R}^{d-p}
    \ar[r]
    &[-10pt]
    0
  \end{tikzcd}
$$
the Darboux-condition \eqref{PullbackOfTranverseDarbouxFrameComponentsVanishes} implies that the projected pullback $e$ of $E$ \eqref{PullbackProjectedFrameField}
is a co-frame field on $\Sigma$ (cf. \cite[(1.13)]{GriffithsHarris79}):
\begin{equation}
  \label{TheEmbeddingCondition}
  \mathllap{
  \underset{
    \sigma 
      \in 
    \CoverOf{\Sigma}
  }{\forall}
  \;\;\;\;\;
  }
  \begin{tikzcd}[
    row sep=15pt, column sep=large
  ]
    &&&
    \mathbb{R}^{1+p}
    \\[4pt]
    T_{\sigma} 
    \CoverOf{\Sigma}
    \ar[
      r,
      hook,
      "{
        \mathrm{d}\CoverOf{\phi}
      }"
    ]
    \ar[
      drrr,
      rounded corners,
      to path={
           ([yshift=+00pt]\tikztostart.south)  
       -- ([yshift=-26pt]\tikztostart.south)
       -- node[yshift=6pt]{
           \scalebox{.7}{
             $
               0 
               \;=\;
               \phi^\ast 
               \overline{P}
               E
             $
           }
       }
         ([xshift=+0pt]\tikztotarget.west)
      }
    ]
    \ar[
      urrr,
      rounded corners,
      to path={
           ([yshift=+00pt]\tikztostart.north)  
       -- ([yshift=+25pt]\tikztostart.north)
       -- node[yshift=+6pt]{
           \scalebox{.7}{
             $
               e 
               \;=\;
               \phi^\ast 
               P
               E
             $
           }
          }
          node[yshift=-4pt]{
            \scalebox{.7}{
             $\sim$
        }
         }
         ([xshift=+0pt]\tikztotarget.west)
      }
    ]
    &
    T_{\phi(\sigma)}   \CoverOf{X}
    \ar[
      r,
      "{ E }",
      "{ \sim }"{swap}
    ]
    &
    \mathbb{R}^{1+d}
    \ar[
      r,
      "{ \sim }"{swap}
    ]
    &
    \mathbb{R}^{1+p}
    \times
    \mathbb{R}^{d-p}
    \ar[
      u,
      "{ P }"{swap}
    ]
    \ar[
      d,
      "{ \overline{P} }"
    ]
    \\
    &&&
    \mathbb{R}^{d-p}
    \,.
  \end{tikzcd}
\end{equation}
Notice that this situation of Darboux co-frames:
\begin{itemize}[leftmargin=.8cm]
\item[\bf (i)]
is just what was eventually called the ``embedding condition'' in the ``super-embedding''-literature, cf. \cite[(2.6-9)]{Bandos11}\cite[(5.13-14)]{BandosSorokin23} -- noticing that this is crucially stronger than just the top part of \eqref{TheEmbeddingCondition} which was  the original ``geometrodynamical condition'' of \cite[(2.23)]{BPSTV95};
\item[\bf (ii)]
justifies the terminology  ``tangential'' and ``transversal'' in Ntn. \ref{TangentialAndTransversal}, because with a Darboux co-frame given, the co-frame fields $E^a$ at $\phi(\Sigma) \subset X$ carrying a tangential or transversal index according to 
\eqref{TangentialTransversalDecomposition}
are exactly those which are tangential or transversal to the immersed manifold $\Sigma$, respectively. 
\end{itemize}
\end{remark}

\begin{definition}[\bf Pseudo-Riemannian immersion]
  \label{NonDegenerateImmersion}
  We say that an immersion
  $\phi : \Sigma \xhookrightarrow{\;} X$ into a pseudo-Riemannian manifold $(X,\mathrm{d}s^2)$ is itself {\it pseudo-Riemannian}  if the pullback form $\phi^\ast \mathrm{d}s^2$ is still a pseudo-Riemannian metric.
\end{definition}

\begin{lemma}[\bf Existence of Darboux co-frames]
\label{ExistenceOfDarbouxCoframes}
Given a pseudo-Riemannian immersion $\phi : \Sigma \xhookrightarrow{\;} X$ (Def. \ref{NonDegenerateImmersion}) into a pseudo-Riemannian manifold $(X,\mathrm{d}s^2)$, 
then a Darboux co-frame field (Def. \ref{CoFrameFields}) exists.
\end{lemma}
\begin{proof}
  On the complement of $\Sigma$ in $X$ the Darboux condition is trivial and we may use the construction of general orthonormal co-frame fields from Lem. \ref{ExistenceOfOrthonormalCoframes}. By the argument there, what remains is just to construct a Darboux co-frame locally on open neighborhoods around each point $\phi(\sigma)$ for $\sigma \in \Sigma$. 

  Now by classical facts: There exists an open neighborhood $U_\sigma$, $\phi(\sigma) \in U_\sigma \subset X$, where the immersion $\phi$ restricts to an embedding of the manifold $\phi(\Sigma) \cap U$ (e.g. \cite[Thm. 4.12]{Boothby75}), and there exists a further open neighborhood $U'_\sigma$, $\phi(\sigma) \in U'_\sigma \subset U_\sigma \subset X$ carrying a ``slice chart'' $x_\sigma \,:\,U'_\sigma \hookrightarrow \mathbb{R}^{1,d}$ for $\phi$ that identifies $\phi(\Sigma) \cap U'$ with a rectilinear hyperplane in an open subset of $\mathbb{R}^{1,d}$ (e.g. \cite[Thm. 5.8]{Lee12}).
  
  This implies that as we apply the Gram-Schmidt process \eqref{AnOrthonormalFrame} to this slice coordinate frame $x_\sigma$, the coefficient matrix $\big(E_a^\mu\big)_{a,\mu \in 0}^d$ is {\it block-diagonal}, and hence so is its inverse $\big(E^a_\mu\big)_{a,\mu \in 0}^d$ \eqref{InverseFrameCoefficientMatrix}. 
  But this means that the corresponding
  $
    E 
    \;:=\;
    E^\bullet_\mu \, \mathrm{d}x^\mu
  $ \eqref{AnOrthonormalCoframe}
  satisfies the Darboux property \eqref{PullbackOfTranverseDarbouxFrameComponentsVanishes}, since
  $
    \big(\phi^\ast E^a\big)(\partial_\mu)
    \;=\;
    E^a_\mu
  $
  vanishes when $a$ and $\mu$ are not in the same block, with one being transversal and the other tangential.
\end{proof}

\medskip

\noindent
{\bf Second fundamental form.}
Now given a pseudo-Riemannian immersion $\phi : \Sigma \xhookrightarrow{\;} X$ into a pseudo-Riemannian manifold and an adapted choice of Darboux co-frame field $E$ on $X$ (via Lem. \ref{ExistenceOfDarbouxCoframes}) with respect to some open cover $\CoverOf{X} \twoheadrightarrow X$, let 
$$
  \Omega
  \;\in\;
  \Omega^1_{\mathrm{dR}}
  \big(
    \CoverOf{X}
    ;\,
    \mathfrak{so}(1,d)
  \big)
  \;\;\;\;\;\;
  \mbox{with components}
  \;\;\;\;\;\;
  \big(
    \Omega^{a b}
    \defneq
    -
    \Omega^{b a} \big)_{a,b = 0}^{d}   
    \;\in\;
    \Omega^1_{\mathrm{dR}}(\CoverOf{X})
$$

\newpage 
\noindent be the unique torsion-free connection for $E$, in that
\begin{equation}
  \label{BosonicTorsionConstraint}
  \mathrm{d}
  E^a
  +
  \Omega^a{}_b\, E^b
  \;=\;
  0
  \,.
\end{equation}
Denote the pullback of the tangential and transversal components of this connection, respectively, by:
\begin{equation}
  \label{PullbackOfConnectionForm}
  \def\arraystretch{1.2}
  \def\arraycolsep{3pt}
  \begin{array}{cccl}
    \omega^a{}_b
    &:=&
    \phi^\ast
    \Omega^a{}_b
    &
    \scalebox{.85}{
      \def\arraystretch{.9}
      \begin{tabular}{l}
      for tangential $a$ 
      \\
      and tangential $b$,
      \end{tabular}
    }
    \\[+7pt]
    e^{b_1}
    \SecondFundamentalForm
      ^a
      _{b_1 b_2}
    &:=&
    \phi^\ast \Omega^a{}_{b_2}
    &
    \scalebox{.85}{
      \def\arraystretch{.9}
      \begin{tabular}{l}
      for transversal $a$ 
      \\
      and tangential $b_2$.
      \end{tabular}
    }
  \end{array}
\end{equation}
Then the Darboux-condition on $E$ \eqref{TheEmbeddingCondition} implies that the pullback of the torsion constraint \eqref{BosonicTorsionConstraint} to $\Sigma$ is equivalent to the following two equations:
\begin{equation}
  \label{PullbackOfBosonicTorsionConstraint}
  \phi^\ast
  \big(
    \mathrm{d}
    E^a
    +
    \Omega^a{}_b \, E^b
    \;=\;
    0
  \big)
  \quad 
    \Leftrightarrow
  \quad
  \left\{\!\!\!
  \def\arraystretch{1.4}
  \begin{array}{ll}
    \mathrm{d}
    e^a
    +
    \omega^a{}_b 
    \,
    e^b
    \;=\;
    0
    &
    \mbox{for tangential $a$}
    \\
    \mbox{I\hspace{-1.5pt}I}^a_{b_1 b_2}
    e^{b_1} e^{b_2}
    \;=\;
    0
    &
    \mbox{for transversal $a$.}
  \end{array}
  \right.
\end{equation}
Here the first line just says that $\omega$ is the torsion-free connection for $e$ on $\Sigma$, while the second line says that 
\begin{equation}
  \label{SecondFundamentalFormIsSymmetric}
  \SecondFundamentalForm^a_{b_1 b_2}
  \;=\;
  \SecondFundamentalForm^a_{b_2 b_1}
\end{equation}
is a symmetric tensor on $\Sigma$. As such this is historically known as the {\bf second fundamental form} of the immersion $\phi$ (e.g., \cite[p. 819]{BergerBryantGriffiths83}\cite[(II.2.12)]{Chavel93}).

\medskip

\noindent
{\bf Reformulating the Darboux condition.} 
We now re-formulate the Darboux coframe condition (Def. \ref{DarbouxCoFrame}) in a form that generalizes naturally to $\sfrac{1}{2}$BPS super-immersions.

\begin{proposition}[\bf The Darboux condition reformulated]
  \label{ReformulatingTheDarbouxCondition}
  Let $(X, \mathrm{d}s^2)$ be a smooth (pseudo-)Riemannian manifold of dimension $1 + p$ 
  and $\phi : \Sigma \xrightarrow{\;} X$ an immersion 
  \eqref{ImmersionDiagram}
  of a smooth manifold.
  Then a relativistic local coframe field $E : T \CoverOf{X} \xrightarrow{\;} \mathbb{R}^{1+d}$ on $X$ (Def. \ref{OrthonormalCoframes}) is Darboux for $\phi$ (Def. \ref{DarbouxCoFrame})
  if and only if there exists
  $$
    \mathrm{Sh}
    \,\in\,
    C^\infty\Big(
      \CoverOf{\Sigma}
      ;\;
      \mathrm{Hom}_{\mathrm{Vect}}
      \big(
        {P}(\mathbb{R}^{1+d})
        ,\,
        \overline{P}
        (\mathbb{R}^{1+d})
      \big)
    \Big)
    \,,
    \;\;\;
    \mbox{
      \rm
      $P,\overline{P}$
      as in 
      Ntn. \ref{TangentialAndTransversal}
    }
  $$
  such that
  \begin{itemize}[
    leftmargin=.9cm,
    topsep=1pt,
    itemsep=3pt
  ]
    \item[\bf (a)]
    $\phi^\ast( P E' )$
    is a relativistic local coframe field on $\Sigma$ {\rm (Def. \ref{CoFrameFields})},
    \item[\bf (b)]
    $\phi^\ast\big(\overline{P} E'\big)
    \,=\,
    \mathrm{Sh} \cdot \phi^\ast(P E')
    \,,
    $
  \end{itemize}
  for all local coframe fields $E'$ on $X$ that are in the same transversal gauge-orbit as $E$:
  $$
    E' 
    \;=\;
    U \cdot E
    \,,
    \;\;\;
    \mbox{\rm for}
    \;\;\;
    U
    \;\in\; 
    C^\infty\big(
      \CoverOf{X}
      ;\,
      \mathrm{O}
      \big(
        \overline{P}\mathbb{R}^{1+d}
      \big)
      \,
    \big)
    \,.
  $$  
\end{proposition}
This is an elementary argument, and yet the implications are somewhat profound (cf. Rem. \ref{OutlookOnSuperDarboux} below):
\begin{proof}
  Noticing that by assumption that
  \begin{equation}
    \label{CommutatorPwithU}
    P \circ U
    \;=\;
    P
    \;\;\;
    \mbox{and}
    \;\;\;
    \overline{P} \circ U
    \;=\;
    U \circ \overline{P}
    \,,
  \end{equation}  
  the key point is that the second condition equivalently says that $\mathrm{Sh}$ takes values in $\mathrm{O}\big( \overline{P}(\mathbb{R}^{1+p}) \big)$-invariant maps:
  $$
    \def\arraystretch{1.5}
    \begin{array}{cll}
      &
      \phi^\ast\big(
        \overline{P}
        U E
      \big)
      \;=\;
      \mathrm{Sh}
      \cdot
      \phi^\ast\big(
        P U E
      \big)
      \\
      \Leftrightarrow
      &
      \phi^\ast\big(
        
        U \overline{P} E
      \big)
      \;=\;
      \mathrm{Sh}
      \cdot
      \phi^\ast\big(
        P E
      \big)      
      &
      \proofstep{
        by
        \eqref{CommutatorPwithU}
      }
      \\
      \Leftrightarrow
      &
      \phi^\ast\big(
        \overline{P}
        E
      \big)
      \;=\;
      \phi^\ast(U)^{-1}
      \cdot 
      \mathrm{Sh}
      \cdot
      \phi^\ast\big(
        P E
      \big)      
      \,.
    \end{array}
  $$
  But since the only fixed point of $\mathrm{O}\big(\overline{P}\mathbb{R}^{1+p}\big)$ is the origin $0 \in \overline{P} \mathbb{R}^{1+d}$ this implies that 
  \begin{equation}
    \label{BosonicShearingMapVanishes}
    \mathrm{Sh}
    \;=\;
    0
    \,,
  \end{equation}
  whereby the second condition above is equivalently the Darboux condition $\phi^\ast\big(\overline{P} E\big) = 0$ \eqref{PullbackOfTranverseDarbouxFrameComponentsVanishes}, whence the first condition follows by Rem. \ref{TheDarbouxCondition}.
\end{proof}
\begin{remark}[\bf Outlook on the supergeometric generalization]
  \label{OutlookOnSuperDarboux}
 $\,$

  \noindent {\bf (i)} The formulation of the Darboux condition in Prop. \ref{ReformulatingTheDarbouxCondition} makes immediate sense also for super-immersion (recalled as Def. \ref{ImmersionOfSuperManifolds} below),
  but in this generality the strong implication \eqref{BosonicShearingMapVanishes} turns out to be relaxed. 
  
  \vspace{1mm} 
  \noindent {\bf (ii)} It is this extra freedom in choosing a {\it shear map} $\mathrm{Sh}$ expressing the pullback of the transversal super-coframe in terms of the tangential super-coframe 
  which becomes the source of higher gauge fields on super $p$-branes (discussed in \S\ref{BPSSuperImmersions}), in particular of the B-field on M5-branes (discussed in \S\ref{SuperEmbeddingConstruction}).
\end{remark}

\newpage 
\subsection{$\sfrac{1}{2}$BPS Super-immersions}
\label{BPSSuperImmersions}

Now we pass to super-geometry and specifically to super-spacetimes and their (higher) super Cartan geometry. Our notation follows \cite[\S 2]{GSS24-SuGra} to which we refer the reader for review, references and further discussion.

\medskip

\noindent
{\bf Supergeometric generalization by internalization.}
Many differential-geometric concepts generalize straightfowardly to supergeometry, by just interpreting their algebraic formulation verbatim in superalgebra (a general process called ``internalization'' in category theory). This is the case for instance for the notion of immersions \eqref{ImmersionDiagram}:
\begin{definition}[{\bf Super-immersion}, {e.g. \cite[above Thm. 4.4.3]{Varadarajan04}}]
  \label{ImmersionOfSuperManifolds}
  A map of smooth super-manifolds $\phi : \Sigma \xrightarrow{\;} X$ is an {\it immersion} if its differential at each point $\sigma \in \bosonic\Sigma \xhookrightarrow{\;} \Sigma$ is an injective map of super-vector spaces
  $$
      T_{\sigma} \Sigma
        \xhookrightarrow{\; \mathrm{d}\phi \;}
      T_{\phi(\sigma)} X
      \,.
  $$
\end{definition}

\smallskip
\noindent
{\bf On the general nature of super-Darboux coframe fields.}
However, the case of Lorentzian (super-)spacetimes is a little different: The {\it verbatim} generalization of co-frame structure (Def. \ref{CoFrameFields})
to supergeometry modeled on $\mathbb{R}^{1,d\vert \mathbf{N}}$  
would ask for a reduction of the structure group to the ortho-symplectic supergroup 
$$
  \mathrm{OSp}\big(
    1,d \,\vert\, \mathbf{N} 
  \big)
  \;\defneq\;
  \mathrm{O}\big(
    \mathbb{R}^{1,d\vert \mathbf{N}}
  \big)
  \xhookrightarrow{\quad}
  \mathrm{GL}\big(
    \mathbb{R}^{1,d\vert \mathbf{N}}
  \big)
  \,.
$$
For this notion the above discussion of Darboux co-frames would generalize verbatim.
But instead, for Lorentzian super-spacetimes one asks, of course, for {\it further} reduction to just the Spin-group (the ``external automorphisms'' of $\mathbb{R}^{1,d\vert \mathbf{N}}$, cf. \cite[Prop. 6]{HS18}):

\vspace{-.5cm}
\begin{equation}
  \label{IteratedReductionOfStructureGroups}
  \begin{tikzcd}[
    row sep=-4pt
  ]
  \mathrm{Spin}(1,d)
  \ar[r, hook]
  &
  \mathrm{O}\big(\mathbb{R}^{1,d\vert \mathbf{N}}\big)
  \ar[r, hook]
  &
  \mathrm{GL}\big(
    \mathbb{R}^{1,d \vert \mathbf{N}}
  \big)
  \\
  \mathclap{
  \scalebox{.7}{
    \color{darkblue}
    \bf
    \def\arraystretch{.9}
    \begin{tabular}{c}
      Spin-group
    \end{tabular}
  }
  }
  &
  \mathclap{
  \scalebox{.7}{
    \color{darkblue}
    \bf
    \def\arraystretch{.9}
    \begin{tabular}{c}
      ortho(-symplectic)
      \\
      super-group
    \end{tabular}
  }
  }
  &
  \mathclap{
  \scalebox{.7}{
    \color{darkblue}
    \bf
    \def\arraystretch{.9}
    \begin{tabular}{c}
      general-linear
      \\
      super-group
    \end{tabular}
  }
  }
  \end{tikzcd}
\end{equation}
which means that one is now dealing with co-frames for stronger $G$-structures.
Therefore, the general existence proof 
(Prop. \ref{ExistenceOfDarbouxCoframes})
for Darboux coframe fields 
does not pass to super-immersions into super-spacetimes. 

\medskip

\noindent
{\bf BPS Super-immersions.}
Instead, the existence of Darboux coframes now becomes a {\it condition} on the super-immersion. This is essentially the condition known in the literature as ``super-embedding'' (cf. Rem. f\ref{SuperEmbeddingConditionInTheLiterature}). Since it is not really about embeddings but about immersions (by Rem. \ref{ImmersionsAndEmbeddings}), and here specifically those that preserve ``half of the local supersymmetry'', and since we will streamline the definition a little, we shall instead speak of {\it $\sfrac{1}{2}$BPS super-immersions}.

\smallskip

\noindent
To that end, consider
  \begin{itemize}[
    leftmargin=.8cm,
    topsep=1pt,
    itemsep=2pt
  ]
  \item[\bf (i)]
  $X$ a super-spacetime (cf. \cite[\S 2.2.2]{GSS24-SuGra}) locally modeled on a Minkowski super-space $\mathbb{R}^{1,d \vert \mathbf{N}}$ for a real $\mathrm{Pin}(1,d)$-representation $\mathbf{N}$ with canonical Clifford generators $\big(\Gamma_a : \mathbf{N} \to \mathbf{N}\big)_{a = 0}^d$, 
  \item[\bf (ii)]
  $p \leq d$ such that 
  \begin{equation}
    \label{TheProjectorInGeneral}
    P
    \;:=\;
    \tfrac{1}{2}
    \big(
      1
      \,+\,
      \Gamma_{p+1}
      \cdot
      \Gamma_{p+2}
      \cdots
      \Gamma_{d}
    \big)
    \;:\;
    \mathbb{R}^{1,d \,\vert\, \mathbf{N}}
    \xrightarrow{\quad}
    \mathbb{R}^{1,d\vert \mathbf{N}}
  \end{equation}
  is a projector ($P \circ P = P$) with complementary projector denoted $\overline{P} := 1 - P$,
  \item[\bf (iii)]
  $\Sigma$ a super-manifold locally modeled on $\mathbb{R}^{1,p \,\vert P(\mathbf{N})}$.
\end{itemize}
and notice that then the action of $\mathrm{Spin}(d-p) \xhookrightarrow{\;} \mathrm{Spin}(1,p) \times \mathrm{Spin}(d-p) \hookrightarrow \mathrm{Spin}(1+d)$ on $\mathbf{N}$ evidently commutes with $P$ and with $\overline{P}$, which allows to regard
\begin{equation}
  \label{ResidualTransverseSpinReps}
  P(\mathbf{N})
  ,\,
  \overline{P}(\mathbf{N})
  \;\;
  \in
  \;\;
  \mathrm{Rep}_{\mathbb{R}}\big(
    \mathrm{Spin}(d-p)
  \big)
  \,.
\end{equation}

\begin{definition}[\bf \sfrac{1}{2}BPS super-immersion]
  \label{BPSImmersion}

  \noindent
  In the above situation, we call a super immersion   
  $
    \phi : \Sigma \xhookrightarrow{\;} X
  $
  (Def. \ref{ImmersionOfSuperManifolds})
  a {\it \sfrac{1}{2}BPS immersion} if
  it admits the super-analog of a Darboux co-frame field in the form Prop. \ref{ReformulatingTheDarbouxCondition}, namely if there exists an orthonormal super co-frame field $(E,\Psi)$ on $X$
  which is\footnote{
    Beware that, by the discussion around \eqref{IteratedReductionOfStructureGroups}, a choice of such super Darboux co-frames is, even locally, more than just a choice of Lorentz- (i.e.: Spin-)gauge: In contrast to the bosonic situation, not every local super-coframe field need to be Spin-gauge equivalent to a super-Darboux coframe field on super-spacetime. 
  } {\it super-Darboux} for $\phi$
  in that there is a ``super-shear map''
  \begin{equation}
    \label{SuperShearMap}
    \mathrm{Sh}
    \,\in\,
    C^\infty\Big(
      \CoverOf{\Sigma}
      ;\,
      \mathrm{Hom}_{\mathbb{R}}
      \big(
        {P}
        (\mathbb{R}^{1,d\vert \mathbf{N}})
        ,\,
        \overline{P}
        (\mathbb{R}^{1,d\vert \mathbf{N}})
      \big)
    \Big)
  \end{equation}
  such that
  \begin{align}
    \label{PullbackOfTransversalCoFrameIsCoFrame}
    \mbox{\bf (a)}
    \hspace{.3cm}
    &
    (e',\psi')
    \;:=\;
    \phi^\ast \big(P (E',\Psi')\big)
    \;\;
    \mbox{
      is a local super co-frame field on $\Sigma$
    }
    \\
    \label{ConditionOnSuperShearMap}
    \mbox{\bf (b)}
    \hspace{.3cm}
    &
     \phi^\ast\big(\overline{P}(E',\Psi')  \big) 
       \,=\, 
     \mathrm{Sh}
     \cdot 
     \phi^\ast
     \big(P (E',\Psi')\big)
     \,\defneq\,
     \mathrm{Sh}
     \cdot (e',\psi')
     \,.
  \end{align}
  for all super-coframe fields $(E',\Psi')$ in the same transversal gauge orbit as $(E, \Psi)$:
  \begin{equation}
    \label{TransversalSpinTransformation}
    (E', \Psi')
    \;=\;
    U \cdot
    (E, \Psi)
    \,,
    \;\;\;\;
    \mbox{for}
    \;\;\;
    U 
    \,\in\,
    C^\infty\big(
      \CoverOf{X}
      ;\,
      \mathrm{Spin}(d-p)
      \,
    \big)
    \,.
  \end{equation}
\end{definition}
In supergeometric generalization of Prop. \ref{ReformulatingTheDarbouxCondition},
we may re-cast the above super-Darboux condition as follows:
\begin{lemma}[\bf Reformulation of super-Darboux condition]
  \label{ReformulationOfSuperDarbouxCondition}
  The condition 
  \eqref{ConditionOnSuperShearMap}
  on a super-shear map $\mathrm{Sh}$ \eqref{SuperShearMap} is equivalently its $\mathrm{Spin}(d-p)$-equivariance
  $$
     \phi^\ast\big(\overline{P}(E',\Psi')  \big) 
       \,=\, 
     \mathrm{Sh}
     \cdot 
     \phi^\ast
     \big(P (E',\Psi')\big)
     \hspace{.6cm}
     \Leftrightarrow
     \hspace{.6cm}
    \phi^\ast(U)
    \cdot
    \mathrm{Sh}
      \;=\;
    \mathrm{Sh}
    \cdot 
    \phi^\ast(U)
    \,.
  $$
\end{lemma}
\begin{proof}
Noticing  that the transversal spin-action
\eqref{TransversalSpinTransformation} commutes with the projection
\eqref{TheProjectorInGeneral}
\begin{equation}
  \label{TransversalSpinActionCommutesWithProjector}
  U \circ P
  \,=\,
  P \circ U
  ,\,
  \;\;\;\;
  U \circ \overline{P}
  \,=\,
  \overline{P} \circ U
\end{equation}
as well as with pullback to the worldvolume, in that
\begin{equation}
  \label{TransversalSpinActionCommutesWithPullback}
  \phi^\ast \big(
    U \cdot (-)
  \big)
  \;=\;
  \phi^\ast(U)
  \cdot
  \phi^\ast(-)
  \hspace{.5cm}
  \mbox{i.e.:}
  \hspace{.4cm}
  \begin{tikzcd}[
    row sep=24pt, column sep=large
  ]
    T\CoverOf{\Sigma}
    \ar[
      d,
      "{
        \mathrm{d}\CoverOf{\phi}
      }"
    ]
    \ar[
      r,
      "{
        \scalebox{1.15}{$($}
          p_{{}_{T\CoverOf{\Sigma}}}
          ,\,
          \phi^\ast(E,\Psi)
        \scalebox{1.1}{$)$}
      }"
    ]
    &[45pt]
    \CoverOf{\Sigma}
    \times
    \mathbb{R}^{1,d\vert\mathbf{N}}
    \ar[
      d,
      "{
        \phi
         \times
        \mathrm{id}
      }"{description}
    ]
    \ar[
      dr,
      "{
        \phi^\ast(U)
      }"{sloped}
    ]
    \\
    T\CoverOf{X}
    \ar[
      r,
      "{
        \scalebox{1.15}{$($}
          p_{{}_{T \CoverOf{X}}},
          (E,\Psi)
        \scalebox{1.15}{$)$}
      }"
    ]
    &
    \CoverOf{X}
    \times
    \mathbb{R}^{1,d\vert\mathbf{N}}
    \ar[
      r,
      "{ U }"
    ]
    &
    \mathbb{R}^{1,d\vert\mathbf{N}}
    \mathrlap{\,,}
  \end{tikzcd}
\end{equation}
\vspace{1pt}
we have the following sequence of logical equivalences:
$$
  \def\arraystretch{1.7}
  \begin{array}{cll}
  &
  \phi^\ast
  \big(
    \overline{P}
    \,
    (E', \Psi')
  \big)
    \;=\;
  \mathrm{Sh}
  \cdot 
  \phi^\ast \big(
    P
    \,
    (E',\Psi')
  \big)
  \\
  \Leftrightarrow
  &
  \phi^\ast
  \big(
    \overline{P}
    \,
    U  (E,\Psi)
  \big)
    \;=\;
  \mathrm{Sh}
  \cdot 
  \phi^\ast \big(
    P
    \,
    U  (E,\Psi)
  \big)
  &
  \proofstep{
    by
    \eqref{TransversalSpinTransformation}
  }
  \\
  \Leftrightarrow
  &
  \phi^\ast
  \big(
    U
    \overline{P}
    \,
    (E,\Psi)
  \big)
    \;=\;
  \mathrm{Sh}
  \cdot 
  \phi^\ast \big(
    U
    P
    (E,\Psi)
  \big)  
  &
  \proofstep{
    by
    \eqref{TransversalSpinActionCommutesWithProjector}
  }
  \\
  \Leftrightarrow
  &
  \phi^\ast(U)
  \cdot
  \phi^\ast
  \big(
    \overline{P}
    (E,\Psi)
  \big)
    \;=\;
  \mathrm{Sh}
  \cdot 
  \phi^\ast(U)
  \cdot
  \phi^\ast \big(
    P
    (E,\Psi)
  \big)  
  &
  \proofstep{
    by
    \eqref{TransversalSpinActionCommutesWithPullback}
  }
  \\
  \Leftrightarrow
  &
  \phi^\ast(U)
  \cdot
  \mathrm{Sh}
    \cdot
    \phi^\ast\big(
      P(E,\Psi)
    \big)
    \;=\;
  \mathrm{Sh}
  \cdot 
  \phi^\ast(U)
  \cdot
  \phi^\ast \big(
    P
    (E,\Psi)
  \big)  
  &
  \proofstep{
    by
    \eqref{TransversalSpinTransformation}
  }
  \\
  \Leftrightarrow
  &
  \phi^\ast(U)
  \cdot
  \mathrm{Sh}
    \;=\;
  \mathrm{Sh}
  \cdot 
  \phi^\ast(U)
  &
  \proofstep{
    by
    \eqref{PullbackOfTransversalCoFrameIsCoFrame}.
  }
  \end{array}
$$

\vspace{-.4cm}
\end{proof}

It follows that most components of a super-shear map vanish identically -- as in the bosonic case \eqref{BosonicShearingMapVanishes} -- but, remarkably, the odd-odd component may be non-trivial:
\begin{proposition}[\bf Components of $\sfrac{1}{2}$BPS immersions]
  \label{ComponentsOfBPImmersion}
  A super-immersion $\phi : \Sigma^{1+p} \xrightarrow{\;} X^{1+d}$ 
  is a $\sfrac{1}{2}$BPS immersion
  (Def. \ref{BPSImmersion}) 
  if and only if the pullback of any super-Darboux coframe field $(E,\Psi)$ is of the form
  \begin{equation}
    \left.
    \def\arraystretch{1.6}
    \begin{array}{ll}
      \label{BPSImmersionInComponents}
      \phi^\ast (P E)
      \;=\;
      e
      &
      \phi^\ast (\overline{P}E)
      \;=\;
      0
      \\
      \phi^\ast (P \Psi)
      \;=\;
      \psi
      &
      \phi^\ast 
      (\overline{P} \Psi)
      \;=\;
      \mathrm{Sh}_{11} 
        \cdot 
      \psi
    \end{array}
   \!\! \right\}
    \hspace{.7cm}
    \Leftrightarrow
    \hspace{.7cm}
    \left\{\!\!
    \def\arraystretch{1.6}
    \begin{array}{l}
      \phi^\ast E
      \;=\;
      e
      \\
      \phi^\ast \Psi
      \;=\;
      \psi
      \,+\,
      \mathrm{Sh}_{11} 
        \cdot 
      \psi\,,
    \end{array}
    \right.
  \end{equation}
  for some
  \vspace{1mm} 
  \begin{equation}
    \label{TheNontrivialShearComponent}
    \mathrm{Sh}_{11} 
      \,\in\,
    C^\infty\big(
      \CoverOf{\Sigma}
      ;\,
      \mathrm{Hom}_{
        \mathrm{Rep}_{\mathbb{R}}
        \scalebox{1.1}{$($}
          \mathrm{Spin}(d-p)
        \scalebox{1.1}{$)$}
      }
      \big(
        P \mathbf{N}
        ,\,
        \overline{P}\mathbf{N}
      \big)
    \big)
    \,.
  \end{equation}
\end{proposition}
\begin{proof}
  Since $(e,\psi) := \phi^\ast(E,\Psi)$ is a super-coframe field by assumption \eqref{PullbackOfTransversalCoFrameIsCoFrame}, the pullback of $\overline{P}(E,\Psi)$ to $\CoverOf{\Sigma}$ may uniquely be expanded in $(e,\psi)$:
  \begin{equation}
    \label{ComponentOfSupershearMap}
    \def\arraystretch{1.5}
    \begin{array}{l}
      \phi^\ast (\overline{P}E)
      \;=\;
      \mathrm{Sh}_{00}
      \cdot e
      \,+\,
      \mathrm{Sh}_{01}
      \cdot \psi
      \\
      \phi^\ast (\overline{P}\Psi)
      \;=\;
      \mathrm{Sh}_{10}
      \cdot e
      \,+\,
      \mathrm{Sh}_{11}
      \cdot \psi
      \,,
    \end{array}
  \end{equation}
  and by \eqref{ConditionOnSuperShearMap} the coefficients are just the components of the super-shear map $\mathrm{Sh}$ \eqref{SuperShearMap}. 

  Now, by Lem. \ref{ReformulationOfSuperDarbouxCondition}, $\mathrm{Sh}$ and hence its components are $\mathrm{Spin}(d-p)$-equivariant, and as such they map between the following $\mathrm{Spin}(d-p)$-representations:
  $$
    \begin{tikzcd}[
      column sep=1pt,
      row sep=-2pt
    ]
      \mathrm{Sh}_{00}
      &:&
      (1+p)
      \ar[r]
      &[+15pt]
      \mathbf{(d-p)}
      \\
      \mathrm{Sh}_{01}
      &:&
      P(\mathbf{N}) 
      \ar[r]
      &[+15pt]
      \mathbf{(d-p)}
      \\
      \mathrm{Sh}_{10}
      &:&
      (1+p)
      \ar[r]
      &
      \overline{P}(\mathbf{N})
      \\
      \mathrm{Sh}_{11}
      &:&
      P(\mathbf{N})
      \ar[r]
      &
      \overline{P}(\mathbf{N})
      \,.
    \end{tikzcd}
  $$
  Here $(1+p) := P(\mathbb{R}^{1,d})$ denotes the trivial $\mathrm{Spin}(d-p)$-representation of dimension $1+p$, 
  while $\mathbf{(d-p)} := \overline{P}(\mathbb{R}^{1,d})$ denotes the vectorial irrep (via the defining irrep of $\mathrm{SO}(d-p)$), 
  and $P(\mathbf{N})$, $\overline{P}(\mathbf{N})$ are regarded as $\mathrm{Spin}(d-p)$-representations via \eqref{ResidualTransverseSpinReps}.

  But since these maps are $\mathrm{Spin}(d-p)$-equivariant, Schur's lemma  (e.g. \cite[Prop. 1.16]{EtingofEtAl11}) says that they are trivial between non-isomorphic irrep summands.
  This manifestly implies that $\mathrm{Sh}_{00} = 0$. Similarly, since the $\mathrm{Spin}(d-p)$-representations $P(\mathbf{N})$ and $\overline{P}(\mathbf{N})$
  are spinorial (in that their fixed subspace of $-1 \in \mathrm{Spin}(d-p)$ is zero)
  they do not contain a vectorial summand like $\mathbf{(d-p)}$ (which is fixed by the central element $-1$), also $\mathrm{Sh}_{01} = 0$ and $\mathrm{Sh}_{10} = 0$.
\end{proof}

\begin{remark}[\bf Existence of fermionic shear]
  The noteworthy point in Prop. \ref{ComponentsOfBPImmersion} is that, in general, the one component $\mathrm{Sh}_{11}$ in \eqref{ComponentOfSupershearMap} of the super-shear map {\it need not} vanish, since $P(\mathbf{N})$ and $\overline{P}(\mathbf{N})$ may contain the same $\mathrm{Spin}(d-p)$ irreps. Concretely, we see this below for the example of the M5-brane, where these two representations are in fact isomorphic, cf. \eqref{IrrepDecompositionOf6dSpinors} below.
  Remarkably, this freedom in BPS super-immersions is the origin of the worldvolume higher gauge fields, discussed for the M5 in \S\ref{TheM5EquationsOfMotion}.
\end{remark}

\begin{remark}[\bf The ``super-embedding''-condition]
\label{SuperEmbeddingConditionInTheLiterature}
In summary, Prop. \ref{ComponentsOfBPImmersion} says in particular that a $\sfrac{1}{2}$BPS immersion (Def. \ref{BPSImmersion}) comes with the following structure:
\begin{equation}
  \label{BPSImmersionDiagram}
  \begin{tikzcd}[row sep=20pt]
    \mathllap{
      \scalebox{.7}{
        \color{darkblue}
        \bf
        \def\arraystretch{.9}
        \begin{tabular}{c}
          Brane
          \\
          worldvolume
          \\
          supermanifold
        \end{tabular}
      }
    }
    \Sigma
    \ar[
      rr,
      "{ \phi }",
      "{
        \scalebox{.7}{
          \color{darkgreen}
          \bf
          super-immersion
        }
      }"{swap}
    ]
    &[+15pt]
    &[-25pt]
    X
    \mathrlap{
      \scalebox{.7}{
        \color{darkblue}
        \bf
        \def\arraystretch{.9}
        \begin{tabular}{c}
          Target
          \\
          spacetime
          \\
          supermanifold
        \end{tabular}
      }
    }
    &[-25pt]
    \\[-15pt]
    0
    \ar[
      r,
      <-|,
      "{ \phi^\ast }"
    ]    
    &
    \overline{P} E
    \mathrlap{
      \scalebox{.7}{
        \color{darkblue}
        \bf
        \def\arraystretch{.9}
        \begin{tabular}{c}
          Transversal part
          \\
          of bosonic co-frame
        \end{tabular}
      }
    }
    \\
    &
    \mathllap{
      \scalebox{.7}{
        \color{darkblue}
        \bf
        \def\arraystretch{.9}
        \begin{tabular}{c}
          Super-Darboux
          \\
          co-frame field
        \end{tabular}
      }
    }
    (E,\Psi)
    \ar[
      rr,
      <-|,
      "{ U }"
    ]
    \ar[
      d,
      |->,
      "{ (P,P) }"
    ]
    \ar[
      u,
      |->,
      "{  
         ( \overline{P}, 0)   
      }"{swap}
    ]
    &&
    \big(\tilde E,\tilde \Psi\big)    
    \mathrlap{
      \scalebox{.7}{
        \color{darkblue}
        \bf
        \def\arraystretch{.9}
        \begin{tabular}{c}
          Any spacetime
          \\
          super co-frame
        \end{tabular}
      }
    }
    \\
    \mathllap{
      \scalebox{.7}{
        \color{darkblue}
        \bf
        \def\arraystretch{.9}
        \begin{tabular}{c}
          Worldvolume
          \\
          super co-frame
        \end{tabular}
      }
    }
    (e,\psi)
    \ar[
      r,
      <-|,
      "{ \phi^\ast }"
    ]    
    &
    (P E, P\Psi)
    \mathrlap{
      \scalebox{.7}{
        \color{darkblue}
        \bf
        \def\arraystretch{.9}
        \begin{tabular}{c}
          Tangential part
          \\
          of super-coframe
        \end{tabular}
      }
    }
  \end{tikzcd}
\end{equation}
\noindent
This is what is broadly known as the ``super-embedding''-condition, in the literature. Specifically:
\begin{itemize}[
  topsep=1pt,
  itemsep=2pt,
  leftmargin=.8cm
]
\item[\bf (i)] 
The condition $\phi^\ast P E = e$ in 
\eqref{BPSImmersionDiagram}

is the {\bf basic embedding condition} of \cite[(6)]{HoweSezgin97a}\cite[(4)]{HSW97a}\cite[(2)]{HoweRaetzelSezgin98}, 
which earlier was known 

as the {\bf geometrodynamical condition} \cite[(2.23)]{BPSTV95};

\item[\bf (ii)] 
With the additional condition $\phi^\ast \overline{P} E = 0$ in 
\eqref{BPSImmersionDiagram}

this is the {\bf superembedding condition} of \cite[(4.36-37)]{Sorokin00}, see also \cite[(2.6-9)]{Bandos11}\cite[(5.13-14)]{BandosSorokin23}.

\item[\bf (iii)] 
The further condition $\phi^\ast P \Psi = \psi$ in 
\eqref{BPSImmersionDiagram}

is tacitly introduced in \cite[(4.46)]{Sorokin00}, reviewed in \cite[(5.26)]{BandosSorokin23}.
\end{itemize}
\noindent
In comparing to these references, notice that 
\begin{itemize}[
  topsep=1pt,
  itemsep=2pt,
  leftmargin=.8cm
]
\item[\bf (iv)]
Our expression
$
   \phi^\ast 
   \big(
     P U
     ,\,
     \overline{P}U
   \big)
   \,=:\,
   u
$

corresponds to the {\bf harmonics} in \cite[\S 2.1]{BPSTV95}\cite[(29)]{HSW98}\cite[(4.11)]{Sorokin00}, 
where $U \in C^\infty\big(X;\, \mathrm{Spin}(1,d)\big)$.
\end{itemize}
\end{remark}

\smallskip

In closing this section, we just notice that,
in addition to the super co-frame field, a super-immersion also pulls back the spin-connection $\Omega$ on target spacetime:

\begin{notation}[\bf Second fundamental super-form]
  Given a $\sfrac{1}{2}$BPS super-immersion $\phi$ \eqref{BPSImmersion} into a super-spacetime $X$ with spin connection $\Omega$, then in supergeometric generalization of \eqref{PullbackOfConnectionForm} the pullback of $\Omega$ is uniquely expanded as follows:
  \begin{equation}
  \label{PullbackOfSuperConnectionForm}
  \def\arraystretch{1.2}
  \def\arraycolsep{3pt}
  \begin{array}{cccl}
    \omega^a{}_b
    &:=&
    \phi^\ast
    \Omega^a{}_b
    &
    \scalebox{.85}{
      \def\arraystretch{.9}
      \begin{tabular}{l}
        for transversal $a$ 
        \\
        and tangential $b$
      \end{tabular}
    }
    \\[+5pt]
    e^{b_1}
    \SecondFundamentalForm
      ^a
      _{b_1 b_2}
    +
    \psi^\beta
    \SecondFundamentalForm
      ^a
      _{\beta \, b_2}
    &:=&
    \phi^\ast \Omega^a{}_{b_2}
    &
    \scalebox{.9}{
      \def\arraystretch{.85}
      \begin{tabular}{l}
        for transversal $a$ 
        \\
        and tangential $b_2$.
      \end{tabular}
    }
  \end{array}
\end{equation}
\end{notation}

\medskip

\section{M5-brane super-immersions}
\label{SuperEmbeddingConstruction}

Here we give a streamlined account of the specialization of the ``super-embedding''-construction (\S\ref{SuperEmbeddings})
to the case of the M5-brane (due to \cite{HoweSezgin97b}\cite[\S 5.2]{Sorokin00}, reviewed in \cite[\S 5]{BandosSorokin23}), focusing on the derivation of the $H_3$-flux density (cf. Rem. \ref{TransversalFermionicShear}) and its Bianchi identity and (non/self-duality) equations of motion (Prop. \ref{BianchiIdentityOnM5BraneInComponents}) which drive the discussion of flux quantization on the M5 (in \S\ref{FluxQuantizationOnM5Branes}, as introduced in \S\ref{IntroductionAndOverview}).

\medskip 
Our key move to make the notoriously intricate derivation more transparent is (not to use a matrix representation for the spinors but) to algebraically carve out 
(in \S\ref{SpinorsIn6dFrom11d})
the worldvolume spin representation $2 \cdot \mathbf{8}$ by the same tangential/transversal projection operators 
that enter the definition of BPS super-immersions (Def. \ref{BPSImmersion}) in the first place.

\subsection{Self-dual tensors in 6d}

For reference and completeness, we first briefly record some properties of self-dual tensors in 6d. In this section indices run through $0, 1, \cdots, 5$.
 
\begin{notation}[\bf Tensors in 6d]
$\,$
\begin{itemize}[
  leftmargin=.6cm,
  topsep=1pt,
  itemsep=4pt
]

\item 
$\epsilon_{0 1 \cdots 5} \,=\, + 1$
and
$\epsilon^{0 1 \cdots 5} \,=\, - 1$
as in \eqref{transversalizationOfLeviCivitaSymbol}.

\item $\tilde H_3  := \tfrac{1}{3!} (\tilde H_3)_{a_1 a_2 a_3} \, e^{a_1} e^{a_2} e^{a_3}$  denotes our generic (self-dual) rank-3 tensor

The tilde is in order to distinguish this from the flux density $H_3$ on the M5-brane which is {\it not} actually self-dual, but ``non-linearly self-dual'', see below.

\item 
$\star \tilde H_3 \,:=\, \tfrac{1}{3!} 
\big(\tfrac{1}{3!} \epsilon_{a_1 a_2 a_3 \, b_1 b_2 b_3} (\tilde H_3)^{b_1 b_2 b_3}\big) e^{a_1} e^{a_2} e^{a_3}$ is the Hodge dual tensor.

Notice that below we consider this on super-spacetimes, where the Hodge duality operation makes sense in this form (only) for such differential forms with vanishing fermionic frame component $(\psi^0)$.

Hence the self-duality condition is
\begin{equation}
  \label{SelfDualityCondition}
  \tilde H_3
  \;=\;
  \star \tilde H_3
  \;\;\;\;\;\;
  \Leftrightarrow
  \;\;\;\;\;\;
  (\tilde{H}_3)_{a_1 a_2 a_3}
  \;=\;
  \tfrac{1}{3!}
  \epsilon_{a_1 a_2 a_3 \, b_1 b_2 b_3}
  (\tilde H)_3^{b_1 b_2 b_3}.
\end{equation}

\item The ``square of $H_3$ on two indices'' 
\vspace{2mm} 
\begin{equation}
  \label{SquareOfH3OnTwoIndices}
   (\K)_a^{b} 
  \,:=\, 
  (\tilde H_3)_{a\, \color{darkblue} c_1 c_2} (\tilde H_3)^{ {\color{darkblue} c_1 c_2} \, b}
\end{equation}
plays a key role in relating the self-dual tensor $\tilde H_3$ to the actual flux density $H_3$.

In the following it is often suggestive to regard it as a matrix, such as to write:
\vspace{2mm} 
\begin{equation}
  \label{SquareOfSquareOfH3}
  (\K \cdot \K)_a^b
  \;
  :=
  \;
  (\K)_a^c (\K)_c^a
  \,,
  \;\;\;\;\;\;
  \mathrm{tr}\big(
    \K
  \big)
  \;:=\;
  (\K)_a^a
  \,,
  \;\;\;\;\;\;
  \mathrm{tr}\big(
    \K \cdot \K
  \big)
  \;:=\;
  (\K)_a^b (\K)_b^a
  \,.
\end{equation}

\end{itemize}
\end{notation}

\smallskip 
\begin{lemma}[\bf Trace of square of selfdual 3-form vanishes]
If $\tilde H_3 \,=\, \star \tilde H_3$ then the trace
\eqref{SquareOfSquareOfH3} vanishes:
\begin{equation}
  \label{SquareOfSelfDual3FormVanishes}
  \mathrm{tr}(\K)
  \,=\,
  0
  \;\;\;\;\;\;
  \mbox{\rm in that}
  \;\;\;\;\;\;
  (\tilde H_3)_{a_1 a_2 a_3}
  (\tilde H_3)^{a_1 a_2 a_3}
  \;=\;
  0
  \,.
\end{equation}
\end{lemma}
\begin{proof}
$$
  \def\arraystretch{1.6}
  \begin{array}{lll}
    (\tilde H_3)_{a_1 a_2 a_3}
    (\tilde H_3)^{a_1 a_2 a_3}
    &
    \;=\;
    \tfrac{1}{3!}
    (\tilde H_3)_{a_1 a_2 a_3}
    \epsilon^{
      a_1 a_2 a_3 
      \,
      b_1 b_2 b_3
    }
    (\tilde H_3)_{b_1 b_2 b_3}
    &
    \proofstep{
      by
      \eqref{SelfDualityCondition}
    }
    \\
   & \;=\;
    -
    \tfrac{1}{3!}
    (\tilde H_3)_{b_1 b_2 b_3}
    \epsilon^{
      b_1 b_2 b_3
      \,
      a_1 a_2 a_3 
    }
    (\tilde H_3)_{a_1 a_2 a_3}
    \\
  &  \;=\;
    -
    (\tilde H_3)_{a_1 a_2 a_3}
    (\tilde H_3)^{a_1 a_2 a_3}
    &
    \proofstep{
      by
      \eqref{SelfDualityCondition}
      \,.
    }
  \end{array}
$$

\vspace{-3mm} 
\end{proof}

\begin{lemma}[\bf Squaring over a single index]
If $\tilde H_3 \,=\, \star \, \tilde H_3$ then 
{\rm (cf. \cite[(9)]{HoweSezginWest97})}:
\begin{equation}
  \label{TildeThreeSquaredOverSingleIndex}
  (\tilde H_3)_{
    \color{darkgreen} a_1 a_2 
    \, 
    \color{darkblue}
    c
  }
  (\tilde H_3)^{
    { \color{darkblue} c }
    \, 
    \color{darkorange}
    b_1 b_2
  }
  \;=\;
  +
  \,
  \tensor*{\delta}{
    ^{[ \color{darkorange}b_1}
    _{[ \color{darkgreen} a_1}
  }
  (\K)^{
    { \color{darkorange} b_2 }
    ]
  }_{
    { \color{darkgreen} a_2 }
    ]
  }
  \,.
\end{equation}
\end{lemma}
\begin{proof}
$$
  \def\arraystretch{1.8}
  \begin{array}{ll}
    (\tilde H_3)_{a_1 a_2 \, c}
    (\tilde H_3)^{c \, b_1 b_2}
    \\
    \;=\;
    \Big(
      \tfrac{1}{3!}
      \epsilon_{
        a_1 a_2 \, c
        \,
        e_1 e_2 e_3
      }
      (\tilde H_3)^{e_1 e_2 e_3}
    \Big)
    \Big(
      \tfrac{1}{3!}
      \epsilon^{
        c \, b_1 b_2
        \,
        d_1 d_2 d_3
      }
      (\tilde H_3)_{d_1 d_2 d_3}
    \Big)
    &
    \proofstep{
      by 
      \eqref{SelfDualityCondition}
    }
    \\
    \;=\;
    \tfrac{
      - 5!
    }{
      3! \cdot 3!
    }
    \delta_{
      a_1 a_2
      \,
      e_1 e_2 e_3
    }^{
      b_1 b_2
      \,
      d_1 d_2 d_3    
    }
    (\tilde H_3)^{e_1 e_2 e_3}
    (\tilde H_3)_{d_1 d_2 d_3}
    &
    \proofstep{
     by \eqref{ContractingKroneckerWithSkewSymmetricTensor}
    }
    \\
    \;=\;
    \tfrac{
      + 5!
    }{
      3! \cdot 3!
    }
    \tfrac{
      6 \cdot 2! \cdot 3!
    }{5!}
    \delta^{
      b_1 b_2
    }_{
      [a_1 |e_1|
    }
    \delta^{
      d_1 d_2 d_3
    }_{
      a_2] e_2  e_3
    }
    (\tilde H_3)^{e_1 e_2 e_3}
    (\tilde H_3)_{d_1 d_2 d_3}
    \;+\;
    \tfrac{
      - 5!
    }{
      3! \cdot 3!
    }
    \tfrac{
      3 \cdot 2! \cdot 3!
    }{5!}
    \delta^{
      b_1 b_2
    }_{
      e_1 e_2
    }
    \delta^{
      d_1 d_2 d_3
    }_{
      a_1 a_2 e_3
    }
    (\tilde H_3)^{e_1 e_2 e_3}
    (\tilde H_3)_{d_1 d_2 d_3}
    &
    \proofstep{
      by 
      \eqref{SquareOfSelfDual3FormVanishes}
    }
    \\
    \;=\;
    2
    \,
    \delta^{
      [b_1 
    }_{
      [a_1
    }
    (\tilde H_3)^{b_2] e_2 e_3}
    (\tilde H_3)_{a_2] e_2 e_3}
    \;-\;
    (\tilde H_3)^{b_1 b_2 e_3}
    (\tilde H_3)_{a_1 a_2 e_3}\,.
  \end{array}
$$

\vspace{-3mm} 
\end{proof}

\smallskip 
\begin{lemma}[\bf Square of square is proportional to identity]
  If $\tilde H_3 \,=\, \star \tilde H_3$ then it squared square \eqref{SquareOfSquareOfH3} is
  {\rm (cf. \cite[(8)]{HoweSezginWest97})}:
  \smallskip 
  \begin{equation}
    \label{KSquareIsProportionalToIdentity}
    \K \cdot \K
    \;\;
    =
    \;\;
    \tfrac{1}{6}
    \mathrm{tr}\big(
      \K \cdot \K
    \big)
    \,
    \mathrm{id}
    \;\;\;\;\;\;\;\;
    \mbox{\rm in that}
    \;\;\;\;\;\;\;\;
    (\K)_a^c \, (\K)_c^b
    \;\;
    =
    \;\;
    \tfrac{1}{6}
    \,\delta_a^b\,
    (\K)_{c_1}^{c_2}
    (\K)_{c_2}^{c_1}
    \,.
  \end{equation}
\end{lemma}
\begin{proof}
$$
\hspace{-5mm}
  \def\arraystretch{1.7}
  \begin{array}{lll}
    (\K)_{a}^{c}
    \,
    (\K)_c^b
    &
    \;=\;
    (\tilde H_3)_{a \, d_1 d_2}
    (\tilde H_3)^{d_1 d_2 \, c}
    (\tilde H_3)_{c \, e_1 e_2}
    (\tilde H_3)^{e_1 e_2 \, b}
    &
    \proofstep{
      by def.
    }
    \\
   & \;=\;
    (\tilde H_3)_{a \, d_1 d_2}
    \delta^{
      d_1
    }_{
      e_1
    }
    (\K)^{
      d_2
    }_{
      e_2
    }
    (\tilde H_3)^{e_1 e_2 \, b}    
    &
    \proofstep{
      by
      \eqref{TildeThreeSquaredOverSingleIndex}
    }
    \\
   & \;=\;
    (\tilde H_3)_{a \, d_1 d_2}
    (\tilde H_3)^{d_1 e_2 \, b}    
    (\K)^{
      d_2
    }_{
      e_2
    }
    \\
  &  \;=\;
    -
    \delta^{
      [e_2
    }_{
      [a
    }
    (\K)^{
      b]
    }_{
      d_2]
    }
    (\K)^{
      d_2
    }_{
      e_2
    }
    &
    \proofstep{
      by
      \eqref{TildeThreeSquaredOverSingleIndex}
    }
    \\
 &   \;=\;
    -
    \tfrac{1}{4}
    \delta^{
      e_2
    }_{
      a
    }
    (\K)^{
      b
    }_{
      d_2
    }
    (\K)^{
      d_2
    }_{
      e_2
    }
    +
    \tfrac{1}{4}
    \delta^{
      b
    }_{
      a
    }
    (\K)^{
      e_2
    }_{
      d_2
    }
    (\K)^{
      d_2
    }_{
      e_2
    }
    +
    \tfrac{1}{4}
    \delta^{
      e_2
    }_{
      d_2
    }
    (\K)^{
      b
    }_{
      a
    }
    (\K)^{
      d_2
    }_{
      e_2
    }
    -
    \tfrac{1}{4}
    \delta^{
      b
    }_{
      d_2
    }
    (\K)^{
      e_2
    }_{
      a
    }
    (\K)^{
      d_2
    }_{
      e_2
    }
    \\
  &  \;=\;
    -
    \tfrac{1}{4}
    \,
    (\K)^{
      c
    }_{
      a
    }    
    \,
    (\K)^{
      b
    }_{
      c
    }
    \;+\;
    \tfrac{1}{4}
    \,
    \delta_a^b
    \,
    \mathrm{tr}(\K \K)
    \;+\;
    \tfrac{1}{4}
    (\K)^b_a
    \,
    \underbrace{
      (\K)^{e_2}_{e_2}
    }_{
      \underset{
        \scalebox{.6}{
         \eqref{SquareOfSelfDual3FormVanishes}
        }
      }{=} 
      0
    }
    \;-\;
    \tfrac{1}{4}
    \,
    (\K)_a^c
    (\K)_c^b
    \,.
  \end{array}
$$

\vspace{-6mm} 
\end{proof}

\begin{lemma}[\bf Inverse of identity minus square]
\label{IdMinus2KIInvertible}
If $\tilde H_3 \,=\, \star \, \tilde H_3$ and $\mathrm{tr}(\K \cdot \K)$ \eqref{KSquareIsProportionalToIdentity} satisfies the non-criticality condition
\begin{equation}
  \label{NonCriticalityCondition}
  \mathrm{tr}\big(
    \K \cdot \K
  \big)
  \;\neq\;
  \tfrac{3}{2}
\end{equation}
then {\rm (cf. \cite[(10)]{HoweSezginWest97})} the matrices $(\mathrm{id} \mp 2 \K)$ \eqref{SquareOfH3OnTwoIndices} are invertible \footnote{
  The expression
  $(\mathrm{id} + 2 \K)_{a b}$ is interpreted 
  \cite[(8)]{BBdSS00}
  as (proportional to) the effective metric on M5-branes that is seen by open M2-branes ending on them, the M-theoretic lift of the ``open string metric'' on D-branes.
}, with inverse:
\begin{equation}
  \label{InverseOfIdMinus2K}
  \begin{array}{l}
    \big(
      \mathrm{id}
      \mp
      2 \, \K
    \big)^{-1}
    \;\;
    =
    \;\;
    \frac{1}{
      \mathclap{\phantom{
        \vert^{\vert}
      }}
      1 
        - 
      \sfrac{2}{3}
      \,
      \mathrm{tr}(
        \K \cdot \K
      )
    }
    \big(
      \mathrm{id}
      \pm
      2\, \K
    \big).
  \end{array}
\end{equation}
\end{lemma}
\begin{proof}
$$
  \def\arraystretch{1.3}
  \begin{array}{lll}
    \big(
      \mathrm{id}
      \mp
      2\,
      \K
    \big)
    \cdot
    \big(
      \mathrm{id}
      \pm
      2
      \,
      \K
    \big)
    &
    \;=\;
    \mathrm{id}
    -
    4 \, 
    \K \cdot \K
    \\
   & \;=\;
    \big(
      1 
        -
      \tfrac{2}{3}
      \mathrm{tr}(\K \cdot \K)
    \big)
    \,
    \mathrm{id}
    &
    \proofstep{
      by
      \eqref{KSquareIsProportionalToIdentity}
    }.
  \end{array}
$$

\vspace{-4mm} 
\end{proof}

\begin{lemma}[\bf Anti-Selfduality of cubic form]
If $\tilde H_3 \,=\, \star \tilde H_3$, then the skew-symmetric part of the cubic $(\K)_{a_1}{}^{c} (\tilde H_3)_{c \, a_2 a_3}$ \eqref{SquareOfH3OnTwoIndices} is anti-self-dual {\rm (cf. \cite[(5.82)]{Sorokin00})}:
\begin{equation}
  \label{AntiSelfDualityOfCubicTildeH}
  (\K)_{
    [d_1
  }{}^{
    c
  }
  (\tilde H_3)_{
    |c| 
    d_2 d_3]
  }
  \;\;
  =
  \;\;
  -\tfrac{1}{3!}
    \epsilon_{
      d_1 d_2 d_3
      \,
      a_1 a_2 a_3
    }
  (\K)^{a_1}{}_c
  (\tilde H_3)^{c\, a_2 a_3}.
\end{equation}
\end{lemma}
\begin{proof}
$$
  \def\arraystretch{1.7}
  \begin{array}{ll}
    \tfrac{1}{3!}
    \epsilon_{
      d_1 d_2 d_3
      \,
      a_1 a_2 a_3
    }
    (\tilde H_3)^{
      a_1
      \,
      b_1 b_2
    }
    (\tilde H_3)_{
      b_1 b_2 
      \,
      c
    }
    (\tilde H_3)^{
      c
      \,
      a_2 a_3
    }
    \\
    \;=\;
    \tfrac{
      1
    }{3!}
    \epsilon_{
      d_1 d_2 d_3
      \,
      a_1 a_2 a_3
    }
    (\tilde H_3)^{
      a_1
      \,
      b_1 b_2
    }
    (\tilde H_3)_{
      b_1 b_2 
      \,
      c
    }
    \Big(
    \tfrac{1}{3!}
    \epsilon^{
      c
      \,
      a_2 a_3
      \,
      e_1 e_2 e_3
    }
    (\tilde H_3)_{e_1 e_2 e_3}
    \Big)
    &
    \proofstep{
      by
      \eqref{SelfDualityCondition}
    }
    \\
    \;=\;
    \tfrac{
      -2 \cdot 4!
    }{3!\cdot 3!}
    (\tilde H_3)^{
      a_1
      \,
      b_1 b_2
    }
    (\tilde H_3)_{
      b_1 b_2 
      \,
      c
    }
    \Big(
    \delta_{
      d_1 d_2 d_3
      \,
      a_1
    }
    ^{
      c
      \,
      e_1 e_2 e_3
    }
    (\tilde H_3)_{e_1 e_2 e_3}
    \Big)
    &
    \proofstep{
      by
      \eqref{ContractingKroneckerWithSkewSymmetricTensor}
    }
    \\
    \;=\;
    \tfrac{
      -2 \cdot 4!
    }{3!\cdot 3!}
    \tfrac{
      3!
    }{4!}
    (\tilde H_3)^{
      a_1
      \,
      b_1 b_2
    }
    (\tilde H_3)_{
      b_1 b_2 
      \,
      c
    }
    \Big(
      \delta
        _{a_1}
        ^{e_1}
      \delta_{
        d_1 d_2 d_2
      }^{
        c \, e_2 e_3
      }
    (\tilde H_3)_{e_1 e_2 e_3}
    \Big)
    &
    \proofstep{
      by
      \eqref{SquareOfSelfDual3FormVanishes}
    }
    \\
    \;=\;
    \tfrac{
      -2 \cdot 4!
    }{3!\cdot 3!}
    \tfrac{
      3!
    }{4!}
    3
    (\tilde H_3)^{
      a_1
      \,
      b_1 b_2
    }
    (\tilde H_3)_{
      b_1 b_2 
      \,
      c
    }
    \Big(
      \delta
        _{a_1}
        ^{e_1}
      \delta_{[d_1}^{c}
      \delta_{
        d_2 d_3]
      }^{
        e_2 e_3
      }
    (\tilde H_3)_{e_1 e_2 e_3}
    \Big)
    \\
    \;=\;
    \tfrac{
      -2 \cdot 4!
    }{3!\cdot 3!}
    \tfrac{
      3!
    }{4!}
    3
    (\tilde H_3)^{
      e_1
      \,
      b_1 b_2
    }
    (\tilde H_3)_{
      b_1 b_2 
      \,
      [d_1
    }
    (\tilde H_3)_{
      |e_1| 
      d_2 d_3]
    }
    &
    \proofstep{
      by
      \eqref{ContractingKroneckerWithSkewSymmetricTensor}
    }
    \\
    \;=\;
    -
    (\tilde H_3)_{
      [d_1
      |
      b_1 b_2 
   }
    (\tilde H_3)^{
      b_1 b_2
      \,
      e_1
    }
    (\tilde H_3)_{
      e_1 
      |
      d_2 d_3]
      \,.
    }
  \end{array}
$$

\vspace{-3mm}
\end{proof}

\subsection{Spinors in 6d from 11d}
\label{SpinorsIn6dFrom11d}

Instead of using a matrix representation for the Clifford algebra on the 5-brane, for our proofs in \S\ref{TheM5EquationsOfMotion} it is key to algebraically characterize the worldvolume spin representation $2 \cdot \mathbf{8}_+ \;\underset{\mathclap{\scalebox{.5}{$\mathrm{Spin}(1,5)$}}}{\simeq}\; P(\mathbf{32}) \; \underset{\mathclap{\scalebox{.5}{$\mathrm{Spin}(5)$}}}{\simeq}\; 4 \cdot \mathbf{4}$ (see \eqref{IrrepDecompositionOf6dSpinors}) as the fixed locus $P(\mathbf{32})$ inside the target spin representation $\mathbf{32}$ in 11d. Here we spell out how this works.

\medskip

\noindent
{\bf Spinors in 11d.}
For reference, we begin with briefly recalling the following standard facts (proofs and references may be found in \cite[\S 2.2.1]{GSS24-SuGra}\cite[\S A]{HSS19}):
There exists an $\mathbb{R}$-linear representation $\mathbf{32}$ of $\mathrm{Pin}^+(1,10)$ with generators
\begin{equation}
  \label{The11dMajoranaRepresentation}
  \Gamma_a 
  \;:\;
  \mathbf{32}
  \xrightarrow{\;\;}
  \mathbf{32}
\end{equation}
and equipped with a skew-symmetric bilinear form
\begin{equation}
  \label{TheSpinorPairing}
  \big(\,
    \overline{(-)}
    (-)\,
  \big)
  \;:\;
  \mathbf{32}
  \otimes
  \mathbf{32}
  \xrightarrow{\quad}
  \mathbb{R}
\end{equation}
with the following properties, where as 
usual we denote skew-symmetrized product of $k$ Clifford generators by
  $$
    \Gamma_{a_1 \cdots a_k}
    \;:=\;
    \tfrac{1}{k!}
    \underset{
      \sigma \in
      \mathrm{Sym}(k)
    }{\sum}
    \mathrm{sgn}(\sigma)
    \,
    \Gamma_{a_{\sigma(1)}}
    \cdot
    \Gamma_{a_{\sigma(2)}}
    \cdots
    \Gamma_{a_{\sigma(n)}}
    \,:
  $$

\begin{itemize}[leftmargin=.4cm]
  \item
  the Clifford generators square to plus the Minkowski metric \eqref{MinkowskiMetric}
  \begin{equation}\label{CliffordDefiningRelation}
    \Gamma_a
    \Gamma_b
    +
    \Gamma_b
    \Gamma_a
    \;\;=\;\;
    +2 \, \eta_{a b}
    \,
    \mathrm{id}_{\mathbf{32}}
    \,,
  \end{equation}
  \item 
  the Clifford volume form equals the Levi-Civita symbol 
  \eqref{transversalizationOfLeviCivitaSymbol}:
  
  \begin{equation}
    \label{CliffordVolumeFormIn11d}
    \Gamma_{a_1 \cdots a_{11}}
    \;=\;
    \epsilon_{a_1 \cdots a_{11}}
    \mathrm{id}_{\mathbf{32}}
    \,,
  \end{equation}
  \item the Clifford generators are skew self-adjoint with respect to the pairing \eqref{TheSpinorPairing}
  \begin{equation}
    \label{SkewSelfAdjointnessOfCliffordGenerators}
    \overline{\Gamma_a}
    \;=\;
    - \Gamma_a
    \;\;\;\;\;\;
    \mbox{in that}
    \;\;\;\;\;\;
    \underset{
      \phi,\psi \in \mathbf{32}
    }{\forall}
    \;\;
    \big(
      \overline{(\Gamma_a \phi)}
      \,
      \psi
    \big)
    \;=\;
    -
    \big(
      \overline{\phi}
      \,
      (\Gamma_a \psi)
    \big)
    \,,
  \end{equation}
  so that generally
  \begin{equation}
    \label{AdjointnessOfCliffordBasisElements}
    \overline{\Gamma_{a_1 \cdots a_p}}
    \;=\;
    (-1)^{
      p + p(p-1)/2
    }
    \,
    \Gamma_{a_1 \cdots a_p}
    \,,
  \end{equation}
  \item
  the $\mathbb{R}$-vector space of $\mathbb{R}$-linear  endomorphisms of $\mathbf{32}$ has a linear basis given by the $\leq 5$-index Clifford elements 
  \begin{equation}
    \label{CliffordElementsSpanningLinearMaps}
    \mathrm{End}_{\mathbb{R}}\big(
      \mathbf{32}
    \big)
    \;\;
    =
    \;\;
    \Big\langle
      1
      ,\,
      \Gamma_{a_1}
      ,\,
      \Gamma_{a_1 a_2}
      ,\,
      \Gamma_{a_1, a_2, a_3}
      ,\,
      \Gamma_{a_1, \cdots a_4}
      ,\,
      \Gamma_{a_1, \cdots, a_5}
    \Big\rangle_{
      a_i = 0, 1, \cdots
    }
    \,,
  \end{equation}

  \item
  the $\mathbb{R}$-vector space of {\it symmetric} bilinear forms on $\mathbf{32}$
  has a linear basis given by the expectation values with respect to \eqref{TheSpinorPairing} of the 1-, 2-, and 5-index Clifford basis elements:
  \begin{equation}
    \label{SymmetricSpinorPairings}
    \mathrm{Hom}_{\mathbb{R}}
    \Big(
    (\mathbf{32}\otimes \mathbf{32})_{\mathrm{sym}}
    ,\,
    \mathbb{R}
    \Big)
    \;\;
    \simeq
    \;\;
    \Big\langle
    \big(
      (\overline{-})
      \Gamma_a
      (-)
    \big)
    \,,\;\;
    \big(
      (\overline{-})
      \Gamma_{a_1 a_2}
      (-)
    \big)
    \,,\;\;
    \big(
      (\overline{-})
      \Gamma_{a_1 \cdots a_5}
      (-)
    \big)
    \Big\rangle_{
      a_i = 0, 1, \cdots
    }
  \end{equation}
\end{itemize}

Generally we have the Clifford expansion formula:
\begin{equation}
  \label{GeneralCliffordProduct}
  \Gamma^{a_j \cdots a_1}
  \,
  \Gamma_{b_1 \cdots b_k}
  \;=\;
  \sum_{l = 0}^{
    \mathrm{min}(j,k)
  }
  \pm
  l!
\binom{j}{l}
 \binom{k}{l}
  \,
  \delta
   ^{[a_1 \cdots a_l}
   _{[b_1 \cdots b_l}
  \Gamma^{a_j \cdots a_{l+1}]}
  {}_{b_{l+1} \cdots b_k]}
\end{equation}

\medskip

\noindent
{\bf Spinors in 6d.}
In the literature (\cite[pp. 88]{Sorokin00}\cite[\S B]{BandosSorokin23}) the relevant spinor representations on the M5-brane are usually discussed by explicit matrix presentations.
Here we instead mean to give a transparent  algebraic account by projecting the relevant subrepresentations out of the 11d Majorana representation $\mathbf{32} \in \mathrm{Rep}_{\mathbb{R}}\big(\mathrm{Pin}^+(1,10)\big)$ \eqref{The11dMajoranaRepresentation}:

\medskip

\noindent
{\bf Algebraic reduction of Majorana $\mathbf{32}$ in 11d to Majorana-Weyl $2 \cdot \mathbf{8}$ in 6d.}
Consider an ortho-transversal linear basis of $\mathbb{R}^{1,10}$, decomposed as follows (where we declare the last line in a moment):
\begin{equation}
  \label{TangentialTransversalDecomposition}
  \def\arraystretch{1.3}
  \def\tabcolsep{4pt}
  \begin{array}{cccccccccccl}
    &&
    \mathclap{
      \;\;\;\;\;\;\;
      \overbrace{
        \phantom{------------}
      }^{
        \scalebox{.7}{
          ``tangential directions''
        }
      }
    }
    &&&&&&
    \mathclap{
      \!\!\!
      \overbrace{
        \phantom{----------}
      }^{
        \scalebox{.7}{
          ``transversal directions''
        }
      }
    }
    \\[-10pt]
    0 
    &
    1 
    &
    2 
    &
    3 
    &
    4 
    &
    5 
    &
    5' 
    &
    6 
    &
    7 
    &
    8 
    &
    9 
    &
    \\
    \Gamma_0
    &
    \Gamma_1
    &
    \Gamma_2
    &
    \Gamma_3
    &
    \Gamma_4
    &
    \Gamma_5
    &
    \Gamma_{5'}
    &
    \Gamma_{6}
    &
    \Gamma_{7}
    &
    \Gamma_{8}
    &
    \Gamma_{9}
    &
    \in 
    \;
    \mathrm{Pin}^+(1,10)
    \;\subset\; 
    \mathrm{End}_{\mathbb{R}}(\mathbf{32})
    \\
    \gamma_0
    &
    \gamma_1
    &
    \gamma_2
    &
    \gamma_3
    &
    \gamma_4
    &
    \gamma_5
    & &&&&
    &
    \in
    \;
    \mathrm{Pin}^+(1,5)
    \;\;\,\subset\;
    \mathrm{End}_{\mathbb{R}}
    \big(
      2 
        \cdot 
      \mathbf{8}_+
      \oplus
      2 \cdot
      \mathbf{8}_-
    \big).
  \end{array}
\end{equation}
Observing that
$  \big(
    \Gamma_{5' 6 7 8 9}
  \big)^2
  \;=\;
  (-1)^{5(5-1)/2}
  \;=\;
  +
  1
$
we obtain a projection operator
\begin{equation}
  \label{TheProjectionOperator}
  P
  \;:=\;
  \tfrac{1}{2}
  (1 + \Gamma_{5'6789})
  \;\;
  \in
  \;\;
  \mathrm{End}_{\mathbb{R}}(\mathbf{32})
  \,.
\end{equation}
Since the Dirac adjoint of $\Gamma_{012345}$ is
\eqref{AdjointnessOfCliffordBasisElements}
$$
  \overline{
    \Gamma_{5'6789}
  }
  \;=\;
  \underbrace{
  (-1)^{
    5 
    \,+\, 
    5(5-1)/2 
  }
  }_{-1}
  \Gamma_{5'6789}
  \,,
$$
the Dirac adjoint of the projector \eqref{TheProjectionOperator} is its complementary projector
\begin{equation}
  \label{transversalProjectionOperator}
  \overline{P}
  \;=\;
  \tfrac{1}{2}
  \big(
    1
    -
    \Gamma_{5'6789}
  \big)
  \;\;\;
  \in
  \;\;
  \mathrm{End}_{\mathbb{R}}(\mathbf{32})
  \,,
\end{equation}
which is related to $P$ by the following simple but crucial relations:
\begin{equation}
  \label{CommutativityOfGammaWithP}
  \def\arraystretch{1.4}
  \begin{array}{ll}
    P P \;=\; P\,,
    &
    \overline{P} P \;=\; 0\,,
    \\
    \overline{P} \overline{P}
    \;=\;   
    \overline{P}
    \,,
    &
    P \overline{P} \;=\; 0
  \end{array}
  \;\;\;\;\;\;
  \mbox{and}
  \;\;\;\;\;\;
  \left\{\!\!
  \def\arraystretch{1.4}
  \begin{array}{ll}
    \Gamma^a P 
    \;=\;
    \overline{P} \, \Gamma^a 
    &
    \mbox{
      for tangential $a$
    }
    \\
    \Gamma^a P 
    \;=\;
    P \, \Gamma^a 
    &
    \mbox{
      for transversal $a$.
    }    
  \end{array}
  \right.
\end{equation}

\medskip

Therefore these projection operators \eqref{TheProjectionOperator} \eqref{transversalProjectionOperator} carve out a pair of chiral
representations of $\mathrm{Spin}(1,5) \hookrightarrow \mathrm{Pin}^+(1,10)$ 
inside $\mathbf{32} \,\in\, \mathrm{Rep}_{\mathbb{R}}\big( \mathrm{Spin}(1,10) \big)$. 

But since also
$
  \big(
    \Gamma_{6789}
  \big)^2
  \;=\;
  (-1)^{4(4-1)/2}
  \,=\,
  +1
$
and
$
  \big(
    \Gamma_{5'}
  \big)^2
  \,=\,
  +1
$
there are two further pairs of projectors
$$
  \tfrac{1}{2}
  \big(
    1 
    \pm
    \Gamma_{6789}
  \big)
  \,,
  \;\;\;\;\;\;\;
  \mbox{and}
  \;\;\;\;\;\;\;
  \tfrac{1}{2}
  \big(
    1
    \pm 
    \Gamma_{5'}
  \big)
$$
which commute with these $\mathrm{Spin}(1,5)$-actions on $P(\mathbf{32})$ and on $\overline{P}(\mathbf{32})$, thus decomposing them according to
$$
 \def\arraystretch{1.7}
 \begin{array}{l}
  P
  \;\;=\;\;
  \tfrac{1}{2}\big(
    1 
    +
    \Gamma_{6789}
  \big)
  \,
  \tfrac{1}{2}\big(
    1 
    +
    \Gamma_{5'}
  \big)
  +
  \tfrac{1}{2}\big(
    1 
    -
    \Gamma_{6789}
  \big)
  \,
  \tfrac{1}{2}\big(
    1 
    -
    \Gamma_{5'}
  \big)
  \\
  \overline{P}
  \;\;=\;\;
  \tfrac{1}{2}\big(
    1 
    -
    \Gamma_{6789}
  \big)
  \,
  \tfrac{1}{2}\big(
    1 
    +
    \Gamma_{5'}
  \big)
  +
  \tfrac{1}{2}\big(
    1 
    +
    \Gamma_{6789}
  \big)
  \,
  \tfrac{1}{2}\big(
    1 
    -
    \Gamma_{5'}
  \big)
  \,,
  \end{array}
$$
each into a pair of isomorphic summands, whose isomorphism is given by acting for instance with $\Gamma_6$:
$$
  \def\arraystretch{1.6}
  \begin{array}{l}
  \Gamma_6
  \;
  \tfrac{1}{2}
  \big(
    1
    +
    \Gamma_{6789}
  \big)
  \,
  \tfrac{1}{2}
  \big(
    1
    +
    \Gamma_{5'}
  \big)
  \;=\;
  \tfrac{1}{2}
  \big(
    1
    -
    \Gamma_{6789}
  \big)
  \,
  \tfrac{1}{2}
  \big(
    1
    -
    \Gamma_{5'}
  \big)
  \;
  \Gamma_6
  \\
  \Gamma_6
  \;
  \tfrac{1}{2}
  \big(
    1
    -
    \Gamma_{6789}
  \big)
  \,
  \tfrac{1}{2}
  \big(
    1
    +
    \Gamma_{5'}
  \big)
  \;=\;
  \tfrac{1}{2}
  \big(
    1
    +
    \Gamma_{6789}
  \big)
  \,
  \tfrac{1}{2}
  \big(
    1
    -
    \Gamma_{5'}
  \big)
  \;
  \Gamma_6
  \,.
  \end{array}
$$
In conclusion this exhibits (cf. e.g. \cite[Lem. 4.12]{HSS19}) the irrep decomposition
\begin{equation}
  \label{TheChiralSpinRepresentations}
  P(\mathbf{32})
  \;\;\;
  \underset
    {
      \mathclap{
        \scalebox{.6}{
          $\mathrm{Spin}(1,5)$
        }
      }
    }
    {\simeq}
  \;\;\;
  2 \cdot \mathbf{8}_+
  \,,
  \;\;\;\;\;\;
  \overline{P}(\mathbf{32})
  \;\;\;
  \underset
    {
      \mathclap{
        \scalebox{.6}{
          $\mathrm{Spin}(1,5)$
        }
      }
    }
    {\simeq}
  \;\;\;
  2 \cdot \mathbf{8}_-
  \;\;\;\;\;\;\;
  \in
  \,
  \mathrm{Rep}_{\mathbb{R}}
  \big(
    \mathrm{Spin}(1,5)
  \big)
\end{equation}
with respect to the tangential Clifford generators $\Gamma_0,\, \Gamma_1, \cdots \Gamma_5 \,\in\, \mathrm{Pin}^+(1,10)$, which as such we denote $\gamma_0, \, \cdots, \gamma_5$:
\begin{equation}
  \label{TangentialCliffordGenerators}
  \gamma_a
  \;\;\;:\;\;
  \begin{tikzcd}[sep=2pt]
    2 \cdot \mathbf{8}_+
    &
    \underset{
      \mathclap{
        \scalebox{.55}{
          $\mathrm{Spin}(1,5)$
        }
      }
    }
    {
      \simeq
    }
    &
    P(\mathbf{32})
    \ar[
      ddrr
    ]
    &\phantom{-----}&
    P(\mathbf{32})
    &
    \underset{
      \mathclap{
        \scalebox{.55}{
          $\mathrm{Spin}(1,5)$
        }
      }
    }
    {
      \simeq
    }
    &
    2 \cdot \mathbf{8}_+
    \\[-7pt]
    \oplus
    &&
    \oplus
    &&
    \oplus
    &&
    \oplus
    \\
    2 \cdot \mathbf{8}_-
    &
    \underset{
      \mathclap{
        \scalebox{.55}{
          $\mathrm{Spin}(1,5)$
        }
      }
    }
    {
      \simeq
    }
    &
    \overline{P}(\mathbf{32})
    \ar[
      uurr,
      crossing over,
      "{
        \Gamma_a
      }"{description}
    ]
    &&
    \overline{P}(\mathbf{32})
    &
    \underset{
      \mathclap{
        \scalebox{.55}{
          $\mathrm{Spin}(1,5)$
        }
      }
    }
    {
      \simeq
    }
    &
    2 \cdot \mathbf{8}_-
  \end{tikzcd}
  \;\;\;\;\;
  \mbox{
    for tangential $a$.
  }
\end{equation}
With \eqref{CommutativityOfGammaWithP} this implies
that the projection operators 
\eqref{TheProjectionOperator} \eqref{transversalProjectionOperator} 
also serve to project out these tangential Clifford generators:
\begin{equation}
  \label{ProjectionOfCliffordElements}
  \overline{P} \, \Gamma^a \, P
  \;=\;
  \left\{\!\!
  \def\arraystretch{1.3}
  \begin{array}{ll}
    \gamma^a{}_{\vert_{2 \cdot \mathbf{8}_+}}
    &
    \mbox{for tangential $a$}
    \\
    0 & \mbox{for transversal $a$}
    \,,
  \end{array}
  \right.
  \;\;\;\;\;\;\;\;\;\;
  P \, \Gamma^a \, \overline{P}
  \;=\;
  \left\{\!\!
  \def\arraystretch{1.3}
  \begin{array}{ll}
    \gamma^a{}_{\vert_{2 \cdot \mathbf{8}_-}}
    &
    \mbox{for tangential $a$}
    \\
    0 & \mbox{for transversal $a$}.
  \end{array}
  \right.
\end{equation}
Therefore we may think of $P$ as also acting on $\mathbb{R}^{1,10}$ by projection to $\mathbb{R}^{1,5}$, hence as acting on all of the super-vector space $\mathbb{R}^{1,10\vert \mathbf{32}}$
\begin{equation}
  \label{SuperProjectionOperator}
  \begin{tikzcd}
    \mathbb{R}^{
      1,10\vert \mathbf{32}
    }
    \ar[
      r,
      ->>
    ]
    \ar[
      rr,
      rounded corners,
      to path={
           ([yshift=+00pt]\tikztostart.south)  
        -- ([yshift=-08pt]\tikztostart.south)
        -- node[]{
          \scalebox{.8}{
            \colorbox{white}{
              $P$
            }
          }
        }
           ([yshift=-08pt]\tikztotarget.south)
        -- ([yshift=-00pt]\tikztotarget.south)
      }
    ]
    &
    \mathbb{R}^{1,5\vert 2 \cdot \mathbf{8}}
    \ar[
      r,
      hook
    ]
    &
    \mathbb{R}^{
      1,10\vert \mathbf{32}
    }
    \mathrlap{\,.}
  \end{tikzcd}
\end{equation}
Below we make much use of this super-projector \eqref{SuperProjectionOperator} for streamlined definitions and computations.

A simple example that will be useful below:

\begin{example}
It follows immediately from \eqref{SymmetricSpinorPairings} that $P(\mathbf{32})$ inherits {\it symmetric} bilinear spinor pairings such as
$$
  \begin{tikzcd}
  \big(
  (2 \cdot \mathbf{8}_+)
  \otimes
  (2 \cdot \mathbf{8}_+)
  \big)_{\mathrm{sym}}
  \ar[
    r,
    hook
  ]
  &
  \big(
  \mathbf{32}
  \otimes
  \mathbf{32}
  \big)_{\mathrm{sym}}
  \ar[
    rrr,
    "{
      \big(\,
        \overline{(-)}
        \,\gamma^a \Gamma^b\,
        (-)
      \big)
    }"
  ]
  &&&
  \mathbb{R}
  \end{tikzcd}
  \;\;\;\;
  \mbox{for}
  \;
  \def\arraystretch{.9}
  \begin{array}{l}
    \mbox{tangential $a$}
    \\
    \mbox{transversal $b$,}
  \end{array}
$$
so that, notably, the polarization identity implies that
\begin{equation}
  \label{VanishingOfOneIndexCliffordElementInQuadraticForm}
  \underset{
    \psi \in P(\mathbf{32})
  }{
    \forall
  }
  \;\;
  \big(\,
    \overline{\psi}
    \, \gamma^a \, \Gamma^b\,
    \psi
  \big)
  (H_1)_a
  \,=\,
  0
  \hspace{.8cm}
  \Leftrightarrow
  \hspace{.8cm}
  H_1 \,=\, 0
  \,,
\end{equation}
which we need below in \eqref{ComputingTorsionConstraintAtPsi2}.

On the other hand, the corresponding symmetric pairing with two tangential indices vanishes identically
\begin{equation}  \label{TangentialExpectationValueOf2IndexCliffordElementVanishes}
  \psi_1,
  \psi_2
  \;\in\;
  P(\mathbf{32})
  \;\;\;\;\;\;
  \Rightarrow
  \;\;\;\;\;\;
  \big(\,
    \overline{\psi_1}
    \,
    \gamma_{a_1 a_2}
    \,
    \psi_2
  \big)
  \;=\;
  0
  \,,
\end{equation}
because
$$
    \big(\,
      \overline{\psi_1}
      \,\gamma_{a_1 a_2}\,
      \psi_2
    \big)
    \;=\;
    \big(\,
      \overline{P \psi_1}
      \,\gamma_{a_1 a_2}\,
      P \psi_2
    \big)
    \;=\;
    \big(\,
      \overline{\psi_1}
      \,
      \overline{P}
      \,\gamma_{a_1 a_2}\,
      P \psi_2
    \big)
    \;=\;
    \big(\,
      \overline{\psi_1}
      \,
      \,\gamma_{a_1 a_2}\,
      \overline{P}
      P
      \,
      \psi_2
    \big)
    \;= \; 0
    \,.
$$
\end{example}

\medskip

\noindent
{\bf Residual $\mathrm{Spin}(5)$-action.}
Moreover, the fact that the transverse Clifford elements commute with $P$ and $\overline{P}$ \eqref{CommutativityOfGammaWithP} immediately implies that $P(\mathbf{32})$ and $\overline{P}(\mathbf{32})$ also inherit an action of the transverse $\mathrm{Spin}(5) \xhookrightarrow{\;} \mathrm{Pin}^+(1,10)$ with respect to the direct product subgroup inclusion
$
  \mathrm{Spin}(1,5)
  \times
  \mathrm{Spin}(5)
  \xhookrightarrow{\;}
  \mathrm{Pin}^+(1,10)
  \,,
$
in fact as $\mathrm{Spin}(5)$-representations they are isomorphic, for instance via multiplication by $\Gamma_0$ (or any other of the tangential $\Gamma_{a \leq 5}$)
$$
  \begin{tikzcd}[
    row sep=0pt, column sep=large
  ]
    P(\mathbf{32})
    \ar[
      out=180-65, 
      in=65, 
      looseness=3, 
    "
      \scalebox{.7}{
      $\mathrm{Spin}(5)$
      }
    "
  ]    
  \ar[
      r,
      "{ \sim }"
    ]
    &
    \overline{P}(\mathbf{32})
    \ar[
      out=180-65, 
      in=65, 
      looseness=3, 
    "
      \scalebox{.7}{
      $\mathrm{Spin}(5)$
      }
    "
  ]    
    \\
    P\Psi
    \ar[
      r,
      |->,
      "{
        \Gamma_0 \cdot (-)
      }"
    ]
    &
    \overline{P}
    \gamma^0
    \Psi
    \\    
  \end{tikzcd}
$$
Moreover, as a $\mathrm{Spin}(5)$-representation, $P(\mathbf{32})$ now decomposes into representations in the images of the four mutually orthogonal projectors
$$
  P_{\sigma_1 \sigma_2}
  \;:=\;
  \tfrac{1}{2}
  \big(
    1
    +
    \sigma_1
    \Gamma_{0 1}
  \big)
  \tfrac{1}{2}
  \big(
    1
    +
    \sigma_2
    \Gamma_{2345}
  \big)
  \,,
  \;\;\;\;
  \sigma_i \in \{\pm 1\}\,.
$$
These are all isomorphic to each other, for instance via
$$
  \def\arraystretch{1.2}
  \begin{array}{l}
  \Gamma_1
  \,
  P_{\sigma_1,  \sigma_2}
  \;=\;
  P_{-\sigma_1, \sigma_2}
  \,
  \Gamma_1
  \,,
  \\
  \Gamma_2
  \,
  P_{\sigma_1,  \sigma_2}
  \;=\;
  P_{\sigma_1, -\sigma_2}
  \,
  \Gamma_2
  \,,
  \end{array}
$$
and thus to be denoted $\mathbf{4} \in \mathrm{Rep}_{\mathbb{R}}\big(\mathrm{Spin}(5)\big)$:
$$
  P(\mathbf{32})
  \;\;
  \simeq_{{}_{\mathbb{R}}}
  \;\;
  \big(
    P_{++} + P_{+-} + P_{-+} + P_{--}
  \big)
  P(\mathbf{32})
  \;\;
  \underset{
    \mathclap{
      \scalebox{.6}{
        $\mathrm{Spin}(5)$
      }
    }
  }{
    \simeq
  }
  \;\;
  4 \cdot \mathbf{4}
  \,.
$$
Consequently, we have   
\begin{equation}
  \label{TransverseCliffordAction}
  \Gamma_a
  \;\;\;:\;\;
  \begin{tikzcd}[sep=-2pt]
    4 \cdot \mathbf{4}
    &
      \underset{
        \mathclap{
        \scalebox{.6}{$
        \mathrm{Spin}(5)
        $}
        }
      }
      {\simeq}
    &
    P(\mathbf{32})
    \ar[
      rr,
      "{
        \Gamma_a
      }"
    ]
    &\phantom{-----}&
    P(\mathbf{32})
    &
      \underset{
        \mathclap{
        \scalebox{.6}{$
        \mathrm{Spin}(5)
        $}
        }
      }
      {\simeq}
    &
    4 \cdot \mathbf{4}
    \\[-5pt]
    \oplus
    &&
    \oplus
    &&
    \oplus
    &&
    \oplus
    \\
    4 \cdot 
    \mathbf{4}
    &
      \underset{
        \mathclap{
        \scalebox{.6}{$
        \mathrm{Spin}(5)
        $}
        }
      }
      {\simeq}
    &
    \overline{P}(\mathbf{32})
    \ar[
      rr,
      "{
        \Gamma_a
      }"
    ]
    &&
    \overline{P}(\mathbf{32})
    &
      \underset{
        \mathclap{
        \scalebox{.6}{$
        \mathrm{Spin}(5)
        $}
        }
      }
      {\simeq}
    &
    2 \cdot \mathbf{4}
  \end{tikzcd}
  \;\;\;\;\;
  \mbox{
    for transverse $a$.
  }
\end{equation}

In summary this identifies the $\mathrm{Spin}(1,5) \times \mathrm{Spin}(5)$-action on $P(\mathbf{32})$ and $\overline{P}(\mathbf{32})$ as:
\vspace{1mm} 
\begin{equation}
  \label{IrrepDecompositionOf6dSpinors}
  2 \cdot \mathbf{8}_+
  \;\;\;
  \underset
    {
      \mathclap{
        \scalebox{.6}{
          $\mathrm{Spin}(1,5)$
        }
      }
    }
    {\simeq}
  \;\;
  P(\mathbf{32})
  \;\;
  \underset
    {
      \mathclap{
        \scalebox{.6}{
          $\mathrm{Spin}(5)$
        }
      }
    }
    {\simeq}
  \;\;
  4 \cdot \mathbf{4}
  \;\;
  \underset
    {
      \mathclap{
        \scalebox{.6}{
          $\mathrm{Spin}(5)$
        }
      }
    }
    {\simeq}
  \;\;
  \overline{P}(\mathbf{32})
  \;\;\,
  \underset
    {
      \mathclap{
        \scalebox{.6}{
          $\mathrm{Spin}(1,5)$
        }
      }
    }
    {\simeq}
  \;\;
  2 \cdot \mathbf{8}_-
  \,.
\end{equation}

\medskip

\noindent {\bf Hodge duality in 6d.}
From the relation 
$
  \Gamma_{ 0 1 2 3 4 5 5' 6 7 8 9}
  \;\;
  =
  \;\;
  1
  \;\;
  \in
  \;
  \mathrm{End}_{\mathbb{R}}(\mathbf{32})
$
\eqref{CliffordVolumeFormIn11d}
it follows analogously
for the tangential Clifford algebra \eqref{TangentialCliffordGenerators} that
\begin{equation}
  \label{6dChiralityOperator}
  \gamma_{ 0 1 2 3 4 5}
  \;\;
  =
  \;\;
  1
  \;\;
  \in
  \;
  \mathrm{End}_{\mathbb{R}}(2 \cdot \mathbf{8})
  \,,
\end{equation}
because
$$
  \phi \;\in\;
  2\cdot \mathbf{8}
  \;\subset\;
  \mathbf{32}
  \;\;\;\;\;\;\;\;
  \Rightarrow
  \;\;\;\;\;\;\;\;
  \def\arraystretch{1.2}
  \begin{array}{lll}
    \Gamma_{012345} \, \phi
    &
    =\;
    \Gamma_{012345}
    \,
    \Gamma_{5' 6789}
    \,
    \phi
    &
    \proofstep{
      by
      \eqref{TheChiralSpinRepresentations}
    }
    \\
   & =\; \phi
    &
    \proofstep{
      by
      \eqref{CliffordVolumeFormIn11d}.
    }
  \end{array}
$$
in fact
$$
  \Gamma_{012345}
  \;=\;
  \Gamma_{012345}
  \Gamma_{0123455'6789}
  \;=\;
  \Gamma_{5'6789}
  \;\;
  \in
  \;\;
  \mathrm{End}_{\mathbb{R}}(\mathbf{32})
  \,.
$$
By \eqref{6dChiralityOperator}, the Hodge duality relations on the 6d Clifford basis elements are as follows:
\begin{equation}
  \label{HodgeDualityOf6dCliffordGenertors}
  \def\arraystretch{1.3}
  \def\arraycolsep{0pt}
  \begin{array}{rcll}
    \gamma^{a_1 a_2 a_3 a_4 a_5 a_6}
    &\;\;=\;\;&
    +
    \;\;\;
    \epsilon^{a_1 a_2 a_3 a_4 a_5 a_6}
    &
    1
    \\
    \gamma^{a_1 a_2 a_3 a_4 a_5}
    &\;\;=\;\;&
    +
    \;\;\;
    \epsilon^{a_1 a_2 a_3 a_4 a_5\, b}
    &
    \gamma_{b}
    \\
    \gamma^{a_1 a_2 a_3 a_4}
    &\;\;=\;\;&
    -
    \tfrac{1}{2}
    \epsilon^{a_1 a_2 a_3 a_4\, b_1 b_2}
    &
    \gamma_{b_1 b_2}
    \\
    \gamma^{a_1 a_2 a_3}
    &\;\;=\;\;&
    -
    \tfrac{1}{3!}
    \epsilon^{
      a_1 a_2 a_3 
      \,
      b_1 b_2 b_3
    }
    &
    \gamma_{b_1 b_2 b_3}
    \\
    \gamma^{a_1 a_2}
    &\;\;=\;\;&
    +
    \tfrac{1}{4!}
    \epsilon^{
      a_1 a_2 
      \,
      b_1 b_2 b_3 b_4
    }
    &
    \gamma_{b_1 b_2 b_3 b_4}
    \\
    \gamma_{a_1}
    &\;\;=\;\;&
    +
    \tfrac{1}{5!}
    \epsilon^{
      a_1 
      \, 
      b_1 b_2 b_3 b_4 b_5
    }
    &
    \gamma_{b_1 b_2 b_3 b_4 b_5}
    \\
    1
    &\;\;=\;\;&
    -
    \tfrac{1}{6!}
    \epsilon^{b_1 b_2 b_3 b_4 b_5 b_6}
    &
    \gamma_{b_1 b_2 b_3 b_4 b_5 b_6}\,.
  \end{array}
\end{equation}

\medskip

\medskip

\noindent
{\bf Special Clifford relations in 6d.}
From \eqref{HodgeDualityOf6dCliffordGenertors} the following Lem. \ref{SelfDualityOf3IndexCoefficientsIn6d} and Lem. \ref{ExpandingAntiChiralOperatorsIntoCliffordElements} are immediate but of key importance:

\begin{lemma}[\bf Self-duality of 3-index coefficients]
\label{SelfDualityOf3IndexCoefficientsIn6d}
The coefficients of $\gamma^{a_1 a_2 a_3}$ are being projected onto their self-dual part:
\begin{equation}
  \label{CoefficientsOf3Index6dGammaAreSelfDual}
  (\tilde H_3)_{a_1 a_2 a_3}
  \gamma^{a_1 a_2 a_3}
  \;=\;
  \tfrac{1}{2}
  \Big(
    (\tilde H_3)_{a_1 a_2 a_3}
    +
    (\star \tilde H_3)_{a_1 a_2 a_3}
  \Big)
  \gamma^{a_1 a_2 a_3}.
\end{equation}
\end{lemma}
\begin{proof}
$\,$

\vspace{-4mm} 
$$
  \def\arraystretch{1.7}
  \begin{array}{lll}
    (\tilde H_3)_{a_1 a_2 a_3}
    \gamma^{a_1 a_2 a_3}
    &
    \;=\;
    (\tilde H_3)_{a_1 a_2 a_3}
    \Big(
    \tfrac{-1}{3!}
    \epsilon^{
      a_1 a_2 a_3
      \,
      b_1 b_2 b_3
    }
    \gamma_{b_1 b_2 b_3}
    \Big)
    &
    \proofstep{
      by
      \eqref{HodgeDualityOf6dCliffordGenertors}
    }
    \\
   & \;=\;
    \Big(
    \tfrac{+1}{3!}
    \epsilon^{
      b_1 b_2 b_3
      \,
      a_1 a_2 a_3
    }
    (\tilde H_3)_{a_1 a_2 a_3}
    \Big)
    \gamma_{b_1 b_2 b_3}
    \\
   & \;=\;
    (\star \tilde H_3)^{b_1 b_2 b_3}
    \gamma_{b_1 b_2 b_3}\;.
  \end{array}
$$

\vspace{-5mm} 
\end{proof}

\begin{lemma}[\bf Expanding anti-chiral operators into 6+5d Clifford elements]
  \label{ExpandingAntiChiralOperatorsIntoCliffordElements}
  The $\mathbb{R}$-vector space of linear maps $2 \cdot \mathbf{8}_{\pm} \xrightarrow{\;} 2 \cdot \mathbf{8}_{\mp}$ is spanned by products with any transverse Clifford elements in $\mathrm{Spin}(5)$
  of the tangential 1-index and the self-dual combination of tangential 3-index Clifford elements \eqref{TangentialCliffordGenerators}:
  \begin{equation}
    \label{ExpandingChiralLinearMapsInCliffordElements}
    \mathrm{Hom}_{\mathbb{R}}
    \big(
      2\cdot \mathbf{8}_\pm
      ,\,
      2\cdot \mathbf{8}_\mp
    \big)
    \;\;
    \simeq
    \;\;
    \Big\langle
      \gamma_{a_1}
      ,\,
      \tfrac{1}{2}
      \big(
      \gamma_{a_1 a_2 a_3}
      +
      \tfrac{1}{3!}
      \epsilon_{
        a_1 a_2 a_3
        \,
        b_1 b_2 b_3
      }
      \gamma^{b_1 b_2 b_3}
      \big)
    \Big\rangle_{a_i \in \{0,1, \cdots, 6\}}
    \cdot
    \mathrm{Spin}(5)
    \,.
  \end{equation}
\end{lemma}
\begin{proof}
 By \eqref{CliffordElementsSpanningLinearMaps} any such linear map is the linear combination of $\mathrm{Pin}^+(1,10)$-elements $\Gamma_{a_1 \cdots a_{\leq 5}}$, and among these appear precisely only those with an odd number of tangential indices, 
 by \eqref{TangentialCliffordGenerators} and \eqref{TransverseCliffordAction}, where by 6d Hodge duality \eqref{HodgeDualityOf6dCliffordGenertors}
 those with 5 tangential indices
 and those with anti self-dual combinations of 3-indices 
 may be omitted.
\end{proof}

\smallskip
The following claim \eqref{ConjugatingGammaWithH3} is implicit in \cite[p. 2]{HoweSezgin97b} and explicit in \cite[(5.67)]{Sorokin00}\cite[(5.77)]{BandosSorokin23} (stated there for a specific matrix representation); we spell out a proof. 

\begin{lemma}[\bf Conjugating $\gamma$ with $H_3$]
If $\tilde H_3 \;=\; \star \tilde H_3$ 
then 
\begin{equation}
  \label{ConjugatingGammaWithH3}
    \Big(
      \tfrac{1}{3!}
      (\tilde H_3)_{b_1 b_2 b_3}
      \gamma^{b_1 b_2 b_3}
    \Big)
    \,
    \gamma_a
    \,
    \Big(
      \tfrac{1}{3!}
      (\tilde H_3)_{c_1 c_2 c_3}
      \gamma^{c_1 c_2 c_3}
    \Big)
  \;\;
  =
  \;\;
  -2
  \,
  (\tilde H_3)_{
    a
    \,
    b_1 b_2
  }
  (\tilde H_3)^{
    a'
    \,
    b_1 b_2
  }
  \,
  \gamma_{a'}
  \;\;
  \defneq
  \;\;
  -2
  \,
  (\K)_a^{a'}
  \,
  \gamma_{a'}
  \,.
\end{equation}
\end{lemma}
\begin{proof}
First, Gamma-expansion gives
$$
  \def\arraystretch{1.6}
  \begin{array}{ll}
    \Big(
      \tfrac{1}{3!}
      (\tilde H_3)_{b_1 b_2 b_3}
      \gamma^{b_1 b_2 b_3}
    \Big)
    \,
    \gamma_a
    \,
    \Big(
      \tfrac{1}{3!}
      (\tilde H_3)_{c_1 c_2 c_3}
      \gamma^{c_1 c_2 c_3}
    \Big)
    &
    =\;
    -
    (\tilde H_3)_{a \, c_1 c_2}
    (\tilde H_3)^{b \, c_1 c_2}
    \,
    \gamma_b
    \;+\;
    \tfrac{1}{6}
    \underbrace{
      (\tilde H_3)_{c_1 c_2 c_3}
      (\tilde H_3)^{c_1 c_2 c_3}
    }_{
     \underset{
       \scalebox{.65}{
         \eqref{SquareOfSelfDual3FormVanishes}
       }
     }{=} 
     0
    }
    \,
    \gamma_a
    \\
  &  \;\;\;\;
    \;+\;
    \tfrac{1}{6}
    (\tilde H_3)_{c_1 c_2 c_3}
    (\tilde H_3)_{a c_4 c_5}
    \,
    \gamma^{c_1 \cdots c_5}
    \;-\;
    \tfrac{1}{4}
    (\tilde H_3)^{b\, c_1 c_2 }
    (\tilde H_3)_{b}{}^{c_3 c_4 }
    \,
    \gamma_{a\, c_1 \cdots c_4}
    \\
  &  \;\;\;\;
    \underbrace{
    \;-\;
    \tfrac{1}{2}
    (\tilde H_3)^{b_1 c_1 c_2}
    (\tilde H_3)^{b_2}{}_{c_1 c_2}
    \,
    \gamma_{a\, b_1 b_2}
    \mathrlap{\,,}
    }_{= 0}
  \end{array}
$$
where over the braces we noticed terms that vanish for symmetry reasons.
Applying \eqref{ContractionOfTildeH3On2IndicesWith5IndexGamma} to the remaining 5-index terms and then simplifying the coefficients of the only remaining 1-index term yields the claimed result.
\\
(A computer algebra check is available in \cite{AncillaryFile}.)
\end{proof}

\begin{lemma}[\bf Contraction with $\gamma$]
If $\tilde H_3 \,=\, \star \tilde H_3$ then
\begin{equation}
\label{ContractionOfTildeH3On2IndicesWith5IndexGamma}
  \def\arraystretch{1.7}
  \begin{array}{l}
    (\tilde H_3)_{a \, c_4 c_5}
    \gamma^{c_1 \cdots c_5}
    \;=\;
    \tfrac{4! \cdot 2!}{3!}
    \,
    (\tilde H_3)^{b_1 b_2 b_3}
    \,
    \delta
      ^{ c_1 c_2 c_3 \, a' }
      _{ b_1 b_2 b_3 \, a  }
    \,
    \gamma_{a'}
  \end{array}
\end{equation}
\end{lemma}
\begin{proof}
$$
  \def\arraystretch{1.7}
  \begin{array}{lll}
    (\tilde H_3)_{a \, c_4 c_5}
    \gamma^{c_1 \cdots c_5}
    &
    \;=\;
    \,
    \Big(
    \tfrac{1}{3!}
    (\tilde H_3)^{b_1 b_2 b_3}
    \epsilon_{
      a \, c_4 c_5
      \,
      b_1 b_2 b_3
    }  
    \Big)
    \Big(
      \epsilon^{c_1 c_2 c_3 c_4 c_5 a'}
      \gamma_{a'}
    \Big)
    &
    \proofstep{
      by
      \eqref{HodgeDualityOf6dCliffordGenertors}
    }
    \\
   & \;=\;
    -
    \tfrac{4! \cdot 2!}{3!}
    \,
    (\tilde H_3)^{b_1 b_2 b_3}
    \,
    \delta
      ^{ c_1 c_2 c_3 \, a' }
      _{ a \, b_1 b_2 b_3  }
    \,
    \gamma_{a'}
    &
    \proofstep{
      by
      \eqref{ContractingKroneckerWithSkewSymmetricTensor}
    }.
  \end{array}
$$

\vspace{-4mm} 
\end{proof}

\subsection{Super-flux on M5 immersions}
\label{TheM5EquationsOfMotion}
With the above preliminaries in hand we are now ready to work out  the characterization of those $\sfrac{1}{2}$BPS super-immersions that correspond to M5-branes.

\begin{definition}[\bf M5-brane super-immersion]
  \label{M5SuperImmersion}
  Given a super-spacetime $\big(X, (E,\Psi,\Omega)\big)$ of super-dimension $(1,10)\,\vert\, \mathbf{32}$,
  we say that an {\it M5-brane super-immersion} into $X$ is a $\sfrac{1}{2}$BPS super-immersion (``super-embedding'', Def. \ref{BPSImmersion}) of a super-manifold $\Sigma$ of bosonic dimension $1+5$, hence \eqref{TangentialCliffordGenerators} of super-dimension 
  $1+5 \,\vert\, 2\cdot \mathbf{8}$.
  $$
    \begin{tikzcd}[column sep=large]
      \mathllap{
        \scalebox{.7}{
          \color{darkblue}
          \bf
          \def\arraystretch{.9}
          \begin{tabular}{c}
            M5 brane
            \\
            super-worldvolume
          \end{tabular}
        }
      }
      \Sigma^{1,5\vert 2\cdot \mathbf{8}}
      \ar[
        rr,
        "{ \phi }",
        "{
          \scalebox{.7}{
            \color{darkgreen}
            \bf
            \def\arraystretch{.85}
            \begin{tabular}{c}
              $\sfrac{1}{2}$BPS
              \\
              super-immersion
            \end{tabular}
          }
        }"{swap}
      ]
      &&
      X^{1,10\vert \mathbf{32}}
      \mathrlap{
        \scalebox{.7}{
          \color{darkblue}
          \bf
          \def\arraystretch{.}
          \begin{tabular}{c}
            11d supergravity
            \\
            super-spacetime
          \end{tabular}
        }
      }
    \end{tikzcd}
  $$
  
\end{definition}

\begin{remark}[\bf Transversal fermionic shear of M5 super-immersions]
\label{TransversalFermionicShear}
Given an M5 super-immersion $\phi$ (Def. \ref{M5SuperImmersion}) the assumption 
\eqref{BPSImmersionDiagram}
that 
 $(e,\psi)$ is a co-frame field on $\Sigma$ implies (cf. \cite[(5.69)]{BandosSorokin23}) that there exist unique components fields
 $\slashed{\tilde H}$ and $\tau_a$ on $\Sigma$
 which parametrize the transversal shear of $\phi$ in the fermionic directions:
\begin{equation}
  \label{transversalFermionicComponent}
  \phi^\ast \Psi
  -
  \psi
  \;=\;
  \phi^\ast(\overline{P}\Psi)
  \;=\;
  \slashed{\tilde H} \psi
  +
  \tau_a 
  \,
  e^a
  \,.
\end{equation}
These components pointwise map the fermionic tangent space $T^{\mathrm{odd}}_{\sigma} \Sigma \cong 2 \cdot \mathbf{8}_+$ of $\Sigma$ into the transverse target space copy of $2 \cdot \mathbf{8}_-$:
\begin{equation}
  \label{ShearingOfSuperEmbeddings}
  \begin{tikzcd}[
    column sep=68pt
  ]
    &&
    \overbrace{
    2\cdot \mathbf{8}_+
    \oplus
    2\cdot \mathbf{8}_-
    }^{
      \mathbf{32}
    }
    \ar[
      d,
      ->>,
      shift right=20pt
    ]
    \\
    T_{\sigma}\Sigma
    \ar[
      urr,
      bend left=5pt,
      "{
        \phi^\ast
        \Psi
        \;=\;
        (
          \psi
          ,\,
          \overbrace{
            \scalebox{1}{$
            \slashed{\tilde H}
            \psi
            \,+\,
            \tau_a e^a
            $}
          }^{
            \mathclap{
            \scalebox{.7}{
              \color{darkblue}
              \bf
              odd shear 
            }
            }
          }
        )      
      }"{sloped}
    ]
    \ar[
      rr,
      "{
        \psi
      }",
      "{
        \scalebox{.7}{
          \color{darkgreen}
          \bf
          odd co-frame
        }
      }"{swap}
    ]
    &&
    2\cdot \mathbf{8}_+
    \phantom{
      \oplus
      2\cdot \mathbf{8}_-
    }
  \end{tikzcd}
\end{equation}
Notice that as such we may act on $\psi$ also with transversal Clifford generators, and have with \eqref{TheChiralSpinRepresentations}:
\begin{equation}
  \label{ChiralityOfOddWorldvolumeFrame}
  \psi
  \;=\;
  P(\psi)
  \,,
  \;\;\;\;\;\;\;\;\;
  \tau_a
  \;=\;
  \overline{P}(\tau_a)
  \,.
\end{equation}

Of course, here $\slashed{\tilde H} \defneq \mathrm{Sh}_{11}$ and $\tau \defneq \mathrm{Sh}_{01}$ are components \eqref{ComponentOfSupershearMap} of the super-shear map \eqref{SuperShearMap}, and the BPS immersion condition 
of Def. \ref{BPSImmersion}
implies (Prop. \ref{ComponentsOfBPImmersion}) that 
\begin{equation}
  \label{TauVanishes}
  \tau = 0
  \,.
\end{equation}
However, below we still keep $\tau$ around, to show where it would appear, cf. Rem. \ref{SimplificationOfWorldvoluemTorsionConstraint} and  Rem. \ref{RecoveringTheOrdinaryWorldvolumeBianci} below.

\smallskip 
This remaining freedom in \eqref{ShearingOfSuperEmbeddings} of M5 super-immersions $\phi$ to ``shear'' along the odd directions by a component $\slashed{\tilde H}$ turns out to reflect the degrees of freedom of the flux density $H_3$ on the brane's worldvolume (cf. \eqref{H3AsFunctionOfTildeH3} in Prop. \ref{BianchiIdentityOnM5BraneInComponents} below).
\end{remark}

We offer the following two concrete examples of M5-brane super-immersions/embeddings (no such examples seem to exist in the literature). The first is the archetypical example which is simple but instructive, the second is arguably the next-archetypical example and already fairly laborious to construct:
\begin{example}[\bf Flat M5-brane immersion]
\label{FlatM5Branes}
Consider target space to be flat Minkowski super-spacetime $X \defneq \mathbb{R}^{1,10\vert \mathbf{32}}$ (or some toroidal compactification thereof)
with its canonical coordinate functions $(X^a)_{a = 0}^{10}$ and $(\theta^\alpha)_{\alpha=1}^{32}$ and co-frame field (e.g. \cite[Ex. 2.76]{GSS24-SuGra}):
\begin{equation}
  \label{CanonicalCoframeOn11dSuperMinkowski}
  \def\arraystretch{1.5}
  \begin{array}{l}
    E^a
    \;:=\;
    \mathrm{d}X^a
    \,+\,
    \big(
      \overline{\Theta}
      \,\Gamma^a\,
      \mathrm{d}\Theta
    \big)
    \\
    \Psi^\alpha
    \;:=\;
    \mathrm{d}\Theta^\alpha
    \,.
  \end{array}
\end{equation}
Then consider the brane worldvolume to be the sub-supermanifold \eqref{SuperProjectionOperator}
$$
  \Sigma 
  \,\defneq\, 
  \mathbb{R}^{1,6 \vert 2 \cdot \mathbf{8}}
  \,:=\,
  P\big(
    \mathbb{R}^{1,10 \vert \mathbf{32}}
  \big)
  \xhookrightarrow{\quad}
  \mathbb{R}^{1,10\vert \mathbf{32}}
$$
with canonical coordinate functions
$x := P \circ X$, $\theta := P \circ \Theta$, and  with the super-immersion given by
\begin{equation}
  \label{FlatSuperImmersion}
  \def\arraystretch{.9}
  \begin{tikzcd}[
    column sep=0pt,
    row sep=-3pt
  ]
    \mathbb{R}^{1,6\vert 2 \cdot \mathbf{8}}
    \ar[
      rr,
      "{ \phi }"
    ]
    &&
    \mathbb{R}^{1,10 \vert \mathbf{32}}
    \\
    x^a &\overset{\phi^\ast}{\longmapsfrom}& X^a
    \\
    \theta
    +
    \slashed{\tilde H}_3
    \cdot
    \theta
    &\overset{\phi^\ast}{\longmapsfrom}&
    \Theta
  \end{tikzcd}
\end{equation}
for {\it constant}
\vspace{2mm} 
\begin{equation}
  \label{AssumingConstantH3}
  (H_3)_{a_1 a_2 a_3} : \Sigma \xrightarrow{\;} \mathbb{R}
  \,,
  \;\;\;
  \mbox{with}
  \;\;\;
  \frac{\partial}{\partial x^a}
  (H_3)_{a_1 a_2 a_3} \;=\; 0
  \,,\;\;\;\;\;
  \frac{\partial}{\partial \theta^\alpha}
  (H_3)_{a_1 a_2 a_3} \;=\; 0
  \,.
\end{equation}
From this, we find the pullback of the given co-frame field to be
$$
  \def\arraystretch{1.6}
  \begin{array}{lll}
    \phi^\ast E^a
    &
    \;=\;
    \phi^\ast \big(
      \mathrm{d} X^a
      +
      (\overline{\Theta}
      \,\Gamma^a\,
      \mathrm{d}\Theta)
    \big)
    &
    \proofstep{
      by
      \eqref{CanonicalCoframeOn11dSuperMinkowski}
    }
    \\
    &\;=\;
    \mathrm{d}x^a
      +
      \big(
        \overline{
          (\theta + \slashed{\tilde H}_3 \theta)
        } 
      \,\Gamma^a\,
      \mathrm{d}
      (\theta + \slashed{\tilde H}_3 \theta)
    \big)
    &
    \proofstep{
      by
      \eqref{FlatSuperImmersion}
    }
    \\
    &\;=\;
    \left\{\!\!
    \def\arraystretch{1.5}
    \begin{array}{ll}
    \mathrm{d}x^a
    \,+\,
    \big(
      \delta^a_{a'}
      -
      2(\K)^a_{a'}
    \big)
    \big(
      \overline{\theta}
      \,\gamma^{a'}\,
      \mathrm{d}\theta
    \big)
    &
    \vert
    \;
    \mbox{\small  tangential $a$}
    \\
    \big(
      \overline{\theta}
      \,
      \{
        \slashed{\tilde H}_3,
        \Gamma^a
      \}
      \,
      \mathrm{d}
      \theta
    \big)
    \,=\,
    0
    &
    \vert
    \;
    \mbox{\small  transversal $a$}
    \end{array}
    \right.
    &
    \proofstep{
      \def\arraystretch{1.1}
      \def\tabcolsep{-5pt}
      \begin{tabular}{l}
      by 
      \eqref{AssumingConstantH3},
      \\\eqref{CommutativityOfGammaWithP}
      \&
      \eqref{ConjugatingGammaWithH3}
      \end{tabular}
    }
  \end{array}
  \;\;\;\;
$$
and
$$  
  \def\arraystretch{1.2}
  \begin{array}{lll}
    \phi^\ast
    \Psi
    &
    \;=\;
    \phi^\ast
    \mathrm{d}\Theta
    &
    \proofstep{
      by \eqref{CanonicalCoframeOn11dSuperMinkowski}
    }
    \\
    &\;=\;
    \mathrm{d}
    \big(
     \theta 
       + 
     \slashed{\tilde H}_3 
     \theta
    \big)
    &
    \proofstep{
      by \eqref{FlatSuperImmersion}    
    }
    \\
    &\;=\;
    \mathrm{d}\theta
    \,+\,
    \slashed{\tilde H}_3
    \,
    \mathrm{d}\theta
    &
    \proofstep{
      by 
      \eqref{AssumingConstantH3}.
    }
  \end{array}
$$
These show, by
\eqref{BPSImmersionInComponents},
that $(E,\Psi)$ 
\eqref{CanonicalCoframeOn11dSuperMinkowski}
is indeed a Darboux super-coframe (Def. \ref{BPSImmersion}) for this $\phi$ \eqref{FlatSuperImmersion}, which is hence indeed an M5-brane super immersion (Def. \ref{M5SuperImmersion}).

We also see in this example that the induced torsion on the worldvolume is the canonical super-torsion {\it plus} a correction by $\slashed{\tilde H}_3$:
$$
  \mathrm{d}e^a
  \;\defneq\;
  \mathrm{d}
  \Big(
    \mathrm{d}
    x^a
    +
    \big(
      \delta^a_{a'}
      -
      2(\K)^a_{a'}
    \big)
    \big(
      \overline{\theta}
      \,\gamma^{a'}\,
      \mathrm{d}\theta
    \big)
  \Big)
  \;=\;
    \big(\,
      \overline{
      \mathrm{d}\theta}
      \,\gamma^{a}\,
      \mathrm{d}\theta
    \big)
    -
   2(\K)^a_{a'}
    \big(\,
      \overline{
      \mathrm{d}\theta}
      \,\gamma^{a'}\,
      \mathrm{d}\theta
    \big)
  \,.
$$
This is a general phenomenon, as we see in \eqref{WorldvolumeTorsionForVanishingTau} below.
\end{example}
\begin{example}[\bf Holographic M5-Brane immersion]
  \label{HolographicM5BraneSuperImmersions}
  The super-immersion of the M5-brane superworldvolume in parallel to the horizon of its own ``black'' M5-brane super-$\mathrm{AdS}_7$-geometry turns out to be an M5 super-immersion (hence a super-embedding) according to Def. \ref{M5SuperImmersion}, iff its distance from the horizon is exactly the throat radius in Poincar{\'e} coordinates. This is already a non-trivial computation which we relegate to \cite{GSS24-AdS7}.
\end{example}

\newpage

\noindent
{\bf The torsion-constraint on M5 super-immersions.}
Given a super-immersion $\phi : \Sigma \xrightarrow{\;} X$, the pullback of the bulk torsion constraint to the worldvolume $\Sigma$
has equivalently the following tangential and transverse components:
\begin{equation}
  \label{PullbackOfTorsionConstraint}
  \def\arraystretch{2.4}
  \begin{array}{cl}
    &
    \phi^\ast
    \Big(
      \big(\,
      \overline{\Psi}
      \,\Gamma^a\,
      \Psi
      \big)
    \;=\;
    \mathrm{d}
    \,
    E^a
    +
    \Omega^a{}_b
    \, 
    E^b
    \Big)
    \\
    \underset{
      \scalebox{.7}{
        \eqref{BPSImmersionDiagram}
      }
    }{
      \Leftrightarrow
    }
    &
    \left\{\!\!
    \def\arraystretch{1.5}
    \begin{array}{ll}
      \big(\,
        \overline{
          (\phi^\ast \Psi)
        }
        \,
        \gamma^a 
        \,
        (\phi^\ast \Psi)
      \big)
      \;=\;
      \mathrm{d}
      \, e^a
      +
      \omega^a{}_b
      \,
      e^b
      &
      \mbox{ \small 
        for tangential $a$,
      }
      \\
      \big(\,
        \overline{
          (\phi^\ast \Psi)
        }
        \,
        \Gamma^a 
        \,
        (\phi^\ast \Psi)
      \big)
      \;=\;
      +
      \,
      \SecondFundamentalForm
        ^a
        _{
          b_1 b_2
        }
      \,
      e^{b_1} e^{b_2}
      +
      \SecondFundamentalForm
        ^a
        _{
          \beta 
          \,
          b
        }
        \,
        \psi^\beta 
        e^b
      &
      \mbox{\small 
        for transversal $a$,
      }
    \end{array}
    \right.
  \end{array}
\end{equation}
where in the second line we made explicit the second fundamental super-form $\SecondFundamentalForm$ from \eqref{PullbackOfSuperConnectionForm}.

This constraint \eqref{PullbackOfTorsionConstraint}
is hence a condition to be satisfied by any $\sfrac{1}{2}$BPS super-immersion, and as such we refer to it as the {\it worldvolume torsion constraint}.

\begin{lemma}[\bf The worldvolume torsion constraint in components]
\label{FormOfTranversalFermionicImmersion}
The transversal worldvolume torsion constraint \eqref{PullbackOfTorsionConstraint} is equivalent to the following set of conditions:
\begin{itemize}[
  leftmargin=.8cm,
  itemsep=2pt,
  topsep=1pt
]
\item[\bf (i)]
The bosonic component of the 2nd fundamental form 
$\SecondFundamentalForm$ \eqref{PullbackOfSuperConnectionForm}
of $\phi$
is symmetric in its tangential indices, as in the classical case \eqref{SecondFundamentalFormIsSymmetric}:
\begin{equation}
  \label{SecondFundamentalSuperFormIsSymmetricInBosonicIndices}
  \SecondFundamentalForm
    ^a
    _{b_1 b_2}
  \;=\;
  \SecondFundamentalForm
    ^a
    _{b_2 b_1}
  \;\;\;\;\;
  \begin{array}{l}
    \mbox{\rm \small 
for transversal $a$}
    \\
    \mbox{\rm \small  and tangential $b_i$};
  \end{array}
\end{equation}
\item[\bf (ii)]
the fermionic shear component $\tau_b$ 
\eqref{transversalFermionicComponent}
is (over-)determined by the equations:
\begin{equation}
  \label{EquationForTau}
  (\tau_b)_\alpha
  \;=\;
  \tfrac{1}{2}
  (
    \Gamma_{5'}
    \SecondFundamentalForm
  )
    ^{5'}
    _{ \alpha\, b }
  \;=\;
  \tfrac{1}{2}
  (
    \Gamma_{6}
    \SecondFundamentalForm
  )
    ^{6}
    _{ \alpha\, b }
  \;=\;
  \tfrac{1}{2}
  (
  \Gamma_{7}
  \SecondFundamentalForm
  )
    ^{7}
    _{ \alpha\, b }
  \;=\;
  \tfrac{1}{2}
  (
  \Gamma_{8}
  \SecondFundamentalForm
  )
    ^{8}
    _{ \alpha\, b }
  \;=\;
  \tfrac{1}{2}
  (
  \Gamma_{9}
  \SecondFundamentalForm
  )
    ^{9}
    _{ \alpha\, b }
  \,;
\end{equation}

\item[\bf (iii)]
the fermionic shear component $\slashed{\tilde H}$ \eqref{transversalFermionicComponent} takes the form 
\begin{equation}
  \label{SettingTheSlashHComponent}
  \slashed{\tilde H}
  \;=\;
  \slashed{\tilde H}_3
  \;:=\;
  \tfrac{1}{3!}
  (\tilde H_3)_{a_1 a_2 a_3}
  \gamma^{a_1 a_2 a_3}
  \;\;\;\;\;
  \mbox{\rm for any}
  \;\;\;\;\;
  \tilde H_3 
    \;=\; 
  \star \tilde H_3
  \,,
\end{equation}
\end{itemize}
with which the tangential part of \eqref{PullbackOfTorsionConstraint} equivalently says that
\begin{itemize}[
  leftmargin=.8cm,
  itemsep=2pt,
  topsep=1pt
]
\item[\bf (iv)]
the induced worldvolume torsion is:
\begin{equation}
  \label{WorldvolumeTorsion}
  \mathrm{d}
  \,
  e^a
  +
  \omega^a{}_b
  \, 
  e^b
  \;=\;
  \big(
   \overline{\tau_{b_1}}
   \gamma^a
   \tau_{b_2}
  \big)
  e^{b_1}
  e^{b_2}
  \;-\;
    2
    \big(
      \overline{\tau_b}
      \,
      \gamma^a
      \slashed{\tilde H}_3
      \,
      \psi
    \big)
    e^{b}
  \;+\;
  \big(
    \delta
      ^{ a  }
      _{ a' }
    -2
    (\K)
      ^{ a  }
      _{ a' }
  \big)
  \big(\,
    \overline{\psi}
    \,\gamma^{a'}
    \psi
  \big)
  \,,
\end{equation}
where the notation on the far right is from \eqref{SquareOfH3OnTwoIndices}.

\end{itemize}
\end{lemma}
\begin{proof}
The fermionic co-frame components of the transversal torsion constraint are, at face value:
\begin{equation}
  \label{TransversalTorsionConstraintComponents}
  \def\arraystretch{1.6}
  \begin{array}{lll}
    &
    \big(\,
    \overline{
      (\phi^\ast \Psi)
    }
    \,\Gamma^a\,
    (\phi^\ast \Psi)
    \big)
    \;-\;
    \SecondFundamentalForm
      ^a
      _{ b_1 b_2 }
      e^{b_1}
      e^{b_2}
    \;-\;
    \SecondFundamentalForm
      ^a
      _{ \beta \, b }
      \psi^\beta
      e^{b}
    \;=\;
    0
    \\
    &
    \;\Leftrightarrow\;
    \left\{\!\!\!\!
    \def\arraystretch{1.8}
    \begin{array}{l}
      \scalebox{.7}{
        \color{gray}
        $(\psi^0)$
      }
      \;
      \Big(
      \big(
        \overline{\tau_{b_1}}
        \,\Gamma^a\,
        \tau_{b_2}
      \big)
      -
      \SecondFundamentalForm
        ^a
        _{b_1 b_2}
      \Big)
      e^{b_1}
      e^{b_2}
      \;=\;
      0
      \\
      \scalebox{.7}{
        \color{gray}
        $(\psi^1)$
      }
      \;
      \Big(
      2
      \big(
        \overline{
          (
            \psi + \slashed{\tilde H}\psi
          )
        }
        \,
        \Gamma^a
        \,
        \tau_b
      \big)
      \;-\;
      \psi^\alpha
      \,
      \SecondFundamentalForm
        ^{a}
        _{ \alpha\,  b }
      \Big)
      e^b
      \;=\;
      0
      \\
      \scalebox{.7}{
        \color{gray}
        $(\psi^2)$
      }
      \;
      \big(
        \overline{
          (\psi + \slashed{\tilde H}\psi)
        }
        \,\Gamma^a\,
        (\psi + \slashed{\tilde H}\psi)
      \big)
      \;=\;
      0
      \,.
    \end{array}
    \right.
  \end{array}
\end{equation}
We analyze these components in turn:

\medskip

\noindent
\fbox{
  \bf
  The transversal torsion constraint at $\psi^0$
}
Observing that for transversal $a$ we have
$$
  \def\arraystretch{1.8}
  \begin{array}{lll}
    \big(
      \overline{\tau_{b_1}}
      \,\Gamma^a\,
      \tau_{b_2}
    \big)
    &
    \;=\;
    \big(
      \overline{
        ( \overline{P} \tau_{b_1} )
      }
      \,\Gamma^a\,
      \overline{P}\tau_{b_2}
    \big)
    &
    \proofstep{
      by
      \eqref{ChiralityOfOddWorldvolumeFrame}
    }
    \\
    &\;=\;
    \big(
      \overline{
        \tau_{b_1}
      }
      P
      \,\Gamma^a\,
      \overline{P}\tau_{b_2}
    \big)
    &
    \proofstep{
      by
      \eqref{SkewSelfAdjointnessOfCliffordGenerators}
    }
    \\
    & \;=\;
    0
    &
    \proofstep{
      by 
      \eqref{CommutativityOfGammaWithP},
    }
  \end{array}
$$
the $(\psi^0)$-component 
in \eqref{TransversalTorsionConstraintComponents}
says 
equivalently 
that 
(cf. \cite[(4.59)]{Sorokin00})
the 2nd fundamental form is symmetric in its tangential indices,
$$
  \SecondFundamentalForm
    ^a
    _{ [b_1 b_2] }
  \;=\;
  0
  \,,
$$
as in the classical situation \eqref{SecondFundamentalFormIsSymmetric}.

\medskip
\noindent
\fbox{
  \bf
  The transversal torsion constraint at $\psi^1$
}
Observing that the first summand of the $(\psi^1)$-component in \eqref{TransversalTorsionConstraintComponents} 

\vspace{1mm} 
\noindent is
\vspace{2mm} 
$$
  \def\arraystretch{1.4}
  \begin{array}{lll}
    \big(\,
    \overline{
      (
        \psi 
        + 
        \slashed{\tilde H}_3 \psi
      )
    }
    \,
    \Gamma^a
    \,
    \tau_b
    \big)
        & 
    \;=\;
    \big(\,
    \overline{
      \psi
    }
    \,
    \overline{P}
    \,
    (1 + \slashed{\tilde H}_3)
    \,
    \Gamma^a
    \,
    \overline{P} \, \tau_b
    \big)
    &
    \proofstep{
      by 
      \eqref{ChiralityOfOddWorldvolumeFrame}
      \&
      \eqref{AdjointnessOfCliffordBasisElements}
    }
    \\
    & 
    \;=\;
    \big(\,
    \overline{
      \psi
    }
    \,
    \overline{P}
    \,
    \Gamma^a
    \,
    \overline{P} \, \tau_b
    \big)
    &
    \proofstep{
      by
      \eqref{CommutativityOfGammaWithP}
    }
    \\
    & 
    \;=\;
    \big(\,
    \overline{
      \psi
    }
    \,
    \Gamma^a
    \,
    \tau_b
    \big)
    &
    \proofstep{
      by
      \eqref{ChiralityOfOddWorldvolumeFrame},
    }
  \end{array}
$$
the condition says equivalently that
(cf. \cite[(5.63)]{Sorokin00})
$$
  2
    \big(\,
    \overline{
      \psi
    }
    \,
    \Gamma^a
    \,
    \tau_b
  \big)  
  \;=\; 
  \psi^\alpha
  \,
  \SecondFundamentalForm
    ^{\, a}
    _{\alpha \, b}
  \,,
  \;\;\;\;\;\;\;
  \mbox{hence equivalently that}
  \;\;\;\;\;\;\;
  \big(
    \Gamma^a
    \tau_b
  \big)_\alpha
  \;=\;
  \tfrac{1}{2}
  \SecondFundamentalForm
    ^a
    _{\alpha \, b}
  \,.
$$
Acting on this equation with either of the transverse $\Gamma_a$ shows its equivalence to the claimed equations \eqref{EquationForTau}.

\medskip
\noindent
\fbox{
  \bf
  The transversal torsion constraint at $\psi^2$
}
By \eqref{ExpandingChiralLinearMapsInCliffordElements} and the required $\mathrm{Spin}(5)$-equivariance \eqref{TheNontrivialShearComponent},
the most general form of $\slashed{\tilde H}$ is
$$
  \slashed{\tilde H}
  \;=\;
  \underbrace{
  \mathclap{
    \phantom{
      \tfrac{1}{1}
    }
  }
  (\tilde H_1)_{a} \gamma^a
  }_{
    \slashed{\tilde H}_1
  }
  \;
  +
  \;
  \underbrace{
  \tfrac{1}{3!}
  (\tilde H_3)_{a_1 a_2 a_3}
  \gamma^{a_1 a_2 a_3}
  }_{
    \slashed{\tilde H}_3
  }
  \,.
$$
With that, the $(\psi^2)$-component in \eqref{TransversalTorsionConstraintComponents} equals
$$
  \def\arraystretch{1.6}
  \begin{array}{ll}
    \Big(\,
      \overline{
      \big(
        \psi 
        +
        \slashed{\tilde H}_1 \psi 
        +
        \slashed{\tilde H}_3
        \psi
      \big)
      }
      \, \Gamma^a \,
      \big(
        \psi 
        +
        \slashed{\tilde H}_1 \psi
        +
        \slashed{\tilde H}_3
        \psi
      \big)      
    \Big)
    &
    \proofstep{
      by 
      \eqref{transversalFermionicComponent}
    }
    \\
    \;=\;
    \Big(
      \overline{\psi}
      \big(
        (1 - \slashed{H}_1 + \slashed{\tilde H}_3)
        \,\Gamma^a\,
        (1 + \slashed{H}_1 + \slashed{\tilde H}_3)
      \big)
      \psi
    \Big)
    &
   \proofstep{
     by 
     \eqref{AdjointnessOfCliffordBasisElements}
   }
   \\
   \;=\;
    \Big(
      \overline{\psi}
      \;
        \overline{P}
        (1 - \slashed{H}_1 + \slashed{\tilde H}_3)
        \,\Gamma^a\,
        (1 + \slashed{H}_1 + \slashed{\tilde H}_3)
        P
        \;
      \psi
    \Big)
    &
    \proofstep{
      by
      \eqref{ChiralityOfOddWorldvolumeFrame}.
    }
  \end{array}
$$
Now multiplying out, we obtain:
\begin{equation}
 \label{ComputingTorsionConstraintAtPsi2}
 \def\arraystretch{1.6}
 \begin{array}{ll}
   \overline{P}
    \big(
      1
      -
      \slashed{\tilde H}_1
      +
      \slashed{\tilde H}_3
    \big)
    \,\Gamma^a\,
    \big(
      1
      +
      \slashed{\tilde H}_1
      +
      \slashed{\tilde H}_3
    \big)
    P
    \\
  \;=\;
   \overline{P}
    \big(
      1
      +
      \slashed{\tilde H}_3
    \big)
    \,\Gamma^a\,
    \big(
      1
      +
      \slashed{\tilde H}_3
    \big)
    P
     \\
     \hspace{.5cm}
   \;-\;
   \overline{P}
    \slashed{\tilde H}_1
    \,\Gamma^a\,
    \big(
      1
      +
      \slashed{\tilde H}_1
      +
      \slashed{\tilde H}_3
    \big)
    P
   \;+\;
   \overline{P}
    \big(
      1
      +
      \slashed{\tilde H}_1
      +
      \slashed{\tilde H}_3
    \big)
    \,\Gamma^a\,
    \slashed{\tilde H}_1
    P
  \\
  \;=\;
  \overline{P}
  \Big(
    \Gamma^a 
    +
    \{
      \Gamma^a
      ,
      \slashed{\tilde H}_3
    \}
    +
    \slashed{\tilde H}_3
    \Gamma^a
    \slashed{\tilde H}_3
  \Big)
  P
  \\
  \hspace{.5cm}
  +
  \;
  \overline{P}
  \big(
    [\slashed{\tilde H}_1, \Gamma^a]
    \;-\;
    \slashed{\tilde H}_1
    \Gamma^a
    \slashed{\tilde H}_3
    +
    \slashed{\tilde H}_3
    \Gamma^a
    \slashed{\tilde H}_1
  \big)
  P    
  \\[+6pt]
  \;=\;
  \left\{\!\!
  \def\arraystretch{1.5}
  \begin{array}{ll}
    \overline{P}
    \Big(
    \gamma^a 
    +
    (\tilde H_3)_{a \, b_1 b_2}
    \gamma^{a_1 a_2}
    +
    \slashed{\tilde H}_3
    \gamma^a
    \slashed{\tilde H}_3
    \Big)
    P
    +
    \mathcal{O}(\slashed{H}_1)
    &
    \mbox{for tangential $a$}
    \\
    2
    \,
    \overline{P}
    \,
    \slashed{H}_1
    P
    \,
    \Gamma^a
    &
    \mbox{for transversal $a$}
  \end{array}
  \right.
  &
  \proofstep{
    by
    \eqref{CommutativityOfGammaWithP}
    .
  }

\end{array}
\end{equation}
The last line  makes it manifest,
cf. \eqref{VanishingOfOneIndexCliffordElementInQuadraticForm},
that the transversal component vanishes iff $\slashed{H}_1 = 0$  (cf. the argument via matrix representations in \cite[(15)]{HoweSezgin97b}\cite[(5.66)]{Sorokin00}) --- which proves \eqref{SettingTheSlashHComponent}, remembering \eqref{CoefficientsOf3Index6dGammaAreSelfDual}
--- in which case the terms denoted $\mathcal{O}(\slashed{H}_1)$ 
in \eqref{ComputingTorsionConstraintAtPsi2}
vanish.

\smallskip

\noindent
\fbox{\bf The worldvolume torsion.}
The remaining tangential component in \eqref{ComputingTorsionConstraintAtPsi2}
is the $(\psi^2)$-component of the worldvolume torsion  as claimed in \eqref{WorldvolumeTorsion} (cf. \cite[(5.71)]{Sorokin00}):
$$
\def\arraystretch{1.6}
\begin{array}{lll}
  \overline{P}
    \Big(
    \gamma^a 
    +
    (\tilde H_3)_{a \, b_1 b_2}
    \gamma^{a_1 a_2}
    +
    \slashed{\tilde H}_3
    \gamma^a
    \slashed{\tilde H}_3
    \Big)
    P
  &
  \;=\;
  \overline{P}
    \Big(
    \gamma^a 
    +
    \slashed{\tilde H}_3
    \gamma^a
    \slashed{\tilde H}_3
    \Big)
    P
  &
  \proofstep{
    by
    \eqref{CommutativityOfGammaWithP}
  }
  \\
  &\;=\;
  \overline{P}
  \Big(
  \big(
    \delta_{a'}^a
    -
    2
    (\K)_{a'}^a
  \big)
  \gamma^{a'} 
  \Big)
  P
  &
  \proofstep{
    by
    \eqref{ConjugatingGammaWithH3}.
  }
  \end{array}
$$
It just remains to check the other summands of the worldvolume torsion \eqref{WorldvolumeTorsion}: Recalling  $\phi^\ast \Psi \;=\; \psi + \slashed{\tilde H}_3 \psi + \tau_a \, e^a$, the $(e^2)$-term is immediate, and the 
$(e^1 \, \psi^1)$-component is obtained as follows:
$$
  \def\arraystretch{1.6}
  \begin{array}{lll}
    \big(
      \overline{
        (1 + \slashed{\tilde H}_3)
        \psi
      }
      \,\gamma^a\,
      \tau_b
    \big)
    e^b
    -
    \big(
      \overline{
        \tau_b
      }
      \,\gamma^a\,
      (1 + \slashed{\tilde H}_3)
      \psi
    \big)
    e^b
    &
    \;=\;
    -2
    \big(
      \overline{\tau_b}
      \,\gamma^a\,
      (1 + \slashed{\tilde H}_3)
      \psi
    \big)
    e^{b}
    &
    \proofstep{
      by 
      \eqref{SymmetricSpinorPairings}
    }
    \\
&    \;=\;
    -2
    \big(
      \overline{ 
        (\overline{P}\tau_b)
      }
      \,\gamma^a\,
      (1 + \slashed{\tilde H}_3)
      P\psi
    \big)
    e^{b}
    &
    \proofstep{
      by 
      \eqref{ChiralityOfOddWorldvolumeFrame}
    }
    \\
  &  \;=\;
    -2
    \big(
      \overline{\tau_b}
      \,
      P
      \,\gamma^a\,
      (1 + \slashed{\tilde H}_3)
      P\psi
    \big)
    e^{b}
    &
    \proofstep{
      by
     \eqref{transversalProjectionOperator}
    }
    \\
  &  \;=\;
    -2
    \,
    \big(
      \overline{\tau_b}
      \,\gamma^a
      \slashed{\tilde H}_3
      \psi
    \big)
    e^b
    &
    \proofstep{
      by
      \eqref{transversalProjectionOperator}
      .
    }
  \end{array}
$$
This completes the proof.
\end{proof}

\begin{remark}[\bf Simplification of the worldvolume torsion constraint]
  \label{SimplificationOfWorldvoluemTorsionConstraint}
  
  Under the assumption $\tau = 0$ \eqref{TauVanishes} the condition \eqref{EquationForTau} is equivalently just the statement that 
  \begin{equation}
    \label{FermionicComponentOfSecondFundamentalFormVanishes}
    \SecondFundamentalForm^a_{\alpha \, b} 
      \,=\, 
    0
    \,.
  \end{equation}
  Hence in this case the conditions \eqref{SecondFundamentalSuperFormIsSymmetricInBosonicIndices} and \eqref{EquationForTau} are jointly equivalent to 
  $
    \big(
    e^{b'}
    \SecondFundamentalForm^a_{b' \, b}
    +
    \psi^\alpha
    \SecondFundamentalForm^a_{\alpha\, b}
    \big)
    e^b
    \;=\;
    0
  $,
  which in turn may equivalently be re-phrased as
  $
    \overline{P}
    \big(
      \mathrm{d}^\Omega e
    \big)
    \;=\;
    0
    \,,
  $
  where $\mathrm{d}^\Omega$ denotes the exterior covariant derivative with respect to $\phi^\ast \Omega$.

  In summary, if $\tau = 0$ then the worldvolume torsion constraint 
  \eqref{PullbackOfTorsionConstraint}
  according to Lem. \ref{FormOfTranversalFermionicImmersion} simplifies to the following set of statements:
  \def\arraystretch{1.5}
  \begin{align}
      &
      \overline{P}
      \big(
        \mathrm{d}^\Omega e
      \big)
      \;=\;
      0
      \,,
      \\
      &
      \slashed{\tilde H}
      \;=\;
      \slashed{\tilde H}_3
      \;:=\;
      \tfrac{1}{3!}
      (\tilde H_3)_{a_1 a_2 a_3}
      \gamma^{a_1 a_2 a_3}
      \,,
      \\
      \label{WorldvolumeTorsionForVanishingTau}
      &
      \mathrm{d}
      \,
      e^a
      +
      \omega^a{}_b 
      \,
      e^b
      \;=\;
  \big(
    \delta
      ^{ a  }
      _{ a' }
    -2
    (\K)
      ^{ a  }
      _{ a' }
  \big)
  \big(
    \overline{\psi}
    \,\gamma^{a'}
    \psi
  \big)
  \,.
  \end{align}
\end{remark}

\medskip

\noindent
{\bf The super flux densities.}
Given an M5 super-immersion $\phi : \Sigma \xrightarrow{\;} X$ (Def. \ref{M5SuperImmersion})
with induced super co-frame field $(e,\psi)$ on $\Sigma$ \eqref{BPSImmersionDiagram}, consider a differential 3-form with only bosonic co-frame components on $\Sigma$:
\begin{equation}
  \label{Super3FluxDensity}
  H^s_3
  \;:=\;
  H_3 
  \;:=\;
  \tfrac{1}{3!}
  (H_3)_{a_1 a_2 a_3}
  \,
  e^{a_1} e^{a_2} e^{a_3}
  \;\;
  \in
  \;\;
  \Omega^1_{\mathrm{dR}}
  \big(
    \Sigma
    ;\,
    b^2\mathbb{R}
  \big)
  \,.
\end{equation}
Here the super-script $(-)^s$ is to indicate that this is the super 3-flux analog of the super 4-flux density $G_4^s$ 
\eqref{SuperFluxDensitiesOf11dSugra}. Hence the would-be contribution $H_3^0$ analogous to $G_4^0$ in \eqref{SuperFluxDensitiesOf11dSugra} vanishes (cf. \cite[p. 91]{Sorokin00}). (But a candidate for non-vanishing $H_3^0$ does appear on ``exceptional-geometric'' super-spacetime as observed in \cite{FSS20Exc}; we further discuss this in \cite{GSS-Exceptional}).

\medskip

\begin{proposition}[\bf Flux Bianchi identity on M5-brane in components]  \label{BianchiIdentityOnM5BraneInComponents}
  Given a $\sfrac{1}{2}BPS$-immersion of an M5-brane worldvolume {\rm (Def. \ref{M5SuperImmersion})}, the Bianchi identity
  $$
    \mathrm{d}
    H^s_3
    \;=\;
    \phi^\ast
    G_4^s
  $$
  for super flux-densities of the form \eqref{Super3FluxDensity} and \eqref{SuperFluxDensitiesOf11dSugra}
  is equivalent to the following set of conditions:
  \begin{itemize}[
    topsep=1pt,
    itemsep=2pt,
    leftmargin=.8cm  
  ]
  \item[\bf (i)]
  the ordinary but torsion-ful 
  Bianchi identity
  \begin{equation}
    \label{BosonicBianchiComponentOnSuperspace}
    \tfrac{1}{3!}
    \nabla_{[a_1}
    (H_3)_{a_2 a_3 a_4]}
    \,+\,
    \tfrac{1}{2!}
    (H_3)_{[a_1 a_2 |b|}
    \big(
      \overline{\tau_{a_3}}
      \,\gamma^b\,
      \tau_{a_4]}
    \big)
    \;=\;
    \tfrac{1}{4!}
    (\phi^\ast G_4)_{[a_1 \cdots a_4]}\;;
  \end{equation}
  \item[\bf (ii)]
  a rheonomy equation\footnote{
    Solving equation \eqref{RheonomyForH3} determines the values of the super-flux $H_3$ on the super-manifold $\Sigma$ from its value on the bosonic body $\bosonic{\Sigma}$, hence its ``flow'' across superspace (whence ``rheonomy'' as in \cite[\S III.3.3]{CDF91}).
  } for $H_3$: 
  \begin{equation}
    \label{RheonomyForH3}
    \tfrac{1}{3!}
    \psi^\alpha 
    \nabla_\alpha
    (H_3)_{a_1 a_2 a_3}
    \;=\;
    (H_3)_{[a_1 a_2 |b|}
    \big(
      \overline{\tau_{a_3]}}
      \,\gamma^b
      (1 + \slashed{\tilde H}_3)
      \psi
    \big)
    \;+\;
    \big(
      \overline{\tau_{[a_1}}
      \,
      \gamma_{a_2 a_3]}
      (\psi + \slashed{\tilde H}_3 \psi)
   \big)\,;
 \end{equation}
  
  \item[\bf (iii)]
  the 3-flux density $H_3$ is the following function of the transverse fermionic immersion component $\tilde H_3$ \eqref{SettingTheSlashHComponent}:
  \begin{equation}
  \label{H3AsFunctionOfTildeH3}
  (H_3)_{a b c}
  \;=\;
  \tfrac
    {
      -4
    }
    {
    \mathclap{\phantom{
      \vert^{\vert}
    }}
    1 
    -
    \sfrac{2}{3}
    \,
    \mathrm{tr}(\K \cdot \K)
  }
  \Big(
    \delta
      ^{ a  }
      _{ a' }
    +
    2
    \,
    (\K)
      ^{ a' }
      _{ a  }
  \Big)
  (\tilde H_3)_{a' b c}
  \,,
  \end{equation}
  where this is well-defined, hence wherever $\mathrm{tr}(\K \cdot \K)$ \eqref{KSquareIsProportionalToIdentity} satisfies the non-criticality condition \eqref{NonCriticalityCondition}.
  \end{itemize}
\end{proposition}
\begin{proof}
The fermionic co-frame components of the Bianchi identity are, at face value:
$$
  \def\arraystretch{1.8}
  \begin{array}{l}
    \mathrm{d}
    \big(
      \tfrac{1}{3!}
      (H_3)_{a_1 a_2 a_3}
      \,
      e^{a_1} \, e^{a_2} \, e^{a_3}
    \big)
    \,-\,
    \tfrac{1}{4!}
    (\phi^\ast G_4)_{
      a_1 \cdots a_4
    }
    e^{a_1} \cdots e^{a_4}
    \,-\,
    \tfrac{1}{2}
    \big(\,
      \overline{(\phi^\ast \Psi)}
      \gamma_{a_1 a_2}
      (\phi^\ast \Psi)
    \big)
    e^{a_1} \, e^{a_2}
    \;=\;
    0
    \\
    \;\Leftrightarrow\;
    \left\{\!\!
    \def\arraystretch{1.8}
    \begin{array}{l}
      \scalebox{.8}{
        \color{gray}
        $(\psi^0)$
      }
      \;
      \Big(
      \tfrac{1}{3!}
      \nabla_{a_1}
      (H_3)_{a_2 a_3 a_4}
      \,+\,
      \tfrac{1}{2}
      (H_3)_{a_1 a_2 b}
      \big(
        \overline{\tau_{a_3}}
        \,\gamma^b\,
        \tau_{a_4}
      \big)
      \,-\,
      \tfrac{1}{4!}
      (\phi^\ast G_4)_{a_1 \cdots a_4}
      \Big)
      e^{a_1} \cdots e^{a_4}
      \;=\;
      0
      \\
      \scalebox{.8}{
        \color{gray}
        $(\psi^1)$
      }
      \;
      \Big(
        \tfrac{1}{3!}
        \psi^\alpha 
        \nabla_\alpha
        (H_3)_{a_1 a_2 a_3}
        \,-\,
        (H_3)_{a_1 a_2 b}
        \big(
          \overline{\tau_{a_3}}
          \,\gamma^b
          (1 + \slashed{\tilde H}_3)
          \psi
        \big)
        \,-\,
        \big(
          \overline{\tau_{a_1}}
          \,
          \gamma_{a_2 a_3}
          (\psi + \slashed{\tilde H}_3 \psi)
       \big)
       \Big)
       e^{a_1}
       \,
       e^{a_2}
       \,
       e^{a_3}
       \;=\;
       0
      \\
      \scalebox{.8}{
        \color{gray}
        $(\psi^2)$
      }
      \;
    \Big(
    \tfrac{1}{2}
    (H_3)_{a_1 a_2 a_3}
    \big(
      \delta
        ^{a_3}
        _{a'_3}
      -2
      (\tilde H_3)_{a'_3 \, b_1 b_2}
      (\tilde H_3)^{a_3 \, b_1 b_2}
    \big)
    \big(\,
      \overline{\psi}
      \,\gamma^{a'_3}\,
      \psi
    \big)
    \,-\,
    \big(
      \overline{
      (
        \psi
        +
        \slashed{\tilde H}_3
        \psi
      )
      }
      \,
      \gamma_{a_1 a_2}
      \,
      (
        \psi
        +
        \slashed{\tilde H}_3
        \psi
      )
      \big)
    \Big)
    e^{a_1}
    e^{a_2}
    \;=\;
    0
    \end{array}
    \right.
  \end{array}
$$
where we used the expression \eqref{WorldvolumeTorsion} for the worldvolume torsion tensor,
\begin{itemize}[
  leftmargin=.6cm,
  itemsep=2pt
]
\item with which the $(\psi^0)$-component is obvious,
\item 
while in the $(\psi^1)$-component we also used
\eqref{transversalFermionicComponent},
\item and in the $(\psi^2)$-component 
we in addition used
\eqref{WorldvolumeTorsion}.
\end{itemize}

\vspace{1mm}
\noindent
Hence it just remains to further unwind:

\smallskip

\noindent
\fbox{
  \bf
  The 
  Flux Bianchi at 
  $(\psi^2)$
}
for which we crucially use that $\tilde H_3$ is self-dual \eqref{SettingTheSlashHComponent}:

\vspace{1mm} 
\noindent For the summand on the right notice that
\begin{equation}
  \label{RightSideOfPsi2ComponentOfFluxBianchi}
  \def\arraystretch{1.8}
  \begin{array}{lll}
    &
      \tfrac{1}{2}
      \big(
        \overline{
        (\psi + \slashed{\tilde H}_3 \psi)
        }
        \,
        \gamma_{a_1 a_2}
        \,
        (\psi + \slashed{\tilde H}_3 \psi)
      \big)
      e^{a_1} e^{a_2}
      &
    \\
    &
    =\;
    \tfrac{1}{2}
    \Big(
    \big(
      \overline{\psi}
      \,
      \slashed{\tilde H}_3
      \,
      \gamma_{a_1 a_2}
      \,
      \psi
    \big)
    +
    \big(
      \overline{
        \psi
      }
      \,
      \gamma_{a_1 a_2}
      \,
      \slashed{\tilde H}_3
      \,
      \psi
    \big)
    \Big)
    e^{a_1} e^{a_2}
    &
    \proofstep{
      by
      \eqref{TangentialExpectationValueOf2IndexCliffordElementVanishes}
    }
    \\
   & =\;
    \frac{1}{2}
    \Big(
    \overline{
      \psi
    }
    \big(
      -2
      (\tilde H_3)_{a_1 a_2 \,  b}
      \gamma^{b}
      +
      \tfrac{1}{3}
      (\tilde H_3)^{b_1 b_2 b_3}
      \gamma_{
        a_1 a_2 
        \,
        b_1 b_2 b_3
      }
    \big)
    \psi
    \Big)
    \,
    e^{a_1}
    e^{a_2}
    &
    \proofstep{
      by
      \eqref{GeneralCliffordProduct}
    }
    \\
    & 
    =\;
    -4
    (\tilde H_3)_{a_1 a_2 \, b}
    \,
    \tfrac{1}{2}
    \big(\,
      \overline{
        \psi
      }
      \,
      \gamma^b
      \,
      \psi
    \big)
    \,
    e^{a_1} e^{a_2}
    \,,
  \end{array}
\end{equation}
where in the last step we used
$$
  \def\arraystretch{1.4}
  \begin{array}{lll}
    (\tilde H_3)^{b_1 b_2 b_3}
    \Gamma_{a_1 a_2 \, b_1 b_2 b_3}
    &
    \;=\;
    (\tilde H_3)^{b_1 b_2 b_3}
    \epsilon_{a_1 a_2 \, b_1 b_2 b_3\, c}
    \gamma^c
    &
    \proofstep{
      by 
      \eqref{HodgeDualityOf6dCliffordGenertors}
    }
    \\
    & \;=\;
    -
    3!
    (\tilde H_3)_{a_1 a_2 \, c}
    \,
    \gamma^c
    &
    \proofstep{
      by 
      \eqref{CoefficientsOf3Index6dGammaAreSelfDual}.
    }
  \end{array}
$$
This way the Bianchi identity at $(\psi^2)$ is equivalent to \eqref{H3AsFunctionOfTildeH3}, as claimed (cf. \cite[(7)]{HoweSezginWest97}\cite[(5.80-81)]{Sorokin00}):
\begin{equation}
  \def\arraystretch{1.8}
  \begin{array}{lll}
  &
  \big(
    \delta^{a'}_a 
    -2
    \,
    (\K)
      ^{a'}
      _a
  \big)
  (H_3)_{a' b c}
  \;=\;
  -4
  \,
  (\tilde H_3)_{a b c}
  &
  \proofstep{
    by \eqref{RightSideOfPsi2ComponentOfFluxBianchi}
  }
  \\
  \Leftrightarrow
  &
  (H_3)_{a b c}
  \;=\;
  \frac
    {
      -4
    }
    {
    \mathclap{\phantom{
      \vert^{\vert}
    }}
    1 
    -
    \sfrac{2}{3}
    \,
    \mathrm{tr}(\K \cdot \K)
  }
  \big(
    \delta
      ^{ a  }
      _{ a' }
    +
    2
    \,
    (\K)
      ^{ a' }
      _{ a  }
  \big)
  (\tilde H_3)_{a' b c}
  &
  \proofstep{
    by
    \eqref{InverseOfIdMinus2K}.
  }
  \end{array}
\end{equation}

\vspace{-15pt}
\end{proof}

\begin{remark}[\bf Non/Self-duality of the 3-flux density]
  \label{NonSelfDualityOf3FluxDensity}
  Given an M5 super-immersion $\phi$ \eqref{M5SuperImmersion}, then the decomposition of its 3-flux density $H_3$ \eqref{H3AsFunctionOfTildeH3} 
  into a self dual and an anti self-dual summand is (cf. \cite[(18-19)]{HoweSezginWest97})
  $$
    (H_3)_{a b c}
    \;=\;
  \frac
    {
      -4
    }
    {
    \mathclap{\phantom{
      \vert^{\vert}
    }}
    1 
    -
    \sfrac{2}{3}
    \,
    \mathrm{tr}(\K \cdot \K)
  }
    \Big(
      \underbrace{
        \tilde H_{a b c}
      }_{
        \mathclap{
        \scalebox{.7}{
          \rm
          \color{darkblue}
          \bf
          self-dual
        }
        }
      }
      \;+\;
      2 
      \underbrace{
        (\K)_a^{a'}
        \tilde H_{a b c}
      }_{
        \scalebox{.7}{
          \rm
          \color{darkblue}
          \bf
          anti self-dual
        }
      }
    \Big)
    \,.
  $$
  Namely, we already know of course \eqref{SettingTheSlashHComponent} that the first summand is self-dual, but with \eqref{H3AsFunctionOfTildeH3} it follows that the second summand is skew-symmetric in its indices, whence its anti self-duality follows by \eqref{AntiSelfDualityOfCubicTildeH}. In summary this means that (cf. \cite[p. 6]{HoweSezgin97b}\cite{HoweSezginWest97}\cite[p. 92]{Sorokin00})
  the 3-flux density $H_3$ on the M5-brane:
  \begin{itemize}[
    leftmargin=.6cm,
    topsep=1pt,
    itemsep=2pt
  ]
  \item is in general not self-dual, 
  \item albeit determined by its self-dual part, 
  \item with its anti self-dual part a higher order ($\geq$ cubic) function of the self-dual part,
    away from criticality \eqref{NonCriticalityCondition},
  
  \item so that $H_3$ is (only) asymptotically self-dual as its absolute value goes to zero (i.e., in the small field limit).
  \end{itemize}

\smallskip

\noindent Note, therefore, that the extensive attention that self-duality constraints on flux densities in 6d have found in the literature (e.g., \cite{HenneauxTeitelboim88}\cite{Witten97b}\cite{BelovMoore06}\cite{Sen20}\cite{Lambert19b}\cite{AndrianopoliEtAl22}\cite{RSvdW21}) does not actually apply to the full M5-brane worldvolume theory (where the would-be duality condition is at face value more complicated and yet all implied by the $\sfrac{1}{2}$BPS super-immersion condition), but only asymptotically to its small field limit (which in the further limit of decoupled gravity is famously expected to be a $\mathcal{N} =(2,0)$ $D=6$ CFT, e.g. \cite{SaemannSchmidt18}).

\end{remark}

\begin{remark}[\bf Recovering the ordinary worldvolume Bianchi identity]
  \label{RecoveringTheOrdinaryWorldvolumeBianci}
  We see from \eqref{BosonicBianchiComponentOnSuperspace} that the expected Bianchi identity $\mathrm{d} \, H_3 \,=\,\phi^\ast G_4$ \eqref{TwistedBianchiIdentityInIntroduction} 
  on the ordinary bosonic worldvolume $\bosonic{\Sigma} \xhookrightarrow{\;} \Sigma$ is recovered (only) if $\tau = 0$, which is indeed implied by the $\sfrac{1}{2}$BPS immersion condition, cf. Rem. \ref{TransversalFermionicShear} above 
  (more implicitly assumed for instance in \cite[p. 4]{HoweLambertWest98}).
\end{remark}

With the previous two remarks we have arrived at the desired conclusion that the M5's worldvolume Bianchi identity {\it when promoted to super-space} subsumes both the ordinary Bianchi identity on $H_3$ as well as its (non/self-duality) equation of motion. As explained in \S\ref{IntroductionAndOverview}, this suggests that to obtain the completed field content on the M5-brane we are to flux-quantize this super-flux $H_3^s$ on the super-worldvolume.

\smallskip 
In the following \S\ref{FluxQuantizationOnM5Branes}, we discuss some consequences of such flux-quantization.

\section{Flux quantization on M5-branes}
\label{FluxQuantizationOnM5Branes}

Having established (with the result in \S\ref{SuperEmbeddingConstruction}) that the flux quantization \eqref{M5BraneFluxQuantizationInIntroduction} of the super-worldvolume 
super-flux on the M5-brane  provides a full global completion of the M5's on-shell field content, here we discuss some properties of the resulting flux quantized fields. 

\begin{itemize}[
  leftmargin=.4cm,
  topsep=1pt,
  itemsep=2pt
]
\item[--]
\S\ref{GaugePotentialsOnM5Branes} establishes the backwards-compatibility of the flux quantized fields, locally, to the traditional formulas for gauge potentials,
\item[--] \S\ref{ResultingSolitons} analyzes the resulting worldvolume charges of singular strings and solitonic membranes inside the M5, under the assumption that 
flux quantization is in co-Homotopy theory.
\end{itemize}

\subsection{Gauge potentials on the M5-brane}
\label{GaugePotentialsOnM5Branes}

The construction of flux-quantized fields in \eqref{GaugePotentialAsHomotopy} not only constrains the flux densities to reflect quantized charges, but does so by constructing 
the gauge potentials that exhibit the corresponding higher gauge field.
Globally the nature of the gauge potentials crucially depends on the chosen flux quantization law $\mathcal{A}$; but locally, on a good open cover $\CoverOf{X}$
by contractible charts $U_i \xhookrightarrow{\iota_i} X$,
the charges necessarily vanish (due to the assumption that the classifying space $\mathcal{A}$ is connected, hence with an essentially unique basepoint
$0_{\mathcal{A}} : \ast \xrightarrow{\;} \mathcal{A}$)
and the nature of the gauge potentials becomes independent of the choice of $\mathcal{A}$, as it should be:
\begin{equation}
  \label{OnChargGaugePotentialDiagram}
  \begin{tikzcd}[column sep=huge]
    &[+10pt]
    &[-10pt]
    \ast
    \ar[
      d,
      "{
        0_{\mathcal{A}}
      }"
    ]
    \\
    \mathllap{
      \scalebox{.7}{
        \color{darkblue}
        \bf
        \def\arraystretch{.9}
        \begin{tabular}{c}
          Contractible
          \\
          chart
        \end{tabular}
      }
    }
    U_i
    \ar[
      d,
      hook,
      "{ \iota_i }"
    ]
    \ar[
      urr,
      "{
        \scalebox{.7}{
          \color{darkgreen}
          \bf
          vanishing charge
        }
      }"{sloped, pos=.3},
      "{
        \sim
      }"{sloped, description},
      "{
        \scalebox{.6}{$
        \mathrm{hmtp}
        $}
      }"{swap, sloped},
      "{\ }"{swap, pos=.2, name=s}
    ]
    &&
    \mathcal{A}
    \ar[
      d,
      "{
        \mathbf{ch}_{\mathcal{A}}
      }"
    ]
    \\
    \mathllap{
      \scalebox{.7}{
        \color{darkblue}
        \bf
        \def\arraystretch{.9}
        \begin{tabular}{c}
          Domain super-space
          \\
          (worldvolume)
        \end{tabular}
      }
    }
    \Sigma
    \ar[
      r,
      "{
        \vec F
      }"{name=t},
      "{
        \scalebox{.7}{
          \color{darkgreen}
          \bf
          flux density
        }
      }"{swap}
    ]
    &
    \Omega^1_{\mathrm{dR}}
    \big(
      -
      ;\,
      \mathfrak{a}
    \big)_{\mathrm{clsd}}
    \ar[
      r,
      "{
        \eta^{\scalebox{.7}{$\shape$}}
      }"
    ]
    &
    \shape
    \,
    \Omega^1_{\mathrm{dR}}
    \big(
      -
      ;\,
      \mathfrak{a}
    \big)_{\mathrm{clsd}}    
    \ar[
      from=s,
      to=t,
      Rightarrow,
      "{
        \GaugePotential{A}
        \mathrlap{
          \scalebox{.7}{
            \color{orangeii}
            \bf
            Local gauge potentials
          }
        }
      }"{description}
    ]
  \end{tikzcd}
\end{equation}

\vspace{1mm} 
\noindent By definition of the moduli object on the bottom right (see \cite[p. 26]{SS24-Flux}\cite[Ex. 2.55]{GSS24-SuGra} following \cite[Def. 9.1]{FSS23Char}), 
the homotopy filling the diagram \eqref{OnChargGaugePotentialDiagram} is a concordance (deformation) of closed $\mathfrak{a}$-valued differential 
forms from zero to the given flux densities $\vec F$:
\vspace{-3mm} 
\begin{equation}
  \label{NullConcordance}
  \GaugePotential{A}
  \;\in\;
  \Omega^1_{\mathrm{dR}}
  \big(
    U_i
    \times [0,1]
    ;\,
    \mathfrak{a}
  \big)_{\mathrm{clsd}}
  \quad 
  \mbox{s.t.}
  \quad 
  \left\{\!\!
  \def\arraystretch{1.2}
  \begin{array}{l}
    \iota^\ast_0(\GaugePotential{A}\,)
    \;=\;
    0
    \\
    \iota^\ast_1(\GaugePotential{A}\,)
    \;=\;
    \vec F
  \end{array}
  \right.
  \in
  \Omega^1_{\mathrm{dR}}\big(
    U_i
    ;\,
    \mathfrak{a}
  \big)_{\mathrm{clsd}}
  \,.
\end{equation}
This is just the notion of coboundary in $\mathfrak{a}$-valued de Rham cohomology (\cite[Def. 6.3]{FSS23Char}), in fact it reduces to ordinary de Rham coboundaries in the abelian case where $\mathfrak{a}$ is the $L_\infty$-algebra which is $\mathbb{R}$ concentrated in some degree (\cite[Prop. 6.4]{FSS23Char}). When $\mathfrak{a}$ is not abelian, as in the case of interest where $\mathfrak{a} = \mathfrak{l}_{{}_{S^4}} S^7$ \eqref{WhiteheadOfHHopfFib}, then a little work is needed to extract more concise coboundary data underlying the concordances \eqref{NullConcordance}. This is what we do now for the worldvolume fields on the M5-brane (Prop. \ref{BFieldGaugePotentialFromBoncordance} below), showing how it reproduces the traditional local formulas and thereby providing these with a rule for their global completion.

\medskip

\noindent
{\bf Deriving the form of local gauge potentials on the M5-brane.}
For the case of 11d supergravity, we had shown in
\cite[Props. 1.1, 2.28]{GSS24-SuGra} how the traditional formulas for the local C-field gauge potentials are indeed reproduced by the homotopy theory in \eqref{OnChargGaugePotentialDiagram} and as such become amenable to global completion. Here we extend this analysis to include the B-field on M5-brane worldvolumes, showing how it reproduces the traditional formulas \eqref{TraditionalFormulaForBFieldGaugePotential}.
To this end:

\noindent
{\bf (i)}
Recall the fiberwise Stokes Theorem (e.g. \cite[Lem. 6.1]{FSS23Char}\cite[Lem. 6.1]{FSS23Char}) for differential forms $\widehat{F}$ on a cylinder manifold:
\vspace{-2mm} 
\begin{equation}
    \label{FiberwiseStokes}
    \mathrm{d}
    \int_{[0,1]}
    \widehat F
    \;\;=\;\;
    \iota_1^\ast \widehat F
    \,-\,
    \iota_0^\ast \widehat F
    \,-\,
    \int_{[0,1]}
    \mathrm{d}
    \widehat{F}
    \,,
    \hspace{.7cm}
    \mbox{on}
    \hspace{.5cm}
    \begin{tikzcd}
      X
      \ar[
        r, 
        hook,
        shift left=3pt,
        "{\iota_0}"
      ]
      \ar[
        r, 
        hook,
        shift right=3pt,
        "{\iota_1}"{swap}
      ]
      &
      X \times [0,1]
      \,.
    \end{tikzcd}
\end{equation}

\smallskip

\noindent
{\bf (ii)}
Recall from \eqref{DiagrammaticEOMs}
that closed $\mathfrak{l}_{S^4} S^7$-valued differential forms
are given by
\begin{equation}
  \label{ClosedQuatHopfFibValuedForms}
  \begin{tikzcd}[row sep=small]
    \Omega^1_{\mathrm{dR}}\big(
      -;
      \,
      \mathfrak{l}_{{}_{S^4}} S^7
    \big)_{\mathrm{clsd}}
    \ar[
      d,
      ->>
    ]
    \ar[
      r,
      equals
    ]
    &
    \scalebox{.9}{$
    \left\{
    \def\arraystretch{1.2}
    \begin{array}{l}
      \color{purple}
      H_3 
        \,\in\, 
      \Omega^1_{\mathrm{dR}}(-; b^2 \mathbb{R})
      \\
      G_7 
        \,\in\, 
      \Omega^1_{\mathrm{dR}}(-; b^6 \mathbb{R})
      \\
      G_4 
        \,\in\, 
      \Omega^1_{\mathrm{dR}}(-; b^3 \mathbb{R})
    \end{array}
    \middle\vert
    \def\arraystretch{1.2}
    \begin{array}{l}
      \color{purple}
      \mathrm{d}
      \,
      H_3 
        \;=\;
      G_4
      \\
      \mathrm{d}
      \, G_7 \;=\;
      \tfrac{1}{2} G_4 \, G_4
      \\
      \mathrm{d} \, G_4
      \;=\;
      0
    \end{array}
    \right\}
    $}
    \ar[
      d,
      shift right=42pt
    ]
    \\
    \Omega^1_{\mathrm{dR}}\big(
      -;
      \,
      \mathfrak{l}S^4
    \big)_{\mathrm{clsd}}
    \ar[
      r,
      equals
    ]
    &
    \scalebox{.9}{$
    \left\{
    \def\arraystretch{1.2}
    \begin{array}{l}
      G_7 
        \,\in\, 
      \Omega^1_{\mathrm{dR}}(-; b^6 \mathbb{R})
      \\
      G_4 
        \,\in\, 
      \Omega^1_{\mathrm{dR}}(-; b^3 \mathbb{R})
    \end{array}
    \middle\vert
    \def\arraystretch{1.2}
    \begin{array}{l}
      \mathrm{d}
      \, G_7 \;=\;
      \tfrac{1}{2} G_4 \, G_4
      \\
      \mathrm{d} \, G_4
      \;=\;
      0
    \end{array}
    \right\}
    $}
  \end{tikzcd}
\end{equation}


\noindent
{\bf (iii)}
Observe that a null-concordance \eqref{NullConcordance}
of such data, hence a closed $\mathfrak{l}_{s^4} S^7$-valued differential form on the cylinder manifold $U_i \times [0,1]$ which on one boundary pulls back to the given form data and on the other boundary to zero, is:
\vspace{-2mm} 
\begin{equation}
  \label{NullConcordanceOfFluxDensities}
    \def\arraystretch{1.5}
    \def\arraycolsep{2pt}
    \begin{array}{ccl}
      \big(
        \widehat{G}_4
        ,\,
        \widehat{G}_7
        ,\,
        \widehat{H}_3
      \big)
      \;\in\;
      \Omega^1_{\mathrm{dR}}\big(
        U_i
        \times
        [0,1]
        ;\,
        \mathfrak{l}_{S^4}S^7
      \big)_{\mathrm{clsd}}
    \end{array}
    \hspace{.4cm}
    \mbox{s.t.}
    \hspace{.4cm}
    \left\{
    \def\arraystretch{1.5}
    \def\arraycolsep{2pt}
    \begin{array}{l}
      \iota_0^\ast 
      \big(
        \widehat{G}_4
        ,
        \widehat{G}_7
        ,
        \widehat{H_3}
      \big)
      \;=\;
      0
      \\
      \iota_1^\ast 
      \big(
        \widehat{G}_4
        ,
        \widehat{G}_7
        ,
        \widehat{H_3}
      \big)
      \;=\;
      (G_4, G_7, H_3)
      \,.
    \end{array}
    \right.
\end{equation}

Similarly, given a pair $\big(\widehat{G}_4, \widehat{G}_7, \widehat{H}_3\big)$, $\big(\widehat{G}'_4, \widehat{G}'_7, \widehat{H}'_3\big)$ of such null-concordances for the same $(G_4, G_7, H_3)$, a {\it concordance-of-concordances} between them is 
\vspace{-2mm} 
\begin{equation}
  \label{ConcordanceOfConcordances}
  \big(
    \doublehat{G}_4
    ,
    \doublehat{G}_7
    ,
    \doublehat{H}_3
  \big)
  \;\in\;
  \Omega^1_{\mathrm{dR}}\big(
    U_i 
      \times
    [0,1]_t
      \times
    [0,1]_s
    ;\,
    \mathfrak{l}_{S^4} S^7
  \big)_{\mathrm{clsd}}
  \hspace{.4cm}
  \mbox{s.t.}
  \hspace{.4cm}
  \left\{
  \def\arraystretch{1.5}
  \begin{array}{l}
    \iota_{s=0}^\ast
    \big(
      \doublehat G_4
      ,
      \doublehat G_7
      ,
      \doublehat H_3
    \big)
    \;=\;
    \big(
      \widehat G_4
      ,
      \widehat G_7
      ,
      \widehat H_3
    \big)
    \\
    \iota_{s=1}^\ast
    \big(
      \doublehat G_4
      ,
      \doublehat G_7
      ,
      \doublehat H_3
    \big)
    \;=\;
    \big(
      \widehat G'_4
      ,
      \widehat G'_7
      ,
      \widehat H'_3
    \big)
    \\
    \iota_{t=0}^\ast
    \big(
      \doublehat G_4
      ,
      \doublehat G_7
      ,
      \doublehat H_3
    \big)
    \;=\;
    0
    \\
    \iota_{t=1}^\ast
    \big(
      \doublehat G_4
      ,
      \doublehat G_7
      ,
      \doublehat H_3
    \big)
    \;=\;
    \mathrm{pr}^\ast_{U_i}
    \big(
      G_4
      ,
      G_7
      ,
      H_3
    \big)
    \,,
  \end{array}
  \right.
\end{equation}
where in the last line $\mathrm{pr}_{U_i} : U_i \times [0,1]_s \twoheadrightarrow U_i$ is the canonical  projection.

\medskip

\begin{proposition}[\bf B- \& C-field gauge potentials are local 
$\mathfrak{l}_{{}_{S^4}}S^7$-valued de Rham null coboundaries]
\label{BFieldGaugePotentialFromBoncordance}
Given flux densities $(G_4, G_7, H_3) \,\in\, \Omega^1_{\mathrm{dR}}\big(U_i;\, \mathfrak{l}_{{}_{S^4}}S^7\big)_{\mathrm{clsd}}$ \eqref{ClosedQuatHopfFibValuedForms},

\noindent {\bf (i)} there is a natural surjection from their null concordances
\eqref{NullConcordanceOfFluxDensities}
to triples of gauge potentials as shown here: 
\smallskip 
\begin{equation}
  \label{GaugePotentialsFromNullConcordances}
  \hspace{-2mm} 
 \left\{\!\!\!\!\!
 \adjustbox{raise=4pt}{
  \begin{tikzcd}[
    column sep=27pt
  ]
    &&
    \ast
    \ar[
      dr,
      bend left=15
    ]
    &[-15pt]
    \\
    U_i
    \ar[
      rr,
      "{
        (G_4,\, G_7,\, H_3)
      }"{swap},
      "{\ }"{name=t}
    ]
    \ar[
      urr,
      bend left=15,
      "{\ }"{swap,name=s}
    ]
    \ar[
      from=s,
      to=t,
      Rightarrow,
      dashed, 
      darkorange,
      "{
        (
          \widehat{G}_4,
          \widehat{G}_7,
          \widehat{H}_3
        )
      }"{pos=.0}
    ]
    &&
    \Omega^1_{\mathrm{dR}}\big(
      -;\,
      \mathfrak{l}_{{}_{S^4}}S^7
    \big)_{\!\mathrm{clsd}}
    \ar[
      r,
      shorten <= -5pt,
      "{
        \eta^{\,\scalebox{.55}{$\shape$}}
      }"{swap, pos=.3}
    ]
    &
    \shape
    \big(
    \Omega^1_{\mathrm{dR}}\big(
      -;\,
      \mathfrak{l}_{{}_{S^4}}S^7
    \big)_{\!\mathrm{clsd}}
    \big)
  \end{tikzcd}
  }
  \!\!\!\!\!\!\right\}
  \; 
  \twoheadrightarrow
  \;
  \left\{\!
  \def\arraystretch{1.3}
  \def\arraycolsep{1pt}
  \begin{array}{rcl}
    C_3 
      &\in& 
    \Omega^3_{\mathrm{dR}}(U_i)
    \\
    C_6 
      &\in&
    \Omega^3_{\mathrm{dR}}(U_i)
    \\
    B_2 
      &\in&
    \Omega^2_{\mathrm{dR}}(U_i)
  \end{array}
  \;
  \middle\vert
  \;
  \def\arraystretch{1.3}
  \def\arraycolsep{1pt}
  \begin{array}{rcl}
     \mathrm{d}
     \,
     C_3
     &=&
     G_4
     \\
     \mathrm{d}\, C_6
     &=&
     G_7 - 
     \tfrac{1}{2}
     C_3 \, G_4
     \\
     \mathrm{d}\, B_2
     &=&
     H_3 - C_3
  \end{array}
 \! \right\}
\end{equation}

\smallskip 
\noindent {\bf (ii)} This surjection takes concordances-of-concordances \eqref{ConcordanceOfConcordances} to gauge equivalences of gauge potentials of the following form
\begin{equation}
  \label{GaugeEquivalencesOfGaugePotentials}
  (C_3,\, C_6,\, B_2)
  \;\;
  \sim
  \;\;
  (C'_3,\, C'_6,\, B'_2)
  \hspace{.6cm}
  \Leftrightarrow
  \hspace{.6cm}
  \exists
  \left.
  \def\arraystretch{1.2}
  \begin{array}{l}
    C_2 \;\in\;
    \Omega^2_{\mathrm{dR}}(U_i)
    \\
    C_5
    \;\in\;
    \Omega^5_{\mathrm{dR}}(U_i)
    \\
    B_1
    \;\in\;
    \Omega^1_{\mathrm{dR}}(U_i)
  \end{array}
  \!\!\! \right\}
  \;\;
  \scalebox{.9}{\rm such that}
  \;\;
  \left\{\!\!\!
  \def\arraystretch{1.2}
  \begin{array}{l}
    \mathrm{d}
    \,
    C_2 \;=\;
    C'_3 - C_3
    \\
    \mathrm{d}
    \,
    C_5
    \;=\;
    C'_6 - C_6 - \tfrac{1}{2}
    C'_3 \, C_3
    \\
    \mathrm{d}
    \,
    B_1
    \;=\;
    B'_2 - B_2
    +
    C_2\,.
  \end{array}
  \right.
\end{equation}

\end{proposition}
\begin{proof}
We take the map $\big(\widehat{G}_4, \widehat{G}_7, \widehat{H}_3\big) \,\mapsto\, \big(C_3, C_6, B_2  \big)$ to be given on the C-field flux densities as in \cite[(70)]{GSS24-SuGra} 
\begin{equation}
  \label{FromNullCoboundaryC3C6}
  C_3
  \;:=\;
  \int_{[0,1]}
  \widehat G_4
  ,\,
  \hspace{.8cm}
  C_6
  \;:=\;
  \int_{[0,1]}
  \Bigg(
    \widehat G_7
    -
    \tfrac{1}{2}
    \underbrace{
    \bigg(
      \int_{[0,-]}
      \widehat G_4
    \bigg)
    }_{
      \color{gray}
      \widehat{C}_3
    }
    \widehat G_4
  \Bigg)
\end{equation}
and analogously on the B-field flux density to be:
\begin{equation}
  \label{B2FromNullConcordance}
  B_2
  \;:=\;
  \int_{[0,1]} \widehat{H}_3
  \,.
\end{equation}
To see that this satisfies the relations on the right of \eqref{GaugePotentialsFromNullConcordances}: For the first two lines this is \cite[(70)]{GSS24-SuGra}, while for the last line we compute as follows:
$$
  \def\arraystretch{1.5}
  \begin{array}{lll}
  \mathrm{d}
  \,
  B_2
  &
  \;=\;
  \mathrm{d} \int_{[0,1]} \widehat{H}_3
  &
  \proofstep{
    by \eqref{B2FromNullConcordance}
  }
  \\
 & \;=\;
  \iota_1^\ast \widehat{H}_3
  - 
  \int_{[0,1]}
  \mathrm{d}
  \widehat{H}_3
  &
  \proofstep{
    by \eqref{FiberwiseStokes}
  }
  \\
 & \;=\;
  H_3
  -\int_{[0,1]}
  \widehat{G}_4
  &
  \proofstep{
    by
    \eqref{NullConcordanceOfFluxDensities}
  }
  \\
&  \;=\;
  H_3
  -
  C_3
  &
  \proofstep{
    by \eqref{FromNullCoboundaryC3C6}.
  }
  \end{array}
$$
To see that this map is indeed a surjection, we exhibit a section
$\big(C_3, C_6, B_2\big) \mapsto \big( \widehat{G}_4, \widehat{G}_7, \widehat{H}_3\big)$
extending the lifts of $(C_3, C_6)$ from \cite[(72)]{GSS24-SuGra} to $H_3$ as follows:
$$
  \left.
  \def\arraystretch{1.5}
  \begin{array}{rcl}
    \widehat{G}_4
    &:=&
    t\, 
    G_4
    +
    \mathrm{d}t
    \,
    C_3
    \\
    \widehat{G}_7
    &:=&
    t^2
    \,
    G_7
    \,+\,
    2 t\mathrm{d}t
    \,
    C_6
    \\
    \widehat H_3
    &:=&
    t H_3 
     + 
    \mathrm{d}t\, B_2
  \end{array}
 \!\!\! \right\}
  \scalebox{.9}{
    \def\arraystretch{1}
    \begin{tabular}{c}
      which indeed
      \\
      satisfies
    \end{tabular}
  }
  \left\{\!\!\!
  \def\arraystretch{1.3}
  \begin{array}{rcl}
    \mathrm{d}\big(
      t\, G_4
      \,+\,
      \mathrm{d}t
      \,
      C_3
    \big)
    &=&
    0
    \\
    \mathrm{d}
    \big(
      t^2
      \,
      G_7
      \,+\,
      2 t\mathrm{d}t
      \,
      C_6
    \big)
    &=&
    \tfrac{1}{2}
    \big(
    t\, 
    G_4
    +
    \mathrm{d}t
    \,
    C_3
    \big)
    \big(
    t\, 
    G_4
    +
    \mathrm{d}t
    \,
    C_3
    \big)
    \\
    \mathrm{d}
    \big(
      t H_3 + \mathrm{d}t\, B_2 
    \big)
    &=&
    \big(
     t \, G_4
     \,+\,
     \mathrm{d}t\, C_3
    \big)
    \,.
  \end{array}
  \right.
$$

Finally, to see that the map respects equivalences, 
consider a pair of null-concordances
$\big(\widehat G_4, \widehat{G}_7, \widehat{H}_3 \big)$, $\big(\widehat G'_4, \widehat G'_7, \widehat H'_3 \big)$ with a concordance-of-concordances $\big( \doublehat{G}_4 ,\, \doublehat{G}_7 ,\, \doublehat{H}_3 \big)$ between them,
and produce a gauge equivalence as in \eqref{GaugeEquivalencesOfGaugePotentials} from this by taking, as in \cite[(73)]{GSS24-SuGra},
\begin{equation}
  \label{2ndCoboundariesFrom2ndConcordances}
  C_2 
    \;:=\;
  {\displaystyle \int}_{\!\!\!\! s \in [0,1]} \;
  {\displaystyle \int}_{\!\!\!\! t \in [0,1]}
  \doublehat{G}_4
  \,,
  \hspace{.5cm}
  C_5 
    \;:=\;
 {\displaystyle \int}_{\!\!\!\! s \in [0,1]}\;
 {\displaystyle \int}_{\!\!\!\! t \in [0,1]}
  \Bigg(
    \doublehat{G}_7
    -
    \tfrac{1}{2}
    \bigg(
    {\displaystyle \int}_{\!\!\!\! t' \in [0,-]}
    \doublehat{G}_4
    \bigg)
    \doublehat{G}_4
  \Bigg)
  -
  \tfrac{1}{2}
  C_2 \, C_3
  \,,
\end{equation}
and now in addition 
\vspace{2mm} 
\begin{equation}
  \label{B1FromConcordanceOfConcordances}
  B_1
  \;:=\;
  \int_{s \in [0,1]}
  \int_{t \in [0,1]}
  \doublehat{H}_3
  \,.
\end{equation}

\smallskip 
\noindent That this indeed satisfies the required relations follows for $C_2$ and $C_5$
by \cite[(74)]{GSS24-SuGra} and for $B_1$ by the following computation:
\newpage 
$$
  \def\arraystretch{1.8}
  \begin{array}{lll}
  \mathrm{d}
  \,
  B_1
  &
  \;\defneq\;
  \mathrm{d}
  \,
  \int_{s \in [0,1]}
  \int_{t \in [0,1]}
  \doublehat{H}_3
  &
  \proofstep{
    by
    \eqref{B1FromConcordanceOfConcordances}
  }
  \\
  & \;=\;
  \iota^\ast_{s=1}
  \int_{t \in [0,1]}
  \doublehat{H}_3
  \,-\,  
  \iota^\ast_{s=0}
  \int_{t \in [0,1]}
  \doublehat{H}_3
  \,-\,
  \int_{s \in [0,1]}
  \mathrm{d}
  \int_{t \in [0,1]}
  \doublehat{H}_3
  &
  \proofstep{
    by 
    \eqref{FiberwiseStokes}
  }
  \\
  &\;=\;
  \int_{t \in [0,1]}
  \iota^\ast_{s=1}
  \doublehat{H}_3
  \,-\,  
  \int_{t \in [0,1]}
  \iota^\ast_{s=0}
  \doublehat{H}_3
  \;-\;
  \int_{s \in [0,1]}
  \widehat H_3
  \,+\,
  \int_{s \in [0,1]}
  \int_{t \in [0,1]}
  \mathrm{d}
  \doublehat{H}_3
  &
  \proofstep{
    by 
    \eqref{FiberwiseStokes}
  }
  \\
  &\;=\;
  \int_{t \in [0,1]}
  \widehat{H}'_3
  -
  \int_{t \in [0,1]}
  \widehat{H}_3
  + 
  \int_{s \in [0,1]}
  \int_{t \in [0,1]}
  \doublehat{G}_4
  &
  \proofstep{
    by \eqref{ConcordanceOfConcordances}
  }
  \\
  &\;=\;
  B'_2
  -
  B_2
  + 
  C_2
  &
  \proofstep{
    by
    \eqref{B2FromNullConcordance}
    \&
    \eqref{2ndCoboundariesFrom2ndConcordances}.
  }
  \end{array}
$$

\vspace{-5mm} 
\end{proof}

\begin{remark}[\bf Reproducing the traditional local gauge potential]
  \label{ReproducingTraditionalLocalGaugePotentials}
  Equation \eqref{GaugePotentialsFromNullConcordances} in Prop. \ref{BFieldGaugePotentialFromBoncordance} says, in particular, that the traditional formula \eqref{TraditionalFormulaForBFieldGaugePotential} for the gauge potential on the M5-brane is reproduced locally.
\end{remark}

\subsection{Skyrmions and Anyons on M5}
\label{ResultingSolitons}

Finally, we spell out key consequences of quantizing the $H_3$-flux density on M5-branes in co-Homotopy cohomology theory, according to \eqref{TwistedCohomotopyInIntroduction}, highlighting how this leads to quantum observables of skyrmions and of anyonic quantum states \eqref{GelfandRaikovTheorem}, thus supporting the idea that the completed field content on the M5-brane generically reflects properties of strongly coupled/correlated quantum systems.

\medskip 
The method of non-perturbative quantization of the topological charge sector which we use here is that of \cite{SS23-Obs}, and the approach to anyonic quantum states is broadly that of \cite{SS22-Config}\cite{SS23-DefectBranes},
but where in these previous discussions we focused on intersecting brane configurations, here we give an alternative construction with single M5-branes that is supported by the above super-space analysis.

\medskip

\noindent
{\bf M5-branes near the horizon of their own black brane solution.}
In doing so, we focus on the special but important case where the pullback of the background M-brane charge to the M5-brane worldvolume trivializes, $\phi^\ast (c_3, c_6) \simeq 0$. This may happen if the background charge vanishes by itself in the first place, for instance if the target spacetime is flat Minkowski space $\mathbb{R}^{1,10}$ (as in Ex. \ref{FlatM5Branes}), or else if the M5-worldvolume does not wrap cycles on which the background charge is supported; for instance if it stretches along the asymptotic boundary of the near horizon geometry $\mathrm{AdS}_7 \times S^4$ of its own ``black'' brane background  (as discussed in \cite{CKvP98}\cite{PST99}\cite[\S 3.3]{CKKTvP98}, in microscopic resolution of AdS/CFT duality).

\smallskip

Since this pullback charge serves as the twist of the twisted 3-co-Homotopy on the worldvolume, under Hypothesis H \eqref{TwistedCohomotopyInIntroduction},
it follows that in this case the M5 worldvolume B-field is quantized in plain 3-co-Homotopy (recalling from ftn. \ref{OnTangentialTwists} that for the purpose of the present discussion we are disregarding further tangential twists of Cohomotopy):
\vspace{1mm} 
\begin{equation}
  \label{Twisted3CohomotopyReducingToPlain}
  \hspace{-.8cm}
  \def\arraycolsep{1pt}
  \def\arraycolsep{0pt}
  \begin{array}{ccccc}
    \scalebox{.7}{
      \def\arraystretch{.9}
      \color{darkblue}
      \bf
      \begin{tabular}{c}
        If background charge
        \\
        vanishes on worldvolume
      \end{tabular}
    }
    &
    \scalebox{.7}{
      \def\arraystretch{.9}
      \color{darkgreen}
      \bf
      \begin{tabular}{c}
        then
      \end{tabular}
    }
    &
    \scalebox{.7}{
      \def\arraystretch{.9}
      \color{darkblue}
      \bf
      \begin{tabular}{c}
        the charge-twisted 
        \\
        3-Cohomotopy
      \end{tabular}
    }
    &
    \scalebox{.7}{
      \def\arraystretch{.9}
      \color{darkgreen}
      \bf
      \begin{tabular}{c}
        reduces
        \\
        to
      \end{tabular}
    }
    &
    \scalebox{.7}{
      \def\arraystretch{.9}
      \color{darkblue}
      \bf
      \begin{tabular}{c}
        plain
        \\
        3-Cohomotopy
      \end{tabular}
    }
    \\[+5pt]
    \phi^\ast(c_3, c_6)
    \,\simeq\, 0
    &\Rightarrow&
    \pi^{3+\phi^\ast(c_3 + c_6)}(\Sigma)
    &=&
    \pi^{3}(\Sigma)
    \\[5pt]
    &&
  \left\{
  \hspace{-4pt}
  \adjustbox{raise=3pt}{
  \begin{tikzcd}
    &&
    S^7
    \ar[
      d,
      "{
        h_{\mathbb{H}}
      }"
    ]
    \\
    \Sigma
    \ar[
      r,
      "{ \phi }"{swap}
    ]
    \ar[
      urr,
      dashed,
      "{
        b_2
      }"
    ]
    &
    X
    \ar[
      r,
      "{
        (c_3, c_6)
      }"{swap}
    ]
    &
    S^4
  \end{tikzcd}
  }
  \hspace{-2pt}
  \right\}_{
    \mathrlap{
      \!\!\!\big/\mathrm{rel.hmtp.}
    }}
    &&
  \left\{
  \hspace{-4pt}
  \begin{tikzcd}
    &
    S^3
    \ar[
      r
    ]
    \ar[
      dr,
      phantom,
      "{
        \scalebox{.7}{
          \color{gray}
          (pb)
        }
      }"
    ]
    \ar[
      d
    ]
    &
    S^7
    \ar[
      d,
      "{
        h_{\mathbb{H}}
      }"
    ]
    \\
    \Sigma
    \ar[
      dr,
      "{
        \phi
      }"{swap}
    ]
    \ar[
      r
    ]
    \ar[
      ur,
      dashed,
      "{ b_2 }"
    ]
    &
    \ast
    \ar[
      r,
      "{
        0
      }"{description}
    ]
    \ar[
      d,
      Rightarrow,
      shorten=4pt,
      start anchor={[xshift=7pt]},
      end anchor={[xshift=-7pt]},
    ]
    &
    S^4
    \\[-10pt]
    &
    X
    \ar[
      ur,
      "{
        (c_3, c_6)
      }"{swap, sloped}
    ]
  \end{tikzcd}
  \hspace{-2pt}
  \right\}_{
    \mathrlap{\!\!\big/\mathrm{rel.hmtp.}}
    }
  \end{array}
\end{equation}

\smallskip

\noindent
{\bf Self-dual string charge quantization.}
With this and in  direct generalization of traditional Dirac monopole charge quantization 
\vspace{1mm} 
$$
  \scalebox{.7}{
    \color{darkblue}
    \bf
    \def\arraystretch{.9}
    \begin{tabular}{c}
      Charges of Dirac monopole
      \\
      by trad. flux quantization
    \end{tabular}
  }
  H^2\big(
    \underbrace{
    \mathbb{R}^{1,3}
    \setminus
    \mathbb{R}^{0,1}
    }_{
      \mathclap{
      \scalebox{.7}{
        \color{darkblue}
        \bf
        \def\arraystretch{.9}
        \begin{tabular}{c}
          Spacetime around 
          \\
          magnetic monopole
        \end{tabular}
      }
      }
    }
    ;\,
    \mathbb{Z}
  \big)
  \;\simeq\;
  H^2\big(
    \underbrace{
    \mathbb{R}^{0,1}
    \times
    \mathbb{R}_{> 0}
    }_{\color{gray} 
      \mathrm{contractible}
    }
    \times 
    S^2
    ;\,
    \mathbb{Z}
  \big)
  \;\simeq\;
  H^2\big(
    S^2
    ;\,
    \mathbb{Z}
  \big)
  \;\simeq\;
  \mathbb{Z}
  \scalebox{.7}{
    \color{darkblue}
    \bf
    \def\arraystretch{.9}
    \begin{tabular}{c}
      Integer number of
      \\
      magnetic monopoles
    \end{tabular}
  }
$$
the archetypical charge quantization situation on an M5-brane  is that for its monopole string solution (the {\it self-dual string soliton} \cite[\S 3]{HoweLambertWest98}), hence with M5-worldvolume topology of the form (using \cite[(3.16)]{HoweLambertWest98})
\begin{equation}
  \label{SelfDualStringInWorldvolujme}
  \Sigma
  \;\;\simeq\;\;
  \mathbb{R}^{1,5}
  \setminus
  \mathbb{R}^{1,1}
  \;\;\simeq\;\;
  \mathbb{R}^{1,1}
    \times
  \mathbb{R}_{> 0}
    \times 
  S^3
  \,.
\end{equation}
The corresponding charge quantization law implied by Hypothesis H, via  
\eqref{Twisted3CohomotopyReducingToPlain},
gives string charges in the 3-Cohomotopy of the 3-sphere. Via the {\it Hopf degree theorem} this is canonically identified with the ordinary integral cohomology of the 3-sphere (cf. \cite[(35)]{FSS20-H}) 
\begin{equation}
  \label{HopfDegreeTheorem}
  \mbox{
    $M^3$ 
    a closed orientable 3-fold
  }
  \;\;\;\;\;
  \Rightarrow
  \;\;\;\;\;
  \pi^3\big(M^3\big)
  \;\simeq\;
  H^3(M^3;\, \mathbb{Z})
\end{equation}
and thereby given by the group of integers: 
\smallskip 
$$
  {
    \scalebox{.7}{
      \color{darkblue}
      \bf
      \def\arraystretch{.9}
      \begin{tabular}{c}
        Charges of self-dual string
        \\
        according to 
        \\
        Hypothesis H
      \end{tabular}
    }
  }
  \def\arraystretch{1.6}
    \pi^3\big(
      \Sigma
    \big)
    \;\defneq\;
    \pi^3\big(
      \underbrace{
      \mathbb{R}^{1,1}
       \times
      \mathbb{R}_{> 0}
      }_{
        \mathrm{contractible}
      }
       \times 
      S^3
    \big)
    \\
    \quad 
    \underset{
      \mathclap{
      \raisebox{-10pt}{
        \scalebox{.7}{
          \def\arraystretch{.9}
          \begin{tabular}{c}
            \color{darkgreen}
            \bf
            Hopf degree
            \\
            \color{darkgreen}
            \bf theorem
            \eqref{HopfDegreeTheorem}
          \end{tabular}
        }
        }
      }
    }{
    \;\simeq\;
    }
    \quad 
    \pi^3\big(
      S^3
    \big)
    \\
    \;\simeq\;
    H^3(S^3) \;\simeq\; \mathbb{Z}
  {
    \scalebox{.7}{
      \color{darkblue}
      \bf
      \def\arraystretch{.9}
      \begin{tabular}{c}
        coincide with the 
        \\
        traditional
        charges 
        \\
        in ordinary cohomology.
      \end{tabular}
    }
  }
$$
Hence in this case the charges in co-Homotopy coincide with those seen in ordinary cohomology and the traditionally (and often tacitly) expected charge quantization law of the self-dual string is recovered.

\smallskip

However, even in this situation, the cohomotopical perspective provides further insight:

\smallskip

\noindent
{\bf Pixelated flux.}
The authors of \cite{HLLSZ19} suggest (cf. their Fig. 2) to think of the $N$ units of flux through the 3-sphere surrounding an integer-charged brane as being witnessed by a distribution of $N$ points (``pixels'') {\it on} the sphere. While this perspective is suggested by the physical picture, it is not really supported by the mathematics of charge quantization in ordinary cohomology.

\smallskip 
However, for charge quantization in co-Homotopy, the original {\bf Pontrjagin theorem} establishes a natural bijection between (cohomotopical) charges and actual branes in the guise of (cobordism classes of) submanifolds of the domain space, which in this situation reproduces much the picture of pixelation:

\smallskip 
Concretely (see \cite[Prop. 3.24 \& pp. 13]{SS23-Mf}\cite[\S 2.1]{SS20Tadpole}), for $S \xhookrightarrow{\;} M$ a smooth submanifold of co-dimension $3$ inside a closed smooth manifold $M$, then a {\it normal framing} of $S$ (namely a trivialization of its normal bundle) canonically induces a map $M \xrightarrow{\;} S^3$ (given by accordingly projecting a tubular neighborhood of $S$ onto $\mathbb{R}^3 \simeq S^3 \setminus \{s_0\}$ and mapping its complement to $s_0$), which under Hypothesis H we understand as assigning the {\it Cohomotopy charge} carried by the brane $S$.

\smallskip 
Now Pontrjagin's theorem means that two such branes $S$, $S'$ carry the same Cohomotopy charge iff they are (normally framed-)cobordant, namely that assigning Cohomotopy charge gives a bijection:
\begin{equation}
  \label{PontrjaginTheorem}
  \begin{tikzcd}[
    column sep=40pt
  ]
    \scalebox{.7}{
      \color{darkblue}
      \bf
      \def\arraystretch{.9}
      \begin{tabular}{c}
        Cobordism classes of
        \\
        normally framed submanifolds
        \\
        of co-dimension 3
      \end{tabular}
    }
    \mathrm{Cob}^3_{\mathrm{Fr}}\big(
      \Sigma
    \big)
    \ar[
      rr,
      "{
        \scalebox{.7}{
          \color{darkgreen}
          \bf
          \def\arraystretch{.9}
          \begin{tabular}{c}
            assign brane charge
          \end{tabular}
        }
      }",
      "{ \sim }"{swap}
    ]
    &&
    \pi^3(\Sigma)
    \scalebox{.7}{
      \color{darkblue}
      \bf
      \def\arraystretch{.9}
      \begin{tabular}{c}
        3-Cohomotopy
      \end{tabular}
    }
  \end{tikzcd}
\end{equation}
thus accurately identifying branes by their charges, and {\it vice versa}.

\hspace{-.8cm}
\begin{tabular}{p{10cm}}
But if we apply this correspondence to the above situation, then it identifies (a) integer flux $N \in \mathbb{Z}$ through the $S^3$ around the self-dual string soliton with (b) $\vert N \vert$ points on $S^3$ (the co-dimension 3 submanifold) each carrying charge $\mathrm{sgn}(N)$ (one of two distinct normal framings, reflecting branes and anti-branes, respectively, cf. \cite[p. 12]{SS23-Mf}\cite[p. 12]{SS20Tadpole}). This Cohomotopy theory substantiates the picture of \cite[Fig. 2]{HLLSZ19} (adapted on the right).
\end{tabular}
\hspace{-7pt}
\adjustbox{raise=2pt}{
\begin{minipage}{7.3cm}
  \begin{equation}
  \label{Pixelation}
  \begin{tikzcd}[
    row sep=2pt,
    column sep=5pt
  ]
    \pi^3\big(S^3\big)
    \ar[
      rr,
      "{ \sim }"
    ]
    &&
    \mathrm{Cob}^3_{\mathrm{Fr}}\big(
      S^3
    \big)
    \\
    N &\longmapsto&
    \left[
    \raisebox{-23pt}{
      \includegraphics[width=2cm]{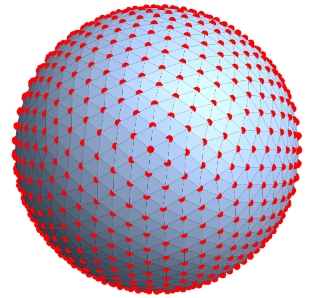}
    }
    \right]
  \end{tikzcd}
  \end{equation}
\end{minipage}
}

\medskip

\noindent
{\bf Skyrmions on M5.}
Alongside {\it singular branes} like the  self-dual string \eqref{SelfDualStringInWorldvolujme} (the flux field it sources would be singular if the actual locus of the singular string were not removed from the M5-worldvolume), flux quantization allows to discuss genuine {\it solitonic branes} which source a flux density that is non-singular everywhere and is instead topologically stabilized by the constraint that it vanishes at infinity (\cite[\S 2.2]{SS24-Flux}). For example, in ordinary electromagnetism there are (recalled in \cite[\S 2.1]{SS24-Flux}) besides the singular branes being the Dirac monopoles (which remain hypothetical) also solitonic branes being {\it Abrikosov vortices} (which are experimentally well observed).

\smallskip

Another example of solitonic objects in this sense are {\it Skyrmions} (cf. \cite{ANW83}\cite{RhoZahed16}\cite{Manton22}), which are solitons in the $\mathrm{SU}(2)$-valued pion field (possibly including higher vector meson field contributions), that at least approximate baryon bound states in confined quantum chromodynamics (and which in the Witten-Sakai-Sugimoto model are essentially identified with wrapped D4-branes, cf. \cite{Sugimoto16}).

\smallskip 
Concretely, the solitonic nature of Skyrmions requires  that the spatial pion field $f : \mathbb{R}^3 \xrightarrow{\;} \mathrm{SU}(2)$ takes the trivial value $1 \in \mathrm{SU}(2)$ ``at infinity''. This may neatly be formalized (cf. \cite[Rem. 2.3 \& p. 14]{SS23-Mf}) by ``adjoining the point at infinity'' to $\mathbb{R}^3$ via passage to its {\it Alexandroff one-point compactification} $\mathbb{R}^3_{\cup \{\infty\}} \,\simeq\, S^3$ and 
requiring $f$ to extend to a map on this compactification such that there it takes the literal point $\infty$ to $1$:
$$
  \begin{tikzcd}[
    row sep=-3pt
  ]
    \mathllap{
    \scalebox{.7}{
      \color{darkblue}
      \bf
      \def\arraystretch{.9}
      \begin{tabular}{c}
        Euclidean space with
        \\
        point at infinity adjoined
      \end{tabular}
    }
    }
    \mathbb{R}^3_{\cup \{\infty\}}
    \ar[
      rr,
      "{ f }",
      "{
        \scalebox{.7}{
          \color{darkgreen}
          \bf
          pion field
        }
      }"{swap}
    ]
    &&
    \mathrm{SU}(2)
    \\
    \infty 
      &\longmapsto&
    1 \,.
  \end{tikzcd}
$$
The {\it baryon number} of such as Skyrmion configuration (e.g., \cite[(4.26)]{BMS10}\cite[(4.26)]{Manton22}) is then the homotopy class of this map, under the Hopf degree theorem \eqref{HopfDegreeTheorem}, noticing that $\mathrm{SU}(2) \underset{\color{orangeii}\mathrm{homeo}}{\,\simeq\,} S^3$:
\vspace{0mm} 
\begin{equation}
  \label{BaryonNumberAsCohomotopyCharge}
  \scalebox{.7}{
    \color{darkblue}
    \bf
    \def\arraystretch{.9}
    \begin{tabular}{c}
      Baryon number
      \\
      of Skyrmion 
    \end{tabular}
  }
  [f]
  \;\in\;
  \pi_0
  \Big(
  \mathrm{Maps}^{\ast/}\big(
    \mathbb{R}^3_{\cup \{\infty\}}
    ,\,
    \mathrm{SU}(2)
  \big)
  \Big)
  \;
  \simeq
  \;
  \pi^3\big(
    \mathbb{R}^3_{\cup \{\infty\}}
  \big)
  \;
  \simeq
  \;
  \pi^3\big(
    S^3
  \big)
  \;\
  \simeq
  \;
  H^3(S^3;\, \mathbb{Z})
  \;
  \simeq
  \;
  \mathbb{Z}
  \,.
\end{equation}
This makes explicit that baryon number in Skyrmion theory is equivalently the Cohomotopy charge embodied by the pion field, of just the same form as the solitonic B-field charge on an M5-brane worldvolume domain of the form 

\vspace{-.6cm}
\begin{equation}
  \label{MinkowskianWorldvolume}
  \Sigma 
    \,\defneq\, 
  \mathbb{R}^{1,2} \times \mathbb{R}^3_{\cup \{\infty\}}
  \,,
\end{equation}
\vspace{-.5cm}

\noindent
in that:\footnote{
  A similar consideration appears in \cite[p. 7]{Intriligator00}, albeit not for the B-field flux in 3-Cohomotopy (maps to $S^3$) but for the worldvolume scalar fields regarded as maps to $S^4$ as in  \eqref{CFieldFluxPulledBackToNM5Breanes}.
}
$$
  \pi^3(\Sigma)
  \;\;
  \defneq
  \;\;
  \pi^3\big(
    \underbrace{
      \mathbb{R}^{1,2}
    }_{\color{gray} 
      \mathclap{
        \mathrm{contractible}
      }
    }
    \times
    \,
    \mathbb{R}^3_{\cup \{\infty\}}
  \big)
  \;\;
  \simeq
  \;\;
  \pi^3\big(
    \mathbb{R}^3_{\cup \{\infty\}}
  \big)
  \;\;
  \simeq
  \;\;
  \pi^3\big(
    S^3
  \big)
  \;\;
  \simeq
  \;\;
  H^3(S^3;\, \mathbb{Z})
  \;\;
  \simeq
  \;\;
  \mathbb{Z}
  \,.
$$

\smallskip

Beyond these sets of charges, flux quantization provides us with their {\it moduli spaces}:

\medskip

\noindent
{\bf Moduli spaces of Skyrmions.} 
Flux quantized fields in the form \eqref{GaugePotentialAsHomotopy} do not just form a set, but naturally a (supergeometric) higher groupoid (a ``space'', review includes \cite[\S 1]{FSS23Char}) whose (higher) morphisms are the (higher) gauge transformations (hence which is the finite version of the on-shell BRST complex of the higher gauge theory). Concretely, this is the homotopy fiber product of the smooth super-set of on-shell flux densities with the moduli space of charges, where the latter is, in our situation \eqref{Twisted3CohomotopyReducingToPlain}, the pointed mapping space
\begin{equation}
  \label{CohomotopyModuliSpace}
  \scalebox{.7}{
    \color{darkblue}
    \bf
    \def\arraystretch{.9}
    \begin{tabular}{c}
      Moduli space of
      \\
      co-Homotopy charges
    \end{tabular}
  }
  \begin{tikzcd}
    \widetilde\boldpi^3(\Sigma)
    \;:=\;
    \mathrm{Maps}^{\ast/}
    \big(
      \Sigma
      ,\,
      S^3
    \big)
    \ar[
      r,
      ->>,
      "{
        \pi_0
      }"
    ]
    &
    \widetilde{\pi}^3\big(
      \Sigma
    \big)
    \scalebox{.7}{
      \color{darkblue}
      \bf
      \def\arraystretch{.9}
      \begin{tabular}{c}
        Set of
        \\
        co-Homotopy charges
      \end{tabular}
    }
  \end{tikzcd}
\end{equation}
Therefore, from now on we understand $\Sigma$ as being pointed by a ``point at infinity'', and its co-Homotopy to be the corresponding ``reduced'' co-Homotopy $\widetilde \pi^3$ (classes of point-preserving maps) such as to implement any constraints of charges vanishing at infinity.\footnote{Notice that the point at infinity may be disjoint, denoted $\Sigma_{\sqcup \{\infty\}} := \Sigma \sqcup \{\infty\}$, whereby reduced co-Homotopy subsumes plain co-Homotopy: $\widetilde \pi^3(\Sigma_{\sqcup \{\infty\}}) \,=\, \pi^3(\Sigma)$ (similarly for any other generalized cohomology theory).}

\smallskip 
This entails that for $(X, \infty_{{}_X})$ and $(Y, \infty_{{}_Y})$ a pair of spaces with designated points at infinity, their appropriate product space is the ``smash product''
$$
  X \wedge Y
  \;\;
  :=
  \;\;
  \frac
    { X \times Y }
    {
      X 
      \!\times\! \{\infty_{{}_Y}\}
      \,\cup\,
      \{\infty_{{}_X}\}
      \!\times\!
      Y
    }
  \;\;\;\;
  \mbox{
    pointed by
    $(\infty_{{}_X}, \infty_{{}_Y})$
  }
$$
in terms of which our Minkowskian M5 worldvolume \eqref{MinkowskianWorldvolume} properly reads as follows:\footnote{Recall that this yoga of pointed spaces just serves to neatly encode fall-off conditions on the fields under consideration.}
\begin{equation}
  \label{MinkowskianWorldvolumeAsPointedSpace}
  \def\arraystretch{1.6}
  \begin{array}{cl}
  \Sigma
  \;\;
  \defneq
  \;\;
  \mathbb{R}^{1,2}_{\sqcup \{\infty\}}
  \,\wedge\,
  \mathbb{R}^{3}_{\cup \{\infty\}}
\quad 
  \Rightarrow
  &
  \widetilde \pi^3(\Sigma)
  \;\simeq\;
  \widetilde {\pi}^3
  \big(
    \mathbb{R}^3_{\cup \{\infty\}}
  \big)
  \;\simeq\;
  \mathbb{Z}
 \end{array}
  \qquad \qquad 
\adjustbox{
  margin=-3pt,
  fbox,
  raise=-11pt
}{
\hspace{-8pt}
\begin{tikzpicture}
\node
 at (0,1) {
   \scalebox{1}{
     $\mathbb{R}^{1,1}_{\phantom{A}}$
   }
 };
\node
 at (.5,1) {
   \scalebox{1}{
     $\times$
   }
 };
\node
 at (1,1) {
   \scalebox{1}{
     $\mathbb{R}^3_{\phantom{A}}$
   }
 };
\node
 at (1.5,1) {
   \scalebox{1}{
     $\times$
   }
 };
\node
 at (2,1) {
   \scalebox{1}{
     $\mathbb{R}^{1}_{\phantom{A}}$
   }
 };
\draw[
  line width=10pt,
  fill=olive,
  draw=olive
]
  (-.35,.5) -- (0+.3,.5);  
\draw[
  line width=10pt,
  fill=olive,
  draw=olive
]
  (2-.25,.5) -- (2+.2,.5);  
\end{tikzpicture}
\hspace{-10pt}
}
\end{equation}
where the ``brane diagram'' on the right indicates the extension of the Skyrmion as seen via its Cohomotopy charges.

\smallskip
Remarkably, the charged points that, so far, appeared only through their {\it net number} in $\mathbb{Z}$ now ``come to life'' as we pass to their co-Homotopy moduli space \eqref{CohomotopyModuliSpace},
in that this happens to equivalently be (the group completion $\mathbb{G}$ of) the {\it configuration space of points} in space, $\mathrm{Conf}(\mathbb{R}^3)$, whose elements are finite subsets of $\mathbb{R}^3$ and whose paths are continuous {\it motions} of these (by \cite[Thm. 1]{Segal73}):
\vspace{2mm} 
\begin{equation}
  \label{SegalTheorem1}
  \scalebox{.7}{
    \color{darkblue}
    \bf
    \def\arraystretch{.9}
    \begin{tabular}{c}
      co-Homotopy moduli
      \\
      space of Skyrmions
    \end{tabular}
  }
  \widetilde\boldpi^3\big(
    \mathbb{R}^3_{\cup \{\infty\}}
  \big)
  \;\;
  \underset{\color{orangeii}
    \mathrm{hmtp}
  }{\simeq}
  \;\;
  \mathbb{G}
  \big(
    \mathrm{Conf}
    (
      \mathbb{R}^3
)
  \big)
  \scalebox{.7}{
    \color{darkblue}
    \bf
    \def\arraystretch{.9}
    \begin{tabular}{c}
      Configuration space of
      \\
      points and anti-points.
    \end{tabular}
  }
\end{equation}
Here the group completion
$$
  \mathbb{G}\big(
    \mathrm{Conf}(\mathbb{R}^3)
  \big)
  \;:=\;
  \Omega 
  \big(
    B_{\sqcup}
    \mathrm{Conf}(\mathbb{R}^3)
  \big)
$$

\vspace{1mm} 
\noindent is with respect to the topological monoid structure on $\mathrm{Conf}(\mathbb{R}^3)$ via (suitably adjusted) disjoint union of configurations, hence by adjoining ``anti-points'' (carrying negative unit charge) to the ordinary points (carrying positive unit charge) in the configurations.

\begin{remark}[\bf Torsion contribution in flux quantization]
\label{TorsionContribution}
While the classifying spaces 
$S^3$ (for co-Homotopy)
and $B^3 \mathbb{Z}$ (for ordinary cohomology)
have the same rational homotopy type
$$
  \mathfrak{l}
  S^3
  \;\simeq\;
  \mathfrak{l}
  B^3 \mathbb{Z}
  \,,
$$
so that both would qualify as valid charge quantization laws on the M5-brane if its coupling to the background C-field were ignored, they differ drastically in torsion contributions. It is the unbounded torsion in the homotopy groups of $S^3$, hence in the higher homotopy groups of spheres, which conspires to make 3-co-Homotopy reflect these configuration spaces 
\eqref{SegalTheorem1}
of solitonic branes. Cf. also Rem. \ref{NonAbelianEffects} below.
\end{remark}

\newpage 
\medskip

Of course, with the plain Minkowskian worldvolume topology used so far \eqref{MinkowskianWorldvolumeAsPointedSpace}, the worldvolume Skyrmions actually live in 1+5 dimensions and are themselves 2-branes instead of vortex-like: It remains to pass to a suitable double dimensional reduction in the remaining two spatial directions in order to model anyonic defects.

\medskip

\noindent
{\bf Topological quantum observables on Skyrmions.}
To that end, first consider KK-compactifying the background spacetime, and jointly the M5-brane worldvolume, on a circle $S^1_A$,
\vspace{1mm} 
\begin{equation}
  \label{IIACompactification}
  \Sigma
  \;\;\defneq\;\;
  \mathbb{R}^{1,1}_{\sqcup \{\infty\}}
    \,\wedge\,
  \mathbb{R}^3_{\cup \{\infty\}}
    \,\wedge\, 
  (S^1_{A})_{\sqcup \{\infty\}}
  \hspace{1cm}
\adjustbox{
  margin=-3pt,
  fbox,
  raise=-10pt
}{
\hspace{-8pt}
\begin{tikzpicture}
\node
 at (0,1) {
   \scalebox{1}{
     $\mathbb{R}^{1,1}_{\phantom{A}}$
   }
 };
\node
 at (.5,1) {
   \scalebox{1}{
     $\times$
   }
 };
\node
 at (1,1) {
   \scalebox{1}{
     $\mathbb{R}^3_{\phantom{A}}$
   }
 };
\node
 at (1.5,1) {
   \scalebox{1}{
     $\times$
   }
 };
\node
 at (2,1) {
   \scalebox{1}{
     $S^1_{A}$
   }
 };
\draw[
  line width=10pt,
  fill=olive,
  draw=olive
]
  (-.35,.5) -- (0+.3,.5);  
\draw[
  line width=10pt,
  fill=olive,
  draw=olive,
  draw opacity=.4
]
  (2-.25,.5) -- (2+.2,.5);  
\end{tikzpicture}
\hspace{-10pt}
}  
\end{equation}

\vspace{1mm} 
\noindent which we think of as the ``M-theory circle'' fiber in M/IIA duality.
Now $\widetilde \boldpi^3(\Sigma)$ being a {\it moduli space} (instead of just a set) of charges implies (by ``closed smash monoidal structure'', cf. \cite[\S A.2]{SS23-Obs}) that the moduli on this KK-compactified worldvolume are equivalently the free loop space $\mathcal{L}$ of the (group completed) configuration space:
$$
  \def\arraystretch{1.6}
  \begin{array}{ll}
  \widetilde{\boldpi}^3(\Sigma)
  &\;\simeq\;
  \mathrm{Maps}^{\ast/}\big(
    (S^1_A)_{
      \sqcup \{\infty\}
    }
    \wedge
    \mathbb{R}^3_{\cup \{\infty\}}
    ,\,
    S^3
  \big)
  \;
    \simeq
  \;
  \mathrm{Maps}
  \Big(
    S^1_A
    ,\,
    \mathrm{Maps}^{\ast/}
    \big(
      \mathbb{R}^3_{\cup \{\infty\}}
      ,\,
      S^3
    \big)
  \Big)
  \\
  &\;\defneq\;
  \mathcal{L}
  \Big(
  \widetilde{\boldpi}^3\big(\mathbb{R}^3_{\cup \{\infty\}}\big)
  \Big)
  \;\simeq\;
  \mathcal{L}
  \big(
  \mathbb{G}
  \mathrm{Conf}(\mathbb{R}^3)
  \big)
  \,.
  \end{array}
$$
However, in the strongly-coupled regime that we are after,
the circle $S^1_A \,\simeq\, (\mathbb{R}^1_{A})_{\cup \{\infty\}}$ is meant to ``decompactify'', which must mean that we are to fix the asymptotic charges $c_\infty$ at $\infty \in (\mathbb{R}^1_A)_{\cup \infty}$, hence to consider as the (topological) non-perturbative moduli space of Skyrmions on M5 
the space of (not free) loops of configurations (but) based at $c_\infty$:
$$
  \Omega_{c_\infty}
  \Big(
    \widetilde{\boldpi}^3\big(
      \mathbb{R}^3_{\cup \{\infty\}}
    \big)
  \Big)
  \;\simeq\;
  \Omega_{c_\infty}
  \mathbb{G}\big(
    \mathrm{Conf}(\mathbb{R})
  \big)
  \,.
$$
With this charge moduli space  understood as the topological sector of the integrated BRST-complex of the field theory, the topological quantum observables should correspond to compactly supported complex functions on this space, and hence {\it higher observables} should be given by the homology groups of this space \cite[\S 3]{SS23-Obs}:
\begin{equation}
  \label{QuantumObservables}
  \mathrm{Obs}_\bullet
  \;:=\;
  H_\bullet
  \Big(
    \Omega_{c_\infty}
    \mathbb{G}
    \mathrm{Conf}(\mathbb{R}^3)
    ;\,
    \mathbb{C}
  \Big)
  \;\;
  \in
  \;\;
  \mathrm{HopfAlg}^{\mathbb{Z}}_{\mathbb{C}}
  \,.
\end{equation}
In fact, under concatenation of loops, this graded group forms a non-commutative graded {\it Pontrjagin-Hopf algebra}, and by comparison with the case of Yang-Mills theory  \cite[Thm. 3.1]{SS23-Obs} we may regard this as the algebra of topological quantum observables on our higher gauge theory, regarded as the topological sector of discrete light cone quantization of the system \cite[\S 4]{SS23-Obs}.

\medskip

This turns out to be particularly interesting for ``open'' M5-branes:

\medskip

\noindent
{\bf Open M5-branes.}
To complete the desired dimensional reduction of the M5-worldvolume to 1+3 dimensions, consider now $\ZTwo$-orbifolding one of the dimensions transverse to the Skyrmions, via the reflection action (the sign representation of $\ZTwo$) as in heterotic M-theory (Ho{\v r}ava-Witten theory \cite{HoravaWitten96}), so that the M5 worldvolume \eqref{IIACompactification} now becomes the orbifold
\begin{equation}
  \label{OpenM5BraneWorldvolume}
  \Sigma
  \;\;
    \defneq
  \;\;
  \mathbb{R}^{1,1}_{\sqcup \{\infty\}}
    \,\wedge\,
  \big(
    \mathbb{R}^2
    \,\times\,
    \mathbb{R}^1_{H}
    \!\sslash\!
    \ZTwo
  \big)_{\cup \{\infty\}}
    \,\wedge\, 
  (S^1_{A})_{\sqcup \{\infty\}}
  \hspace{.6cm}
  \adjustbox{}{
\adjustbox{
  margin=-3pt,
  fbox,
  raise=-10pt
}{
\hspace{-8pt}
\begin{tikzpicture}
\node
 at (0,1) {
   \scalebox{1}{
     $\mathbb{R}^{1,1}_{\phantom{A}}$
   }
 };
\node
 at (.5,1) {
   \scalebox{1}{
     $\times$
   }
 };
\node
 at (1,1) {
   \scalebox{1}{
     $\mathbb{R}^2_{\phantom{A}}$
   }
 };
\node
 at (1.5,1) {
   \scalebox{1}{
     $\times$
   }
 };
\node
 at (2,1) {
   \scalebox{1}{
     $S^1_{H}$
   }
 };
\node
 at (2.5,1) {
   \scalebox{1}{
     $\times$
   }
 };
\node
 at (3,1) {
   \scalebox{1}{
     $S^1_{A}$
   }
 };
\draw[
  line width=10pt,
  fill=olive,
  draw=olive
]
  (-.35,.5) -- (0+.3,.5);  
\draw[
  line width=10pt,
  fill=olive,
  draw=olive,
  draw opacity=.4
]
  (3-.25,.5) -- (3+.2,.5);  
\end{tikzpicture}
\hspace{-10pt}
}  
  }
\end{equation}

\vspace{2mm}
\noindent 
where the orbifolded torus $S^1_A \times S^1_H \!\sslash\! \ZTwo$ is that from \cite[Fig. 2]{HoravaWitten96}. Essentially this compactification of the worldvolume of M5-branes  is also considered for discussion of Skyrmions in \cite{ILP18} (except that here we take $S^1_A$ to be tangential instead of transversal to the Skyrmions, in order to reduce them to anyons, below).

\smallskip 
Such M5-brane configurations wrapping $S^1_H \!\sslash\! \ZTwo$ are known as ``open M5-branes'' \cite[Fig. 3]{BGT06}\footnote{
  Strictly speaking, the discussion in \cite{BGT06} is for solitonic/singular M5-branes, while here we are concerned with the analogous situation for the {\it probe} incarnation of the open M5-brane (its ``non-linear sigma-model''), hence without backreaction. 
} (denoted ``M5 worldvolume'' in \cite[(52)]{FSS20Exc}, and by a gray bar in \cite[p. 2]{FSS21Emerge}), whose orbi-singularity may be identified with a {\it non-supersymmetric} 4-brane \cite{KOTY23}, hence here with a 3-brane as we think of the $S^1_A$-factor shrunk away -- cf. the figure below \eqref{GelfandRaikovTheorem}.  This process of obtaining a non-supersymmetric 3-brane by wrapping the M5 on a supersymmetry-breaking torus is similar to the construction in the WSS model for holographic QCD \cite[\S 4]{Witten98} which instead of the $\ZTwo$-orbifold of the torus $S^1_H \!\times\! S^1_A$
envisions compactification on a plain torus but equipped with supersymmetry-breaking spin structure.

\medskip

\noindent
{\bf Charge moduli on open M5-branes.}
The proper way to measure charges on such orbifolds is in proper orbifold cohomology \cite{SS20-Orb} (familiar for D-brane charges measured in orbifold/equivariant K-theory \cite{SzaboValentino10}\cite[Ex. 4.5.4]{SS21-EquBund}), which for the case of flux quantization in co-Homotopy means (\cite[\S 3]{SS20Tadpole}\cite[\S 5.2]{SS20-Orb}\cite{SS20EquChar}) to measure in orbifold/equivariant co-Homotopy, namely with the moduli space \eqref{CohomotopyModuliSpace} in the case \eqref{MinkowskianWorldvolumeAsPointedSpace} replaced by the pointed equivariant mapping space:
\vspace{1mm} 
$$
  \scalebox{.7}{
    \color{darkblue}
    \bf
    \def\arraystretch{.9}
    \begin{tabular}{c}
      Equivariant 
      \\
      co-Homotopy moduli
    \end{tabular}
  }
  \widetilde{\boldpi}^{2,1}_{\ZTwo}
  \big(
    (
    \mathbb{R}^2
    \times 
    \mathbb{R}^1_{H}
    )_{\cup \{\infty\}}
  \big)
  \;\;
  :=
  \;\;
  \mathrm{Maps}^{\ast/}_{\ZTwo}
  \Big(
    (
    \mathbb{R}^2
    \times 
    \mathbb{R}^1_{H}
    )_{\cup \{\infty\}}
    ,\,
    (
    \mathbb{R}^2
    \times 
    \mathbb{R}^1_{H}
    )_{\cup \{\infty\}}
  \Big)
  \scalebox{.7}{
    \color{darkblue}
    \bf
    \def\arraystretch{.9}
    \begin{tabular}{c}
      Equivariant pointed
      \\
      mapping space
    \end{tabular}
  }
$$

\smallskip 
\noindent Now the equivariant generalization of Segal's theorem \eqref{SegalTheorem1} 
says (\cite[Thm. 2]{RourkeSanderson00}) that these orbifold moduli are equivalent to (group completed) equivariant configurations of points, namely to $\ZTwo$-invariant finite subsets of $\mathbb{R}^2 \times \mathbb{R}^1_{H}$:
$$
  \widetilde
  \boldpi^{2,1}\Big(
    \big(
    \mathbb{R}^2
    \times
    \mathbb{R}^1_{H}
    \big)_{\cup \{\infty\}}
  \Big)
  \;\;
  \underset{\color{orangeii}
    \mathrm{hmtp}
  }{
    \simeq
  }
  \;\;
  \mathbb{G}
  \Big(
  \mathrm{Conf}\big(
    \mathbb{R}^2
    \times
    \mathbb{R}^1_{H}
  \big)^{\ZTwo}
  \Big)
  \,.
$$
But this means that any one of the  points (branes) in a configuration
\begin{itemize}[
  leftmargin=.6cm,
  topsep=1pt,
  itemsep=2pt
]
\item either moves freely in the Ho{\v r}ava-Witten bulk $\mathbb{R}^2 \times \mathbb{R}_{> 0} \,\simeq\, \mathbb{R}^3$ 

(exactly mirrored by a point in $\mathbb{R}^2 \times \mathbb{R}_{< 0}$)
\item or is stuck on the  heterotic plane $\mathbb{R}^2 \times \{0\} \,\simeq\, \mathbb{R}^2$,
\end{itemize}
hence it means that the moduli space is now the product of spaces of configurations in the Ho{\v v}ava-Witten bulk and on the heterotic plane (cf. also \cite[(1.2), Thm. 4.1]{Xico06}):
\vspace{-.2cm}
\begin{equation}
  \label{HeteroticConfiurationSpace}
  \scalebox{.7}{
    \color{darkblue}
    \bf
    \begin{tabular}{c}
      Heterotic co-Homotopy
      \\
      charge moduli space
    \end{tabular}
  }
  \widetilde
  \boldpi^{2,1}\Big(
    \big(
    \mathbb{R}^2
    \times
    \mathbb{R}^1_{H}
    \big)_{\cup \{\infty\}}
  \Big)
  \;\;
  \underset{\color{orangeii}
   \mathrm{hmtp} 
  }{\simeq}
  \quad \;\; 
  \underbrace{
  \mathbb{G}
  \mathrm{Conf}(\mathbb{R}^3)
  }_{
    \mathclap{
    \scalebox{.7}{
      \color{darkblue}
      \bf
      \def\arraystretch{.9}
      \begin{tabular}{c}
        Brane configurations
        \\
        in 
        HW-bulk
        $\simeq$
        \eqref{SegalTheorem1}
      \end{tabular}
    }
    }
  }
  \;\;\times\;\;
  \overbrace{
    \mathbb{G}
    \mathrm{Conf}(\mathbb{R}^2)
  }^{
    \mathclap{
    \scalebox{.7}{
      \color{darkblue}
      \bf
      \def\arraystretch{.9}
      \begin{tabular}{c}
        Brane configurations
        \\
        on heterotic plane
      \end{tabular}
    }  
    }
  }
  \,.
\end{equation}

This is most curious, because:

\medskip

\noindent
{\bf Vortex braiding on open M5 branes.}
The configuration space of points in $\mathbb{R}^2$ that has appeared on the right in \eqref{SegalTheorem1} is the classifying space of the Artin braid groups $\mathrm{Br}(n)$ (\cite[\S 7]{FoxNeuwirth62}, cf. \cite[pp. 9]{Williams20}), whose loop space at given $c_\infty = N$ is the braid group itself:
$$
  \mathrm{Conf}(\mathbb{R}^2)
  \;\;
  \simeq
  \;\;
  \textstyle{
    \underset{n \in \mathbb{N}}{\coprod}
  }
  B \mathrm{Br}(N)
  \,,
  \hspace{.5cm}
  \Omega_N 
  \mathrm{Conf}(\mathbb{R}^2)
  \;\simeq\;
  \mathrm{Br}(N)
  \,.
$$
Therefore the algebra of topological light cone quantum observables \eqref{QuantumObservables}
on our open M5-brane fields
now contains the group algebra of the braid group:
\begin{equation}
  \label{BraidGroupAlgebra}
  \begin{tikzcd}
    \mathbb{C}\big[
     \mathrm{Br}(N)
    \big]
    \;\;
    =
    \;\;
    H_0\big(
      \Omega_N
      \mathrm{Conf}(\mathbb{R}^2)
    \big)
    \ar[
      r,
      hook
    ]
    &
    H_\bullet\Big(
      \Omega_{c_\infty}
      \widetilde{\boldpi}^{2,1}
      \big(
        (\mathbb{R}^2
        \times
        \mathbb{R}^1_{H})_{\cup \{\infty\}}
      \big)
    \Big)
    \,.
  \end{tikzcd}
\end{equation}

Looking back through the construction, we see that the braiding happening here is that of the skyrmions in 1+5 dimensions (solitonic 2-branes) which have been dimensionally reduced to solitonic 1-branes in 1+3 dimensions (akin to Abrikosov vortices).

\smallskip

At this point one expects these objects to be anyonic. Indeed, we may now derive that this is the case:

\medskip

\noindent
{\bf Anyonic quantum states on open M5.}
Given a quantum (star-)algebra of observables as in \eqref{BraidGroupAlgebra} the corresponding {\it quantum states} are identified with the positive linear functions $\rho$ on the vector space of observables (these assigning the expectation values $\rho(A)$ of any observable $A$ in the given quantum state, cf. \cite[\S 2.2]{CSS23} in our context):
$$
  \mathrm{QStates}_{\mathrm{vrtx}}
  \;:=\;
  \left\{\!\!
  \begin{tikzcd}
    \mathrm{Obs}
    \ar[
      r,
      "{ \rho }",
      "{\color{darkgreen} 
        \mathrm{linear}
      }"{swap}
    ]
    &
    \mathbb{C}
  \end{tikzcd}
  \;\Big\vert\;
  \underset{A \in \mathrm{Obs}}{\forall} \,
  \rho\big(
    A^\ast A
  \big)
  \;\geq\;
  0
  \in \mathbb{R}
  \subset 
  \mathbb{C}
  \right\}
  \,.
$$
Hence, the quantum states of those vortex charges on the compactified M5-brane are the positive linear functionals on the braid group algebra $\mathbb{C}\big[\mathrm{Br}(N)\big]$ \eqref{BraidGroupAlgebra}.

\smallskip

Now the {\it Gelfand-Raikov theorem} \cite[(2)]{GelfandRaikov43}\cite[Thm. 13.4.5(ii)]{Dixmier77}
says that the positive linear functionals on a group algebra are exactly the expectation values with respect to a cyclic state $\vert \psi \rangle$ of unitary representations $U$ of that group on some Hilbert space:
\vspace{-2mm} 
\begin{equation}
  \label{GelfandRaikovTheorem}
  \hspace{-1cm}
  \begin{tikzcd}[row sep=-10pt, 
    column sep=0pt
  ]
    \scalebox{.7}{
      \color{darkblue}
      \bf
      Braid group representations
    }
    \ar[
      rr,
      phantom,
      "{
        \scalebox{.7}{
          \color{darkgreen}
          \bf
          are identified with
        }
      }"{pos=.4}
    ]
    &&
   \hspace{-1.2cm} 
   \scalebox{.7}{
      \color{darkblue}
      \bf
      \def\arraystretch{.9}
      \begin{tabular}{c}
        topological quantum states
        \\
        of vortex solitons 
        \\
        on open M5-branes
      \end{tabular}
    }
    \\
    \left\{\!\!
      \def\arraystretch{1.3}
      \begin{array}{r}
        \HilbertSpace{H}
        \\
        \mathrm{Br}(N)
        \xrightarrow{U}
        \mathrm{U}(\HilbertSpace{H})
        \\
        \vert\psi\rangle
        \in
        \HilbertSpace{H}
      \end{array}
      \def\arraystretch{1.3}
      \begin{array}{l}
        \mbox{Hilbert space}
        \\
        \mbox{unitary rep}
        \\
        \mbox{cyclic vector}
      \end{array}    
   \! \right\}_{
      \mathrlap{
        \!\!\big/\sim
      }
    }
    \qquad 
    \ar[
      rr,
      "{ \sim }"
    ]
    &&
    \quad 
      \bigg\{\!\!
    \def\arraystretch{1.5}
    \begin{array}{r}
      \mathbb{C}\big[
        \mathrm{Br}(N)
      \big]
      \xrightarrow{ \, \rho \, }
      \mathbb{C}
    \end{array}
    \mbox{
      pos. lin. func.
    }
\!\!   \bigg\}
    \;
    =
    \;
    \mathrm{States}_{\mathrm{vrtx}}
    \\
    \big(
      \HilbertSpace{H}
      ,\,
      U
      ,\,
      \vert \psi \rangle
    \big)
    &\longmapsto&
    \bigg(
    \rho 
      \;: 
    \underset{
      g \in \mathrm{Br}(N)
    }{\sum}
    c_g \cdot g
    \;\;\;\mapsto\;
    \underset{
      g \in \mathrm{Br}(N)
    }{\sum}
    c_g
    \,
    \langle 
      \psi \vert 
      \,
      U(g)
      \,
      \vert \psi
    \rangle
    \bigg).
  \end{tikzcd}
\end{equation}
But here $U$ is a braid group representation which hence exhibits $\HilbertSpace{H}$ as a space of anyon quantum states $\vert \psi \rangle$ (cf. e.g. \cite{NSSFDS08}\cite{Rowell22} and in our context \cite{SS23-DefectBranes}\cite{SS23-ToplOrder}\cite[\S 3]{MySS24}).

\medskip

\hspace{-.9cm}
\begin{tabular}{p{10cm}}
In summary, we have found that flux quantization on M5-branes (under Hypothesis H) makes the topological quantum states of the vortex solitons \eqref{HeteroticConfiurationSpace} on the 
boundary of the
open M5-brane 
\eqref{OpenM5BraneWorldvolume}
(wrapped over $S^1_A$) be anyonic, as expected in strongly-correlated (topologically ordered) quantum systems. We discuss this effect in more detail in the followup \cite{SS24-AbAnyons}.
\end{tabular}
\hspace{25pt}
\adjustbox{
  raise=-2cm
}
{
\begin{tikzpicture}

\draw
  (.8,1.3) -- (3.8,1.3);
\draw
  (.8,-.7) -- (2.9,-.7);

\draw[
  fill=lightgray
]
  (0,0) -- (0,-2) -- (.8,-.7) -- (.8,1.3) -- cycle; 

\draw[
  fill=lightgray
]
  (3,0) -- (3,-2) -- (3+.8,-.7) -- (3+.8,1.3) -- cycle; 

\draw[
  draw=white,
  line width=3
]
  (.05,0) -- (2.95,0);
\draw[
  line width=.8
]
  (0,0) -- (3,0);
\draw[
  line width=.8
]
  (0,-2) -- (3,-2);
\draw[
  line width=.8
]
  (0,0) -- (0,-2);
\draw[
  line width=.8
]
  (3,0) -- (3,-2);

\draw[
  fill=olive,
  draw=olive
]
  (3.3,-.4)
  ellipse (.03 and .05);
\draw[
  fill=olive,
  draw=olive
]
  (3.5,+.1)
  ellipse (.03 and .05);
\draw[
  fill=olive,
  draw=olive
]
  (3.6,-.5)
  ellipse (.03 and .05);

\node
  at (1.4,-1.2)
  {
    \scalebox{.7}{
      \colorbox{white}{
      \hspace{-15pt}
      \color{darkblue}
      \bf
      \def\arraystretch{.7}
      \begin{tabular}{c}
        open M5-brane
        \\
        wrapped on 
        $S^1_A$
      \end{tabular}
      \hspace{-14pt}
      }
    }
  };

\node[
  rotate=57
]
  at (3.6,-1.3)
  {
    \scalebox{.65}{
      \colorbox{white}{
      \hspace{-17.3pt}
      \color{gray}
      \bf
      \def\arraystretch{.7}
      \begin{tabular}{c}
        non-susy
        \\
        3-brane
      \end{tabular}
      \hspace{-14.5pt}
    }
    }
  };

\node
  at
  (4.3, -.2)
  {
    \scalebox{.65}{
      \color{olive}
      \bf
      \def\arraystretch{.8}
      \begin{tabular}{c}
        anyonic
        \\
        vortices
      \end{tabular}
    }
  };

\node
  at (1.4,-2.2)
  {
    \scalebox{.8}{
        \hspace{-7pt}
        $S^1_H
        \!\sslash\!\ZTwo$
        \hspace{-7pt}
    }
  };

\end{tikzpicture}
}

\begin{remark}[\bf Non-abelian solitonic effects through flux quantization in non-abelian cohomology]
  \label{NonAbelianEffects}
$\,$
  
\noindent {\bf (i)}  For appreciating these results, notice that even though we do not consider ``coincident'' brane worldvolumes here --- on which only common lore \cite[p. 8]{Witten96} expects non-abelian gauge groups, such as $\mathrm{SU}(2)$ --- a degree of non-abelianness is introduced even on the single M5-brane by quantizing its flux in a {\it non-abelian cohomology theory} \cite[\S II]{FSS23Char} like 3-Cohomotopy. 

 \vspace{1mm}  
\noindent {\bf (ii)}  Concretely (cf. \cite[Rem. 2.1]{FSS23Char}), where ordinary 3-cohomology $H^3(-;\mathbb{Z}) \simeq H^1\big(-; B^2 \mathbb{Z}\big) = \pi_0 \mathrm{Maps}\big(-; B B^2 \mathbb{Z}\big)$ 
  is an abelian cohomology theory in that the Eilenberg-MacLane space $B^2 \mathbb{Z} = K(\mathbb{Z},2)$ is an abelian $\infty$-group (an $\infty$-loop space), this is not the case for 3-Cohomotopy $\pi^3(-) = H^1(-;\Omega S^3) = \pi_0 \mathrm{Map}\big(-;B \Omega S^3\big)$, since the $\infty$-group $\Omega S^3$ is just a braided $\infty$-group but not abelian (it admits only two deloopings), cf. Rem. \ref{TorsionContribution} above.

\vspace{1mm} 
\noindent {\bf (iii)}  Specifically, it is through the equivalence of underlying homotopy types $\mathrm{SU}(2) \,\simeq\, S^3$ that the 3-cohomotopically flux-quantized B-field on the single M5-brane has Skyrmion-like field configurations even in the absence of an ordinary $\mathrm{SU}(2)$-gauge field. More discussion of this phenomenon on M5-branes is in \cite{FSS20TwistedString}.
  
\end{remark}

\medskip

\section{Conclusion}

\noindent
We have established that and how the on-shell field content on M5-branes is to be completed by a choice of flux quantization law (p. \pageref{GeneralRuleOfFluxQuantization}) for the higher gauge field on super-space (where the flux Bianchi identity already implies the duality equation of motion, Prop. \ref{BianchiIdentityOnM5BraneInComponents}).

\smallskip 
In doing so, we relied on a rigorous and streamlined re-derivation (in \S\ref{SuperEmbeddingConstruction}) of the flux sector of the ``super-embedding''-construction of the M5-brane sigma-model. We suggest that our concise natural geometric re-formulation of the ``super-embedding''-conditions (Def. \ref{BPSImmersion} in \S\ref{SuperEmbeddings}, cf. Rem. \ref{SuperEmbeddingConditionInTheLiterature}) helps with understanding the phase spaces of super $p$-brane sigma-models --- in the present case but also in its generalizations such as to ``super-exceptional'' geometry (to which we turn in \cite{GSS-Exceptional}).

\smallskip 
With this in hand, we used previous results to find that among the admissible flux quantization laws on the M5-brane is a form of twisted co-Homotopy theory \eqref{M5BraneFluxQuantizationInIntroduction}, thereby showing that this yields exact global field content on the M5 without the need of further constraints (which had previously remained open).

\smallskip 
Assuming this choice (``Hypothesis H''), we have discussed, in \S\ref{FluxQuantizationOnM5Branes}, some key examples of the resulting moduli spaces 
and quantum states
of topological charges, highlighting that Skyrmion-like solitonic charges appear quite generically on M5-branes, and that on ``open'' M5-branes the quantum states of the resulting vortex-like solitons in 1+3 dimensions are anyonic \eqref{GelfandRaikovTheorem}. (This complements previous results \cite{SS23-DefectBranes}\cite{SS23-ToplOrder} where anyonic quantum statistics was argued for {\it intersecting} M5-brane configurations, which however are not as readily connected to a full on-shell superspace model as considered here. Notice that previously the realization of anyonic brane states had remained conjectural, cf. \cite[p. 65]{deBoerShigemori13}.)  

\smallskip

This result supports the general idea that the M5-brane may serve as a
much needed general model for otherwise elusive quantum phenomena in the strongly coupled/correlated non-perturbative regime and may point the way to a more microscopically detailed form of holography in high-energy and solid-state physics.

\smallskip 
We discuss 
applications of the present result to holography in \cite{GSS24-AdS7}, to quantum materials in \cite{SS24-AbAnyons}  
and plan to discuss an application
to  exceptional supergravity in
\cite{GSS-Exceptional}.


\newpage

\begin{thebibliography}{10}

\bibitem[ANW93]{ANW83}
G. Adkins, C. Nappi, and E. Witten, 
{\it \color{darkblue} Static Properties of Nucleons in the Skyrme Model}, Nucl. Phys. B {\bf 228} (1983), 
552-566,
[\href{https://doi.org/10.1016/0550-3213(83)90559-X}{\tt doi:10.1016/0550-3213(83)90559-X}],
[\href{http://inspirehep.net/record/190174}{\tt inspire:190174}].

\bibitem[APPS97]{APPS97}
M. Aganagic, J. Park, C. Popescu, and J. Schwarz, 
{\it \color{darkblue} World-Volume Action of the M Theory Five-Brane}, Nucl. Phys. B {\bf 496} (1997),
191-214, [\href{https://doi.org/10.1016/S0550-3213(97)00227-7}{\tt doi:10.1016/S0550-3213(97)00227-7}],
[\href{https://arxiv.org/abs/hep-th/9701166}{\tt arXiv:hep-th/9701166}].



\bibitem[ABJ08]{ABJ08}
O. Aharony, O. Bergman, and D. Jafferis, {\it \color{darkblue} Fractional M2-branes}, J. High Energy Phys. {\bf 0811} (2008) 043,
[\href{https://doi.org/10.1088/1126-6708/2008/11/043}{\tt doi:10.1088/1126-6708/2008/11/043}],
[\href{https://arxiv.org/abs/0807.4924}{\tt arXiv:0807.4924}].

\bibitem[AGMOY00]{AGMOY00}
O. Aharony, S. Gubser, J. Maldacena, H. Ooguri and Y. Oz, 
{\it \color{darkblue} Large $N$ Field Theories, String Theory and Gravity}, Phys. Rept. {\bf 323} 
(2000), 183-386,
[\href{https://doi.org/10.1016/S0370-1573(99)00083-6}{\tt doi:10.1016/S0370-1573(99)00083-6}],
\newline 
[\href{https://arxiv.org/abs/hep-th/9905111}{\tt arXiv:hep-th/9905111}].


\bibitem[AFFK15]{AFFK15}
J. Alicea, M. Fisher, M. Franz, and Y.-B. Kim, 
{\it \color{darkblue} Strongly Interacting Topological Phases}, report on Banff workshop \underline{\href{http://www.birs.ca/events/2015/5-day-workshops/15w5051}{15w5051}} (2015),
[\href{https://ncatlab.org/nlab/files/AliceaEtAl-InteractingTopPhases.pdf}{\tt ncatlab.org/nlab/files/AliceaEtAl-InteractingTopPhases.pdf}]

\bibitem[Al85a]{Alvarez85a}
O. Alvarez, 
{\it \color{darkblue} Cohomology and Field Theory}, talk at: {\it Symposium on Anomalies, Geometry, Topology}, Argonne IL (28-30 March 1985), 
[\href{https://inspirehep.net/conferences/965785}{\tt spire:965785}].

\bibitem[Al85b]{Alvarez85b}
O. Alvarez, 
{\it \color{darkblue} Topological quantization and cohomology}, Commun. Math. Phys. {\bf 100} 2 (1985), 279-309, \newline 
[\href{https://projecteuclid.org/euclid.cmp/1103943448}{\tt euclid:cmp/1103943448}].


\bibitem[Anc]{AncillaryFile}
{\it \color{darkblue} Ancillary {\tt Mathematica} notebook} with mechanized computations referred to in the text:
\newline
[\href{https://ncatlab.org/schreiber/show/Flux+Quantization+on+M5-Branes#Anc}{\tt ncatlab.org/schreiber/show/Flux+Quantization+on+M5-Branes\#Anc}]

\bibitem[An${}^+$22]{AndrianopoliEtAl22}
L. Andrianopoli, C. A. Cremonini, R. D’Auria, P. A. Grassi, R. Matrecano, R. Noris, L. Ravera, and M. Trigiante,
{\it \color{darkblue}M5-brane in the superspace approach}, 
Phys. Rev. D {\bf 106} 2 (2022) 026010, 
[\href{https://arxiv.org/abs/2206.06388}{\tt arXiv:2206.06388}],
[\href{https://doi.org/10.1103/PhysRevD.106.026010}{\tt doi:10.1103/PhysRevD.106.026010}].

\bibitem[AL12]{AL12}
R. Aurich and S. Lustig, 
{\it \color{darkblue} A survey of lens spaces and large-scale CMB anisotropy}, Mon. Not. Roy. Astron. Soc. {\bf 424} (2012), 1556-1562, 
[\href{https://doi.org/10.1111/j.1365-2966.2012.21363.x}{\tt doi:10.1111/j.1365-2966.2012.21363.x}],
[\href{https://arxiv.org/abs/1203.4086}{\tt arXiv:1203.4086}].

\bibitem[BaSh10]{BakulevShirkov10}
A. P. Bakulev and D. Shirkov, 
{\it \color{darkblue} Inevitability and Importance of Non-Perturbative Elements in Quantum Field Theory}, Proceedings of the 6th Mathematical Physics Meeting, Sep. 14–23 (2010) Belgrade, Serbia, 27–54, 
[\href{https://arxiv.org/abs/1102.2380}{\tt arXiv:1102.2380}].

\bibitem[Ba11]{Bandos11}
I. A. Bandos, 
{\it \color{darkblue} Superembedding approach to Dp-branes, M-branes and multiple D(0)-brane systems}, Phys. Part. Nucl. Lett. {\bf 8} (2011), 149-172, 
[\href{https://doi.org/10.1134/S1547477111030046}{\tt doi:10.1134/S1547477111030046}],
[\href{https://arxiv.org/abs/0912.2530}{\tt arXiv:0912.2530}].


\bibitem[BLNPST97]{BLNPST97}
I. Bandos, K. Lechner, A.  Nurmagambetov, P. Pasti, D. Sorokin, and M. Tonin, 
{\it \color{darkblue} Covariant Action for the Super-Five-Brane of M-Theory}, Phys. Rev. Lett. {\bf 78} (1997), 4332-4334, 
[\href{https://arxiv.org/abs/hep-th/9701149}{\tt arXiv:hep-th/9701149}],
[\href{https://doi.org/10.1103/PhysRevLett.78.4332}{\tt doi:10.1103/PhysRevLett.78.4332}].


\bibitem[BPSTV95]{BPSTV95}
I. Bandos, P. Pasti, D. Sorokin, M. Tonin, and D. Volkov, 
{\it \color{darkblue} Superstrings and supermembranes in the doubly supersymmetric geometrical approach}, Nucl. Phys. B {\bf 446} (1995), 79-118, 
[\href{https://arxiv.org/abs/hep-th/9501113}{\tt arXiv:hep-th/9501113}],
[\href{https://doi.org/10.1016/0550-3213(95)00267-V}{\tt doi:10.1016/0550-3213(95)00267-V}].

\bibitem[BaSo23]{BandosSorokin23}
I. A. Bandos and D. P. Sorokin, 
{\it \color{darkblue} Superembedding approach to superstrings and super-$p$-branes}, in: {\it Handbook of Quantum Gravity}, Springer (2023), 
[\href{https://doi.org/10.1007/978-981-19-3079-9_111-1}{\tt doi:10.1007/978-981-19-3079-9\_111-1}],
[\href{https://arxiv.org/abs/2301.10668}{\tt arXiv:2301.10668}].

\bibitem[BMS10]{BMS10}
R. A. Battye, N. Manton and P. Sutcliffe, {\it \color{darkblue} Skyrmions and Nuclei} (2010) in: {\it The Multifaceted Skyrmion}, World Scientific (2016), 3-39,
[\href{https://doi.org/10.1142/9789814280709_0001}{\tt doi:10.1142/9789814280709\_0001}].

\bibitem[BM06]{BelovMoore06}
D. Belov and G. Moore, 
{\it \color{darkblue}Holographic Action for the Self-Dual Field} 
[\href{https://arxiv.org/abs/hep-th/0605038}{\tt arXiv:hep-th/0605038}]

\bibitem[BBdSS00]{BBdSS00}
E. Bergshoeff, D. S. Berman, J. P. van der Schaar and P. Sundell,
{\it \color{darkblue} A Noncommutative M-Theory Five-brane}, Nucl. Phys. B {\bf 590} (2000), 173-197,
[\href{https://doi.org/10.1016/S0550-3213(00)00476-4}{\tt doi:10.1016/S0550-3213(00)00476-4}],
[\href{https://arxiv.org/abs/hep-th/0005026}{\tt arXiv:hep-th/0005026}].


\bibitem[BGT06]{BGT06}
E. Bergshoeff, G. Gibbons, and P. Townsend, 
{\it \color{darkblue} Open M5-branes}, 
Phys. Rev. Lett. {\bf 97} (2006) 231601,
[\href{https://doi.org/10.1103/PhysRevLett.97.231601}{\tt doi:10.1103/PhysRevLett.97.231601}],
[\href{https://arxiv.org/abs/hep-th/0607193}{\tt arXiv:hep-th/0607193}].


\bibitem[Be08]{Berman08}
D. Berman,
{\it \color{darkblue} M-theory branes and their interactions}, Phys. Rept. {\bf 456} (2008), 89-126, 
[\href{https://arxiv.org/abs/0710.1707}{\tt arXiv:0710.1707}],
[\href{https://doi.org/10.1016/j.physrep.2007.10.002}{\tt doi:10.1016/j.physrep.2007.10.002}].


\bibitem[BLT13]{BLT13}
R. Blumenhagen, D. L{\"u}st, and S. Theisen,
{\it \color{darkblue} Basic Concepts of String Theory}
Springer, New York (2013), 
[\href{https://doi.org/10.1007/978-3-642-29497-6}{\tt doi:10.1007/978-3-642-29497-6}].


\bibitem[BBG83]{BergerBryantGriffiths83}
E. Berger, R. Bryant, and P. Griffiths, 
{\it \color{darkblue} The Gauss equations and rigidity of isometric embeddings}, Duke Math. J. {\bf 50} 3 (1983), 803-892,
[\href{http://dx.doi.org/10.1215/S0012-7094-83-05039-1}{\tt doi:10.1215/S0012-7094-83-05039-1}].


\bibitem[Bo75]{Boothby75}
W. M. Boothby, 
{\it \color{darkblue} An introduction to differentiable manifolds and Riemannian geometry}, Academic Press (1975, 1986), Elsevier (2002),
[\href{https://shop.elsevier.com/books/an-introduction-to-differentiable-manifolds-and-riemannian-geometry-revised/boothby/978-0-08-057475-2}{\tt ISBN:9780121160517}].


\bibitem[Br14]{BrambillaEtAl14}
N. Brambilla et al., 
{\it \color{darkblue} QCD and strongly coupled gauge theories -- challenges and perspectives}, Eur.
Phys. J. C Part. Fields {\bf 74} 10 (2014) 2981, 
[\href{https://doi.org/10.1140/epjc/s10052-014-2981-5}{\tt doi:10.1140/epjc/s10052-014-2981-5}],
[\href{https://arxiv.org/abs/1404.3723}{\tt arXiv:1404.3723}].


\bibitem[BH80]{BrinkHowe80}
L. Brink and P. Howe, 
{\it \color{darkblue} Eleven-Dimensional Supergravity on the Mass-Shell in Superspace}, Phys. Lett. B {\bf 91} (1980), 384-386, [\href{https://doi.org/10.1016/0370-2693(80)91002-3}{\tt doi:10.1016/0370-2693(80)91002-3}].

\bibitem[BMG86]{BMG86}
R. Brooks, F. Muhammad, and S. J. Gates Jr., {\it \color{darkblue}  Matter Coupled to $D=2$ Simple Unidexterous Supergravity, Local $(\mbox{supersymmetry})^2$ and Strings}, Class. Quant. Grav. {\bf 3} 5 (1986), 745-751, 
[\href{https://inspirehep.net/literature/237390}{\tt spire:237390}],
[\href{https://doi.org/10.1088/0264-9381/3/5/005}{\tt doi:10.1088/0264-9381/3/5/005}].

\bibitem[BuSS21]{BurtonSS21}
S. Burton, H. Sati and U. Schreiber
{\it \color{darkblue}Lift of fractional D-brane charge to equivariant Cohomotopy theory},
J. Geometry and Physics,
{\bf 161} (2021) 104034
[\href{https://arxiv.org/abs/1812.09679}{\tt arXiv:1812.09679}],
[\href{https://doi.org/10.1016/j.geomphys.2020.104034}{\tt doi:10.1016/j.geomphys.2020.104034}].


doi:10.1016/j.geomphys.2020.104034

\bibitem[Cartan26]{Cartan26}
{\'E}. Cartan (translated by Vladislav Goldberg from Cartan’s lectures at the Sorbonne in 1926–27), 
{\it \color{darkblue} Riemannian Geometry in an Orthogonal Frame}, World Scientific, Singapore  (2001),
[\href{https://doi.org/10.1142/4808}{\tt doi:10.1142/4808}].

\bibitem[CDF91]{CDF91}
L. Castellani, R. D'Auria, and P. Fr{\'e}, 
{\it \color{darkblue} Supergravity and Superstrings -- A Geometric Perspective}, 
World Scientific, Singapore (1991), 
[\href{https://doi.org/10.1142/0224}{\tt doi:10.1142/0224}].


\bibitem[Ch93]{Chavel93}
I. Chavel, 
{\it \color{darkblue} Riemannian geometry -- A modern introduction}, Cambridge University Press (1993),
\newline 
[\href{https://doi.org/10.1017/CBO9780511616822}{\tt doi:10.1017/CBO9780511616822}].


\bibitem[CG21]{ChenGiron21}
G.-Q. Chen and T. P. Giron, 
{\it \color{darkblue} Weak continuity of curvature for connections in $L^p$}, 
[\href{https://arxiv.org/abs/2108.13529}{\tt arXiv:2108.13529}].

\bibitem[CGK20]{CGK20}
G. Y. Cho, D. Gang, and H.-C. Kim, 
{\it \color{darkblue} M-theoretic Genesis of Topological Phases}, J. High Energy Phys. {\bf 2020}  
(2020) 115, 
[\href{https://doi.org/10.1007/JHEP11(2020)115}{\tt doi:10.1007/JHEP11(2020)115}],
[\href{https://arxiv.org/abs/2007.01532}{\tt arXiv:2007.01532}].

\bibitem[CSS23]{CSS23}
D. Corfield, H. Sati, and Urs Schreiber,
{\it \color{darkblue} Fundamental weight systems are quantum states}, Lett. Math. Phys.
{\bf 113}  (2023) 112,
[\href{https://doi.org/10.1007/s11005-023-01725-4}{\tt doi:10.1007/s11005-023-01725-4}],
[\href{https://arxiv.org/abs/2105.02871}{\tt arXiv:2105.02871}].

\bibitem[CF80]{CF80}
E. Cremmer and S. Ferrara, 
{\it \color{darkblue} Formulation of Eleven-Dimensional Supergravity in Superspace}, Phys. Lett. B {\bf 91} (1980), 61-66, 
[\href{https://doi.org/10.1016/0370-2693(80)90662-0}{\tt doi:10.1016/0370-2693(80)90662-0}].

\bibitem[CJLP98]{CJLP98}
E. Cremmer, B. Julia, H. Lu, and C. Pope, {\it \color{darkblue} Dualisation of Dualities, II: Twisted self-duality of doubled fields and superdualities}, Nucl. Phys. B {\bf 535} (1998), 242-292, 
[\href{https://doi.org/10.1016/S0550-3213(98)00552-5}{\tt doi:10.1016/S0550-3213(98)00552-5}], \newline 
[\href{https://arxiv.org/abs/hep-th/9806106}{\tt arXiv:hep-th/9806106}].

\bibitem[CKKTvP98]{CKKTvP98}
P. Claus, R. Kallosh, J. Kumar, P. K. Townsend, and A. Van Proeyen, {\it \color{darkblue}
Conformal Theory of M2, D3, M5 and `D1+D5' Branes}, J. High Energy Phys. {\bf  9806} (1998) 004,
[\href{https://doi.org/10.1088/1126-6708/1998/06/004}{\tt doi:10.1088/1126-6708/1998/06/004}],
[\href{https://arxiv.org/abs/hep-th/9801206}{\tt arXiv:hep-th/9801206}].


\bibitem[CKvP98]{CKvP98}
P. Claus, R. Kallosh, and A. Van Proeyen, {\it \color{darkblue} 5-brane and superconformal $(0,2)$ tensor multiplet in 6 dimensions}, Nucl. Phys. B {\bf 518} (1998), 117-150,
[\href{https://doi.org/10.1016/S0550-3213(98)00137-0}{\tt doi:10.1016/S0550-3213(98)00137-0}],
[\href{https://arxiv.org/abs/hep-th/9711161}{\tt arXiv:hep-th/9711161}].


\bibitem[dBS13]{deBoerShigemori13}
J. de Boer and M. Shigemori, 
{\it \color{darkblue} Exotic Branes in String Theory}, Phys. Rept. {\bf 532} (2013), 65-118, \newline
[\href{https://doi.org/10.1016/j.physrep.2013.07.003}{\tt doi:10.1016/j.physrep.2013.07.003}],
[\href{https://arxiv.org/abs/1209.6056}{\tt arXiv:1209.6056}].

\bibitem[DEGKS08]{DEGKS08}
E. D'Hoker, J. Estes, M. Gutperle, D. Krym, and P. Sorba, 
{\it \color{darkblue} Half-BPS supergravity solutions and superalgebras}, J. High Energy Phys.
{\bf 0812}  (2008) 047, 
[\href{https://iopscience.iop.org/article/10.1088/1126-6708/2008/12/047}{\tt doi:10.1088/1126-6708/2008/12/047}], \newline 
[\href{https://arxiv.org/abs/0810.1484}{\tt arXiv:0810.1484}].

\bibitem[DFM07]{DFM07}
D.-E. Diaconescu, D. Freed and G. Moore, 
{\it \color{darkblue}The M-theory 3-form and $E_8$-gauge theory}, 
in: {\it Elliptic Cohomology Geometry, Applications, and Higher Chromatic Analogues}, Cambridge University Press (2007),
[\href{https://doi.org/10.1017/CBO9780511721489}{\tt doi:10.1017/CBO9780511721489}],
[\href{https://arxiv.org/abs/hep-th/0312069}{\tt arXiv:hep-th/0312069}].


\bibitem[Dirac1931]{Dirac31}
P.A.M. Dirac, 
{\it \color{darkblue} Quantized Singularities in the Electromagnetic Field}, Proc. Royal Soc. A {\bf 133} (1931), 60-72,
[\href{https://doi.org/10.1098/rspa.1931.0130}{\tt doi:10.1098/rspa.1931.0130}].

\bibitem[Dix77]{Dixmier77}
J. Dixmier, {\it \color{darkblue} Unitary representations of locally compact groups}, Ch. 13 in:
{\it $C^\ast$-algebras}, North Holland (1977),
[\href{https://ncatlab.org/nlab/files/Dixmier-CStarAlgebras-UnitaryReps.pdf}{\tt ncatlab.org/nlab/files/Dixmier-CStarAlgebras-UnitaryReps.pdf}]

\bibitem[Du96]{Duff96}
M. Duff, 
{\it \color{darkblue} M-Theory (the Theory Formerly Known as Strings)}, Int. J. Mod. Phys. A {\bf 11} (1996), 5623-5642, 
[\href{https://doi.org/10.1142/S0217751X96002583}{\tt doi:10.1142/S0217751X96002583}],
[\href{https://arxiv.org/abs/hep-th/9608117}{\tt arXiv:hep-th/9608117}].


\bibitem[Du99]{Duff99}
M. Duff, 
{\it \color{darkblue} The World in Eleven Dimensions: Supergravity, Supermembranes and M-theory}, IoP Publishing  (1999), 
[\href{https://inspirehep.net/literature/513582}{\tt spire:513582}], 
[\href{https://www.crcpress.com/The-World-in-Eleven-Dimensions-Supergravity-supermembranes-and-M-theory/Duff/9780750306720}{\tt ISBN:9780750306720}].

\bibitem[Et11]{EtingofEtAl11}
P. Etingof,  O. Golberg, S. Hensel, T. Liu, A. Schwendner, D. Vaintrob, and E. Yudovina,
{\it \color{darkblue} Introduction to representation theory}, Student Mathematical Library {\bf 59}, 
AMS (2011), [\href{https://bookstore.ams.org/stml-59}{\tt ams:stml-59}],
[\href{https://arxiv.org/abs/0901.0827}{\tt arXiv:0901.0827}].


\bibitem[FGSS20]{FGSS20}
A. Ferraz, K. S. Gupta, G. W. Semenoff, and P. Sodano (eds.), 
{\it \color{darkblue} Strongly Coupled Field Theories for Condensed Matter and Quantum Information Theory}, Springer Proceedings in Physics 239, Springer (2020), 
[\href{https://doi.org/10.1007/978-3-030-35473-2}{\tt doi:10.1007/978-3-030-35473-2}].

\bibitem[FSS13]{FSS13CupCS}
D. Fiorenza, H. Sati, and U. Schreiber,
{\it \color{darkblue} Extended higher cup-product Chern-Simons theories}, J. Geom. Phys. {\bf 74} (2013), 130-163, [\href{https://doi.org/10.1016/j.geomphys.2013.07.011}{\tt doi:10.1016/j.geomphys.2013.07.011}],
[\href{https://arxiv.org/abs/1207.5449}{\tt arXiv:1207.5449}].


\bibitem[FSS15a]{FSS15-Stacky}
D. Fiorenza, H. Sati, and U. Schreiber
{\it \color{darkblue} A higher stacky perspective on Chern-Simons theory}, in {\it Mathematical Aspects of Quantum Field Theories}, Springer (2015),
153-211, 
[\href{https://doi.org/10.1007/978-3-319-09949-1_6}{\tt doi:10.1007/978-3-319-09949-1\_6}],
[\href{https://arxiv.org/abs/1301.2580}{\tt arXiv:1301.2580}].

\bibitem[FSS15b]{FSS15-M5WZW}
D. Fiorenza, H. Sati, and U. Schreiber,
{\it \color{darkblue} The WZW term of the M5-brane and differential cohomotopy}, 
J. Math. Phys. {\bf 56} (2015)  102301,
[\href{https://doi.org/10.1063/1.4932618}{\tt doi:10.1063/1.4932618}],
[\href{https://arxiv.org/abs/1506.07557}{\tt arXiv:1506.07557}].


\bibitem[FSS20a]{FSS20Exc}
D. Fiorenza, H. Sati, and U. Schreiber, 
{\it \color{darkblue} Super-exceptional geometry: Super-exceptional embedding construction of M5},
J. High Energy Phys. {\bf 2020}  (2020) 107, 
[\href{https://doi.org/10.1007/JHEP02(2020)107}{\tt doi:10.1007/JHEP02(2020)107}], \newline 
[\href{https://arxiv.org/abs/1908.00042}{\tt arXiv:1908.00042}].

\bibitem[FSS20b]{FSS20-H}
D. Fiorenza, H. Sati, and  U. Schreiber, {\it \color{darkblue} Twisted Cohomotopy implies M-theory anomaly cancellation on 8-manifolds}, 
Commun. Math. Phys. {\bf 377} (2020), 1961-2025, 
[\href{https://doi.org/10.1007/s00220-020-03707-2}{\tt doi:10.1007/s00220-020-03707-2}], \newline 
[\href{https://arxiv.org/abs/1904.10207}{\tt arXiv:1904.10207}].



\bibitem[FSS21a]{FSSHopf}
D. Fiorenza, H. Sati, and U. Schreiber,
{\it \color{darkblue} Twisted Cohomotopy implies M5 WZ term level quantization},
Commun. Math. Phys. {\bf 384} (2021), 403-432,
[\href{https://doi.org/10.1007/s00220-021-03951-0}{\tt doi:10.1007/s00220-021-03951-0}],
[\href{https://arxiv.org/abs/1906.07417}{\tt arXiv:1906.07417}].

\bibitem[FSS21b]{FSS21Emerge}
D. Fiorenza, H. Sati, and U. Schreiber, 
{\it \color{darkblue} Super-exceptional embedding construction of the heterotic M5: Emergence of SU(2)-flavor sector},
J. Geom. Phys. {\bf 170} (2021) 104349, [\href{https://arxiv.org/abs/2006.00012}{\tt arXiv:2006.00012}], \newline
[\href{https://doi.org/10.1016/j.geomphys.2021.104349}{\tt doi:10.1016/j.geomphys.2021.104349}].


\bibitem[FSS21c]{FSS20TwistedString}
D. Fiorenza, H. Sati, and U. Schreiber,
{\it \color{darkblue} Twisted cohomotopy implies twisted String structure on M5-branes},
J. Math. Phys. {\bf 62} (2021) 042301,
[\href{https://doi.org/10.1063/5.0037786}{\tt doi:10.1063/5.0037786}],
[\href{https://arxiv.org/abs/2002.11093}{\tt arXiv:2002.11093}].

\bibitem[FSS22]{FSS22-Twistorial}
D. Fiorenza, H. Sati, and U. Schreiber,
{\it \color{darkblue}Twistorial Cohomotopy implies Green-Schwarz anomaly cancellation},
Rev. MMath. Phys. {\bf 34} 05 (2022) 2250013
[\href{https://doi.org/10.1142/S0129055X22500131}{\tt doi:10.1142/S0129055X22500131}], [\href{https://arxiv.org/abs/2008.08544}{\tt arXiv:2008.08544}].



\bibitem[FSS23]{FSS23Char}
D. Fiorenza, H. Sati, and U. Schreiber,
{\it \color{darkblue} The Character map in Nonabelian Cohomology --- Twisted, Differential and Generalized},
World Scientific, Singapore (2023),
[\href{https://doi.org/10.1142/13422}{\tt doi:10.1142/13422}],
[\href{https://arxiv.org/abs/2009.11909}{\tt arXiv:2009.11909}].


\bibitem[FN62]{FoxNeuwirth62}
R. H. Fox and L. Neuwirth, 
{\it \color{darkblue} The braid groups}, Math. Scand. {\bf 10} (1962), 119-126,
\newline
[\href{https://doi.org/10.7146/math.scand.a-10518}{\tt doi:10.7146/math.scand.a-10518}].

\bibitem[FMS07]{FMS07}
D. Freed, G. Moore, and G.  Segal, {\it \color{darkblue} Heisenberg Groups and Noncommutative Fluxes}, 
Annals Phys. {\bf 322} (2007), 236-285, 
[\href{https://doi.org/10.1016/j.aop.2006.07.014}{\tt doi:10.1016/j.aop.2006.07.014}],
[\href{https://arxiv.org/abs/hep-th/0605200}{\tt arXiv:hep-th/0605200}].



\bibitem[FKLW03]{FKLW03}
M. Freedman, A. Kitaev, M. Larsen, and Z. Wang, 
{\it \color{darkblue} Topological quantum computation}, Bull. Amer. Math. Soc. {\bf 40} (2003), 31-38, 
[\href{https://doi.org/10.1090/S0273-0979-02-00964-3}{\tt doi:10.1090/S0273-0979-02-00964-3}],
[\href{https://arxiv.org/abs/quant-ph/0101025}{\tt arXiv:quant-ph/0101025}].

\bibitem[GM98]{GanorMotl98}
O. Ganor and L. Motl, 
{\it \color{darkblue} Equations of the $(2,0)$ Theory and Knitted Fivebranes}, 
J. High Energy Phys. {\bf  9805}  (1998) 009, 
[\href{https://doi.org/10.1088/1126-6708/1998/05/009}{\tt doi:10.1088/1126-6708/1998/05/009}],
[\href{https://arxiv.org/abs/hep-th/9803108}{\tt arXiv:hep-th/9803108}].

\bibitem[GH86]{GatesHishino86}
S. J. Gates Jr. and H. Nishino,
{\it \color{darkblue} $D=2$ superfield supergravity, local $(\mbox{sypersymmetry})^2$ and non-linear 2 models}, Classical and Quantum Gravity {\bf 3} 3 (1986), 391-399,
[\href{https://doi.org/10.1088/0264-9381/3/3/013}{\tt doi:10.1088/0264-9381/3/3/013}],
[\href{https://inspirehep.net/literature/217518}{\tt spire:217518}].



\bibitem[Gr43]{GelfandRaikov43}
I. Gelfand and D. Raikov, {\it \color{darkblue} 
Irreducible unitary representations of locally bicompact groups}, 
Rec. Math.  (Mat. Sbornik) {\bf 13(55)} (1943), 301–316, [\href{https://www.mathnet.ru/eng/sm6181}{\tt mathnet:sm6181}].


\bibitem[GS23]{GS23}
G. Giotopoulos and H. Sati, {\it \color{darkblue} Field Theory via Higher Geometry I: Smooth Sets of Fields}, 
\newline 
[\href{https://arxiv.org/abs/2312.16301}{\tt arXiv:2312.16301}].


\bibitem[GSS24a]{GSS24-SuGra}
G. Giotopoulos, H. Sati and U. Schreiber, 
{\it \color{darkblue} Flux Quantization on 11d Superspace},
J. High Energy Phys. {\bf 2024} 82 (2024)
[\href{https://doi.org/10.1007/JHEP07(2024)082}{\tt doi:10.1007/JHEP07(2024)082}],
[\href{https://arxiv.org/abs/2403.16456}{\tt arXiv:2403.16456}].

\bibitem[GSS24b]{GSS24-AdS7}
G. Giotopoulos, H. Sati, and U. Schreiber, 
{\it \color{darkblue} Holographic M-Brane Super-Embeddings},
J. High Energy Physocs (2024, in print)
[\href{https://arxiv.org/abs/2408.09921}{\tt arXiv:2408.09921}].

\bibitem[GSS24c]{GSS-Exceptional}
G. Giotopoulos, H. Sati, and U. Schreiber, 
{\it \color{darkblue} M5 Super-Embeddings into Exceptional Super-Spacetime} (in preparation).

\bibitem[GSS24d]{GSS24-SuperSet}
G. Giotopoulos, H. Sati, and U. Schreiber, {\it \color{darkblue} Field Theory via Higher Geometry II: Super-sets of fermionic fields} (in preparation).


\bibitem[Gi20]{Giron20}
T. P. Giron, 
{\it \color{darkblue} On the Analysis of Isometric Immersions of Riemannian Manifolds}, PhD thesis,
Oxford U., (2020), [\href{https://ora.ox.ac.uk/objects/uuid:cae2f41d-c5a1-4138-9dec-5e824d21044e}{\tt uuid:cae2f41d-c5a1-4138-9dec-5e824d21044e}].

\bibitem[GKSTY02]{GKSTY02}
E. Gorbatov, V. Kaplunovsky, J. Sonnenschein, S. Theisen, S. Yankielowicz, 
{\it On Heterotic Orbifolds, M Theory and Type $I'$ Brane Engineering}, 
JHEP 0205:015 (2002) 
[\href{https://arxiv.org/abs/hep-th/0108135}{\tt arXiv:hep-th/0108135}],
[\href{https://doi.org/10.1088/1126-6708/2002/05/015}{\tt doi:10.1088/1126-6708/2002/05/015}].


\bibitem[GS21]{GS21}
D. Grady and H. Sati, 
{\it \color{darkblue} Differential cohomotopy versus differential cohomology for M-theory and differential lifts of Postnikov towers}, J. Geom. Phys. {\bf 165} (2021) 104203, 
[\href{https://doi.org/10.1016/j.geomphys.2021.104203}{\tt doi;10.1016/j.geomphys.2021.104203}],
[\href{https://arxiv.org/abs/2001.07640}{\tt arXiv:2001.07640}].

\bibitem[GS22]{GS22}
D. Grady and H. Sati, 
{\it \color{darkblue} Ramond–Ramond fields and twisted differential K-theory}, 
Adv. Theor. Math. Phys. {\bf 26} (2022), 1097–1155,
[\href{https://dx.doi.org/10.4310/ATMP.2022.v26.n5.a2}{\tt doi:10.4310/ATMP.2022.v26.n5.a2}], 
[\href{https://arxiv.org/abs/1903.08843}{\tt 
	arXiv:1903.08843}]. 

\bibitem[GH79]{GriffithsHarris79}
P. Griffiths and J. Harris, 
{\it \color{darkblue} Algebraic geometry and local differential geometry}, Ann. scient.  l'{\'E}cole Norm. Sup., 
Serie 4, {\bf 12} 3 (1979), 355-452, 
[\href{http://www.numdam.org/item/?id=ASENS_1979_4_12_3_355_0}{\tt numdam:ASENS\_1979\_4\_12\_3\_355\_0}].

\bibitem[Gu77]{Guggenheimer77}
H. Guggenheimer, 
{\it \color{darkblue} Differential Geometry}, Dover (1977), 
[\href{https://store.doverpublications.com/products/9780486634333}{\tt ISBN:9780486634333}],  
[\href{https://archive.org/details/differentialgeom0000gugg/}{\tt ark:/13960/t9t22sk9n}].



\bibitem[HL23]{HanLewicka23}
Q. Han and M. Lewicka, {\it \color{darkblue} Isometric immersions and applications}, 
Notices AMS {\bf 2023}, 
[\href{https://arxiv.org/abs/2310.02566}{\tt arXiv:2310.02566}].



\bibitem[HLS18]{HartnollLucasSachdev18}
S. Hartnoll, A. Lucas, and S. Sachdev, {\it \color{darkblue} Holographic quantum matter}, 
MIT Press (2018), \newline 
[\href{https://mitpress.ublish.com/book/holographic-quantum-matter}{\tt ISBN:9780262348010]}],
[\href{https://arxiv.org/abs/1612.07324}{\tt arXiv:1612.07324}].

\bibitem[HLLSZ19]{HLLSZ19}
J. Heckman, C. Lawrie, L. Lin, J. Sakstein, and G. Zoccarato, {\it \color{darkblue} Pixelated Dark Energy}, 
Fortsch. Phys. {\bf 67} 11 (2019) 1900071, 
[\href{https://doi.org/10.1002/prop.201900071}{\tt doi:10.1002/prop.201900071}],
[\href{https://arxiv.org/abs/1901.10489}{\tt arXiv:1901.10489}].


\bibitem[HR18]{HeckmanRudelius18}
J. Heckman and T. Rudelius, 
{\it \color{darkblue} Top Down Approach to 6D SCFTs}, J. Phys. A: Math. Theor. {\bf 52} (2018) 093001, 
[\href{https://doi.org/10.1088/1751-8121/aafc81}{\tt doi:10.1088/1751-8121/aafc81}],
[\href{https://arxiv.org/abs/1805.06467}{\tt arXiv:1805.06467}].

\bibitem[HT88]{HenneauxTeitelboim88}
M. Henneaux and C. Teitelboim, 
{\it \color{darkblue}  Dynamics of chiral (self-dual) $p$-forms}, 
Phys. Lett. B {\bf 206}  (1988), 650-654,
[\href{https://doi.org/10.1016/0370-2693(88)90712-5}{\tt doi:10.1016/0370-2693(88)90712-5}].


\bibitem[HW15]{HollandsWald15}
S. Hollands and R. Wald, {\it \color{darkblue} Quantum fields in curved spacetime}, 
Phys. Rept. {\bf 574} (2015), 1-35, \newline 
[\href{https://doi.org/10.1016/j.physrep.2015.02.001}{\tt doi:10.1016/j.physrep.2015.02.001}],
[\href{https://arxiv.org/abs/1401.2026}{\tt arXiv:1401.2026}].

\bibitem[HW96]{HoravaWitten96}
P. Ho{\v r}ava and E. Witten, {\it \color{darkblue} Heterotic and Type I string dynamics from eleven dimensions}, Nucl. Phys. B {\bf 460} (1996), 506 -524, 
[\href{https://doi.org/10.1016/0550-3213(95)00621-4}{\tt doi:10.1016/0550-3213(95)00621-4}],
[\href{https://arxiv.org/abs/hep-th/9510209}{\tt arXiv:hep-th/9510209}].


\bibitem[HLW98]{HoweLambertWest98}
P. S. Howe, N. Lambert, and P. West, 
{\it \color{darkblue} The Self-Dual String Soliton}, Nucl. Phys. B {\bf 515} (1998), 203-216, 
[\href{https://doi.org/10.1016/0550-3213(94)90586-X}{\tt doi:10.1016/0550-3213(94)90586-X}],
[\href{https://arxiv.org/abs/hep-th/9709014}{\tt arXiv:hep-th/9709014}].

\bibitem[HRS98]{HoweRaetzelSezgin98}
P. S. Howe, O. Raetzel, and E. Sezgin, 
{\it \color{darkblue} On Brane Actions and Superembeddings}, J. High Energy Phys. {\bf 9808} (1998) 011,
[\href{https://doi.org/10.1088/1126-6708/1998/08/011}{\tt 
doi:10.1088/1126-6708/1998/08/011}], 
[\href{https://arxiv.org/abs/hep-th/9804051}{\tt arXiv:hep-th/9804051}].


\bibitem[HS97a]{HoweSezgin97a}
P. S. Howe and E. Sezgin, {\it \color{darkblue} Superbranes}, Phys. Lett. B {\bf 390} (1997), 133-142, 
[\href{https://arxiv.org/abs/hep-th/9607227}{\tt arXiv:hep-th/9607227}],
[\href{https://doi.org/10.1016/S0370-2693(96)01416-5}{\tt doi:10.1016/S0370-2693(96)01416-5}].

\bibitem[HS97b]{HoweSezgin97b}
P. Howe and E. Sezgin, 
{\it \color{darkblue} $D=11$, $p=5$}, Phys. Lett. B {\bf 394} (1997), 62-66, [\href{https://arxiv.org/abs/hep-th/9611008}{\tt arXiv:hep-th/9611008}],\newline 
[\href{https://doi.org/10.1016/S0370-2693(96)01672-3}{\tt doi:10.1016/S0370-2693(96)01672-3}].

\bibitem[HSW97a]{HSW97a}
P. Howe, E. Sezgin and P. West, 
{\it \color{darkblue}Covariant Field Equations of the M Theory Five-Brane}, 
Phys. Lett. B {\bf 399} (1997), 49-59, 
[\href{https://doi.org/10.1016/S0370-2693(97)00257-8}{\tt doi:10.1016/S0370-2693(97)00257-8}], [\href{https://arxiv.org/abs/hep-th/9702008}{\tt arXiv:hep-th/9702008}].

\bibitem[HSW97b]{HoweSezginWest97}
P. S. Howe, E. Sezgin, and P. C. West, {\it \color{darkblue} The six-dimensional self-dual tensor}, Phys. Lett. B {\bf 400} (1997), 255-259,
[\href{https://doi.org/10.1016/S0370-2693(97)00365-1}{\tt doi:10.1016/S0370-2693(97)00365-1}],
[\href{https://arxiv.org/abs/hep-th/9702111}{\tt arXiv:hep-th/9702111}].

\bibitem[HSW98]{HSW98}
P. S. Howe, E. Sezgin and P. C. West, 
{\it \color{darkblue}Aspects of Superembeddings}, in: {\it Supersymmetry and Quantum Field Theory}, Lecture Notes in Physics {\bf 509}, Springer (1998),
[\href{https://doi.org/10.1007/BFb0105230}{\tt doi:10.1007/BFb0105230}],
[\href{https://arxiv.org/abs/hep-th/9705093}{\tt arXiv:hep-th/9705093}].


\bibitem[HSS19]{HSS19}
J. Huerta, H. Sati, and U. Schreiber,
{\it \color{darkblue} Real ADE-equivariant (co)homotopy and Super M-branes}, 
Comm. Math. Phys. {\bf 371} (2019), 425–524,
[\href{https://doi.org/10.1007/s00220-019-03442-3}{\tt doi:10.1007/s00220-019-03442-3}],
[\href{https://arxiv.org/abs/1805.05987}{\tt arXiv:1805.05987}].


\bibitem[HS18]{HS18}
J. Huerta and U. Schreiber,
{\it \color{darkblue} M-theory from the superpoint}, Lett.  Math. Physics {\bf 108} (2018), 2695–2727, 
[\href{https://doi.org/10.1007/s11005-018-1110-z}{\tt doi;10.1007/s11005-018-1110-z}],
[\href{https://arxiv.org/abs/1702.01774}{\tt arXiv:1702.01774}].

\bibitem[In00]{Intriligator00}
K. Intriligator, 
{\it \color{darkblue} Anomaly Matching and a Hopf-Wess-Zumino Term in 6d, $\mathcal{N}=(2,0)$ Field Theories}, Nucl. Phys. B {\bf 581} (2000), 
257-273,  
[\href{https://doi.org/10.1016/S0550-3213(00)00148-6}{\tt doi:10.1016/S0550-3213(00)00148-6}],
[\href{https://arxiv.org/abs/hep-th/0001205}{\tt arXiv:hep-th/0001205}].



\bibitem[ILP18]{ILP18}
T. Ivanova, O. Lechtenfeld, and A. Popov, {\it \color{darkblue} Skyrme model from 6d $\mathcal{N}=(2,0)$ theory}, Phys. Lett. B {\bf 783} (2018), 222-226, 
[\href{https://doi.org/10.1016/j.physletb.2018.06.052}{\tt doi:10.1016/j.physletb.2018.06.052}],
[\href{https://arxiv.org/abs/1805.07241}{\tt arXiv:1805.07241}].

\bibitem[KOTY23]{KOTY23}
J. Kaidi, K. Ohmori, Y. Tachikawa, and K. Yonekura, 
{\it \color{darkblue} Non-supersymmetric heterotic branes}, \newline 
[\href{https://arxiv.org/abs/2303.17623}{\tt arXiv:2303.17623}]

\bibitem[Ka21]{Kayban21}
N. Kayban, 
{\it \color{darkblue} Riemannian Immersions and Submersions} (2021), 
\newline
[\href{https://ncatlab.org/nlab/files/Kayban-RiemannianImmersions.pdf}{\tt ncatlab.org/nlab/files/Kayban-RiemannianImmersions.pdf}]

\bibitem[La19a]{Lambert19}
N. Lambert, {\it \color{darkblue} Lessons from M2’s and Hopes for M5’s}, 
Proc.  {\it LMS-EPSRC Durham Symposium: Higher Structures in M-Theory}, Aug. 2018, 
Fortschr. Phys. {\bf 67}, (2019), 8-9, 
[\href{https://doi.org/10.1002/prop.201910011}{\tt doi:10.1002/prop.201910011}],
[\href{https://arxiv.org/abs/1903.02825}{\tt arXiv:1903.02825}].

\bibitem[La19b]{Lambert19b}
N. Lambert, 
{\it \color{darkblue} $(2,0)$ Lagrangian Structures}, 
Phys. Lett. B {\bf 798} (2019) 134948, 
[\href{https://arxiv.org/abs/1908.10752}{\tt arXiv:1908.10752}],
\newline 
[\href{https://doi.org/10.1016/j.physletb.2019.134948}{\tt doi:10.1016/j.physletb.2019.134948}].


\bibitem[Lee12]{Lee12}
J. Lee, 
{\it \color{darkblue} Introduction to Smooth Manifolds}, Springer (2012), 
[\href{https://doi.org/10.1007/978-1-4419-9982-5}{\tt doi:10.1007/978-1-4419-9982-5}].

\bibitem[Lee18]{Lee18}
J. Lee, 
{\it \color{darkblue} Introduction to Riemannian Manifolds}, Springer (2018), 
[\href{https://doi.org/10.1007/978-3-319-91755-9}{\tt doi:10.1007/978-3-319-91755-9}].


\bibitem[Mant22]{Manton22}
N. Manton, 
{\it \color{darkblue} Skyrmions -- A Theory of Nuclei}, World Scientific (2022),
[\href{https://doi.org/10.1142/q0368}{\tt doi:10.1142/q0368}].

\bibitem[MRS12]{MRS12}
P. Mastrolia, M. Rigoli, and A. G. Setti, 
{\it \color{darkblue} Some formulas for immersed submanifolds} in: {\it Yamabe-type Equations on Complete, Noncompact Manifolds}, Birkh{\"a}user (2012), 
[\href{https://doi.org/10.1007/978-3-0348-0376-2}{\tt doi:10.1007/978-3-0348-0376-2}].

\bibitem[MaSa15]{MathaiSati}
V. Mathai H. Sati,
{\it \color{darkblue} Higher abelian gauge theory associated to gerbes on noncommutative deformed M5-branes and S-duality}, 
J. Geom. Phys. {\bf 92} (2015), 240-251,
[\href{https://doi.org/10.1016/j.geomphys.2015.02.019}{\tt 
 doi:10.1016/j.geomphys.2015.02.019}], \newline 
[\href{https://arxiv.org/abs/1404.2257}{\tt 
arXiv:1404.2257}].

\bibitem[MiSc06]{MiemiecSchnakenburg06}
A. Miemiec and I. Schnakenburg, {\it \color{darkblue} Basics of M-Theory}, Fortsch. Phys.
{\bf 54} (2006),
5-72, \newline 
[\href{https://doi.org/10.1002/prop.200510256}{\tt doi:10.1002/prop.200510256}],
[\href{https://arxiv.org/abs/hep-th/0509137}{\tt arXiv:hep-th/0509137}].

\bibitem[MySS24]{MySS24}
D. J. Myers, H. Sati, and U. Schreiber, {\it \color{darkblue} Topological Quantum Gates in Homotopy Type Theory}, Commun. Math. Phys. (2024, in print), [\href{https://arxiv.org/abs/2303.02382}{\tt arXiv:2303.02382}].

\bibitem[NSS${}^+$08]{NSSFDS08}
C. Nayak, S. H. Simon, A. Stern, M. Freedman, and S. Das Sarma, 
{\it \color{darkblue} Non-Abelian Anyons and Topological Quantum Computation}, Rev. Mod. Phys. {\bf 80} 
(2008) 1083, 
[\href{https://doi.org/10.1103/RevModPhys.80.1083}{\tt doi:10.1103/RevModPhys.80.1083}],
[\href{https://arxiv.org/abs/0707.1889}{\tt arXiv:0707.1889}].


\bibitem[O'N83]{ONeil83}
B. O'Neill, 
{\it \color{darkblue} Semi-Riemannian Geometry With Applications to Relativity}, Pure and Applied Mathematics {\bf 103}, Academic Press (1983), [\href{https://shop.elsevier.com/books/semi-riemannian-geometry-with-applications-to-relativity/oneill/978-0-12-526740-3}{\tt ISBN:9780125267403}].

\bibitem[PST97]{PST97}
P. Pasti, D. Sorokin, and M. Tonin, {\it \color{darkblue} Covariant Action for a $D=11$ Five-Brane with the Chiral Field}, Phys. Lett. B {\bf 398} (1997) 41-46,  
[\href{https://doi.org/10.1016/S0370-2693(97)00188-3}{\tt doi:10.1016/S0370-2693(97)00188-3}],
[\href{https://arxiv.org/abs/hep-th/9701037}{\tt arXiv:hep-th/9701037}].

\bibitem[PST99]{PST99}
P. Pasti, D. Sorokin, and M. Tonin, 
{\it \color{darkblue}  Branes in Super-AdS Backgrounds and Superconformal Theories}, Proceedings, {\it International Workshop on Supersymmetries and Quantum Symmetries}, Moscow (1999), 
[\href{https://arxiv.org/abs/hep-th/9912076}{\tt arXiv:hep-th/9912076}],
[\href{https://inspirehep.net/literature/511348}{\tt inspire:511348}].


\bibitem[PB20]{PetruninBarrera20}
A. Petrunin and S. Z. Barrera, 
{\it \color{darkblue} What is differential geometry? curves and surfaces}, 
[\href{https://arxiv.org/abs/2012.11814}{\tt arXiv:2012.11814}].

\bibitem[Re15]{Rebhan15}
A. Rebhan, 
{\it \color{darkblue} The Witten-Sakai-Sugimoto model: A brief review and some recent results}, 
EPJ Web of Conferences {\bf 95}  (2015) 02005, 
[\href{https://doi.org/10.1051/epjconf/20159502005}{\tt doi:10.1051/epjconf/20159502005}],
[\href{https://arxiv.org/abs/1410.8858}{\tt arXiv:1410.8858}].

\bibitem[RZ16]{RhoZahed16}
M. Rho and I. Zahed,
{\it \color{darkblue} The Multifaceted Skyrmion}, World Scientific (2016), [\href{https://doi.org/10.1142/9710}{\tt doi;10.1142/9710}].

\bibitem[RSvdW21]{RSvdW21}
D. Rist, C. Saemann, and M. van der Worp,
{\it  \color{darkblue} Towards an M5-Brane Model III: Self-Duality from Additional Trivial Fields}, 
J. High Energ. Phys. {\bf 2021} (2021) 36,
[\href{https://doi.org/10.1007/JHEP06(2021)036}{\tt doi:10.1007/JHEP06(2021)036}],\newline 
[\href{https://arxiv.org/abs/2012.09253}{\tt arXiv:2012.09253}].




\bibitem[Ro07]{Rogers07}
A. Rogers, 
{\it \color{darkblue} Supermanifolds: Theory and Applications}, World Scientific (2007), 
[\href{https://doi.org/10.1142/1878}{\tt doi:10.1142/1878}].

\bibitem[RS00]{RourkeSanderson00}
C. Rourke and B. Sanderson, 
{\it \color{darkblue} Equivariant Configuration Spaces}, J. London Math. Soc. {\bf 62} (2000), 544-552,
[\href{https://doi.org/10.1112/S0024610700001241}{\tt doi:10.1112/S0024610700001241}],
[\href{https://arxiv.org/abs/math/9712216}{\tt arXiv:math/9712216}].


\bibitem[Ro22]{Rowell22}
E. C. Rowell, 
{\it \color{darkblue} Braids, Motions and Topological Quantum Computing}, 
[\href{https://arxiv.org/abs/2208.11762}{\tt arXiv:2208.11762}]


\bibitem[S{\"a}S18]{SaemannSchmidt18}
C. Saemann and L. Schmidt, 
{\it \color{darkblue} Towards an M5-Brane Model I: A 6d Superconformal Field Theory}, J. Math. Phys. {\bf 59} (2018) 043502,
[\href{https://doi.org/10.1063/1.5026545}{\tt doi:10.1063/1.5026545}],
[\href{https://arxiv.org/abs/1712.06623}{\tt arXiv:1712.06623}].

\bibitem[S{\"a}S20]{SaemannSchmidt20}
C. Saemann and L. Schmidt, 
{\it \color{darkblue} Towards an M5-Brane Model II: Metric String Structures}, Fortschr. Phys. {\bf 68} (2020) 2000051, [\href{https://doi.org/10.1002/prop.202000051}{\tt doi:10.1002/prop.202000051}],
[\href{https://arxiv.org/abs/1908.08086}{\tt arXiv:1908.08086}].

\bibitem[Sa06]{Sati06}
H. Sati, 
{\it \color{darkblue} Duality symmetry and the form fields of M-theory},
  J. High Energy Phys. {\bf 0606} (2006) 062, \newline 
[\href{https://doi.org/10.1088/1126-6708/2006/06/062}{\tt 
doi:10.1088/1126-6708/2006/06/062}], 
[\href{https://arxiv.org/abs/hep-th/0509046}{\tt 
arXiv:hep-th/0509046}]. 


\bibitem[Sa10]{Sati10}
H. Sati,
{\it \color{darkblue} Geometric and topological structures related to M-branes},
in: 
{\it \color{darkblue} Superstrings, Geometry, Topology, and $C^\ast$-algebras},
Proc. Symp. Pure Math. {\bf 81}, AMS, Providence (2010), 181-236,
[\href{https://doi.org/10.1090/pspum/081}{\tt doi:10.1090/pspum/081}],
[\href{https://arxiv.org/abs/1001.5020}{\tt arXiv:1001.5020}].


\bibitem[Sa11]{Sati-Aust}
H. Sati, 
{\it \color{darkblue} Geometric and topological structures related to M-branes II: Twisted String and String${}^c$
structures}, 
J. Australian Math. Soc. {\bf 90} (2011), 93-108, 
[\href{https://www.cambridge.org/core/journals/journal-of-the-australian-mathematical-society/article/geometric-and-topological-structures-related-to-mbranes-ii-twisted-string-and-stringc-structures/222FDC37835511262CDEC577A820B46A}{\tt 
doi:10.1017/S1446788711001261}], 
[\href{https://arxiv.org/abs/1007.5419}{\tt arXiv:1007.5419}]. 
 	
 

\bibitem[Sa13]{Sati13}
H. Sati, {\it \color{darkblue} Framed M-branes, corners, and topological invariants}, 
J. Math. Phys. {\bf 59} (2018) 062304, \newline 
[\href{https://doi.org/10.1063/1.5007185}{\tt doi:10.1063/1.5007185}], 
[\href{https://arxiv.org/abs/1310.1060}{\tt arXiv:1310.1060}].

\bibitem[Sa19]{Sati19}
H. Sati, 
{\it \color{darkblue} Six-dimensional gauge theories and (twisted) generalized cohomology},     
[\href{https://arxiv.org/abs/1908.08517}{\tt arXiv:1908.08517}].


\bibitem[SS20a]{SS20Tadpole}
H. Sati and U. Schreiber,
{\it \color{darkblue} Equivariant Cohomotopy implies orientifold tadpole cancellation}, 
J. Geom. Phys.
{\bf 156} (2020) 103775,
[\href{https://doi.org/10.1016/j.geomphys.2020.103775}{\tt doi:10.1016/j.geomphys.2020.103775}],
[\href{https://arxiv.org/abs/1909.12277}{\tt arXiv:1909.12277}].


\bibitem[SS20b]{SS20-Orb}
H. Sati and U. Schreiber:
{\it \color{darkblue} Proper Orbifold Cohomology},
[\href{https://arxiv.org/abs/2008.01101}{\tt arXiv:2008.01101}].

\bibitem[SS20c]{SS20EquChar}
H. Sati and Urs Schreiber:
{\it \color{darkblue} The character map in equivariant twistorial Cohomotopy
implies the Green-Schwarz mechanism with heterotic M5-branes},
[\href{https://arxiv.org/abs/2011.06533}{\tt arXiv:2011.06533}].

\bibitem[SS21a]{SS21M5Anomaly}
H. Sati and U. Schreiber, 
{\it \color{darkblue} Twisted Cohomotopy implies M5-brane anomaly cancellation},
Lett. Math. Phys. {\bf 111}  (2021) 120, 
[\href{https://doi.org/10.1007/s11005-021-01452-8}{\tt doi:10.1007/s11005-021-01452-8}],
[\href{https://arxiv.org/abs/2002.07737}{\tt arXiv:2002.07737}].

\bibitem[SS21b]{SS21-EquBund}
H. Sati and U. Schreiber, 
{\it \color{darkblue} Equivariant Principal $\infty$-Bundles},
Cambridge University Press (2025, in print)
[\href{https://arxiv.org/abs/2112.13654}{\tt arXiv:2112.13654}].

\bibitem[SS22]{SS22-Config}
H. Sati and U. Schreiber,
{\it \color{darkblue} Differential Cohomotopy implies intersecting brane observables via configuration spaces and chord diagrams},
Adv. Theor. Math.  Phys. {\bf 26} 4 (2022), 957-1051,
[\href{https://arxiv.org/abs/1912.10425}{\tt arXiv:1912.10425}],
\newline 
[\href{https://dx.doi.org/10.4310/ATMP.2022.v26.n4.a4}{\tt doi:10.4310/ATMP.2022.v26.n4.a4}].

\bibitem[SS23a]{SS23-Mf}
H. Sati and U. Schreiber, 
{\it \color{darkblue} M/F-Theory as Mf-Theory}, Rev. Math. Phys. {\bf 35} (2023) 2350028, \newline 
[\href{https://doi.org/10.1142/S0129055X23500289}{\tt arXiv:10.1142/S0129055X23500289}],
[\href{https://arxiv.org/abs/2103.01877}{\tt arXiv:2103.01877}].



\bibitem[SS23b]{SS23-DefectBranes}
H. Sati and U. Schreiber,
{\it \color{darkblue} Anyonic Defect Branes and Conformal Blocks in Twisted Equivariant Differential K-Theory}, 
Rev. Math. Phys. {\bf 35} 06 (2023) 2350009,
[\href{https://doi.org/10.1142/S0129055X23500095}{\tt doi:10.1142/S0129055X23500095}],
[\href{https://arxiv.org/abs/2203.11838}{\tt arXiv:2203.11838}].

\bibitem[SS23c]{SS23-ToplOrder}
H. Sati and U. Schreiber, 
{\it \color{darkblue} Anyonic topological order in TED K-theory}, 
Rev. Math. Phys. (2023)
{\bf 35} 03 (2023) 2350001,
[\href{https://doi.org/10.1142/S0129055X23500010}{\tt doi:10.1142/S0129055X23500010}],
[\href{https://arxiv.org/abs/2206.13563}{\tt arXiv:2206.13563}].

\bibitem[SS24a]{SS24-Phase}
H. Sati and U. Schreiber:
{\it \color{darkblue} Flux Quantization on Phase Space},
Ann. Henri Poincar{\'e} {\bf 26} (2024) 895–919,
\newline 
[\href{https://doi.org/10.1007/s00023-024-01438-x}{\tt doi:10.1007/s00023-024-01438-x}],
[\href{https://arxiv.org/abs/2312.12517}{\tt arXiv:2312.12517}].

\bibitem[SS24b]{SS23-Obs}
H. Sati and U. Schreiber:
{\it \color{darkblue} Quantum Observables of Quantized Fluxes},
Ann. Henri Poincar{\'e} (2024)
\newline
[\href{https://arxiv.org/abs/2312.13037}{\tt arXiv:2312.13037}],
[\href{https://doi.org/10.1007/s00023-024-01517-z}{\tt doi:10.1007/s00023-024-01517-z}].


\bibitem[SS24c]{SS24-AbAnyons}
H. Sati and U. Schreiber,
{\it \color{darkblue} Abelian Anyons on Flux-Quantized M5-Branes},
[\href{https://arxiv.org/abs/2408.11896}{\tt arXiv:2408.11896}].

\bibitem[SS25]{SS24-Flux}
H. Sati and U. Schreiber,
{\it \color{darkblue} Flux quantization},
Encyclopedia of Mathematical Physics 2nd ed., 
{\bf 4} (2025) 281-324
[\href{https://arxiv.org/abs/2402.18473}{\tt arXiv:2402.18473}], 
[\href{doi:10.1016/B978-0-323-95703-8.00078-1}{\tt doi:10.1016/B978-0-323-95703-8.00078-1}].


\bibitem[SV22]{SV2}	
H. Sati and A. Voronov, 
{\it \color{darkblue} Mysterious Triality and M-theory}, 
[\href{https://arxiv.org/abs/2212.13968}{\tt arXiv:2212.13968}]. 

\bibitem[Se73]{Segal73}
G. Segal, 
{\it \color{darkblue} Configuration-spaces and iterated loop-spaces}, 
Invent. Math. {\bf 21} (1973), 213-221, \newline 
[\href{https://doi.org/10.1007/BF01390197}{\tt doi:10.1007/BF01390197}].

\bibitem[Sen20]{Sen20}
A. Sen, 
{\it \color{darkblue} Self-dual forms: Action, Hamiltonian and Compactification}, 
J. Phys. A: Math. Theor. {\bf 53} (2020) 084002,
[\href{https://doi.org/10.1088/1751-8121/ab5423}{\tt doi:10.1088/1751-8121/ab5423}],
[\href{https://arxiv.org/abs/1903.12196}{\tt arXiv:1903.12196}].

\bibitem[SS98]{SezginSundell98}
E. Sezgin and P. Sundell, 
{\it \color{darkblue} Aspects of the M5-Brane}, in: {\it Proceedings of Nonperturbative aspects of strings, branes and supersymmetry} (1998), 369-389, 
[\href{https://inspirehep.net/literature/483085}{\tt inspire:483085}],
[\href{https://arxiv.org/abs/hep-th/9902171}{\tt arXiv:hep-th/9902171}].


\bibitem[So00]{Sorokin00}
D. Sorokin, 
{\it \color{darkblue} Superbranes and Superembeddings}, Phys. Rept. {\bf 329} (2000), 1-101, 
[\href{https://arxiv.org/abs/hep-th/9906142}{\tt arXiv:hep-th/9906142}],
[\href{https://doi.org/10.1016/S0370-1573(99)00104-0}{\tt doi:10.1016/S0370-1573(99)00104-0}].

\bibitem[St64]{Sternberg64}
S. Sternberg, 
{\it \color{darkblue} Lectures on differential geometry}, 
Prentice-Hall (1964),
AMS (1983),
[\href{https://bookstore.ams.org/chel-316}{\tt ams:chel-316}].

\bibitem[Su16]{Sugimoto16}
S. Sugimoto, 
{\it \color{darkblue} Skyrmion and String theory}, chapter 15 in: 
M. Rho, I. Zahed (eds.), 
{\it The Multifaceted Skyrmion}, World Scientific (2016),
[\href{https://doi.org/10.1142/9710}{\tt doi;10.1142/9710}].

\bibitem[SV10]{SzaboValentino10}
R. Szabo and A. Valentino, 
{\it \color{darkblue} Ramond-Ramond Fields, Fractional Branes and Orbifold Differential K-Theory}, Commun. Math. Phys. {\bf 294} (2010), 647-702,
[\href{https://doi.org/10.1007/s00220-009-0975-1}{\tt doi:10.1007/s00220-009-0975-1}],
[\href{https://arxiv.org/abs/0710.2773}{\tt arXiv:0710.2773}].


\bibitem[Ts04b]{Tsimpis04b}
D. Tsimpis, 
{\it \color{darkblue} 11D supergravity at $\mathcal{O}(\ell^3)$}, J. High Energy Phys. {\bf 0410}
(2004) 046,
[\href{https://arxiv.org/abs/hep-th/0407271}{\tt arXiv:hep-th/0407271}],
[\href{https://doi.org/10.1088/1126-6708/2004/10/046}{\tt doi:10.1088/1126-6708/2004/10/046}].


\bibitem[Va04]{Varadarajan04}
V. Varadarajan, {\it \color{darkblue} Supersymmetry for mathematicians: An introduction}, Courant Lecture Notes in Mathematics {\bf 11}, American Mathematical Society (2004), 
[\href{http://dx.doi.org/10.1090/cln/011}{\tt doi:10.1090/cln/011}].

\bibitem[Wi20]{Williams20}
L. Williams, 
{\it \color{darkblue} Configuration Spaces for the Working Undergraduate}, Rose-Hulman Undergrad. Math. J. {\bf 21} 1 (2020) 8,
[\href{https://scholar.rose-hulman.edu/rhumj/vol21/iss1/8}{\tt rhumj:vol21/iss1/8}],
[\href{https://arxiv.org/abs/1911.11186}{\tt arXiv:1911.11186}].

\bibitem[Wi96]{Witten96}
E. Witten, 
{\it \color{darkblue} Bound States Of Strings And $p$-Branes}, Nucl. Phys. B {\bf 460} (1996), 335-350, 
\newline
[\href{https://doi.org/10.1016/0550-3213(95)00610-9}{\tt doi:10.1016/0550-3213(95)00610-9}],
[\href{https://arxiv.org/abs/hep-th/9510135}{\tt arXiv:hep-th/9510135}].

\bibitem[Wi97a]{Witten97a}
E. Witten, 
{\it  \color{darkblue} On Flux Quantization In M-Theory And The Effective Action}, 
J. Geom. Phys. {\bf 22} (1997), 1-13, 
[\href{https://doi.org/10.1016/S0393-0440(96)00042-3}{\tt doi:10.1016/S0393-0440(96)00042-3}],
[\href{https://arxiv.org/abs/hep-th/9609122}{\tt arXiv:hep-th/9609122}]. 

\bibitem[Wi97b]{Witten97b}
E. Witten, 
{\it \color{darkblue} Five-Brane Effective Action In M-Theory}, 
J. Geom. Phys. {\bf 22} (1997), 103-133, \newline 
[\href{https://doi.org/10.1016/S0393-0440(97)80160-X}{\tt doi:10.1016/S0393-0440(97)80160-X}],
[\href{https://arxiv.org/abs/hep-th/9610234}{\tt arXiv:hep-th/9610234}].



\bibitem[Wi98]{Witten98}
E. Witten, {\it \color{darkblue} Anti-de Sitter Space, Thermal Phase Transition, And Confinement In Gauge Theories}, Adv. Theor. Math. Phys. {\bf 2} 3 (1998), 505-532,
[\href{https://dx.doi.org/10.4310/ATMP.1998.v2.n3.a3}{\tt doi:10.4310/ATMP.1998.v2.n3.a3}],
[\href{https://arxiv.org/abs/hep-th/9803131}{\tt arXiv:hep-th/9803131}].



\bibitem[Wi04]{Witten04}
E. Witten, {\it \color{darkblue} Conformal field theory in four and six dimensions}, in: {\it Topology, Geometry and Quantum Field Theory} 
Proceedings of the 2002 Oxford Symposium in Honour of the 60th Birthday of Graeme Segal, 
London Math. Soc. (2004), 405-420, 
[\href{https://doi.org/10.1017/CBO9780511526398.017}{\tt doi:10.1017/CBO9780511526398.017}],
[\href{https://arxiv.org/abs/0712.0157}{\tt arXiv:0712.0157}].

\bibitem[Wi09]{Witten09}
E. Witten, 
{\it \color{darkblue} Geometric Langlands From Six Dimensions}, 
[\href{https://arxiv.org/abs/0905.2720}{\tt arXiv:0905.2720}].



\bibitem[Xi06]{Xico06}
M. Xicot{\'e}ncatl, {\it \color{darkblue} On $\mathbb{Z}_2$-equivariant loop spaces}, in: {\it Recent developments in algebraic topology}, Contemp. Math. {\bf 407}, AMS (2006), 183—191,[\href{https://bookstore.ams.org/CONM/407}{\tt ams:CONM/407}].


\bibitem[Za88]{Zandi88}
A. Zandi, 
{\it \color{darkblue} Minimal immersions of surfaces in quaternionic projective space}, Tsukuba J. Math. {\bf 12} 2 (1988), 423-440, [\href{https://www.jstor.org/stable/43686661}{\tt jstor:43686661}].

\bibitem[ZCZW19]{ZCZW19}
B. Zeng, X. Chen, D.-L. Zhou, and X.-G. Wen, 
{\it \color{darkblue} Quantum Information Meets Quantum Matter -- From Quantum Entanglement to Topological Phases of Many-Body Systems}, Quantum Science and Technology (QST), Springer (2019), 
[\href{https://doi.org/10.1007/978-1-4939-9084-9}{\tt doi:10.1007/978-1-4939-9084-9}],
[\href{https://arxiv.org/abs/1508.02595}{\tt arXiv:1508.02595}].

\bibitem[ZLSS15]{ZLSS15}
J. Zaanen, Y. Liu, Y.-W. Sun, and K. Schalm, {\it \color{darkblue} Holographic Duality in Condensed Matter Physics}, Cambridge University Press (2015), 
[\href{https://doi.org/10.1017/CBO9781139942492}{\tt doi;10.1017/CBO9781139942492}].


\end{thebibliography}
\end{document}